\documentclass[a4paper,11pt]{book}
\usepackage[utf8]{inputenc}
\usepackage{xspace}			   
\usepackage[english,frenchb]{babel}
\usepackage[T1]{fontenc}
\usepackage{lmodern}
\usepackage[margin=28mm,includeheadfoot,bindingoffset=10mm]{geometry}
\usepackage{graphicx,color} 
\usepackage[svgnames]{xcolor} 
\addto\captionsfrenchb{}
\addto\captionsfrench{}
\addto\captionsfrenchb{}
\usepackage{pdfpages}
\usepackage{amsmath, mathtools,amsthm}
\usepackage{amssymb,amsfonts}  
\usepackage{bm,bbm}            
\usepackage{upgreek}           
\usepackage[e]{esvect}         
\usepackage{esint}		       
\usepackage{etoolbox}          
\usepackage{calc}              
\usepackage{icomma}		       
\usepackage{versions}           
\usepackage{tikz}
\usetikzlibrary{arrows}
\usetikzlibrary{decorations.pathmorphing}

\tikzstyle{every picture}=[level distance = 8mm, baseline=-0.5ex] 

\tikzstyle{prop}=[shape=circle,minimum size=6mm, draw=black!80, fill=green!30]
\tikzstyle{propScal}=[shape=rectangle,minimum size=6mm, draw=black!80, fill=green!30]

\tikzset{
    photon/.style={decorate, decoration={snake,amplitude=1pt,segment length=5pt}},
}

\usepackage{slantsc}
\makeatletter
\patchcmd{\chaptermark}{\MakeUppercase}{\scshape\slshape}{}{}%
\patchcmd{\sectionmark}{\MakeUppercase}{\scshape\slshape}{}{}%
\makeatother
\numberwithin{equation}{section} 
\numberwithin{figure}{chapter}
\numberwithin{table}{chapter}
\usepackage{titlesec}                        
\titleformat{\chapter}[display]{\Huge\sffamily\bfseries}%
	{\chaptername~\thechapter}{1ex}{}
\titleformat{\section}[hang]{\Large\sffamily\bfseries}%
	{\rlap{\thesection}}{2em}{}
\titleformat{\subsection}[hang]{\large\sffamily\bfseries}%
	{\rlap{\thesubsection}}{3em}{}
\usepackage[linktocpage=true]{hyperref}

\usepackage{slashed}
\usepackage{mathrsfs}
\usepackage{appendix}

\usepackage{lipsum} 
\usepackage{blindtext}
\blindmathtrue

\renewcommand{\d}{\text{d}}

\newcommand{\prop}{\begin{tikzpicture} 
\draw [thick] (-1.2,0)--(-0.8,0);
\end{tikzpicture}}

\newcommand{\vertx}{\begin{tikzpicture}
\draw [thick] (-1.2,0)--(-1,0)--(-1,0)--(-0.85,0.15);
\draw [thick] (-1,0)--(-0.85,-0.15);
\end{tikzpicture}}

\newcommand{\vertxyuk}{\begin{tikzpicture}
\draw [thick] (-1.2,0)--(-1,0)--(-1,0)--(-0.80,0.20);
\draw [dashed] (-1,0)--(-0.80,-0.20);
\end{tikzpicture}}

\newcommand{\res}{\text{res}}
\newcommand{\E}{\text{Erfi}}
\newcommand{\simall}[2]{\underset{#1\rightarrow#2}{\sim}}
\newcommand{\derD}[1]{\overrightarrow{#1}}
\newcommand{\derG}[1]{\overleftarrow{#1}}
\newcommand{\I}{\mathcal{I}}

\catcode`,\active
\newcommand*\hyper{\begingroup
        \catcode`\,\active
        \def ,{\mskip\pFqskip\relax}%
        \dohyper
}
\catcode`\,12
\def\dohyper#1#2#3#4{%
        {}_{2}F_{1}\biggl(\genfrac..{0pt}{}{#1,#2}{#3}\Big|#4\biggr)%
        \endgroup
}

\catcode`,\active
\newcommand*\appell{\begingroup
        \catcode`\,\active
        \def ,{\mskip\pFqskip\relax}%
        \doappell
}
\catcode`\,12
\def\doappell#1#2#3#4#5#6{%
        F_{4}\biggl(\genfrac..{0pt}{}{#1,#2}{#3,#4}\Big|#5,#6\biggr)%
        \endgroup
}

\newtheorem{propo}{Proposition}[section]

\newtheorem{defi}{Definition}[section]
\newtheorem{lemm}{Lemma}[section]
\newtheorem{thm}{Theorem}[section]
\newtheorem{coro}{Corollary}[section]

\providecommand\phantomsection{}

\title{Analytical and Geometric approches of non-perturbative Quantum Field Theories}

\author{Pierre J. Clavier${}^{1,2}$\\
\normalsize \it ${}^1$ Sorbonne Universités, UPMC Univ Paris 06, UMR 7589, LPTHE, 75005, Paris, France\\
\normalsize \it $^2$ CNRS, UMR 7589, LPTHE, 75005, Paris, France }

\date{Defended the 3${}^{\text{rd}}$ of July 2015}

\begin{document}
%
%
%

\maketitle


\selectlanguage{english} 


%
%

\vspace{2.5cm}

\begin{flushright}
 {\it A la mémoire de André Clavier.}
\end{flushright}

\vspace{3.5cm}

\noindent 
{\it
Si l'eau pouvait éteindre un brasier amoureux, \\
Ton amour qui me brûle est si fort douloureux, \\
Que j'eusse éteint son feu de la mer de mes larmes. \\

Pierre de Marbeuf. Et la mer et l'amour ont l'amer pour partage.}

\thispagestyle{empty}

%
%
%

\frontmatter
\tableofcontents

\chapter*{Remerciements \markboth{{\scshape Remerciements}}{{\scshape Remerciements}}} 

\addcontentsline{toc}{chapter}{Remerciements}

En premier lieu, bien sur, je voudrais remercier celui qui a fait des trois ans de ce doctorat une expérience incroyablement enrichissante, mon directeur de recherche, Marc Bellon. Merci Marc d'avoir 
été aussi présent, à l'écoute et toujours près à m'aider quand j'en avais besoin, mais également d'avoir accepté que je travaille seul sur certains projets. J'ai le sentiment d'avoir beaucoup appris 
sous votre \'egide et espère que nous continuerons à collaborer~: il me reste encore tellement \`a apprendre !

Il me faut \'egalement exprimer ma reconnaissance \`a l'\'egard des rapporteurs, les professeurs Dominique Manchon et Dirk Kreimer, pour le temps qu'ils ont pass\'e \`a lire ce manuscrit, et pour leurs 
commentaires qui ont largement contribu\'e \`a am\'eliorer la qualit\'e globale de ce texte.

Enfin, merci \`a tous les autres membres du jury d'avoir accept\'e d'en faire partie. Merci donc \`a Bernard Julia, Lo\"ic Foissy, Fidel Schaposnik et Christian Brouder pour \^etre pr\'esent le jour de 
la soutenance, j'esp\`ere ne pas leur donner de raison de regretter le d\'eplacement.

Merci également à ceux qui ont pass\'e du temps \`a relire cette th\`ese, qui ont permis de lui donner une finition bien plus lisse que la première version. Merci donc \`a Nguyen Viet Dang, Marc 
Bellon,Astrid Cadet, Annick Clavier et Sophie Dorfman pour leurs efforts. Il va sans dire que je suis l'unique responsable de toutes les erreurs, impr\'ecisions, b\^etises, fautes et autres coquecigrues 
qui ont surv\'ecu \`a toutes ces lectures.

La recherche n'est pas une activité solitaire. Je voudrais donc tr\`es chaleureusement remercier tout ceux avec qui j'ai eu la chance de discuter et d'\'echanger et qui ont tous permis à ma réflexion 
de se construire. En premier lieu bien sur, merci à Nguyen Viet Dang, Christian Brouder et Fr\'ed\'eric H\'elein de m'avoir convaincu de travailler sur le formalisme de Batalin—Vilkovisky, le cinqui\`eme 
chapitre de ce manuscrit n'aurait pas exist\'e sans eux. Et merci \'egalement d'avoir accepter de prendre le temps d'expliquer tant de choses à un pauvre physicien int\'eress\'e par des choses un peu 
formelles. Merci également, p\^ele-m\^ele, à Camille Laurent—Gengoux, Serguei Barannikov, Sylvain Lavau, Sofian Teber, Olivier Bouillot, Lucien Heurtier,  Guillaume Bossard, Robin Zegers, Vincent 
Rivasseau... 

Je souhaite \'egalement remercier les membres du LPTHE, et en premier lieu ses directeurs successifs Olivier Babelon et Benoit Douçot pour avoir cr\'e\'e un tel climat favorable \`a la recherche. Merci 
\'egalement aux secr\'etaires Isabelle Nicolai, Françoise Got et Annie Richard pour leur patience vis-à-vis du th\'esard ignorant des choses administratives que j'ai \'et\'e. Et, bien sur, merci à tous 
les chercheurs et \`a toutes les chercheuses du labo pour toute les discussions autour de la machine à caf\'e.

Je suis également tr\`es reconnaissant aux anciens doctorants du labo de m'avoir accueilli comme ils l'ont fait, et d'avoir pris sous leurs ailes le jeune th\'esard (si si, th\'esard!) effray\'e que 
j'\'etais. Merci donc à Emmanuele, Thomas (le grand) et Thibault. Un merci tout particulier \`a Gautier, sans qui tant d'heure auraient stupidement \'et\'e perdues à travailler. Merci \'egalement \`a 
Bruno pour \^etre le grand fr\`ere spirituel de tous les th\'esards du labo.

Une pens\'ee aussi pour les compagnons de gal\`ere, ceux qui ont commenc\'es leurs doctorats en même temps que moi. Courage Sacha, Trésa, Harold et Tianhan, c'est bient\^ot \`a votre tour de passer \`a 
la casserole ! Merci aussi aux p'tits jeunes pour avoir rendu plus agr\'eable l'ann\'ee où j'\'etais le vieux du groupe. Merci aux bizuts Luc et Fr\'ed\'eric, merci à Matthieu et Oscar pour être de tels sujets d'amusement à travers leurs 
th\`emes de recherches, merci à Thomas (le petit) pour avoir repris les s\'eminaires des doctorants du LPTHE, et un grand merci à Hugo pour avoir \'et\'e la hotline mathematica/LaTeX pendant quasiment 
un an.

Je me permets maintenant de remercier tous les copains du M2. Les collocs d'abord~: merci Pierre, Ced, Laetitia et Nico pour ces soir\'ees \`a parler de films, de maths, de musique, de physique, de 
politique, de maths encore... Merci aussi \`a Hanna, Corentin (et N\'emo), Thomas, Lulu <3 (pardon, Docteur Lulu <3), Yung-Feng pour toutes ces soir\'ees, d\'elires, et souvenir partag\'es ! Ce n'est 
pas fini !

Plus loin dans le temps, merci aux copains du magist\`ere, qui ont \'et\'e les premiers \`a me montrer que je n'\'etais pas fou (ou au moins pas seul) \`a vouloir faire de la recherche en physique.  
En premier lieu bien sur, merci à Etienne pour \^etre plus fou que moi, et \`a Julien (et Isa, f\'elicitations au fait!) pour Londres et le reste. Merci aussi Matthieu pour les TP de L3 et bien sur un 
grand merci à Maud, Benjamin, Laetitia, Yannick, Miguel, Samuel, Erwan, Thibault...

Toute ma reconnaissance va également \`a ma famille, qui a su vivre avec la honte d'avoir un physicien théoricien en son sein. Un merci tout particulier \`a mes parents, Annick et Jean-Philippe, pour 
m'avoir toujours encourag\'e dans mes choix, pour avoir toujours cru en moi, sans doute bien plus que moi. Je fais aussi un gros bisou à mes petites sœurs No\'emie et Suzanne, les remercie d'\^etre 
mes fans num\'ero un et un bis, et aussi d'avoir fait des \'etudes litt\'eraires. Merci \'egalement \`a mes grands-parents~: Sophie (d\'esol\'e de ne pas sauver le monde), Andr\'e et Yolande (merci pour 
avoir fait de l'Anfenière ce havre de paix toutes ces ann\'ees) et Jacques. Merci aux oncles  et tantes, et aux cousins-cousines pour faire de cette famille un vrai clich\'e de famille nombreuse. Vous 
\^etes super ! Enfin, puisque les anciennes dettes se doivent d'\^etre pay\'ees, merci à No\'emie et Suzanne pour avoir (parfois) accept\'e de mettre la table \`a ma place pour que je puisse réviser mes 
maths quand j'\'etais en pr\'epa.

Enfin, et surtout, merci \`a Astrid. Pour tout.


\chapter*{Résumé long\markboth{{\scshape Résumé long}}{{\scshape Résumé long}}} 

\addcontentsline{toc}{chapter}{Résumé long}

 \noindent\hrulefill \\
{\it
It is a tale \\
Told by an idiot, full of sound and fury, \\
Signifying nothing. \\

William Shakespeare. MacBeth (MacBeth, Act 5, Scene 5).}

 \noindent\hrulefill

\vspace{1.5cm}

\paragraph{Introduction} ~~\\

Ce texte débute par une réflexion sur la différence entre physiciens et mathématiciens. J'y expose
mon point de vue~: les mathématiciens cherchent avant tout à trouver de la beauté dans leurs constructions, et la façon de construire 
peut compter plus que le résultat lui-même. A l'inverse, le physicien se doit de comparer ses résultats avec des données qui lui sont 
extérieures. Ainsi, il me semble que le physicien est plus intéressé par son résultat que par la beauté formelle de sa démonstration.

C'est pourquoi je me considère comme un physicien mathématicien, au sens de quelqu'un qui fait de la physique mathématique.
Je cherche avant tout à obtenir un résultat, mais n'ayant pas la prétention de décrire un système existant, je peux me permettre de rechercher 
une rigueur mathématique un peu plus qu'il n'est usuel dans la communauté des physiciens. \\

Ceci étant dit, l'introduction contient, ainsi qu'il est d'usage, une brève approche historique de mon sujet, en l'occurence la théorie quantique des 
champs. Je ne remonte pas jusqu'à Démocrite, puisque je commence par rappeler que la naissance de cette théorie est habituellement datée de 1927, avec 
les travaux de Dirac.

Une formulation covariante des travaux de Dirac fut présentée par Pauli et Jordan, toujours en 1927. Cette méthode fut ensuite étendue à tous les champs 
en 1929 par Heisenberg et Pauli. Dirac, peu satisfait par l'approche de Heisenberg et Pauli, présenta une autre construction replaçant les probabilités et les 
observables au centre de sa construction en 1932. Rosenfeld montra la même année que ces deux approches étaient équivalentes.

Les années trente furent à bien des points de vue, un âge d'or de la théorie quantique des champs. Sur le plan expérimental 
le neutron et le positron furent mis en évidence. Il fut montré que la production de paires de particules/antiparticules pouvait être expliquée 
sans faire référence à la théorie de la mer de Dirac. Une théorie de la désintégration $\beta$ fut présentée par Pauli, et une des interactions 
nucléaires par Yukawa. Enfin et surtout, une preuve de ce qui est aujourd'hui appelé le théorème spin-statistique fut donnée par Pauli en 1940. 
Dès cette époque, il fut remarqué que des observables de la théorie avaient une f\^acheuse tendance à diverger quand la théorie des 
pertubations était un peu poussée. La régularisation des résultats était parfois obtenue, au prix de manipulations profondément insatisfaisantes.

Puis vint la guerre et les changements drastiques qu'elle amena. Les Etats-Unis d'Amérique prirent une place considérable sur la scène de la physique 
mondiale (entre autre...), uniquement concurrencés par l'URSS et dans une moindre mesure, par le Royaume-Uni et le Japon.

Après la guerre, Lamb présenta son fameux résultat sur le décalage du spectre hyperfin de l'atome d'hydrogène (the Lamb shift). Il devint alors clair 
que l'explication théorique de ce résultat nécessitait des calculs de précision en électrodynamique quantique. Un premier résultat, non-relativiste, 
fut obtenu par Bethe en 1947. Une version relativiste de ce calcul fut ensuite présentée par Feynman et indépendamment par Schwinger et Weisskopf en 
1948. Ces calculs impliquaient une méthode plus générale pour supprimer les divergences, qu'aujourd'hui nous appellerions cut-off dur.

De l'autre côté du Pacifique, Tomonoga et son équipe obtinrent des résultats approchants. En octobre 1947, Tomonoga présentait un programme 
auto-cohérent de soustraction des divergences, ce que l'on peut aujourd'hui voir comme la promotion de la renormalisation au rang de programme et non 
plus uniquement comme un ensemble disparate de techniques. Le programme de Tomonoga fut étendu en 1948 et 1949 par Schwinger, qui obtint ainsi une 
formulation de l'électrodynamique quantique au premier ordre de la théorie des perturbations sans divergence.

Toutefois, cette approche était connue pour être particulièrement difficile à manipuler~: il était généralement admis que personne d'autre que Schwinger 
n'aurait pu obtenir ses résultats, et que même lui ne pourrait pas aller très haut dans les ordres de la théorie perturbative. Ce fut donc l'approche 
plus intuitive de Feynman, avec en particulier ses diagrammes de Feynman qui devint prépondérante, et l'est encore de nos jours. Les approches de 
Feynamn et de Tomonoga--Schwinger furent montrées équivalentes en 1949 par Dyson.

Avec les travaux de Dyson, la question de la renormalisation à tous les ordres put commencer à être traitée. Les premiers résultats furent obtenus par 
Dyson lui-même en 1949. Toutefois, la question des divergences chevauchantes (overlapping divergences) ne put être traitée qu'en 1951 par Salam, et fut 
considérée comme entièrement résolue après les travaux de Bogoliubov et Parasiuk en 1957, qui furent rigoureusement prouvés en 1966 par Hepp. Notons que 
Zimmermann donna sa forme définitive au formalisme BPHZ en 1969 en simplifiant la preuve de Hepp. Ce formalisme fut réinterprété en 1998 par 
Kreimer comme la marque d'une structure d'algèbre de Hopf des diagrammes de Feynman. Puisque ce sujet est le thème du premier chapitre ce cette thèse, 
je n'irai pas plus avant dans cette approche historique qui demeure, j'en suis conscient, forcément réductrice.

Toutefois, une approche historique de la théorie quantique des champs se doit de citer les travaux de Yang et Mills qui, en 1954, firent la suggestion 
de remplacer le groupe de jauge $U(1)$ de l'électrodynamique quantique par des groupes plus généraux (typiquement non-abeliens). Ces théories, dites de 
jauge, furent montrées renormalisables en 1972 par 't Hooft et Veltman. Enfin, le mécanisme permettant de donner une masse aux bosons de jauge fut 
construit par un grand nombre de personnes~: Nambu, Goldstone, Higgs, Brout, Englert... \\

La suite de l'introduction est une courte discussion sur la renormalisation qui encore aujourd'hui possède un statut ambigu~: est-ce une caractéristique 
essentielle de la Nature et toute (éventuelle) théorie du Tout devra-t-elle être renormalisée, ou est-ce l'indice qu'un changement de paradigme est nécessaire~?
Quelle que soit la réponse à cette question, la renormalisation est aujourd'hui comprise de façon beaucoup plus naturelle qu'à ses débuts. En effet, elle 
peut être vue comme une conséquence de propriétés analytiques que l'on impose aux fonctions de Green. 

Plus physiquement, les divergences rencontrées en théorie quantique des champs peuvent être vues comme une conséquence du fait que, justement, les objets
fondamentaux de la théorie sont des champs. Car alors un diagramme ne peut pas représenter un processus physique~: une infinité (non dénombrable) 
d'interactions se produit, qui n'est pas prise en compte par le calcul de ce seul diagramme. Cette infinité d'interactions va donner un résultat fini, aussi 
est-il réaliste de considérer que ne pas les prendre en compte va donner un résultat divergent. \\

Ensuite, je présente brièvement les équations de Schwinger--Dyson qui sont l'un des deux piliers de cette thèse. Elle 
constituent en fait les équations d'Euler--Lagrange pour les fonctions de Green de la théorie. Le point intéressant à leur sujet est qu'elles sont des 
équations dont les inconnues sont les fonctions de Green entièrement renormalisées~: elles sont donc une des (quelques) façons d'accéder à des 
informations non pertubatives (i.e. allant au-delà de la théorie des perturbations) de la théorie quantique des champs.

Philosophiquement, les équations de Schwinger--Dyson peuvent être vues comme des équations d'auto-similarité. Elle proviennent de la greffe de diagrammes 
de Feynman dans d'autres diagrammes. Etudier si la solution (en terme de diagrammes, et non de fonctions) de système d'équation de Schwinger--Dyson génère une sous-algèbre de 
Hopf de l'algèbre de Hopf de la renormalisation a permis à Foissy, entre 2008 et 2011, de classer les systèmes d'équations de Schwinger--Dyson. \\

L'autre pilier de cette thèse est le formalisme développé par Batalin et Vilkovisky (formalisme BV). Il est en premier lieu une reformulation du 
formalisme BRST, mais sa puissance réside dans sa généralité. En effet, il traite à égalité les théories ayant des symétries ouvertes et celles
ayant des symétries fermées.


\paragraph{Chapitre 1: l'algèbre de Hopf de la renormalisation} ~~\\

Ce chapitre est une présentation de l'algèbre de Hopf des diagrammes de Feynman telle qu'elle me fut présentée par Dominique Manchon lors d'une conférence 
au CIRM en septembre 2012, avec plusieurs aspects proches des articles de Connes et Kreimer. Il commence par une présentation assez formelle de
structures algébriques intervenant dans la construction d'une algèbre de Hopf~:  algèbre, cogèbre, bigèbre et enfin algèbre de Hopf.

Une algèbre est un espace vectoriel contenant une autre opération interne (la multiplication) qui est associative et distributive sur l'opération interne 
de l'espace vectoriel (l'addition). On représente souvent les axiomes d'une algèbre sous la forme graphique suivante, qui représente on associativité.
\begin{figure}[h!]
 \begin{center}
  \begin{tikzpicture}[->,>=stealth',shorten >=1pt,auto,node distance=3cm,thick] 

    \node (1) {$H\otimes H\otimes H$};
    \node (2) [right of=1] {$H\otimes H$};
    \node (3) [below of=2] {$H$};
    \node (4) [below of=1] {$H\otimes H$};

    \path
      (1) edge node [left] {$m\otimes I$} (4)
      (1) edge node [above] {$I\otimes m$} (2)
      (2) edge node [right] {$m$} (3)
      (4) edge node [above] {$m$} (3);

  \end{tikzpicture}
  \label{asso}
 \end{center}
\end{figure}

Une cogèbre est une structure très similaire à une algèbre, pour laquelle les flèches de ce graphique sont simplement inversées. La 
multiplication est remplacée par une comultiplication. Pour comprendre ce qu'est une comultiplication, il suffit de voir la multiplication comme une 
opération qui prend deux éléments de l'algèbre et leur en associe un troisième. A l'inverse donc, la comultiplication va prendre \emph{un} élément de 
la cogèbre et va lui en associer \emph{deux}. Enfin, notons que de même que une algèbre peut posséder un élément neutre, et donc une application unité, 
une cogèbre possède un élément coneutre et une application counité.

Une bigèbre n'est alors rien de plus qu'une structure réunissant en son sein les structures d'algèbre et de cogèbre, plus des axiomes de 
compatibilité entre ces structures. Une algèbre de Hopf est alors un bigèbre avec de plus une application de l'algèbre dans elle-même, appelée l'antipode.
Pour la beauté de cette structure, je présente ci-dessous le diagramme de l'axiome que l'antipode doit satisfaire.
\begin{figure}[h]
 \begin{center}
  \begin{tikzpicture}[->,>=stealth',shorten >=1pt,auto,node distance=3cm,thick] 

    \node (1) {$H$};
    \node (2) [right of=1] {$\mathbb{K}$};
    \node (3) [right of=2] {$H$};
    \node (4) [above left of=2] {$H\otimes H\otimes H$};
    \node (5) [above right of=2] {$H\otimes H$};
    \node (6) [below left of=2] {$H\otimes H$};
    \node (7) [below right of=2] {$H\otimes H$};

    \path
      (4) edge node [above] {$S\otimes I$} (5)
      (1) edge node [left] {$\Delta$} (4)
          edge node [above] {$\varepsilon$} (2)
          edge node [left] {$\Delta$} (6)
      (5) edge node [right] {$m$} (3)
      (2) edge node [above] {$u$} (3)
      (6) edge node [above] {$I\otimes S$} (7)
      (7) edge node [right] {$m$} (3);

  \end{tikzpicture}
  \label{hopf}
 \end{center}
\end{figure}
Il sera dit un peu plus tard pourquoi l'antipode est si importante dans le cadre de la renormalisation, et plus généralement dans le cadre de la théorie 
des algèbres de Hopf.

La sous-partie suivante présente comment une bigèbre graduée connexe (où connexe signifie que la partie de poids zéro de la bigèbre 
est de dimension un) peut toujours être dotée d'un structure d'algèbre de Hopf. Cette sous-partie assez technique se construit sur une idée relativement 
simple~: la preuve consiste à construire une application qui va automatiquement satisfaire les axiomes de l'antipode. Cette construction est 
récursive~: on commence par définir son action sur les éléments de degré zéro, puis on construit son action sur les éléments de degré $n$ à partir de 
son action sur tous les éléments de degré $m<n$.

La dernière sous-partie s'intéresse aux caractères des algèbres de Hopf et donne une première réponse quant à 
l'importance de l'antipode. Les caractères sont des fonctions de l'algèbre de Hopf évaluées dans une certaine algèbre $A$ qui respectent la structure 
de produit de l'algèbre de Hopf.

Ils possèdent une structure de mono\"ide (i.e. un groupe associatif sans inverse) dont l'opération est un produit de convolution défini grâce au 
coproduit de l'algèbre de Hopf et dont l'élément neutre est l'élément neutre de l'algèbre (et de la cogèbre) dans l'algèbre de Hopf. Du fait de 
la coassociativité de ce coproduit, le produit de convolution est associatif. Si on prend alors pour algèbre $A$ l'algèbre de Hopf elle-même, on 
observe que l'antipode, de par sa définition, permet de construire un inverse pour le produit de convolution, permettant ainsi de donner une structure 
de groupe à l'ensemble des caractères d'une algèbre de Hopf évalués sur l'algèbre elle-même. \\

La seconde partie de ce premier chapitre décrit enfin l'algèbre de Connes--Kreimer de la renormalisation. On commence avec des définitions de la théorie 
des graphes, en prenant bien soin d'autoriser les graphes à avoir des pattes externes, puisque ce sont ces pattes externes qui vont permettre d'interpréter 
ces graphes comme représentant des processus physiques (dans une certaine approximation~: une des causes profondes de la renormalisation est justement 
qu'ils ne représentent pas de tels processus). Les autres définitions habituelles données dans cette sous-partie sont celles de sous-graphes, d'arbres et 
de forêts, de sous-graphes couvrants.

On se donne aussi la peine de définir les graphes décorés, ce qui va permettre de construire des théories quantiques des champs non triviales (i.e. avec 
plusieurs types de champs). Ceci amène naturellement à la définition des noeuds (ou sommets) autorisés d'une théorie.

La seconde sous-partie est celle où l'algèbre de Hopf de la renormalisation est introduite. Nous commençons par remarquer que, de toutes les 
graduations évidentes de l'ensemble des graphes décorés d'une certaine théorie, celle qui va nous permettre de construire cette structure est le nombre 
de boucle, que l'on définit soigneusement en différenciant le cas des graphes avec pattes externes des graphes sans (ces derniers étant les graphes du vide 
dans le jargon des physiciens).

Puis nous passons à la description de la structure proprement dite. Le produit est l'union disjointe de graphes. Le coproduit est construit 
en sommant sur tous les sous-graphes tels que la contraction du graphe initial avec ce sous-graphe donne un graphe autorisé par la théorie. En d'autres 
termes~: on somme sur tous les sous-graphes dont le résidu est un noeud autorisé de la théorie. On obtient alors une bigèbre, mais non connexe~: tout les 
graphes de degré zéro (i.e. les graphes arbres) ne sont pas proportionnels. Donc, suivant la discussion de la première partie, on ne peut pas construire 
d'antipode sur ces objets. Nous résolvons ce problème (qui en est effectivement un~: l'antipode permettra la renormalisation) en prenant le quotient de 
notre bigèbre des graphes par l'idéal généré par $(\Gamma-1)$, o\`u $\Gamma$ est un graphe de degr\'e z\'ero et $1$ l'unit\'e.

La sous-partie suivante voit l'introduction d'un nouveau sujet d'importance~: la décomposition de Birkhoff. Elle dit que si les caractères évaluent nos 
graphes sur une algèbre $A$ pouvant s'écrire comme une somme directe, alors l'opérateur de projection sur une de ces sous-algèbres est un opérateur de 
Rota--Baxter. Ceci peut sembler anecdotique, mais le théorème crucial est qu'alors tout caractère évalué sur $A$ peut s'écrire comme un produit de convolution 
de deux caractères, chacun d'eux étant évalué sur une sous-algèbre différente de $A$. Le point clef est bien sûr que les parties de cette décomposition 
sont encore des caractères.

Après cette longue présentation de sujets mathématiques, on passe à une courte discussion physique des implications de ces résultats pour la théorie 
quantique des champs, et en particulier de la renormalisation. Si on se place dans le cadre de la régularisation dimensionnelle, avec 
$\varepsilon$ le paramètre de la régularisation, on prend $A=\mathbb{C}[[\varepsilon,1/\varepsilon]$, où les coefficients sont en fait 
des fonctions des quantités physiques des particules vues comme entrant dans le graphe à évaluer~: les masses, les impulsions, les charges... On peut 
alors écrire $A=\mathbb{C}[[\varepsilon]]\oplus\frac{1}{\varepsilon}\mathbb{C}[1/\varepsilon]:=A_+\oplus A_-$. Il est clair que dans la 
décomposition de Birkhoff sur les caractères évalués sur $A$, la partie évaluée sur $A_+$ aura un sens dans la limite $\varepsilon\rightarrow0$ alors que 
la partie évaluée sur $A_-$ contiendra toutes les singularités du caractère initial. Nous interprètons la première comme la valeur renormalisée  du 
diagramme de Feynman et la seconde comme un contre-terme additif. \\

A la fin de cette seconde partie de notre premier chapitre, nous avons déjà appris beaucoup sur l'algèbre de Hopf de la renormalisation et, plus important 
encore, nous avons compris pourquoi cet objet est essentiel~: parce qu'il encode la combinatoire (connue pour être difficile, au vu de l'extrême 
technicité des preuves de Hepp et Zimmermann de la procédure BPHZ) de la renormalisation. Il reste à construire l'équation du groupe de renormalisation 
et à l'utiliser sur les fonctions de Green de la théorie. Ceci est l'objet de la troisième partie.

Nous présentons l'approche de Connes et Kreimer du groupe de renormalisation, vu comme une relation sur une famille de caractères dépendant d'un 
paramètre (que l'on interprète comme le logarithme de l'impulsion extérieure du diagramme). On démontre leur relation en supposant l'existence d'une 
quantité intermédiaire, prouvée dans l'article de Connes et Kreimer. En ce cas ce 
qu'on appelle l'équation du groupe de renormalisation est une conséquence presque triviale de cette relation~: on la dérive et l'évalue en zéro.

Nous construisons ensuite les fonctions de corrélation de notre théorie (on travaille avec la théorie scalaire en six dimensions d'espace-temps avec une 
intéraction cubique pour simplifier). Elles sont définies comme une série de graphes ayant tous le même résidu, divisés par leur facteur de symétrie. 
On donne alors sans preuve l'action du coproduit de l'algèbre de la renormalisation sur ces fonctions de corrélation.

Des formules suffisamment simples sont obtenues, qui permettent d'exprimer dans la dernière sous-partie de ce chapitre l'action du coproduit sur les fonctions 
de Green de la théorie, où les fonctions de Green sont définies comme les inverses (pour la multiplication point par point, au sens des séries) des 
fonctions de corrélation. On définit pour cela un couplage effectif comme étant le rapport de la fonction à trois points au carré  sur la fonction à deux 
points au cube. Ce couplage vaut $1$ à l'impulsion de référence. On observe alors que l'action du coproduit sur les fonctions de corrélation est graduée 
(dans un sens bien précis) par ce couplage effectif. On peut dans ce cas utiliser une propriété générale des algèbres de Hopf pour déduire que le coproduit 
a une action similaire sur les fonctions de Green.

Ceci permet d'appliquer l'équation du groupe de renormalisation sur les fonctions de Green de la théorie. Comme nous l'avons déjà mentionné, cette 
équation implique de prendre la dérivée par rapport au logarithme de l'impulsion extérieure d'un caractère puis de l'évaluer en zéro. Cette dérivée, 
quand le calcul est mené au bout, amène les fonctions béta et gamma (la dimension anormale) présentes dans l'équation du groupe de renormalisation telle, 
qu'elle est habituellement écrite.

\paragraph{Chapitre 2: équations de Schwinger--Dyson linéaires} ~~\\

Nous parlons d'équations linéaires pour dire que seule 
une fonction de Green intervient dans l'intégrande de l'équation écrite sous forme intégrale. La seule solution exacte connue d'une équation de Schwinger--Dyson 
est pour une équation de ce type~: le modèle de Yukawa sans masse. Cette solution fut trouvée en 1999 par Broadhurst et Kreimer.

Je commence par présenter leur solution en détail. En effet, il est crucial de comprendre précisément le mécanisme à l'oeuvre pour espérer le généraliser.
La première étape consiste à écrire l'équation de Schwinger--Dyson de ce modèle sous une forme différentielle. Ceci est réalisé en dérivant l'équation 
intégrale, et n'est possible que grâce à la linéarité de l'équation. De plus, plusieurs termes conspirent ensemble pour largement simplifier l'équation 
à résoudre.

Cette équation est ensuite réécrite après un changement de variable sous une forme intégrable (au sens de~: que l'on peut intégrer). Toutefois, le point 
clef consiste à résoudre l'équation, non pour la fonction inconnue, mais pour une autre fonction ayant un paramètre défini par l'inconnue d'origine pour variable. 
Une équation différentielle est trouvée pour cette nouvelle fonction, que l'on peut résoudre. En faisant les transformations inverses, on obtient une 
solution paramétrique de l'équation de Schwinger--Dyson.

La dernière sous-partie de ce second chapitre résume les points clefs de ce résultat de Broadhurst et Kreimer. Premièrement, que l'équation de 
Schwinger--Dyson a pu être écrite sous forme différentielle grâce à son caractère linéaire. Ainsi il semble bien que toute tentative de généralisation 
devra se cantonner à des cas linéaires. Par contre, l'intégration et la résolution paramétrique pourront être résolus dans des cas plus 
généraux. \\

La généralisation la plus naturelle consiste à de donner une masse aux champs que l'on renormalise. Ainsi on obtient deux fonctions inconnues~: une fonction 
de renormalisation de la fonction d'onde, comme auparavant, mais aussi une fonction renormalisant la masse de la particule.

Ce cas est traité dans la seconde partie de ce second chapitre. L'équation de Schwinger--Dyson est alors écrite sous la forme d'un système de deux 
équations intégrales couplées. En effet, on peut séparer les contributions provenant de la masse et celles provenant de la fonction d'onde du fait du caractère 
fermionique de la particule~: un des termes (la fonction d'onde) vient avec un $\slashed p$ et l'autre avec simplement un $p$. 

Appliquant ensuite la méthode de Broadhurst et Kreimer, on écrit ces deux équations sous forme différentielle. Malheureusement elles sont couplées et 
l'intégration semble hors de portée. C'est pourquoi on les étudie dans les limites ultraviolette et infrarouge.

Dans la limite ultraviolette, une des équations ne fait plus intervenir que la fonction d'onde. Il se trouve que l'équation obtenue est celle 
du modèle de Yukawa sans masse étudiée lors de la partie précédente. Cette observation dicte le changement de variable à effectuer et fournit une solution 
à la moitié du problème. Toutefois, la seconde moitié ne vient pas sans résistance. En effet, la seconde équation ne se découple pas, et ne peut pas être 
intégrée. Par conséquent, on est contraint de faire une nouvelle approximation~: en développant l'équation pour la fonction de masse, on observe qu'un 
terme est prédominant dans la limite ultraviolette (puisque c'est la limite d'une grande impulsion).

Si on ne garde que ce terme dominant, on obtient une équation relativement simple, qui peut résolue si on l'écrit comme une équation différentielle 
avec les paramètres qui ont permis de résoudre le cas sans masse. Ainsi, on obtient une solution paramétrique de l'équation de Schwinger--Dyson du modèle de 
Yukawa massif dans la limite ultraviolette.

La limite infrarouge est plus simple à résoudre, mais deux sous-cas doivent être distingués~: celui des infrarouges mous et celui des infrarouges 
profonds. Dans le premier cas, l'impulsion n'est pas suffisamment petite pour permettre de négliger les membres de gauche de l'équation de 
Schwinger--Dyson. Nous obtenons un sytème d'équations dont l'une (la seule qui ne porte que sur une fonction inconnue) peut être montrée comme 
n'étant pas intégrable par la méthode de Broadhurst et Kreimer. Toutefois, une solution est obtenue pour cette équation pour une certaine valeur de 
l'impulsion de référence. Quand on injecte cette solution dans la seconde équation, elle devient très simple et peut être résolue exactement. Nous 
obtenons ainsi une solution explicite de l'équation de Schwinger--Dyson du modèle de Yukawa massif, dans la limite infrarouge molle, pour une certaine 
valeur de l'impulsion externe.

Le cas des infrarouges profonds est de loin le plus simple. En effet, dans ce cas, les équations deviennent homogènes (i.e. sans second membre). De plus,
si nous cherchons une solution non triviale, on peut découpler ces équations. Les deux équations obtenues peuvent être résolues très simplement. 
Nous obtenons ainsi une solution explicite de notre système d'équations, qui résoud la limite infrarouge profonde du modèle de Yukawa massif. \\

La troisième partie de ce second chapitre est probablement la plus novatrice. En effet, on y étudie une version linéaire d'un modèle de Wess--Zumino 
à deux superchamps. L'un de ces superchamps est massif et renormalisé et l'autre sans masse et non renormalisé. Nous obtenons ainsi un système 
d'équations de Schwinger--Dyson linéaires.

La magie de ce modèle vient de ce que la supersymétrie frappe deux fois (comme le facteur). En premier lieu, elle nous permet d'invoquer un théorème de non-renormalisation 
qui affirme qu'il suffit de renormaliser une seule fonction, ici la fonction d'onde. Ainsi, on n'aura pas de sytème d'équations couplées comme dans le cas
de Yukawa massif. Dans un second temps, comme il avait déjà été noté dans un article de 2007 par Bellon, Lozano et Schaposnik, un ansatz efficace pour 
attaquer ce système d'équations de Schwinger--Dyson consiste à supposer que toutes les composantes du supermultiplet reçoivent la même contribution. Alors, 
toute les équations de Schwinger--Dyson deviennent équivalentes les unes aux autres, ce qui est une autre conséquence de la supersymétrie.

Après cette analyse, nous nous retrouvons avec une seule équation à résoudre, que l'on peut attaquer avec la méthode de Broadhurst et Kreimer. En effet, 
puisqu'elle est linéaire, on peut l'écrire sous forme différentielle. Ensuite un changement de variable permet de l'intégrer (notons que lors de cette 
étape, il est clair que le théorème de non-renormalisation nous sauve). Dans ce cas, on peut écrire une équation pour une certaine fonction dont la 
variable est un paramètre construit avec la fonction initiale. Et cette dernière équation peut être résolue, ce qui mène à une solution 
paramétrique de l'équation de Schwinger--Dyson de ce modèle de Wess--Zumino. 

A ma connaisssance, cette solution est la seconde découverte pour cette classe de problème (la première étant celle présentée dans la première partie 
de ce chapitre, pour le modèle de Yukawa), et la première avec un terme de masse. Les résultats de la seconde et troisième partie ont fait l'objet d'un 
article récemment publié par Letters in Mathematical Physics. \\

\paragraph{Chapitre 3: modèle de Wess--Zumino I} ~~\\

Le troisième chapitre de cette thèse est la première étude d'un des sujets les plus centraux de ce doctorat~: l'équation de Schwinger--Dyson du modèle de 
Wess--Zumino sans masse. Bien que non-linéaire (dans le sens défini plus haut) elle peut être vue comme la plus simple qui soit. Ce chapitre et le 
suivant vont être consacrés à son étude.

Il commence par la suite directe du chapitre 1. En effet, le premier chapitre finissait sur l'équation du groupe de renormalisation appliquée aux fonctions
de Green pour les graphes à deux et trois pattes externes, dans le cadre de la théorie scalaire en six dimensions avec une interaction cubique. Pour le modèle 
de Wess--Zumino, la fonction de Green pour les graphes à deux pattes ne va pas changer (au sens où elle obéira à la même équation). Par contre, il est 
connu que la fonction à trois points est en fait la fonction identité. En effet la supersymétrie va faire s'annuler tous les graphes à trois pattes  
ayant plus d'une boucle.

Ceci a bien sûr des conséquences sur l'équation de la fonction à deux points. En effet, on peut réécrire l'équation du groupe de renormalisation sans 
la fonction $\beta$. En comparant les deux versions de cette équation, on observe que l'on a $\beta=3\gamma$, où $\gamma$ est la dimension anormale de la 
théorie. Ce résultat n'est certes pas nouveau mais on appréciera cette démonstration très simple, bien plus que celles utilisant des techniques du 
superespace.

On a donc une équation du groupe de renormalisation pour la fonction à deux points, que l'on nomme $G$. Cette équation étant l'une des plus importantes de 
cette thèse, je l'écris donc explicitement~:
\begin{equation*}
 \partial_LG = \gamma(1+3a\partial_a)G.
\end{equation*}
Si on écrit $G$ sous la forme d'une série de Taylor en $L$ (donc les coefficients de cette série sont des fonctions de $a$, la constante de structure 
fine de la théorie), l'équation du groupe de renormalisation donne une récurrence sur ces fonctions. On voit ainsi qu'il suffit de connaître le premier 
(qui n'est autre que la dimension anormale $\gamma$) pour être, au moins en théorie, entièrement capable de reconstruire la fonction de Green $G$. C'est 
donc sur cette fonction d'une seule variable que nos efforts vont se concentrer.

Nous avons construit l'équation du groupe de renormalisation à partir de premiers principes~: la structure de Hopf de la renormalisation. Pour l'équation 
de Schwinger--Dyson, nous avons besoin d'informations supplémentaires, en particulier les intéractions autorisées par le lagrangien. Comme pour le 
modèle supersymétrique étudié lors du chapitre précédent, la supersymétrie nous permet de n'étudier qu'une seule composante de 
supermultiplet, en l'occurence le champ auxiliaire.

Nous pouvons écrire de manière graphique l'équation de Schwinger--Dyson (dans l'approximation à une boucle) de ce champ comme
\begin{equation*}
\left(
\tikz \node[prop]{} child[grow=east] child[grow=west];
\right)^{-1} = 1 - a \;\;
\begin{tikzpicture}[level distance = 5mm, node distance= 10mm,baseline=(x.base)]
 \node (upnode) [style=prop]{};
 \node (downnode) [below of=upnode,style=prop]{}; 
 \draw (upnode) to[out=180,in=180]   
 	node[name=x,coordinate,midway] {} (downnode);
\draw	(x)	child[grow=west] ;
\draw (upnode) to[out=0,in=0] 
 	node[name=y,coordinate,midway] {} (downnode) ;
\draw	(y) child[grow=east]  ;
\end{tikzpicture}.
\end{equation*}
Le propagateur de cette théorie peut être écrit comme le propagateur libre multiplié par la fonction de Green. Si on développe cette fonction de Green en série 
de Taylor selon $L=\ln p^2$ (on n'écrit pas les $\mu^2$ par souci de lisibilité) on doit alors calculer, après l'échange des séries et de l'intégrale 
à boucle, des intégrales dont l'intégrande contient des logarithmes à la puissance $k$. On simplifie ces intégrales en effectuant une transformée de Mellin, 
ce qui n'est qu'un nom très savant pour une simple observation~:
\begin{equation*}
 \left(\ln\frac{p^2}{\mu^2}\right)^k = \left.\left(\frac{\text{d}}{\text{dx}}\right)^k\left(\frac{p^2}{\mu^2}\right)^x\right|_{x=0}.
\end{equation*}
Alors, toutes les intégrales à calculer sont de la forme~:
\begin{equation*} 
 \text{I}(q^2/\mu^2,x,y) = \int\text{d}^4p\frac{1}{\left(p^2/\mu^2\right)^{1-x}[(q-p)^2/\mu^2]^{1-y}},
\end{equation*}
ce qui est l'objet de la sous-partie suivante. Pour l'instant, nous pouvons nous souvenir qu'il est suffisant de connaître la première dérivée de la fonction 
de Green. Par conséquent, nous prenons une dérivée par rapport à l'impulsion extérieure et évaluons le reste en zéro. L'équation de Schwinger--Dyson prend 
alors une forme remarquablement élégante
\begin{equation*}
 \gamma = a\left(1+\sum_{n=1}^{+\infty}\frac{\gamma_n}{n!}\frac{\text{d}^n}{\text{dx}^n}\right)\left(1+\sum_{m=1}^{+\infty}\frac{\gamma_m}{m!}\frac{\text{d}^m}{\text{dy}^m}\right)H(x,y)\bigg|_{x=y=0}
\end{equation*}
avec $H(x,y):=-\frac{1}{\pi^2}\left.\frac{\partial\text{I}(q^2/\mu^2,x,y)}{\partial L}\right|_{L=0}$ la fonction connue sous le nom de transformée de 
Mellin à une boucle.

La dernière sous-partie est consacrée au calcul de cette transformée de Mellin à une boucle. Les détails de ce calcul ne seront pas, 
bien entendu, donnés ici. Nous dirons juste que ce calcul nécessite de passer à la représentation de Schwinger des diagrammes de Feynman, et donc de 
définir le premier et le second polynôme de Symanzik, qui sont calculés avec les arbres et $2$-forêts recouvrants du graphe que l'on cherche à 
calculer.

On ajoute alors une fonction delta de Dirac dans l'intégrande, en suivant l'esprit du théorème de Chang et Wu. Les trois intégrales peuvent alors se 
calculer en utilisant des transformations de variables simples et en identifiant des représentations intégrales de la fonction $\Gamma$ d'Euler. On 
obtient le résultat
\begin{equation*}
 \text{I}(q^2/\mu^2,x,y) = (q^2)^{x+y}\pi^{2}\frac{\Gamma(1+x)\Gamma(1+y)\Gamma(-x-y)}{\Gamma(1-x)\Gamma(1-y)\Gamma(2+x+y)}.
\end{equation*}
On obtient la forme bien connue de la transformée de Mellin à une boucle qui intervient dans l'équation de Schwinger--Dyson~:
\begin{equation*}
 H(x,y) = \frac{\Gamma(1+x)\Gamma(1+y)\Gamma(1-x-y)}{\Gamma(1-x)\Gamma(1-y)\Gamma(2+x+y)}.
\end{equation*}
Enfin, notons que si on utilise la magnifique relation entre le logarithme de la fonction gamma d'Euler et la fonction zéta de Riemann 
\begin{equation*}
 \ln\Gamma(z+1) = -\gamma z+\sum_{k=2}^{+\infty}\frac{(-1)^k}{k}\zeta(k)z^k
\end{equation*}
on obtient pour $H(x,y)$ une forme qui permettra de prouver certains résultats puissants~:
\begin{equation*}
 H(x,y) = \frac{a}{1+x+y}\exp\Bigl(2\sum_{k=1}^{+\infty}\frac{\zeta(2k+1)}{2k+1}\left((x+y)^{2k+1}-x^{2k+1}-y^{2k+1}\right)\Bigr).
\end{equation*}
C'est sur ce résultat que se conclue la première partie de ce troisième chapitre. \\

La seconde partie décrit des résultats obtenus par Marc Bellon et plusieurs de ses collaborateurs dans une série d'articles publiés entre 2005 et 2012. Ces 
articles expliquent comment remplacer la transformée de Mellin à une boucle par une approximation de cette fonction qui rend les calculs considérablement 
plus simples tout en préservant le comportement asymptotique de la fonction anormale.

On observe que la fonction $H(x,y)$ a des pôles en $x,y=-k$ et en $x+y=k$, pour $k$ un entier naturel non nul. Soit $F_k$ la contribution à $\gamma$ venant 
du pôle en $x=-k$. On a la même contribution du pôle en $y=-k$. De même, soit $L_k$ la contribution venant du pôle en $x+y=k$. On obtient alors à partir 
de l'équation du groupe de renormalisation
\begin{align*}
 & \gamma(1 + 3a\partial_a) F_k = -k F_k + 1 \\
 & (k-2\gamma - 3\gamma a\partial_a)L_k = N_k(\partial_{L_1},\partial_{L_2})G(L_1)G(L_2)|_{L_1=L_2=0}
\end{align*}
avec $N_k$ un numérateur qui n'est rien d'autre que le résidu de $H$ en $x+y=k$.

Nous pouvons grâce à ces objets élucider le comportement asymptotique de la dimension anormale. Pour celà, nous remplaçons comme annoncé la fonction 
$H$ par son développement autour de ses trois premiers pôles. Alors on écrit les équations satisfaites par $F_1$ et $L_1$ et une équation de 
Schwinger--Dyson approchée avec ces quantités.

Puis on développe $\gamma$, $F_1$ et $L_1$ en série (en $a$). On calcule aisément les premiers termes de chacune de ces séries, et on fait la supposition 
que ces trois séries ont une croissance rapide pour ne conserver que les termes dominants dans notre système d'équations afin d'avoir trois récurrences 
approchées couplées. On peut les découpler avec une analyse soigneuse de l'ordre de chacun des termes. \\

La troisième partie de ce chapitre est logiquement le calcul des perturbations de ce comportement asymptotique. On pourrait, dans l'absolu, 
calculer les corrections en $1/n$ du résultat précédent avec la même méthode, mais plusieurs faits rendent une telle analyse impraticable. Tout d'abord, 
les ordres suivants vont recevoir des contributions de tous les pôles de la fonction $H(x,y)$ et plus seulement des trois plus proches de l'origine. De 
plus, le découplage des expansions de chacune des fonctions va devenir extrêmement technique. Ainsi, on va plutôt chercher à exploiter notre connaissance 
du comportement asymptotique de la solution.

Pour cela, on définit par récurrence deux suites $A_{n+1}=-(3n+5)A_n$ et $B_{n+1}=3nB_n$ qui encodent les comportements asymptotiques des fonctions 
$F_1$ et $L_1$. On va alors supposer que toutes les fonctions à calculer ont un comportement asymptotique d'un type encodé par les symboles 
$A:=\sum A_na^n$ et $B:=\sum B_na^n$. Ceci est réalisé en prenant l'ansatz
 \begin{align*}
  & F_1 = f +Ag+Bh \\
  & L_1 = l+Am+Bn \\
  & \gamma = a(c+Ad+Be),
 \end{align*}
et les mêmes types de développement pour les fonctions mesurant les contributions des autres pôles de $H$.

Le point clef est que les séries formelles $A$ et $B$ obéissent (asymptotiquement) à des équations différentielles simples. On peut donc écrire les 
équations venant de l'équation de Schwinger--Dyson et de l'équation du groupe de renormalisation comme des équations différentielles où les seules 
fonctions différentiées sont $f,g,h,l,m,c$...

Ensuite il convient de séparer les termes proportionnels à $A$ de ceux proportionnels à $B$. On peut négliger les termes croisés $A^2$, $AB$, 
etc, puisqu'ils correspondraient à des corrections à de très hauts ordres. Puis nous calculons ordre par ordre les solutions des $3\times3$ 
équations ainsi obtenues.

Avec cette méthode, on peut aisément ajouter les contributions des autres pôles de $H(x,y)$. On obtient des systèmes plus compliqués car de 
cinq équations au lieu de trois puisque tous les $F_{k\geq2}$ obéissent à une même équation, tout comme les $L_{k\geq2}$. Avec la méthode exposée plus 
haut, on obtient $3\times5$ équations des cinq équations originales. Ce système est ensuite résolu ordre par ordre.

La principale difficulté, qui limite l'ordre auquel les calculs sont effectués, est l'ajout de numérateurs de plus en plus compliqués. En 
effet, on doit calculer les résidus de $H(x,y)$ à chacun de ses pôles et les inclure (au moins partiellement) dans les équations.

Avec cette méthode, nous sommes parvenus à calculer le quatrième ordre en $a$ des fonctions $c,d$ et $e$, donc le cinquième de la dimension anormale. 
Plusieurs remarques doivent être faites avant de conclure ce chapitre. Tout d'abord, l'annulation des zétas pairs a constitué  une vérification essentielle 
de la justesse de nos calculs. Pas la seule, car cette vérification ne disait rien sur la validité des coefficients rationnels de la solution. Pour ceux-ci,
nous avons comparé nos résultats avec une étude numérique du même problème effectuée plusieurs années auparavant.

En second lieu, les zétas apparaissant pour la première fois dans le dernier ordre calculé de $\gamma$ sont $\zeta(3)^2$ et $\zeta(5)$. D'autres apparaissent 
dans des résultats intermédiaires mais s'annulent dans le résultat final, ce qui fait que les zétas observés sont de poids moins élevés que ce que nous 
attendions. Ceci indique l'existence d'un mécanisme difficile à démontrer dans ce cadre, et qui justifie l'approche du chapitre suivant, par le biais 
de la transformée de Borel.

\paragraph{Chapitre 4: modèle de Wess--Zumino II} ~~\\

Nous allons commencer par résumer les raisons qui nous poussent à réécrire le chapitre précédent, dans un nouveau formalisme~:
\begin{itemize}
 \item[$\bullet$] Les symboles $A$ et $B$ ont un sens peu clair.
 \item[$\bullet$] On a utilisé un développement de la transformée de Mellin en supposant qu'il était exact. Des calculs préliminaires le suggéraient, 
 mais ce n'est que partiellement prouvé.
 \item[$\bullet$] Les relations auxquelles les symboles $A$ et $B$ obéissent ne sont qu'asymptotiques.
 \item[$\bullet$] Nous avons séparé les termes proportionnels à $A$ de ceux proportionnels à $B$ sans vraie justification.
 \item[$\bullet$] Nous n'avons pas d'argument rigoureux permettant de justifier que nous ne nous intéressions pas aux termes croisés comme $AB$ ou $A^2$.
 \item[$\bullet$] Une analyse analytique du contenu transcendantal de ces développements semble irréaliste dans cet ancien formalisme.
\end{itemize}
Dans l'espoir de voir disparaître certaines de ces difficultés, nous allons donc effectuer une transformée de Borel de notre problème. En effet, les séries formelles 
$A$ et $B$ sont Borel sommables.

Après une courte introduction à la théorie de la resommation de Borel, introduction nous permettant en outre de préciser nos notations au lecteur, nous
établissons quelques résultats utiles pour la suite des opérations. \\

On écrit l'équation du groupe de renormalisation pour la fonction $G$ d'une manière permettant de passer au plan de Borel. Dans ce plan, 
notre étude commence en écrivant l'équation du groupe de renormalisation sous la forme d'une équation de point fixe. Puis nous traitons les produits de 
convolutions dans cette équation comme des perturbations, ce qui nous suggère d'écrire la fonction de Green $G$ comme une superposition d'impulsions 
à une certaine puissance~: $(p^2)^{\alpha}$, avec une intégrale sur le paramètre $\alpha$. De plus, comme nos équations comprennent des intégrales de 
convolution, nous souhaitons que la représentation choisie pour $G$ soit indépendante du chemin choisi. On écrit donc $G$ sous la forme d'une intégrale 
de contour~:
\begin{equation*}
 \hat{G}(\xi,L) = \oint_{\mathcal{C}_{\xi}}\frac{f(\xi,\zeta)}{\zeta}e^{3\zeta L}\d\zeta.
\end{equation*}
Et nous allons écrire l'équation du groupe de renormalisation (et, plus tard, celle de Schwinger--Dyson) pour cette fonction à deux variables $f$.

Pour l'équation de Schwinger--Dyson, le seul point remarquable quand on l'écrit dans le plan de Borel avec la paramétrisation de $G$ en intégrale de 
contour est que l'on voit resurgir la transformée de Mellin à une boucle. Ceci suggère bien que notre paramétrisation est pertinente, puisque l'analyse
présentée au chapitre trois montrait clairement l'importance de cette fonction sur le comportement de la dimension anormale $\gamma$.

Le système que nous souhaitons étudier consiste alors en un système de deux équations intégrodifférentielles que l'on écrit car il s'agit de belles équations~:
\begin{align*}
 & 3(\zeta-\xi)f(\xi,\zeta) = \hat{\gamma}(\xi) + \int_0^{\xi}\hat{\gamma}(\xi-\eta)f(\eta,\zeta)\d\eta + 3\int_0^{\xi}\hat{\gamma}'(\xi-\eta)\eta f(\eta,\zeta)\d\eta, \\
 & \hat{\gamma}(\xi) = 1 + 2\int_0^{\xi}\d\eta\oint_{\mathcal{C}_{\eta}}\d\zeta\frac{f(\eta,\zeta)}{\zeta(1+3\zeta)} + \int_0^{\xi}\d\eta\int_0^{\eta}\d\sigma\oint_{\mathcal{C}_{\eta-\sigma}}\d\zeta\frac{f(\eta-\sigma,\zeta)}{\zeta}\left(\oint_{\mathcal{C}_{\sigma}}\d\zeta'\frac{f(\sigma,\zeta')}{\zeta'}H(3\zeta,3\zeta')\right),
\end{align*}
où $\mathcal{C}_{\xi}$ est un chemin enlaçant au moins l'origine et le point $\xi$ du plan complexe.

Nous commençons par vérifier que les calculs effectués dans le plan physique étaient légitimes. Ceci consiste juste à calculer la transformée de Borel 
des symboles $A$ et $B$, et à vérifier que ce que nous avions fait empiriquement dans le plan physique est équivalent, dans le plan complexe, à 
travailler autour de l'origine et de deux singularités de la transformée de Borel de la fonction $\gamma$. Bien que les calculs soient un peu longs, 
la démonstration est, pour l'essentiel, très simple. La seule chose non triviale à vérifier, et qui nécessite un peu de soin, est que nous trouvons les bons facteurs 
combinatoires après les intégrations par parties multiples. \\

La troisième partie de ce quatrième chapitre s'occupe d'une version tronquée du système écrit plus haut, avec le terme quadratique en $f$ dans l'équation 
de Schwinger--Dyson tronqué. On vérifie aisément que cette troncation ne change pas le comportement asymptotique, plus précisément pour $|\xi|>>1, \Re(\xi)>0$ et 
$\Im(\xi)\neq0$. Ceci vient du fait que le comportement asymptotique de la dimension anormale dans le plan physique était déterminé par le premier 
pôle de la transformée de Mellin $H$, pour lequel nous montrerons bientôt que le terme dominant dans l'équation de Schwinger--Dyson est celui linéaire 
en $f$.

Nous pouvons alors résoudre une spécialisation de l'équation du groupe de renoramalisation en $\zeta=-1/3$. On peut injecter cette solution dans 
l'équation de Schwinger--Dyson tronquée. En effet, en déformant le contour d'intégration à l'infini, on voit grâce au lemme de Jordan que 
l'intégrale linéaire en $f$ est juste la fonction que l'on vient de trouver avec notre résolution de l'équation du groupe de renormalisation, au signe 
près.

Nous obtenons alors une équation cohérente avec $\hat{\gamma}(0) = 1$ et $\hat{\gamma}'(0)=-2$. De plus, nous pouvons utiliser ce résultat pour 
retrouver, après une transformée de Borel inverse, une équation satisfaite asymptotiquement par $\gamma$, démontrée dans un article antérieur.

De plus, quand on écrit l'équation de Schwinger--Dyson tronquée en prenant en compte notre solution de l'équation du groupe de renormalisation, on peut 
remarquer que nous avons presque des valeurs multizétas dans le membre de droite. En fait, on peut facilement exprimer cette équation avec 
ces objets. Au delà de cette observation plutôt simple, il convient de remarquer que l'on effectue des sommes sur toute les valeurs multizétas. De plus, 
il apparait qu'entre deux multizétas de même poids, leur importance relative va dépendre de leur profondeur. Ainsi nous comprenons 
plus clairement pourquoi même la combinatoire de la solution asymptotique de l'équation de Schwinger--Dyson est compliquée~: elle est liée à la filtration 
de l'algèbre des multizétas.

Malgré cette jolie réussite, nous n'avons pas résolu l'équation de Schwinger--Dyson tronquée. Trouver une solution ne semble pas complètement sans espoir 
mais n'a pas encore été tenté. A la place, nous avons commencé à regarder une solution numérique de cette équation. Il est rapidement apparu que ces 
équations sont sujets à des instabilitées numériques dûes aux produits de convolution. C'est pourquoi nous avons eu à utiliser des méthodes un peu plus 
sophistiquées que les méthodes les plus basiques, pour des résultats loin d'être optimaux. Notre conclusion est que le comportement asymptotique 
$|\gamma|$ pourrait être constant. \\

Le coeur de ce chapitre est toutefois sa quatrième et dernière partie, que nous allons à présent résumer. Nous commençons par étudier où les singularités 
de $\hat\gamma$, la transformée de Borel de la fonction anormale, peuvent être localisées. En supposant que cette fonction n'a que des singularités 
algébriques on étudie, grâce à l'équation du groupe de renormalisation, les singularités des deux fonctions à une seule variable $f_{\zeta}:\xi\longrightarrow f(\xi,\zeta)$ et 
$\tilde f_{\xi}:\zeta\longrightarrow f(\xi,\zeta)$.

On injecte alors ces résultats dans l'équation de Schwinger--Dyson et on observe que la liberté que l'on a de changer le contour d'intégration 
$\mathcal{C}_{\xi}$ n'est pas cohérente avec les singularités de ces fonctions. Une singularité de $\hat\gamma$ ne peut par 
conséquent n'être que là où l'on ne peut pas changer le contour d'intégration sans changer le résultat de l'intégration~: là où $H(3\zeta,3\zeta')$ a un 
pôle.

Pour les singularités sur l'axe négatif, cela suffit~: nous savons déjà qu'il y a une singularité en $\xi=-1/3$. Elle va induire une singularité en 
$\xi=-2/3$ et le processus va se répéter. Ainsi, $\hat\gamma$ va avoir des singularités en $\xi=-k/3$ pour tout $k$ entier naturel non nul. Le cas 
des singularités en $3\zeta+3\zeta'=k$ sera plus compliqué puisqu'on a des lignes singulières et plus seulement des singularités isolées. Nous 
commençons donc par étudier les singularités isolées.

Tout d'abord, on observe que l'on doit séparer la singularité en $\xi=-1/3$ des autres, puisque pour celle-ci est la seule pour laquelle le 
terme dominant est le terme linéaire en $f$ dans l'équation de Schwinger--Dyson. On peut toutefois trouver des relations entre les coefficients 
dominant de $\hat\gamma$ et de $f$ autour des singularités. Plus intéressant, nous pouvons également trouver l'ordre des singularités (par ordre on entend
la puissance $\alpha$ telle que $\hat\gamma(\xi) \simall{\xi}{\xi_0} c(\xi-\xi_0)^{\alpha}$). Ce calcul repose de façon cruciale sur une séparation de 
$H(3\zeta,3\zeta')$ en une partie singulière et une partie analytique.

Pour les autres singularités de $\hat\gamma$, une approche simple n'est pas possible à cause de la structure plus compliquée des pôles de 
$H(3\zeta,3\zeta')$ dans la moitié du plan complexe avec partie réelle positive. On utilise plutôt une approche comparable à celle présentée dans le 
troisième chapitre de cette thèse qui, si elle n'est pas moins rigoureuse, est moins élégante. On étudie donc les singularités possibles des fonctions 
$\hat L_k$. Il se trouve qu'elles ne peuvent être que là où $\hat\gamma$ l'est (ce qui n'est pas une surprise) mais aussi que en 
$\xi=k/3$, pour $k$ un entier naturel strictement positif.

Ainsi, nous savons que les singularités de $\hat\gamma$ sont en $\mathbb{Z}/3$. De plus, nous pouvons calculer les ordres des singularités positives de 
cette fonction. On obtient le résultat:
\begin{align*}
 \beta_k & = -\frac{2}{3}(k-1) \quad \text{for }\xi=-k/3, k\geq2 \\ 
 \alpha_k & = \frac{2}{3}(k-1) \quad \text{for }\xi=+k/3, k\geq1 \\
 \beta_1 & = -5/3 \quad \text{for }\xi=-1/3.
\end{align*}
Précisons que $\alpha_1=0$ indique que $\hat\gamma$ a une singularité logarithmique en $\xi=1/3$.

Le résultat principal de la quatrième sous-partie de cette quatrième partie de ce quatrième chapitre (!) est la preuve que le développement de $\hat\gamma$ 
autour de n'importe laquelle de ses singularités peut s'écrire comme une somme de produits rationnels de zétas de Riemann impairs. La preuve est un peu 
technique mais nous allons quand même tenter de la présenter ici.

En premier lieu, on écrit un développement pour $f(\xi,\zeta)$ autour de ses singularités dicté par l'équation du groupe de renormalisation~:
\begin{equation*}
 f(\xi,\zeta)\sim \sum_{\substack{r\geq0 \\ s\geq1}} \frac{\psi_{r,s}(\xi)} {\zeta^r(\zeta-\xi_0)^s}
 \sim \sum_{\substack{r\geq0 \\ s\geq1}}\sum_{n\geq0}\frac{\psi_{r,s}^{(n)}} {\zeta^r(\zeta-\xi_0)^s} (\xi-\xi_0)^{\alpha_k+r+s-1+n}.
\end{equation*}
Alors les termes contribuant à n'importe quel ordre du développement de $\gamma$ seront égaux (à des rationnels près) à des dérivées de la fonction 
$H(3\zeta,3\zeta')$ dont on a montré qu'elle peut s'écrire comme l'exponentielle d'un polynôme sans partie constante dont les coefficients sont des 
zétas de Riemann impairs. Ceci suffit pour prouver le résultat énoncé plus haut.

Notons que ce résultat, simple, était déjà hors de portée avec le formalisme détaillé au chapitre trois. Il semblait évident qu'il devait être vrai, mais 
l'éclatement de la transformée de Mellin, écrite comme une somme sur ses singularités, le rendait extrêmement technique à prouver (quoiqu'il soit possible 
que des arguments astucieux puissent être trouvés pour s'éviter les calculs les plus délicats). Toutefois, ce nouveau cadre nous permet de largement dépasser 
ce résultat.

En effet, nous avons aussi mis une borne inférieure sur le poids des zétas apparaissant à l'ordre $n$ dans le développement de $\hat\gamma$ autour de 
ses singularités en $\xi=\pm1/3$. Ce fut réalisé en développant toutes les fonctions intervenant dans l'équation de Schwinger--Dyson autour de 
zéro et autour de leurs singularités. 

Toutefois, ceci ne pourrait pas suffire. En effet, il se trouve que la bonne façon de faire les calculs (qui évite d'avoir à considérer beaucoup de cas 
différents) consiste à définir un poids modifié pour les zétas impairs~: on écrit $W(\zeta(2n+1))=2n$. Donc $W(\zeta(2n+1)\zeta(2p+1)=2n+2p$. On définit 
ensuite le poids d'une série par 
\begin{equation*}
	W( \sum_{p=0}^\infty a_p \xi^p ) = \sup_p \bigl( W(a_p) - p \bigr).
\end{equation*}
Enfin, on définit de la même manière le poids d'une fonction autour d'une singularité en écrivant son développement autour de cette singularité. Alors, 
par récurrence et avec de la sueur, du sang et des larmes nous avons pu prouver que 
\begin{equation*}
  W_{1/3}(\hat\gamma) = W_{-1/3}(\hat\gamma) = 0.
\end{equation*}
On vérifie aisément que ces résultats donnent une borne supérieure sur le poids des zétas pouvant apparaitre pour un ordre donné, et que ce poids est saturé 
par les résultats trouvé avec l'ancien formalisme. Il est donc hautement improbable que ces bornes ne soient pas optimales.

\paragraph{Chapitre 5: formalisme BV} ~~\\

Ce dernier chapitre, contrairement aux trois précédents, ne traite pas des équations de Schwinger--Dyson. A la place, il nous aventure dans l'étrange et 
effrayant pays des théories de jauge. Plus précisément, nous allons présenter une approche du formalisme BV (pour Batalin--Vilkovisky) qui permet de le voir 
comme une théorie de l'intégration.

Ces travaux ont été réalisés avec de nombreux collaborateurs, et en particulier Nguyen Viet Dang, Christian Brouder et Fr\'ed\'eric H\'elein. Ils commencent 
par une courte session de définition concernant les symétries ouvertes et fermées. Une symétrie ouverte est une symétrie dont l'ensemble générateur ne 
forme pas une algèbre de Lie ailleurs que sur le sous-espace de l'espace des configurations où les équations du mouvement sont satisfaites. La 
motivation habituellement donnée pour le formalisme BV est qu'il permet de traiter indistinctement les symétries ouvertes et fermées, ce qui n'est pas 
tout à fait le cas, par exemple, du formalisme BRST puisque le complexe de cohomologie n'existe plus sur tout l'espace des configurations. Un exemple 
de théorie avec des symétries ouvertes est la supergravité sans champ auxiliaire.

Ensuite, nous donnons quelques définitions de géométrie élémentaire. En fait, nous cherchons à ce moment à fixer des notations plutôt qu'à donner 
des définitions précises. Ainsi, on ne définit aucune notion de topologie, d'homologie ni de cohomologie. La définition des variétés (manifold) est 
aussi supposée connue par le lecteur.

On rappelle quelques notions de base de la théorie des fibrés, et en particulier la notion de section. Puis l'algèbre extérieure des formes et celle des 
polyvecteurs est brièvement rappelée. Leurs structures de fibrés sont alors clairement visibles, ce qui permet de rappeler qu'un champ est une section 
de ces fibrés. Soulignons qu'il s'agit en effet de la notion intuitive de champ~: quelque chose qui, à chaque point d'un espace associe une quantité, par 
exemple un vecteur.

Ensuite, la définition d'une variété symplectique est fournie, et celle de variété conormale d'une hypersurface d'une variété $M$, qui est naturellement une 
sous-variété lagrangienne du fibré cotangent de $M$ doté de sa structure symplectique canonique. Enfin, on définit les crochets de Schouten--Nijenhuis 
sur un fibré cotangent shifté, ce qui nécessite la définition des dérivées à gauche et à droite.

La sous-partie suivante consiste en une courte description du programme que nous allons suivre. On rappelle que le but d'une théorie quantique des champs est 
de construire une fonction évaluation qui va d'un espace opérationnel (dans lequel vivent les observables) dans un corps, typiquement $\mathbb{R}$. 
Lorsque l'espace des observables est de dimension finie (ce qui n'est pas le cas intéressant pour les physiciens), la solution typique est de 
multiplier notre opérateur (qui n'est alors qu'une fonction) par une forme volume et de l'intégrer sur tout l'espace.

Cette approche présente l'avantage d'avoir la machinerie de la cohomologie de de Rham à sa disposition. En particulier elle permet de se rendre compte 
que les observables sont plutôt les éléments du premier groupe de la cohomologie de de Rham. Toutefois elle souffre de deux gros défauts. 
En premier lieu, il n'existe pas forcément de mesure sur un espace opérationnel, qui est typiquement de dimension infinie et surtout indénombrable. C'est 
sur ce point que joue le formalisme BV. En second lieu, par définition, pour une théorie de jauge une certaine algèbre de Lie $\mathfrak{g}$ agit 
sur l'espace des opérateurs en laissant invariantes les observables (par définition de la notion d'observables et de théorie de jauge). Ceci implique 
l'existence de sous-espaces pouvant être non-compacts où l'intégrande de l'intégrale sera invariante, ce qui amènera des divergences. Celà 
nécessitera une notion de fixation de jauge, qui est assez naturelle dans le cadre du formalisme BV.

L'idée de base du formalisme BV est de construire un complexe dual au complexe de de Rham, donc un complexe d'homologie pour les polyvecteurs. Alors, 
tous les résultats usuels de l'intégration seront reformulés dans ce nouveau complexe.

On commence par détailler ce que cette construction devra être pour une théorie de jauge usuelle. L'espace des configurations est $V=M\times \pi\mathfrak{g}$ 
(les fantômes de Faddeev--Popov vivent dans $\pi\mathfrak{g}$ qui n'est que l'algèbre $\mathfrak{g}$ shiftée). On va alors travailler dans le 
cotangent shifté $\Pi T^*V$. Un système de coordonnées dans cet espace est $(x^i,c^\alpha, x_i^*,c_\alpha^*)$. \\

Ceci fermait la (longue) première partie, introductive, de ce dernier chapitre. Entrons à présent dans le coeur de notre sujet et construisons le 
laplacien BV, qui sera l'opérateur de bord de notre complexe d'homologie. \\

D'après la discussion précédente, nous voulons un opérateur qui abaisse le degré des polyvecteurs sur lesquels il agit de $1$. Le candidat naturel 
lié au complexe de de Rham est
\begin{equation*}
(\Delta_{\Omega}\alpha)\lrcorner\Omega=d(\alpha\lrcorner\Omega)
\end{equation*}
avec $\Omega$ une forme volume et $\lrcorner$ l'opération de contraction canonique entre formes et polyvecteurs. Il est trivial (mais essentiel) de 
vérifier que $\Delta_{\Omega}^2=0$, ce qui vient de $d^2=0$.

A présent nous pouvons définir (je me permets d'insister~: il s'agit bien d'une \emph{définition}) l'intégration d'un polyvecteur par ce que nous appellerons 
désormais l'intégrale BV comme
\begin{eqnarray*}
\int_{\Pi N^*\Sigma}^{BV} \alpha 
&=& \int_\Sigma \alpha\lrcorner \Omega.
\end{eqnarray*}
Cette définition peut être justifiée, dans une certaine mesure, en se plaçant dans le cas où $M$ est de dimension finie et où $\Sigma$ est donnée 
par $k$ équations. Nous pouvons alors utiliser la théorie de l'intégration usuelle pour réécrire le membre de droite de cette définition.

Nous pouvons facilement voir que c'est une bonne définition car, juste en jouant avec les définitions, nous avons l'équivalent d'un théorème de 
Stokes pour l'intégrale BV:
\begin{equation*}
 \int^{BV}_{\Pi N^*\partial\Sigma}\alpha=-\int^{BV}_{\Pi N^*\Sigma}\Delta\alpha.
\end{equation*}
Ce théorème a pour corollaire immédiat que l'intégrale d'un polyvecteur ne dépend que de la classe d'homotopie de la surface sur lequel on l'intègre. 
Dans le même esprit, il est aisé de prouver que l'intégrale d'un polyvecteur $\Delta$-exact est nulle.

Pourquoi voir le formalisme BV comme une théorie de l'intégration~? Tout simplement parce que l'on peut calculer des intégrales pour des fonctions plus 
compliquées que des fonctions linéaires~! Or dans une théorie avec des symétries ouvertes, on aura typiquement des interactions à quatre fantômes, ce qui 
n'a pas de sens avec l'approche de Faddeev--Popov ou avec le formalisme BRST. D'un autre coté, dans le formalisme BV, ces interactions conduiront juste 
à certains types de fonctions qu'il sera tout à fait légitime de traiter. \\

Avant de pouvoir faire celà, il convient de discuter de la fixation de jauge dans le formalisme BV, et en particulier dans le formalisme BV vu comme une 
théorie de l'intégration. Pour comprendre la procédure, nous allons commencer par le cas d'une théorie de jauge usuelle. Soit donc une théorie 
avec des champs et des fantômes. Après le doublage des variables propre à la théorie BV, on obtient également des versions étoilées des champs et des 
fantômes. Ecrivons le symbole de l'action BRST dans l'espace cotangent shifté $\Pi T^*V$:
\begin{equation*}
 S=S_0+S_E+S_R=S_0+\underset{S_E}{\underbrace{c^\alpha c^\beta C_{\alpha\beta}^\gamma c^*_\gamma}} + 
\underset{S_R}{\underbrace{\rho_\alpha^i c^\alpha x_i^*}}.
\end{equation*}
Le fait que le champ de vecteur BRST soit de carré nul se traduit au niveau de $S$ par 
\begin{equation*}
 \{S,S\} = 0,
\end{equation*}
où $\{,\}$ sont les crochets de Schouten--Nijenhuis~: on dit que $S$ est une solution de l'équation maîtresse classique (EMC).

A présent, nous pouvons nous intéresser à la fixation de jauge que nous savons devoir effectuer dans l'intégrale BV $\int_{\Pi T^*V} e^{iS}$ que pour 
le moment nous comprenons comme étant $\int^{BV}_{\underline{0}\subset T^*\left(X\times\mathfrak{g}\right)}e^{iS}$, c'est-à-dire que l'on intègre sur 
la section nulle de $\Pi T^*V$.

Mais, si $S_0$ possède une variété critique non compacte, le choix de la section nulle comme surface d'intégration est un très mauvais choix~: l'intégrale 
peut diverger. Mais, d'après la discussion faite un peu plus haut, nous savons que l'on peut changer la surface d'intégration (si l'on reste dans la même 
classe d'homologie) sans changer la valeur de l'intégrale si l'intégrande est $\Delta$-fermée. Ainsi, nous allons tout simplement nous placer sur une surface où l'intégrale sera définie et 
interpréter le résultat de cette intégration (si par hasard nous pouvons la faire) comme la valeur de l'observable.

En conclusion~: le choix de jauge dans le formalisme BV n'est que le choix d'une sous-variété (lagrangienne) sur laquelle effectuer l'intégration. Et 
demander à une observable d'être invariante de jauge n'est que lui demander d'être $\Delta$-exacte.

Avec cette compréhension, il semble important de vérifier que l'intégrale $\int_{\Pi T^*V} e^{iS}$ que nous cherchions initialement à calculer est bien 
invariante de jauge. Nous ne pouvons pas encore le montrer simplement~: il nous manque des formules explicites pour le laplacien BV.

Avant de nous tourner vers cette tâche, nous allons discuter un peu d'un détail technique qui n'a pas encore été mentionné~: les ambiguités de Gribov. 
Tout simplement, la surface d'intégration $\Sigma$ doit intersecter exactement une fois chaque ensemble $[A]$ de champs de jauge que l'on peut amener 
à $A$ par une transformation de jauge. En général, il n'existe pas de telle surface $\Sigma$.

Plus précisément, si $\Sigma$ croise pas $[A]$ alors $A$ ne contribuera pas à l'intégrale de chemin, et si elle le croise deux fois (ou plus), $A$ ne 
sera pas surcompté. Gribov a pu déterminer  qu'il est nécessaire de restreindre le domaine d'intégration à un sous-espace de l'espace des configurations. 
Il n'y a pas deux configurations infinitésimalement proches vivant toutes deux dans cette région (la région de Gribov).

Toutefois, cette restriction n'est pas suffisante. Pour éviter le surcomptage de champs non infinitésimalement proches (mais tels que l'on peut quand même 
passer de l'un à l'autre avec une transformation de jauge) il est nécessaire d'effectuer une autre restriction, et de limiter l'intégration au domaine 
modulaire fondamental. Il fut alors montré que dans ce domaine tous les $[A]$ intersectent $\Sigma$ exactement une fois.

Dans la suite, on va supposer que l'on intègre toujours sur le domaine modulaire fondamental, et non sur tout l'espace des configurations. Toutefois, 
nous n'allons pas pour autant changer nos notations. \\

La dernière partie de cette thèse présente le concept d'équations maîtresses. Nous avons déjà présenté l'équation maîtresse classique, nous allons à présent 
parler de leurs contre-parties quantiques.

En premier lieu, dans le cas où l'espace des configurations est de dimension finie, on se dote de sa mesure de Lebesgue $\Omega_0:=dx^1\wedge\dots\wedge dx^N$ 
comme forme volume (sur la partie bosonique de l'espace des configurations). Un calcul simple mais légèrement technique montre que 
 \begin{equation*} 
  \Delta_{\Omega_0} = \frac{\partial}{\partial x^i}\frac{\partial}{\partial x^*_i}.
 \end{equation*}
 On effectue le même calcul sur la partie fermionique de l'espace des configurations, avec le même résultat. On obtient alors le laplacien BV, qui correspond 
 à la mesure de Lebesgue sur l'espace $\Pi T^*V$:
 \begin{equation*} 
  \Delta = \frac{\partial}{\partial x^i}\frac{\partial}{\partial x^*_i} - \frac{\partial}{\partial c^{\alpha}}\frac{\partial}{\partial c^*_{\alpha}}.
 \end{equation*}
C'est cet objet que l'on appellera ``le'' laplacien BV. On peut montrer que les crochets de Schouten--Nijenhuis mesurent l'obstruction de ce laplacien 
à être une dérivée. De plus, il s'agit du seul opérateur différentiel de degré $-1$ invariant sous translation satisfaisant cette condition, ce qui peut servir de base à une 
approche axiomatique du formalisme BV.

A présent, la question que le lecteur intransigeant se pose sans doute est~: pourquoi la mesure de Lebesgue~? En dimension infinie, elle n'existe pas, il 
est donc important de pouvoir construire les laplaciens BV pour d'autres mesures. De plus la notion d'invariance de jauge ne devra pas dépendre de notre 
choix de la forme volume.

Une forme volume différente de la mesure de Lebesgue est forcément proportionnelle à celle-ci. Donc nous pouvons écrire $\Omega = e^{f}\Omega_0$ avec 
$f$ une fonction. En utilisant $d(e^{f}\alpha\lrcorner\Omega_0)=e^{f}\left(d(\alpha\lrcorner\Omega_0) + df\wedge(\alpha\lrcorner\Omega_0)\right)$ alors 
nous voyons que le laplacien BV $\Delta_{\Omega}$ est relié au laplacien BV construit à partir de $\Omega_0$ par la formule
 \begin{equation*}
  \Delta_{\Omega} = \Delta + \{f,-\}.
 \end{equation*}
Nous pourrons alors passer d'un choix de forme volume à un autre.

En premier lieu, nous pouvons maintenant montrer que $\int^{BV}_{\underline{0}\subset T^*\left(X\times\mathfrak{g}\right)}e^{iS}$ est invariante de jauge pour $S$ l'action de BRST écrite sur $\Pi T^*V$. Ceci est obtenu 
par un calcul direct. De plus nous avons
\begin{equation*}
 \Delta(e{iS/\hbar}) = 0\Leftrightarrow \{S,S\} - i\hbar\Delta S=0.
\end{equation*}
Cette dernière équation est la fameuse équation maîtresse quantique (EMQ). L'étude de ses solutions est l'étude de toutes les théories avec une symétrie 
de jauge donnée que l'on peut bâtir; ce qui constitue encore de nos jours un sujet de recherche actif.

Enfin, on montre que $\int^{BV}_{\underline{0}\subset T^*\left(X\times\mathfrak{g}\right)}\mathcal{F}e^{iS}$ est invariante de jauge si
 \begin{equation*}
  \{S,\mathcal{F}\} -i\hbar\Delta\mathcal{F} = 0.
 \end{equation*}
Comme auparavant, il est légitime de se poser la question~: aurait-on trouvé les mêmes conditions pour l'invariance de jauge si l'on avait choisi une 
autre mesure? La réponse est oui. En effet, changer la mesure fait aussi changer l'objet auquel nous allons demander d'être invariant de jauge~: typiquement,
$e^{iS/\hbar}$ va rentrer dans la mesure. Et ce changement va exactement annuler le changement induit par le changement de mesure.

Enfin, comme il a été dit plus haut, deux observables dont la différence est $\Delta$-fermée auront les mêmes valeurs mesurées. Nous en 
déduisons que les observables de la théorie sont plutôt les éléments du groupe d'homologie de l'opérateur
\begin{equation*}
 \{S,-\} - i\hbar\Delta 
\end{equation*}
ce qui est un fait couramment affirmé dans la littérature sur le formalisme BV.

La dernière sous-section étudie la limite classique $\hbar\rightarrow0$. En premier lieu, bien sûr, dans cette limite, l'équation maîtresse quantique 
devient l'équation maîtresse classique. Plus intéressant, on peut aussi vérifier, à partir de la formule explicite pour les crochets de 
Schouten--Nijenhuis, que les solutions de l'équation maîtresse quantique seront les champs de polyvecteurs qui seront nuls partout, sauf sur le lieu des 
points critiques de $S$. Nous retrouvons donc bien le principe de moindre action que nous connaissons et aimons depuis notre tendre enfance.

\paragraph{Conclusions} ~~\\

La conclusion commence par rappeler les résultats de cette thèse déjà résumés. Elle se poursuit par une discussion sur les suites possibles des thèmes 
de recherches explorées, dont voici les grandes lignes.

Beaucoup de choses peuvent être (et ont été) montrées à partir de l'algèbre de Hopf des diagrammes de Feynman. Toutefois, un résultat dont il ne me semble 
pas qu'il ait été montré dans cette approche est la linéarité en $\varepsilon$ de la fonction $\beta$ de la renormalisation dans certaines théories. Ce 
résultat est présent dans certains travaux de physiciens, mais la démonstration est loin d'être satisfaisante.

Pour les équations de Schwinger--Dyson, plusieurs pistes s'offrent à nous. En premier lieu, nous pouvons chercher à généraliser les résultats présentés à la 
fin du quatrième chapitre aux autres singularités de $\hat\gamma$. Ce  qui nécessite l'utilisation de la théorie de la résurgence et du calcul étranger dans 
toute leur puissance.

Il peut aussi sembler pertinent de travailler avec des modèles plus complexes. Nous sommes en ce moment en train d'étudier la transformée de Borel d'un 
système d'équations de Schwinger--Dyson avec une renormalisation de vertex. Il semble que nous devions utiliser une version modifiée des équations 
de Schwinger--Dyson qui fut présentée pour la première fois par Symanzik en 1961.

Enfin, pouvons espérer arriver à certains résultats sur les équations de Schwinger--Dyson de modèles plus proches de la physique, tels 
que la QED (pour la matière condensée) ou même que la QCD. De telles études ne pourront se faire qu'au prix de sacrifices quant à la rigueur de notre 
approche mais me semblent quand même être la raison d'être de nos travaux.

En ce qui concerne le formalisme BV plusieurs pistes s'offrent également à nous. En premier lieu, il semble important d'effectuer quelques calculs dans une jauge bien 
connue d'une théorie relativement simple, afin de vérifier que l'on retrouve les résultats connus. Nous sommes en train d'étudier la QED dans la jauge 
de Coulomb.

De plus, la théorie des ambiguités de Gribov n'est pas entièrement aboutie. Le formalisme BV peut permettre d'étudier ce problème pour des symétries plus générales 
que la symétrie de jauge usuelle. 

Notre première motivation, quand nous avons commencé à nous intéresser au formalisme BV, était de comprendre les récents travaux de Brunetti, Fredenhagen 
et Rejzner. Nous sommes au début de ce programme mais nous conservons l'espoir que nous pourrons comprendre comment le formalisme BV peut être exploité 
dans le cadre de la théorie algébrique des champs quantiques.

Enfin, il reste le projet plus lointain d'étudier les équations de Schwinger--Dyson dans le cadre du formalisme BV. Cette réunion des approches 
analytiques et géométriques de théories quantiques des champs non-perturbatives est sans conteste le but final de mes travaux.

\mainmatter

\setcounter{page}{1}


%
%
\chapter*{Introduction\markboth{{\scshape Introduction}}{{\scshape Introduction}}} 

\addcontentsline{toc}{chapter}{Introduction}


 \noindent\hrulefill \\
 {\it
Ils sont maigres, meurtris, las, harassés. Qu'importe ! \\
Là-haut chante pour eux un mystère profond. \\
A l'haleine du vent inconnu qui les porte \\
Ils ont ouvert sans peur leurs deux ailes. Ils vont. \\

Jean Richepin. Les Oiseaux de Passages.}

 \noindent\hrulefill



\vspace{1.5cm}

If this thesis was a mathematical one, I would have started by explaining why the subject is important, and to which great open problems of mathematics it 
would be linked to. If it was a thesis about theoretical physics, I would have explained which phenomena could be described by the theories presented in it.
Instead, it is a thesis of mathematical physics, so let me start by explaining what do I call mathematical physics.

A nice way to present the difference between mathematicians and physicists is, as often, through a simile. If mathematicians and physicists were 
hikers, the physicists would be only aiming for the top of the mountain, while mathematicans would enjoy to make digressions to see nice beautiful 
landscapes. Therefore, as a mathematical physicist, I want to reach the top of the mountain, but I like the views too much to go straight on. This, 
obviously, often ends up with me not reaching my goals and having only blurry pictures to show to friends when I am back. It is tough to be a 
mathematical physicist!

So why one shall choose this path? Well, I cannot speak for others but, in my case, it is just that questions asked in a (fair enough) well-defined 
way seem easier to understand and thus to tackle than others. And, on the other hand, the idea that I am investigating for understanding the 
underlying principles of Nature keeps me humble. I like the fact that we can be absolutely right on the logical and formal point of view and still wrong 
on the physical one. Maybe a mathematical physicist needs an inch of masochism... \\


Now, it is about time that we start the traditional historical introduction of our subject. It is usually admitted than the quantum field theory was 
born in 1927 with Dirac's famous article \cite{Di27}, despite Jordan's earlier attempts to quantize the electromagnetic field in 1925 and 1926. This 
formalism was given a covariant formulation the same year by Pauli and Jordan. 

A general method of quantizing a field was given in the two important papers \cite{HePa29a,HePa29b}. These papers check that the quantum description of the fields is invariant under the Lorentz 
transformations and contain some of the first check of the consistency of quantum and relativistic principles (two space-like separated events commute). They also might contain the first 
apparition of infinities in physical quantities, such as the self energy of the electron. Dirac was not very happy with the formulations of Pauli and Heisenberg and wrote another paper \cite{Di32} 
in which the key objects are probabilities and observables. This paper would later be recognized as seeding the philosophy of the S-matrix approach of QFT. Rosenfeld showed in \cite{Ro32} that the 
two approaches were equivalent. 

Then came a golden age for the theory of quantum fields, both from the theoretical and the experimental points of view. Let us summarize these successes, without references to the articles that 
are very easy to find otherwise. On the experimental side, evidence for the existence of the neutron (Chadwick, 1932) and the positron (Anderson, 1932) were given. On the theoretical side, the 
demonstration of the possibility of pair production without reference to the hole theory was given by Pauli and Weisskopf in 1934, Fermi gave in 1933 and 1934 a theory for the $\beta$ decay, 
Yukawa built a theory of nuclear forces  published in 1935, and Pauli presented a proof of what is now called the spin-statistic theorem in 1940.

Hence, in the late 30s, the situation was somehow quite ambiguous. On the one hand, many important results have been obtained by the mean of the quantum theory of fields. On the other hand, the 
computations had to be cleaned of the spurious singularities by methods which, already at that time, were judged to be ad hoc procedures. Then came WWII, and the gigantic changes it made into 
the equilibrium of the research poles of the world.

After the war, the measure of the Lamb shift turned to be a perfect playground for precise QED computations. A first (classical) attempt was made by Bethe in \cite{Be47} in 1947. The method to 
get rid of the infinities stemming from the classical field theory is very close to what would be now called a hard cut-off renormalization.

A relativistic version of Bethe's computation was done in 1948 by Feynman in \cite{Fe48} and independently by Schwinger and Weisskopf in \cite{ScWe48}\footnote{notice than the correct result was 
obtained earlier by French and Weisskopf, but published later due the gap from their result to Feynman's and Schwinger's.}.

Similar results were obtained on the other side of the Pacific ocean by a team led by Tomonaga. In a seminar given in October 1947, Tomonaga presented a way to remove the 
infinities of the results using a self-consistent subtraction scheme. This might be seen as the first appearance of the renormalization as a program.

The work of Tomonaga's team on the renormalization program was extended by Schwinger in \cite{Sc48} and \cite{Sc49}. With these two papers, it was 
clear that all the divergences of QED were removed at the first order in the perturbation theory. The open question was therefore: what would happen in 
the higher orders? This question could not be answered with the Schwinger--Tomonaga formalism, that was notoriously difficult to handle.

Indeed, an alternative approach of quantum electrodynamics was presented by Feynman at the Pocono conference in 1948, which would eventually led to his famous diagrammatic formalism. The points of view of Feynman, 
Schwinger and Tomonaga were shown to be equivalent by Dyson in \cite{Dy49}. Since computations are easier in Feynman's approach, it is 
the dominant one, up to this day.

With this new formalism and Dyson's work on the equivalence between the two approaches, the question of renormalizability  
at every order in the perturbation theory could start to be treated. In \cite{Dy49b}, Dyson gave a proof of the renormalizability of the S matrix of 
QED at all orders in perturbation theory, although the problem of overlapping divergences was not fully solved. Progresses toward a solution were done in 1951 by Abdus Salam 
in \cite{Sa51a} and \cite{Sa51b}. The same year, Dyson found its famous argument about the divergence of the perturbative series.

A few years later, Yang and Mills suggested to study quantum field theory for more complicated gauge groups than the $U(1)$ group of QED. Their paper 
\cite{YaMi54} marks the birth of non-abelian gauge theories. Later on, these theories were shown to be renormalizable by 't Hooft and Veltman in 
\cite{HoVe72}, and the way to give a mass to the gauge bosons was explained through an important number of publications by many people: Nambu, Goldstone,
Higgs, Englert, Brout...

A systematic treatment of overlapping divergences was given by Bogoliubov and Parasiuk in their article \cite{BoPa57} of 1957. A proof of their work was 
presented by Hepp \cite{He66} in 1966. This proof was simplified by Zimmermann \cite{Zi69} in 1969. This BPHZ procedure was understood much later in the work of 
Kreimer \cite{Kr98} as the fingerprint of an underlying Hopf algebra. This structure is the subject of the first chapter of this thesis, so we will 
stop here this obviously incomplete historical presentation of the field. Many much more complete historical presentations can be found in the literature,
but I used mainly \cite{Schweber}. \\

As we saw, renormalization has emerged from quantum field theory (and, more particularly, from QED) as the way to give a predictive power to the theory. An important question is: is 
the renormalization an intrinsic feature of physical theories or a mark of our poor understanding of the true face of Nature? To this day, this question remains open but, in many ways, is 
outside the realm of science, since it is likely that an answer depends on one's philosophical view. 

So, philosophical questions aside, what is renormalization? Is it just this conservative solution to the problem of divergences in a quantum field theory? Well, other points of view exist. For
example the removal of the divergences may be seen as a by-product of some analytic properties of the physical observables. In other words: we ask for the integrals to have some analyticity 
properties, which come from physical considerations. This implies to redefine the physical quantities and, when computed with those new quantities, the integrations give finite quantities. This 
point of view is wonderfully explained in \cite{Weinberg96a}.

More physically, the divergence one encounters when working out an explicit Feynman graph can be seen as the consequence of the fact that we are working in a field theory. Therefore, a diagram 
does not reflect a true process: an uncountably infinite number of interactions occurs, which are not taken into account in the computation of the Feynman diagram. All these  
interactions conspire together to give a finite result to the physical process. Hence, it is not absurd that forgetting them shall lead to divergent results. 

On a more formal side, the theory of renormalization has had a great impact on axiomatic and constructive\footnote{in the terminology of Wightman, quoted by Schweber, the constructive field 
theory is an offspring of the axiomatic one: the former aims to rigorously build non-trivial quantum field theory while the later studies the general theory of quantum fields.} field theories. In 
return, these approaches of quantum field theories have produced important results that change the way we see the difference between renormalizable and unrenormalizable theories. Since these subjects 
are highly technical, and since I am far from being an expert, I will not give any further details. \\

Now, one of the two pillars of this Ph.D. is the study of various Schwinger--Dyson (or Dyson--Schwinger) equations. Equations of this kind were first introduced by Dyson in 1949 in the already mentioned \cite{Dy49b}.
They were generalized by Schwinger \cite{Sc51} by the mean of a variational principle. Thus they can be described as the Euler-Lagrange equations for the Green functions. As a matter of fact, the 
Schwinger--Dyson equations are one of the few doors to the non-perturbative regime of a quantum field theory.

Here and all through this thesis, we say that an object is non-perturbative when it is not from any order of the perturbation series. Hence, such objects can either be series of Feynman graphs, or be 
defined without any reference to the perturbation series nor to the Feynman graphs. The Green functions of our theories will be defined at the end of the first chapter as a series of graphs of a given 
type and will therefore fall into this class of objects.

Schwinger--Dyson equations can also be seen as equations of self-similarity. They can be written as (a system of) equations describing how to graft a graph within another one. Then a solution to 
such a system is a series of graphs, which is by construction stable under the grafting operation. The question is whether this series is a Hopf subalgebra of the Hopf algebra of renormalization? 

These two objects will be properly defined in the first chapter of this thesis, but I believe that their meaning is pretty clear: a series is a Hopf subalgebra if the graphs generated by the action 
of the coproduct are still within the series. In the same chapter, it will be clear why this question is of 
importance. Indeed, the renormalization map of Feynman graphs is recursively defined as a twisted antipode. Hence we will see that a Green function can have a renormalizable value (and thus a physical 
meaning) only if it generates a Hopf subalgebra. This remark allowed a full classification of the systems of Schwinger--Dyson equations, performed by Foissy in the articles \cite{Fo08,Fo10,Fo11}. 

We will not use this combinatorial approach of the Schwinger--Dyson equations, but rather focus on some specific equations, and try to see which kind of information can be extracted about the 
(renormalized) Green functions from the Schwinger--Dyson equations. This is the raison d'\^etre of the ``analytical'' in the title of this thesis. \\

The second pillar of this work is the Batalin--Vilkovisky (BV) formalism. The BV formalism was born in the seminal work of Batalin and Vilkovisky \cite{BaVi77,BaVi81b}. It is based on a reformulation 
of the BRST formalism. This formalism was developed by Tyutin \cite{Ty75} and independently by Becchi, Rouet and Stora \cite{BeRoSt76}. It identifies the observables of the theory with the elements
of a certain cohomology group. A brief introduction to this subject can be found on the first appendix of this thesis. 

The interesting point (from a physical point of view) about the Batalin and Vilkovisky's reformulation of the BRST formalism is that their language is still working for theories that cannot be 
treated without great pain with the BRST formalism. From a more mathematical point of view, it leads to a beautiful theory of integration for polyvector fields: it is the origin of the ``geometric'' of the title. Finally, it lets us hope that we will 
one day be able to deal with gauge theories, and not only with the toy models presented in the following work. \\

Now, it is about time to describe what is going to be found by the reader of this thesis. The first chapter starts with the description of algebraic objects that lead to the definition of Hopf 
algebra. Then the Hopf algebra of renormalization is specifically studied. In particular, the Birkhoff decomposition is presented. Then, the renormalization group equation is introduced within 
this formalism and applied on the Green functions.

The second chapter deals with certain class of Schwinger--Dyson equations called linear, a term that will be defined in the header of that chapter. 
We start by presenting a result by Broadhurst and Kreimer on the Schwinger--Dyson equation of the massless Yukawa model. This result is an exact solution 
to this equation. The method is then applied to the massive Yukawa model and the massive Wess--Zumino-like model.

In our third chapter, we will see how to deal with a non-linear Schwinger--Dyson equation, for a massless Wess--Zumino model. First, the asymptotics of the 
solution is given. Then we present a method to exploit the knowledge of this asymptotics to compute its corrections.

To get rid of the formal series arising in the computations of the third chapter, we map in the fourth chapter the Schwinger--Dyson equation and the renormalization group 
equation of the massless Wess--Zumino model to the Borel plane. There, we are able to characterize the singularities of the solution and to extract some information about the 
number theoretic content of the solution.

Finally, the fifth (and last) chapter is somehow distinct to the others. It is a presentation of the BV formalism of gauge theories, seen 
as a theory of integration of polyvector fields. The gauge-fixing procedure is then detailed in this presentation and the quantum master equation has 
a very simple geometrical meaning.

Three appendices are added to these five chapters. The first is a (maybe not so) brief introduction to the geometrical construction of the BRST formalism. 
The second is a set of remark concerning the link between BRST and BV formalism. And the third is the detail of computation of Feynman graphs which 
 will be needed in the next steps of our Schwinger--Dyson program. \\

Before starting the body of the text, I would like to make a short statement about the level of details of this presentation. I had a dream, while writing this thesis and before, to make it 
self-consistent. I think that I did not succeed in this task: due to a lack of space and time \cite{DouglasAdams} I had to use some results without giving a detailed proof. 
Neither did I give a fully detailed presentations of many mathematical subjects that are used or at least mentioned within the text. Some (typically the theory of alien calculus) because I 
believe that the main points of this work can be understood without knowing anything about it, some because I think that they are now general knowledge: homology, cohomology and so on\dots
I apologize for the inconveniences the reader might encounter due to this incompleteness.

%
%
\chapter{The Hopf algebra of renormalization}

 \noindent\hrulefill \\
 {\it
Cet univers désormais sans maître ne lui paraît ni stérile, ni fertile. Chacun des grains de cette pierre, chaque éclat minéral de cette montagne 
pleine de nuit, à lui seul, forme un monde. La lutte elle-même vers les sommets suffit à remplir un cœur d'homme. Il faut imaginer Sisyphe heureux. 
\\

Albert Camus. Le mythe de Sisyphe.}

 \noindent\hrulefill
 
 \vspace{1.5cm}

The first two sections of this chapter are from notes taken at a lecture given by Dominique Manchon at CIRM in September 2012. The third is closer to \cite{CoKr99,CoKr00}. 
A close (but much more complete) set of notes written by the aforementioned author is \cite{Ma04}. Here I will only present material needed for my subject, without any pretension to completeness.

\section{Algebraic preliminaries}

\subsection{Algebra, coalgebra}

Let $\mathbb{K}$ be a field of characteristic zero and $H$ be a $\mathbb{K}$-vector space. Then $(H,m)$ is an associative algebra if $m$ is a linear map
\begin{align}
 m : \quad& H \otimes H \longrightarrow H \nonumber \\
          & (x;y) \longrightarrow m(x;y)
\end{align}
such that
\begin{equation}
 m(m (x;y);z) = m(x;m(y;z))
\end{equation}
which is equivalent to have figure \ref{asso} to commute.
\begin{figure}[h!]
 \begin{center}
  \begin{tikzpicture}[->,>=stealth',shorten >=1pt,auto,node distance=3cm,thick] 

    \node (1) {$H\otimes H\otimes H$};
    \node (2) [right of=1] {$H\otimes H$};
    \node (3) [below of=2] {$H$};
    \node (4) [below of=1] {$H\otimes H$};

    \path
      (1) edge node [left] {$m\otimes I$} (4)
      (1) edge node [above] {$I\otimes m$} (2)
      (2) edge node [right] {$m$} (3)
      (4) edge node [above] {$m$} (3);

  \end{tikzpicture}
  \caption{Illustration of the associativity of an algebra.}
  \label{asso}
 \end{center}
\end{figure}

An algebra is said to be a unitary algebra if it has a unity, that is an element $1\in H$ such that
\begin{equation}
 \forall x \in H, \quad m(1,x) = m(x;1) = x
\end{equation}
In the following, all our algebras will be assumed to be unitary and associative (if not otherwise specified) and we will just write algebra for unitary associative algebra. From the unity $1$ we 
can define the unit map $u:\mathbb{K}\longrightarrow H$ by
\begin{equation}
 \forall k\in\mathbb{K}, \quad u(k)=k1.
\end{equation}
Then the properties of $u$ can be written graphically as
\begin{figure}[h!]
 \begin{center}
  \begin{tikzpicture}[->,>=stealth',shorten >=1pt,auto,node distance=3cm,thick] 

    \node (1) {$\mathbb{K}\otimes H$};
    \node (2) [right of=1] {$H\otimes H$};
    \node (3) [right of=2] {$H\otimes\mathbb{K}$};
    \node (4) [below of=2] {$H$};

    \path
      (1) edge node [above] {$u\otimes I$} (2)
      (3) edge node [above] {$I\otimes u$} (2)
      (2) edge node [right] {$m$} (4)
      (4) edge [<->] node [right] {$\sim$} (1)
      (4) edge [<->] node [left] {$\sim$} (3);

  \end{tikzpicture}
  \caption{Illustration of the unitarity of an algebra.}
  \label{unit}
 \end{center}
\end{figure} \\
Moreover, a subset $H'$ of $H$ is said to be a subalgebra if $m(H'\otimes H')\subseteq H'$.

Now, $(H,\Delta,\varepsilon)$ is said to be a counitary coassociative coalgebra if $\Delta$ and $\varepsilon$ are two maps
\begin{subequations}
 \begin{align}
 \Delta : \quad & H \longrightarrow H \otimes H \\
 \varepsilon : \quad & H \longrightarrow \mathbb{K}
\end{align}
\end{subequations}
with the axioms of the unitary associative algebra reversed. More precisely we have: 
\begin{itemize}
 \item[$\bullet$] the coassociativity
\begin{equation}
 \forall x \in H, \quad (\Delta\otimes I)(\Delta(x)) = (I\otimes\Delta)(\Delta(x)),
\end{equation}
which, with Sweedler's notation $\Delta(x) = \sum x_1\otimes x_2$ is equivalent to
\begin{equation}
 \sum\Delta(x_1)\otimes x_2 = \sum x_1\otimes\Delta(x_2).
\end{equation} 
Graphically, this is represented by the commutativity of figure \ref{coasso}
\begin{figure}[h]
 \begin{center}
  \begin{tikzpicture}[->,>=stealth',shorten >=1pt,auto,node distance=3cm,thick] 

    \node (1) {$H\otimes H\otimes H$};
    \node (2) [right of=1] {$H\otimes H$};
    \node (3) [below of=2] {$H$};
    \node (4) [below of=1] {$H\otimes H$};

    \path
      (4) edge node [left] {$\Delta\otimes I$} (1)
      (2) edge node [above] {$I\otimes \Delta$} (1)
      (3) edge node [right] {$\Delta$} (2)
      (3) edge node [above] {$\Delta$} (4);

  \end{tikzpicture}
  \caption{Illustration of the coassociativity of a coalgebra.}
  \label{coasso}
 \end{center}
\end{figure}
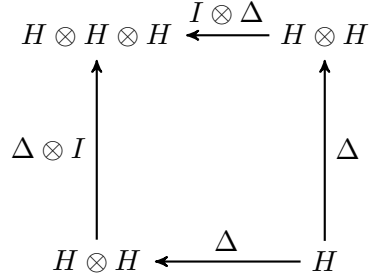
 \item[$\bullet$] the counitarity
\begin{equation}
 \forall x\in H, \quad (I\otimes\varepsilon)(\Delta(x)) \sim x \sim (\varepsilon\otimes I)(\Delta(x)).
\end{equation}
In the same way that the coassociativity property can be pictured by \ref{asso} with the directions of the arrows reversed, we can picture the counitarity of a coalgebra by \ref{unit} with the 
arrows reversed:
\begin{figure}[h]
 \begin{center}
  \begin{tikzpicture}[->,>=stealth',shorten >=1pt,auto,node distance=3cm,thick] 

    \node (1) {$\mathbb{K}\otimes H$};
    \node (2) [right of=1] {$H\otimes H$};
    \node (3) [right of=2] {$H\otimes\mathbb{K}$};
    \node (4) [below of=2] {$H$};

    \path
      (2) edge node [above] {$\varepsilon\otimes I$} (1)
      (2) edge node [above] {$I\otimes \varepsilon$} (3)
      (4) edge node [right] {$\Delta$} (2)
      (1) edge [<->] node [right] {$\sim$} (4)
      (3) edge [<->] node [left] {$\sim$} (4);

  \end{tikzpicture}
  \caption{Illustration of the counitarity of a coalgebra.}
  \label{counit}
 \end{center}
\end{figure}
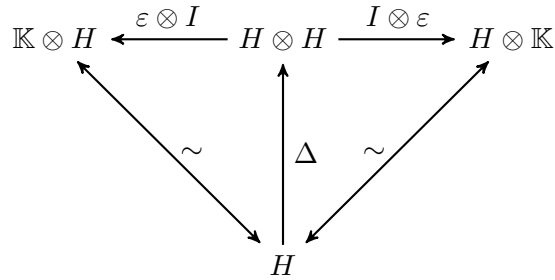
\end{itemize}
We will assume that all the coalgebras that we will encounter are coassociative and counitary (if not otherwise stated) and we will therefore call them simply coalgebra.

Finally, $H'\subseteq H$ is said to be a subcoalgebra of $H$ if $\Delta(H')\subseteq H'\otimes H'$.

\subsection{Bialgebra, Hopf algebra}

After the previous definitions, the gist of a bialgebra becomes clear: it is simply something that put together the axioms of algebra and coalgebra. Hence $(H;m;u;\Delta;\varepsilon)$ is 
a bialgebra if $(H,m,u)$ is an algebra, $(H,\Delta,\varepsilon)$ is a coalgebra and if the figure \ref{bialg} commutes.
\begin{figure}[h]
 \begin{center}
  \begin{tikzpicture}[->,>=stealth',shorten >=1pt,auto,node distance=3cm,thick] 

    \node (1) {$H\otimes H$};
    \node (2) [right of=1] {};
    \node (3) [below of=2] {};
    \node (4) [below of=1] {$H\otimes H\otimes H\otimes H$};
    \node (5) [right of=2] {$H$};
    \node (6) [below of=5] {$H\otimes H$};

    \path
      (1) edge node [left] {$\Delta\otimes\Delta$} (4)
      (1) edge node [above] {$m$} (5)
      (5) edge node [right] {$\Delta$} (6)
      (4) edge node [below] {$(m\otimes m)\circ\tau_{23}$} (6);

  \end{tikzpicture}
  \caption{Compatibility of the algebra and coalgebra structures.}
  \label{bialg}
 \end{center}
\end{figure}
In that figure we have $\tau_{23}:H\otimes H\otimes H\otimes H\longrightarrow H\otimes H\otimes H\otimes H$ defined by
\begin{equation*}
 \tau_{23}(a\otimes b\otimes c\otimes d) = a\otimes c\otimes b\otimes d.
\end{equation*}
Let us notice that this figure which, as stated in its caption, ensures the compatibility between the algebra and coalgebra structures is equivalent to the statement that $\Delta$ is an algebra 
morphism for $(H\otimes H,\tilde{m})$; or equivalently that $m$ is a coalgebra morphism for $(H\otimes H,\tilde{\Delta})$, with $\tilde{m}$ and $\tilde{\Delta}$ defined by:
\begin{align*}
 \tilde{\Delta} & = \tau_{23}\circ(\Delta\otimes\Delta) \\
 \tilde{m} & = (m \otimes m) \circ \tau_{23}.
\end{align*}
Moreover, a bialgebra is said to be graded if, and only if 
\begin{subequations}
 \begin{align}
  & H = \bigoplus_{m\geq0}H_m \\
  & \Delta H_m \subset \bigoplus_{k+l=m}H_k\otimes H_l \\
  & H_k.H_l\subset H_{k+l}
 \end{align}
\end{subequations}
$\forall k;l;m\in \mathbb{N}$. Then
\begin{align*}
  & u(1) = 1_H\in H_0 \\
  & \varepsilon(H_p) = {0}
\end{align*}
$\forall p\geq 1$. Furthermore, if $\text{dim}(H_0)=1$, $H$ is said to be connected. Finally, the bialgebra $H$ is said to be filtered if there exist a collection of $H^n$ such that
\begin{subequations}
 \begin{align}
 & H^0 \subset H^1 \subset \dots \subset H^n \subset \dots; \\
 & H = \bigcup_{n=0}^{+\infty} H^n; \\
 & H^i.H^j \subset H^{i+j}; \\
 & \Delta(H^n) = \sum_{l+m=n}H^l\otimes H^m.
 \end{align}
\end{subequations}
In the following, our Hopf algebra will possesses a natural grading (the loop number), therefore the filtration will be defined by
\begin{equation}
 H^n = \bigoplus_{i=0}^nH_i.
\end{equation}
Now, we can defined a Hopf algebra as a bialgebra together with an application $S:\quad H\longrightarrow H$ (the antipode) such that the figure \ref{hopf} commutes.
\begin{figure}[h]
 \begin{center}
  \begin{tikzpicture}[->,>=stealth',shorten >=1pt,auto,node distance=3cm,thick] 

    \node (1) {$H$};
    \node (2) [right of=1] {$\mathbb{K}$};
    \node (3) [right of=2] {$H$};
    \node (4) [above left of=2] {$H\otimes H$};
    \node (5) [above right of=2] {$H\otimes H$};
    \node (6) [below left of=2] {$H\otimes H$};
    \node (7) [below right of=2] {$H\otimes H$};

    \path
      (4) edge node [above] {$S\otimes I$} (5)
      (1) edge node [left] {$\Delta$} (4)
          edge node [above] {$\varepsilon$} (2)
          edge node [left] {$\Delta$} (6)
      (5) edge node [right] {$m$} (3)
      (2) edge node [above] {$u$} (3)
      (6) edge node [above] {$I\otimes S$} (7)
      (7) edge node [right] {$m$} (3);

  \end{tikzpicture}
  \caption{The axiom of a Hopf algebra.}
  \label{hopf}
 \end{center}
\end{figure}

One can define the notion of a graded Hopf algebra but this will not be needed in this work, see \cite{Ma04} for such subtleties. Finally, a Hopf subalgebra $H'$ of $H$ is a subalgebra and a 
subcoalgebra of $H$, such that $S(H')\subseteq H'$.

\subsection{From bialgebras to Hopf algebras}

In this subsection we will see that if one has a graded connected bialgebra, one can canonically endow it with a Hopf algebra structure. We will need some intermediate steps to reach this result.
\begin{propo} (\cite{Ma04}). \label{red_coprod}
 If $H$ is a graded and connected bialgebra, $\forall n\geq1,\forall x\in H_n$
 \begin{equation}
  \Delta(x) = x\otimes 1 + 1\otimes x + \tilde\Delta(x)
 \end{equation}
 with
 \begin{equation}
  \tilde\Delta(x)\in\bigoplus_{\substack{p+l=n \\ p,l\geq1}} H_p\otimes H_l.
 \end{equation}
 Moreover, $\tilde\Delta$ is a coassociative coproduct. 
\end{propo}
\begin{proof}
 We follow the proof of \cite{Ma04}. From connectedness of $H$ (i.e. the fact that $H_0$ is one dimensional) we can write, for $x\in H_n$, $n\geq1$,
 \begin{equation*}
  \Delta(x) = a(x\otimes1) + b(1\otimes x) + \tilde\Delta x.
 \end{equation*}
 Using the canonical identification between $\mathbb{K}\otimes H$ and $H$ coming from the vector space structure of $H$ we have
 \begin{align*}
  & x \sim (\varepsilon\otimes I)\Delta(x) \sim ax \\
  & x \sim (I\otimes\varepsilon)\Delta(x) \sim bx
 \end{align*}
 from the counit structure of $H$. Then $a=b=1$. The coassociativity property of $\tilde\Delta$ is shown from direct computation
 \begin{align*}
  (\Delta\otimes I)\Delta(x) & = (\Delta\otimes I)(x\otimes 1 + 1\otimes x + \tilde\Delta(x)) \\
                             & = x\otimes 1\otimes 1 + 1\otimes x\otimes1 + \tilde\Delta(x)\otimes 1 + 1\otimes1\otimes x + (\Delta\otimes I)\sum x'\otimes x'' \\
                             & = x\otimes 1\otimes 1 + 1\otimes x\otimes1 + 1\otimes1\otimes x + \sum(x'\otimes x''\otimes 1 + x'\otimes1\otimes x'' + 1\otimes x'\otimes x'') \\
                   \llcorner & + (\tilde\Delta\otimes I)\tilde\Delta(x)
 \end{align*}
 where we have used $\Delta(1)=1\otimes1$ (which can be shown in the same fashion than the first result of this proposition) and the Sweedler's notation for $\tilde\Delta$: 
 $\tilde\Delta(x)=\sum x'\otimes x''$. Similarly we have
 \begin{align*}
  (I \otimes\Delta)\Delta(x) & = (I \otimes\Delta)(x\otimes 1 + 1\otimes x + \tilde\Delta(x)) \\
                             & = x\otimes 1\otimes 1 + 1\otimes x\otimes1 + 1\otimes1\otimes x + 1\otimes\tilde\Delta(x) + (I \otimes\Delta)\sum x'\otimes x'' \\
                             & = x\otimes 1\otimes 1 + 1\otimes x\otimes1 + 1\otimes1\otimes x + \sum(1\otimes x'\otimes x'' + x'\otimes1\otimes x'' + x'\otimes x''\otimes 1) \\
                   \llcorner & + (I \otimes\tilde\Delta)\tilde\Delta(x).
 \end{align*}
 Hence the coassociativity of $\Delta$ implies the coassociativity of $\tilde\Delta$.
\end{proof}
Now, let $H$ be a connected graded bialgebra and $(A;m_A;u_A)$ a unital algebra. Then the convolution product $*_A$ (which we will write $*$ whenever there is no risk of confusion) is an associative product on $\mathcal{L}(H;A)$ (the space of linear maps from 
$H$ to $A$) defined by:
\begin{equation}
 (f*_Ag)(x) = m_A\circ(f\otimes g)\circ\Delta(x) \qquad \forall x\in H.
\end{equation}
\begin{lemm} \label{id_conv}
 The unity of the product $*_A$ is $u_A\circ\varepsilon$, with $u_A$ the unit of $A$ and $\varepsilon$ the counit of $H$.
\end{lemm}
\begin{proof}
 We use $\varepsilon(x)=\lambda1_A$ if $x=\lambda_H\in H_0$ and $0$ otherwise since $H$ is assumed to be connected and graded. Then the only surviving term in $\Delta(x)$ is $x\otimes1$ and
 \begin{equation*}
  f*(u_A\circ\varepsilon)(x) = f(x)u_A(\varepsilon(1)) = f(x)u_A(1) = f(x)
 \end{equation*}
 where in the last step we have used $u_A(1)=1_A$ since $u_A$ is the unit of $A$.
\end{proof}
Let us notice that the main axiom of the Hopf algebra structure (figure \ref{hopf}) can be written 
\begin{equation}
 S *_H I = I *_H S = u\circ\varepsilon.
\end{equation}
Which can be said as ``the antipode $S$ is the inverse for the convolution product $*_H:\mathcal{L}(H;H)\otimes\mathcal{L}(H;H)\longrightarrow\mathcal{L}(H;H)$ of the identity for the product''. 
Therefore, in the remaining part of this subsection, we will show that if a bialgebra is graded and connected the notion of an inverse (for the above convolution product) of the identity (for 
the product) is a well-defined notion.
\begin{propo} \label{van_coprod}
 If $x\in H_n$, then 
 \begin{equation}
  \tilde\Delta^{(k)}(x) = 0 \quad \forall k\geq n+1.
 \end{equation}
 We have used the notation
 \begin{align*}
  & \tilde\Delta^{(1)} := \tilde\Delta \\
  & \tilde\Delta^{(2)} := (I\otimes\tilde\Delta)\circ\tilde\Delta = (\tilde\Delta\otimes I)\circ\tilde\Delta \\
  & \tilde\Delta^{(k)} := (I^{\otimes k-1}\otimes\tilde\Delta)\circ(I^{\otimes k-2}\otimes\tilde\Delta)\circ\dots\circ\tilde\Delta.
 \end{align*}
\end{propo}
\begin{proof}
 We simply write
 \begin{equation*}
  \tilde\Delta^{(k)}(x) = \sum_{(x)}x^{(1)}\otimes...\otimes x^{(k)}
 \end{equation*}
 but each of the $x^{(i)}$ is at least of degree $1$, thus $\tilde\Delta^{(k)}(x)$ is at least of degree $k>|x|$. But $\tilde\Delta^{(k)}(x)$ has to be of the same degree that $x$ (according to 
 proposition \ref{red_coprod}) hence the only solution is that $\tilde\Delta^{(k)}$ vanishes on $H_n$.
\end{proof}
The $\tilde\Delta^{(k)}$ are of importance because if $\alpha$ is an element of $\mathcal{L}(H;A)$ such that $\alpha(1)=0$ then 
\begin{equation} \label{pow_alpha}
 \alpha^{*k}(x) = \left(m^{\otimes k-1}\right)\left(\alpha^{\otimes k}\right)\circ\tilde\Delta^{(k)}(x).
\end{equation}
First, notice that this object is well-defined because the product is associative. Moreover, since $\Delta(1) = 1\otimes1$, it is easy to show that $\alpha^{*k}(1)=0$ by induction on $k$ and 
therefore only the $\tilde\Delta$ term survives in the $\Delta$ within the convolutions product. Then \eqref{pow_alpha} is a consequence of
\begin{lemm}
 \begin{equation}
  \tilde\Delta^{(k)} = \left(I^{\otimes k-1}\otimes\tilde\Delta\right)\circ\tilde\Delta^{(k-1)} = \left(I\otimes\tilde\Delta^{(k-1)}\right)\circ\tilde\Delta.
 \end{equation}
 The first equality being just the definition of $\tilde\Delta^{(k)}$.
\end{lemm}
\begin{proof}
 We prove the second equality of the lemma by induction on $k$. For $k=2$ it is trivial since $\tilde\Delta^{(1)}=\tilde\Delta$. Then, if the lemma is true for a given $k$ we have
 \begin{align*}
  \tilde\Delta^{(k+1)} & = \left(I^{\otimes k}\otimes\tilde\Delta\right)\circ\tilde\Delta^{(k)} \qquad \text{by definition} \\
                       & = \left(I^{\otimes k}\otimes\tilde\Delta\right)\circ\left(I\otimes\tilde\Delta^{(k-1)}\right)\circ\tilde\Delta \qquad \text{by induction} \\
                       & = \left[I\otimes\left(\left(I^{\otimes k-1}\otimes\tilde\Delta\right)\circ\tilde\Delta^{(k-1)}\right)\right]\circ\tilde\Delta \\
                       & = \left(I\otimes\tilde\Delta^{(k)}\right)\circ\tilde\Delta.
 \end{align*}
\end{proof}
The formula \eqref{pow_alpha} is now trivial to show by induction:
\begin{align*}
 \alpha^{*k+1} & = m\left(\alpha\otimes\alpha^{*k}\right)\circ\tilde\Delta \\
               & = \left(m^{\otimes k}\right)\left(\alpha^{\otimes k-1}\right)\left(I\otimes\tilde\Delta^{(k)}\right)\circ\tilde\Delta \qquad \text{by induction} \\
               & = \left(m^{\otimes k}\right)\left(\alpha^{\otimes k-1}\right)\circ\tilde\Delta^{(k+1)}.
\end{align*}
Now, we are able to prove some essential corollaries of the proposition \ref{red_coprod} and of the formula \eqref{pow_alpha}.
\begin{coro} \label{van_pow_alpha}
 Let $\alpha\in\mathcal{L}(H;A)$ such that $\alpha(1)=0$. Then, $\forall n \le k-1$, $\alpha^{*k}$ vanishes on $H_n$.
\end{coro}
\begin{proof}
 This result is a direct consequence of formula \eqref{pow_alpha} and proposition \ref{van_coprod}.
\end{proof}
\begin{coro} \label{clef_colo}
 Any graded connected bialgebra can be endowed with a Hopf algebra structure.
\end{coro}
\begin{proof}
 Let $H$ be a graded connected bialgebra. As noticed earlier, the antipode can be seen as the inverse for the convolution product on $\mathcal{L}(H;H)$ of the identity $I$ of the product $m$.
 Hence we just have to check that such an object can be defined in $H$. We can write the identity
 \begin{equation}
  I = u\circ \epsilon + \alpha
 \end{equation}
 with $\alpha(1)=0$. Then
 \begin{equation*}
  I^{*-1}(x) = (u\circ \epsilon + \alpha)^{*-1}(x) = \sum_{k=0}^{+\infty}(-1)^k\alpha^{*k}(x).
 \end{equation*}
 According to corollary \ref{van_pow_alpha} the last series in locally finite, i.e. its number of non-vanishing term depends of $x$ but is always finite. Hence an antipode is well-defined.
\end{proof}
Moreover, the antipode is canonically built by defining $S(1)=1$ and by induction
\begin{equation}
 S(x) = - x - \sum_{(x)}S(x')x'' = - x - \sum_{(x)}x'S(x'').
\end{equation}
Therefore, in the second section of this chapter we will just need to define a graded connected bialgebra structure over Feynman graphs, and its antipode will be given by the above formula. The 
antipode will be a key object in the taming of the combinatorics of the renormalization process.

\subsection{Characters group}

The group of characters plays a central role in the formalism of Connes and Kreimer of the renormalization process and in particular for the construction of the renormalization group equation. 
Hence let us quickly recall the definition of this object. $G$, the characters group of a Hopf algebra $H$ is a subset of $\mathcal{L}(H,\mathbb{C})$. Elements of $G$ are unital algebra morphisms 
from $H$ to $\mathbb{C}$, that is, $\phi(1)=1$ and $\forall h_1,h_2\in\mathcal{H}$,
\begin{equation}
 \forall \phi:\mathcal{H}\rightarrow\mathbb{C}; \quad \phi\in G \Leftrightarrow \phi(h_1h_2) = \phi(h_1)\phi(h_2) = m_{\mathbb{C}}\left(\phi(h_1)\otimes \phi(h_2)\right).
\end{equation}
In the following, we will write $m$ for $m_{\mathbb{C}}$ and $m_H$ for the product of the Hopf algebra $H$. Now, the product of the convolution group is the usual convolution product of 
$\mathcal{L}(H,\mathbb{C})$:
\begin{equation*}
 \phi_1*\phi_2 = m(\phi_1\otimes\phi_2)\circ\Delta
\end{equation*}
whose identity element is $u\circ\varepsilon$ (the composition of the unit and counit of $H$) according to lemma \ref{id_conv}.
\begin{propo}
 $(G,*)$ is a group with $\phi^{*-1}=\phi\circ S$, $\forall\phi\in G$.
\end{propo}
\begin{proof}
 We will use here some elements of the book \cite{Majid}. We have already shown that $u\circ\varepsilon$ is the identity element of this convolution product. The associativity comes from the 
 coassociativity of $\Delta$:
 \begin{align*}
  (\phi_1*\phi_2)*\phi_3 & = m\left[m(\phi_1\otimes\phi_2)\circ\Delta\otimes\phi_3\right]\circ\Delta \\
                         & = m(m\otimes I)(\phi_1\otimes\phi_2\otimes\phi_3)(\Delta\otimes I)\circ\Delta \qquad \text{from the associativity of } m \\
                         & = m(I\otimes m)(\phi_1\otimes\phi_2\otimes\phi_3)(I\otimes\Delta)\circ\Delta \qquad \text{from the coassociativity of }\Delta \\
                         & = m\left(\phi_1\otimes(\phi_2*\phi_3)\right)\circ\Delta \\
                         & = \phi_1*(\phi_2*\phi_3).
 \end{align*}
 Let us notice that this is true even if the $\phi$s are not characters. Now, for the inverse we have from a direct computation
 \begin{align*}
  (\phi*(\phi\circ S))(x) & = m\left(\phi\otimes(\phi\circ S)\right)\circ\Delta(x) \\
                          & = m(\phi\otimes\phi)\circ(I\otimes S)\circ\Delta(x) \\
                          & = \phi\left[m(I\otimes S)\circ\Delta(x)\right] \qquad \text{since $\phi$ is a character} \\
                          & = \phi(u\circ\varepsilon(x)) \qquad \text{from the axiom of a Hopf algebra} \\
                          & = u\circ\varepsilon(x).
 \end{align*}
 The last line comes from the fact that $\phi(1_H)=1$, $\phi(0_H)=0$ and that $u\circ\varepsilon(x)=x$ if $x\in H_0$ and zero otherwise. Finally, for the closure we have
 \begin{align*}
  (\phi_1*\phi_2)(h_1h_2) & = m(\phi_1\otimes\phi_2)\circ\Delta(m_H(h_1\otimes h_2)) \\
                          & = m(\phi_1\otimes\phi_2)\circ\left[(m_H\otimes m_H)\circ\tau_{23}\left(\Delta(h_1)\otimes\Delta(h_2)\right)\right]
 \end{align*}
 from the axiom of a bialgebra. Now, writing $\Delta(h_1)\otimes\Delta(h_2)=\sum\sum h_1'\otimes h_1''\otimes h_2'\otimes h_2''$ and using the fact that $\phi_1$ and $\phi_2$ are characters we 
 arrive to
 \begin{equation*}
  (\phi_1*\phi_2)(h_1h_2) = \sum\sum \phi_1(h_1')\phi_1(h_2')\phi_2(h_1'')\phi_2(h_2'').
 \end{equation*}
 Changing the order of the products (since the usual product on $\mathbb{C}$ is commutative) we can separate the sums to have
 \begin{equation*}
  (\phi_1*\phi_2)(h_1h_2) = \sum\phi_1(h_1')\phi_2(h_1'')\sum\phi_1(h_2')\phi_2(h_2'') = (\phi_1*\phi_2)(h_1)(\phi_1*\phi_2)(h_2)
 \end{equation*}
 with the product on $\mathbb{C}$ in the rightmost term. Therefore $(\phi_1*\phi_2)$ is a character.
\end{proof}
This ends the preliminaries on Hopf algebras. Let us move to one specific Hopf algebra: the Hopf algebra of renormalization. 

\section{Connes--Kreimer Hopf algebra}

\subsection{Elements of graph theory} \label{graph_theory}

It is not common, in physics literature, to give the definitions of what graphs (and subgraphs, and so on\dots) are. Indeed, it is believed that the intuitive picture is enough: a graph is a bunch 
of vertices and edges (or lines, in physicists' language) joining them. However, we will have to be slightly more precise here in order to build the Hopf algebra of Feynman graphs. In the original 
approach \cite{Kr98} the Hopf algebra of renormalization is seen as a Hopf algebra of decorated rooted trees. This approach imposes to define Feynman graphs as decorated rooted trees (or to map 
Feynman graph to decorated rooted trees). It simplifies the definition of Feynman graphs and makes clear the combinatorial consequence of the Hopf algebra structure, known as the Forest formula 
\cite{BoPa57,He66,Zi69}.

Here, once again, we will rather follow the approach of \cite{Ma04} which presents the advantage of being closer to the physics language. We freely adapted the following definitions from 
\cite{Godsil} but they are quite standard. Thus a (non-planar, non-oriented) graph $\Gamma$ is a doublet of sets $\left(V(\Gamma);e(\Gamma)\right)$ with $V(\Gamma)$ the vertices and $e(\Gamma)$ 
the edges of $\Gamma$. The edges can be external or internal: $e(\Gamma)=\left(E(\Gamma);I(\Gamma)\right)$. Internal edges are unordered pairs of vertices (that are then said to be neighbors) and 
external edges are composed of one vertex. For $x,y\in V(\Gamma)$ we will write $(xy)$ for an element of $I(\Gamma)$ that joins $x$ and $y$ and $(x)$ for an element of $E(\Gamma)$ attached to $x$.

For $x,y\in V(\Gamma)$, a path from $x$ to $y$ is a sequence of vertices starting with $x$ and ending with $y$ such that consecutive vertices in that sequence are neighbors. A graph is said to be 
connected if there is a path between any two vertices. Moreover, given a connected graph, if it stays connected after the removal of any internal lines, it is said to be 1PI (one particle irreducible). 
If the connected components of a graph $\gamma$ are all 1PI, then $\gamma$ is said to be locally 1PI. There is a way to 
give a precise definition of 1PI and locally 1PI graphs in the language of graph theory but we would need to define the deleting operation of an edge, and this would be quite lengthy for a very natural and common 
notion.

Now, a subgraph $\gamma$ of $\Gamma$ is a graph such that
\begin{itemize}
 \item[$\bullet$] $V(\gamma)\subseteq V(\Gamma)$,
 \item[$\bullet$] $e=(xy)\in E(\gamma)\Rightarrow x,y\in V(\gamma)$,
 \item[$\bullet$] $i=(x)\in I(\gamma)\Rightarrow x\in V(\gamma)$.
\end{itemize}
Those three points are really just a complicated way of saying that a subgraph is exactly what we expect it to be. The shortened notation for ``$\gamma$ is a subgraph of $\Gamma$'' is 
$\gamma\subseteq\Gamma$. If in that case $V(\gamma)\neq\emptyset$ and $\gamma\neq\Gamma$ we say that $\gamma$ is a proper subgraph of $\Gamma$, and we write $\gamma\subset\Gamma$. Now, let 
$P\subset V(\Gamma)$. Then $\Gamma(P)$ is the subgraph of $\Gamma$ such that:
\begin{itemize}
 \item[$\bullet$] $V(\Gamma(P))=P$,
 \item[$\bullet$] $I(\Gamma(P))=\{e=(xy)\in I(\Gamma)|x,y\in P\}$,
 \item[$\bullet$] $E(\Gamma(P))=\{e=(xy)\in I(\Gamma),e'=(x')\in E(\Gamma)| x\in P, y\notin P,x'\in P\}$.
\end{itemize}
Hence, intuitively speaking, $\Gamma(P)$ is the ``maximal'' subgraph of $\Gamma$ having the set $P$ for vertices.

If $V(\Gamma)$ can be written as a disjoint union:
\begin{equation}
 V(\Gamma) = P_1\amalg...\amalg P_n
\end{equation}
such that $\Gamma(P_j)$ is connected $\forall j$ then we have a covering $\gamma$ of $\Gamma$:
\begin{equation}
 \gamma = \Gamma(P_1)\amalg...\amalg\Gamma(P_n).
\end{equation}
Obviously, there is in general, more than one covering of a given graph.

Now, the residue $\res(\Gamma)$ of a connected graph $\Gamma$ is the graph with one vertex and $E\left(\res(\Gamma)\right)=E(\Gamma)$. Let $\gamma\subseteq\Gamma$. The contracted graph $\Gamma/\gamma$ is 
the graph obtained by replacing in $\Gamma$ all the connected components of $\gamma$ by their residues. Then it is clear that $\res(\Gamma)=\Gamma/\Gamma$.

Now, we have to define Feynman graphs, which are graphs representing a given interaction process within a physical process. This is nicely detailed in \cite{Ma04} and we will only give the key 
ideas here. In a given physical theory $\tau$ there might be more than one type of particles, therefore we have to label the edges (by non-zero integers) of our Feynman graph to keep track of which 
particle goes on each edge. The number $N_{\tau}$ of possible values of the labels of the edges depend of the theory $\tau$. 

Now, given a vertex $v$, let $T(v)=(n_1;\dots;n_{N_{\tau}})$ its type, with $n_j$ the number of edges of type $j$ attached to $v$. If $e$ is a self-loop attached to $v$, it has to be counted 
twice. In other words, the sum of the $n_j$s of a given vertex has to be equal to the number of half-edges attached to it.
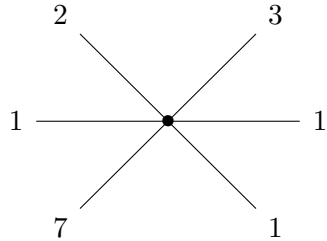
\begin{figure}[h]
 \begin{center}
  \begin{tikzpicture}[shorten >=1pt,node distance=2cm]

  \node (1) at (0,0) {$\bullet$};
  \node (2) [left of=1] {1};
  \node (3) [above left of=1] {2};
  \node (4) [above right of=1] {3};
  \node (5) [right of=1] {1};
  \node (6) [below right of=1] {1};
  \node (7) [below left of=1] {7};

  \draw (0,0)--(2);
  \draw (0,0)--(3);
  \draw (0,0)--(4);
  \draw (0,0)--(5);
  \draw (0,0)--(6);
  \draw (0,0)--(7);

  \end{tikzpicture}
 \caption{A vertex $v$ of valency $6$ with specified edges.}
 \end{center}
\end{figure}

In the above example, if $N_{\tau}=8$, the vertex $v$ is of type $T(v)=(3;1;1;0;0;0;1;0)$. Two half-edges labelled $1$ might belong to the same edge: it is not seen by the type of the vertex.

Finally, a physical theory also specifies the authorized vertices. That is, given a theory $\tau$, we have an number $N_{\tau}$ of edges, a number $v_{\tau}$ and a set
\begin{equation}
 \mathcal{T} = \{\mathcal{T}_1\dots\mathcal{T}_{v_{\tau}}\}
\end{equation}
with the $\mathcal{T}_j$s being sequences of integers. Then the authorized vertices of the theory $\tau$ are the vertices of type within $\mathcal{T}$. Those graphs will be called the Feynman 
graphs of the theory and the vector space spanned by all 1PI Feynman graphs of the theory $\tau$ will be written $V_{\tau}$. Notice that the notions of residue, subgraph, contraction and so on 
have clear generalizations for Feynman graphs.

\subsection{The Hopf algebra structure of Feynman graphs}

There are obvious gradings on Feynman graph. But in physics the perturbative calculations in quantum field theories are made by expanding in the loop number. Therefore we will take the loop number
for grading. Let $\Gamma$ be a connected graph. Then its loop number is
\begin{equation} \label{loop_def}
 l(\Gamma) = |I(\Gamma)| - |V(\Gamma)| + 1.
\end{equation}
Now in physics we also want to allow graphs like \prop~ and their disjoint unions. Such graphs do not fit in our 
presentation of graph theory and, from the formula \eqref{loop_def} their loop number would be $0-0+1=1$. There is some ways to get rid of this difficulty: 
for example we can add vertices of valency two, one for each possible line. However, this darkens the interpretation 
of the vertices to be the possible interactions of the theory, in particular for massless theories. Therefore 
we will rather allow $V(\Gamma)$ to be an empty set and define in this case the loop number to be:
\begin{equation}
 l(\Gamma) = k(\Gamma)-|E(\Gamma)|/2.
\end{equation}
$k(\Gamma)$ is the number of connected components of $\Gamma$ and $|.|$ stands for the cardinality. This last $l(\Gamma)$ is always an integer since if a graph has no vertices it has to 
have an even number of external edges. If $\Gamma=\Gamma_1\amalg\dots\amalg\Gamma_k$ with each $\Gamma_k$ connected its loop number is defined as
\begin{equation}
 l(\Gamma) = \sum_{i=1}^kl(\Gamma_i).
\end{equation}
Now, let $\mathcal{B}_{\tau}$ be the free commutative algebra spanned by $V_{\tau}$. Its product is the disjoint union of graphs which have the empty graph as identity element. In the following, 
the empty graph will therefore be written $1_H$ to be coherent with the notation of the first section (since there will be a Hopf algebra structure soon). Its unit is defined by 
$u(\lambda)=\lambda.1_H$.

The coalgebra structure is not much more complicated (at least when you already know it). The counit is defined by $\varepsilon(1_H)=1$ and $\varepsilon(\Gamma)=0$ if $\Gamma\neq1_H$. The 
crucial step is the coproduct, given for 1PI graphs by
\begin{subequations}
 \begin{align}
  & \Delta(\Gamma) = \sum_{\substack{\gamma\subseteq\Gamma \\ \gamma\text{ locally 1PI} \\ \Gamma/\gamma\in V_{\tau}}}\gamma\otimes\Gamma/\gamma \qquad \text{if }l(\Gamma)\neq0 \\
  & \Delta(\Gamma) = \Gamma\otimes\Gamma \qquad \text{if }l(\Gamma)=0.
 \end{align}
\end{subequations}
Let us notice that the coproduct clearly closes on 1PI graphs. Indeed, if $\Gamma$ then, for any $\gamma\subseteq\Gamma$, $\Gamma/\gamma$ is 1PI. Moreover, 
it is obvious that the coproduct, as the product, respects the grading since $l(\Gamma/\gamma)=l(\Gamma)-l(\gamma)$. Moreover we have the following wonderful proposition
\begin{propo} \cite{Kr97,Ma04} 
 $\Delta$ is coassociative.
\end{propo}
\emph{Elements of proof.}
 The proof is trivial for graphs without loops. If $l(\Gamma)>0$, we directly compute
 \begin{align*}
  (\Delta\otimes I)\Delta(\Gamma) & = \sum_{\substack{\gamma\subseteq\Gamma \\ \Gamma/\gamma\in V_{\tau}}}\sum_{\substack{\delta\subseteq\gamma \\ \gamma/\delta\in V_{\tau}}}\delta\otimes\gamma/\delta\otimes\Gamma/\gamma \\
  (I\otimes \Delta)\Delta(\Gamma) & = \sum_{\substack{\tilde\gamma\subseteq\Gamma \\ \Gamma/\tilde\gamma\in V_{\tau}}}\sum_{\substack{\tilde\delta\subseteq(\Gamma/\tilde\gamma) \\ (\Gamma/\tilde\gamma)/\tilde\delta\in V_{\tau}}}\tilde\gamma\otimes\tilde\delta\otimes(\Gamma/\tilde\gamma)/\tilde\delta.
 \end{align*}
 Proving rigorously that those two sums are equal is quite cumbersome. Instead we will explain why it is true. First for a given graph $\Gamma$, take $\gamma\subseteq\Gamma$ and 
 $\delta\subseteq\gamma$. Then we have an element of the first sum, let us check that an element of the second sum is equal to it. Since $\delta\subseteq\Gamma$ we can take $\tilde\gamma=\delta$. 
 Then $\tilde\delta$ has to be a subgraph of $\Gamma/\delta$. Therefore we can take $\tilde\delta=\gamma/\delta$. Now, we just have to check that 
 $\Gamma/\gamma=(\Gamma/\tilde\gamma)/\tilde\delta=(\Gamma/\delta)/(\gamma/\delta)$. This is true by transitivity of the contraction of graphs.

 Now, let us take $\tilde\gamma\subseteq\Gamma$ and $\tilde\delta\subseteq(\Gamma/\tilde\gamma)$, i.e., an element of the second sum. The only subtlety here is that we have to distinguish 
 between two cases. If $\tilde\delta$ contains the vertex issued from the shrinking of $\tilde\gamma$ in $\Gamma$, we find an element of the first sum equal to 
 $\tilde\gamma\otimes\tilde\delta\otimes(\Gamma/\tilde\gamma)/\tilde\delta$ with the same procedure than we did in the first case. If $\tilde\delta$ does not contains this vertex then it can be 
 seen as an subgraph of $\Gamma$ instead of a subgraph of $\Gamma/\tilde\gamma$. Then we can take $\gamma=\tilde\gamma\amalg\tilde\delta$ and $\delta=\tilde\gamma$ and everything works well.

 As a concluding remark of this incomplete proof, I would advise the reader struggling with this discussion to draw the graphs in a set-theory style to convince himself of the veracity of the 
 above discussion
\begin{flushright}
 $\blacksquare$
\end{flushright}


Now we can extend to $\mathcal{B}_{\tau}$ this coproduct by multiplicativity, i.e., with the property \ref{bialg}. Hence we trivially have
\begin{coro}
 $(\mathcal{B}_{\tau};\amalg;u;\Delta;\varepsilon)$ is a graded bialgebra.
\end{coro}
\begin{proof}
 We have just proven the coassociativity property, and all the other are true by construction.
\end{proof}
However this bialgebra is not connected since the subalgebra of degree $0$ is spanned by all loopless graphs (called tree graphs in the physics literature). Let $\mathcal{I}$ be the ideal 
generated by $1-\gamma$ for $l(\gamma)=0$. We define
\begin{equation}
 H_{\tau} = \mathcal{B}_{\tau}/\mathcal{I}.
\end{equation}
We define the coproduct on $H_{\tau}$ as on $\mathcal{B}_{\tau}$: for 1PI graphs first and then we extend it to $H_{\tau}$ by multiplicativity. Notice that now $\res(\Gamma)\simeq1$ and therefore 
the subgraph $\gamma=\Gamma$ will always appear in the coproduct formula. Hence if we separate the subgraphs $\gamma=\Gamma$ and $\gamma=\emptyset$ from the others we arrive to the celebrated Kreimer's formula:
\begin{subequations}
 \begin{align}
  & \Delta(\Gamma) = 1\otimes\Gamma + \Gamma\otimes1 + \sum_{\substack{\gamma\subset\Gamma \\ \Gamma/\gamma\in V_{\tau}}}\gamma\otimes\Gamma/\gamma \qquad \text{if }l(\Gamma)\neq0 \\
  & \Delta(1) = 1\otimes1 \qquad \text{if }l(\Gamma)=0
 \end{align}
\end{subequations}
since now $1$ is the only loopless graph. Then, from \ref{clef_colo} we know that we can build the antipode inductively:
\begin{equation}
 S(1) = 1, \qquad S(\Gamma) = -\Gamma -  \sum_{\substack{\gamma\subset\Gamma \\ \Gamma/\gamma\in V_{\tau}}}S(\gamma)\Gamma/\gamma.
\end{equation}
One may ask why the Hopf algebra structure is so interesting, and if we could not have stood with the bialgebra structure. It will become clear in a moment that the combinatorics of the 
renormalization procedure are actually encoded into the antipode. In other words: the antipode tells us how to renormalize graphs with overlapping divergences. Therefore, an important part of the 
power of this approach of renormalization lies in the Hopf algebra structure and losing the informations about tree graphs is a small price to pay to tame the dreadful combinatorics of 
renormalization.

\subsection{Birkhoff decomposition}

In physics, we want to compute numbers. Thus we want to have a way to evaluate the Feynman graphs. This is done with the famous Feynman rules, which are maps from $H$ to a target algebra $A$. Before studying 
them in particular, let us present another algebraic concept: the Birkhoff decomposition.

A choice of a renormalization scheme is the choice of a decomposition of the algebra $A$ into two parts:
\begin{equation}
 A = A_-\oplus A_+.
\end{equation}
$A_+$ will be the renormalized part of $A$, and $1_A\in A_+$. In the following, we will work with the dimensional regularization. Let $\varepsilon$ be the complex parameter that measure the difference between the 
space-time dimension and an integer number. It is the same symbol for the counit, but which one is being used shall be clear from the context. Then we will take for $A$ the Laurent series:
\begin{equation}
 A = \mathbb{C}\left.\left[1/\varepsilon,\varepsilon\right]\right].
\end{equation}
Let us notice than this is a commutative algebra. It stays so even is the coefficients are functions of some physical parameters, such as external impulsions, masses, and so on. Then the renormalized observables will be evaluated in the positive part of $A$:
\begin{equation}
 A_- = \frac{1}{\varepsilon}\mathbb{C}\left[1/\varepsilon\right]; \qquad A_+ = \mathbb{C}\left[\left[\varepsilon\right]\right].
\end{equation}
Let $\Pi:A\longrightarrow A_-$ be the projection on $A_-$ parallel to $A_+$.
\begin{lemm} \label{pourquoi} \cite{Kr98}
 $\Pi$ is a Rota--Baxter operator, that is $\forall x,y\in A$:
 \begin{equation}
  \Pi(xy) + \Pi(x)\Pi(y) = \Pi\left(\Pi(x)y+x\Pi(y)\right).
 \end{equation}
\end{lemm}
\begin{proof}
 We decompose $x$ and $y$ into their $A_+$ and $A_-$ parts: $x=x_++x_-$ and $y=y_++y_-$. Then
 \begin{align*}
  \Pi(xy) & = \Pi(x_+y_- + x_-y_+ + x_+y_+ + x_-y_-) \\
          & = \Pi(x_+y_- + x_-y_+) + x_-y_- \\
          & = \Pi(x_+y_- + x_-y_+) + \Pi(x)\Pi(y),
 \end{align*}
 where we have used $x_+y_+\in A_+$ and $x_-y_-\in A_-$, which is obvious from the definitions of $A_+$ and $A_-$. On the other hand, we have
 \begin{equation*}
  \Pi(x)y + \Pi(y)x = x_-(y_-+y_+) + y_-(x_++x_-) = x_-y_+ + y_-x_+ + 2x_-y_-,
 \end{equation*}
 thanks to the commutativity of $A$. Thus, using $\Pi(x_-y_-)=x_-y_-=\Pi(x)\Pi(y)$ we have
 \begin{equation*}
  \Pi\left(\Pi(x)y + \Pi(y)x\right) = \Pi(x_-y_+ + y_-x_+) + 2\Pi(x)\Pi(y) = \Pi(xy) + \Pi(x)\Pi(y).
 \end{equation*}
\end{proof}
Before stating the main theorem of this section, let us recall that $H$, the Hopf algebra of renormalization (we dropped the subscript $\tau$ for readability), is filtered:
\begin{equation*}
 H = \bigcup_{n=0}^{+\infty} H^n
\end{equation*}
with
\begin{equation*}
 H^n :=\bigoplus_{i=0}^nH_i.
\end{equation*}
With $H_i$ the set spanned by 1PI graphs with loop number $i$. Moreover we have identified all the loopless graphs to the empty graph, making $H$ to a connected Hopf algebra.
These two points will be needed for the proof of the following theorem. Let us also precise that $\sum_{(\Gamma)}$ will be a shorthand notation for
\begin{equation*}
 \sum_{\substack{\gamma\subset\Gamma \\ \gamma\neq\emptyset}}.
\end{equation*}  
\begin{thm} \cite{CoKr99,Ma04}
 Let $\phi\in\mathcal{L}(H,A)$ then
 \begin{enumerate}
  \item we can write
  \begin{equation}
   \phi = \phi^{-1}_-*\phi_+
  \end{equation}
  with the inverse being for the convolution product on $\phi\in\mathcal{L}(H,A)$ and
  \begin{align*}
   \phi_- & : H\longrightarrow A_-\oplus k.1_A \\
   \phi_+ & : H\longrightarrow A_+,
  \end{align*}
  $\phi_-(1_H) = \phi_+(1_H) = 1_A$ and, for $\Gamma\in\text{Ker}(\varepsilon)$ (with $\varepsilon$ the counit):
  \begin{subequations}
   \begin{align}
    \phi_-(\Gamma) & = -\Pi\left(\phi(\Gamma) + \sum_{(\Gamma)}\phi_-(\gamma)\phi(\Gamma/\gamma)\right) \\
    \phi_+(\Gamma) & = (I_A - \Pi)\left(\phi(\Gamma) + \sum_{(\Gamma)}\phi_-(\gamma)\phi(\Gamma/\gamma)\right).
   \end{align}
  \end{subequations}
  \item If $\phi$ is a character, then $\phi_-$ and $\phi_+$ are two characters as well.
 \end{enumerate}
\end{thm}
\begin{proof}
 The first point is easy to prove by direct computation. If $\Gamma=1_H=\emptyset$ the assertion is trivially true. For $\Gamma\in\text{Ker}(\varepsilon)$:
 \begin{equation*}
  \phi_+(\Gamma) = \phi(\Gamma) + \sum_{(\Gamma)}\phi_-(\gamma)\phi(\Gamma/\gamma) + \phi_-(\Gamma).
 \end{equation*}
 Thus, still for $\Gamma\in\text{Ker}(\varepsilon)$
 \begin{align*}
  (\phi_-*\phi)(\Gamma) & = m_A(\phi_-\otimes\phi)\Delta(\Gamma) \\
                        & = \phi_-(1)\phi(\Gamma) + \phi_-(\Gamma)\phi(1) + \sum_{(\Gamma)}\phi(\gamma)\phi(\Gamma/\gamma) \qquad \text{by Kreimer's formula} \\
                        & = \phi_+(\Gamma).
 \end{align*}
 We have used $\phi_-(1)=\phi(1)=1_A$ to get to the last line. The second point is more intricate. We prove it by induction on the degree of the elements of $H$. The initialization is 
 trivial. Let us assume that $\forall x,y\in H:|x|+|y|\leq n$ we have
 \begin{equation*}
  \phi_-(xy) = \phi_-(x)\phi_-(y).
 \end{equation*}
 Then, let us take two elements $x,y\in H$ such that $|x|+|y|=n+1$. We have
 \begin{equation*}
  \phi_-(x)\phi_-(y) = \Pi(a)\Pi(b)
 \end{equation*}
 with $a$ and $b$ defined by
 \begin{align*}
  a & := \phi(x) + \sum_{(x)}\phi_-(x')\phi(x'') \\
  b & := \phi(y) + \sum_{(y)}\phi_-(y')\phi(y'').
 \end{align*}
 We use Sweedler's notation to keep the intermediate results within reasonable size. Then, using the fact that $\Pi$ is a Rota--Baxter operator we have
 \begin{align*}
  \phi_-(x)\phi_-(y) & = -\Pi(ab) + \Pi(\Pi(a)b) + \Pi(\Pi(b)a) \\
                     & = -\Pi(ab+\phi_-(x)b + \phi_-(y)a) \qquad\text{since }\Pi(a)=-\phi_-(x); \Pi(b)=-\phi_-(y) \\
                     & = -\Pi\left(\left[\phi(x) + \sum_{(x)}\phi_-(x')\phi(x'')\right]\left[\phi(y) + \sum_{(y)}\phi_-(y')\phi(y'')\right]\right. \\ 
 \llcorner + \phi_-(x)&\left.\left[\phi(y) + \sum_{(y)}\phi_-(y')\phi(y'')\right] + \phi_-(y)\left[\phi(x) + \sum_{(x)}\phi_-(x')\phi(x'')\right]\right) \\
                     & = -\Pi\left(\phi(xy) +\phi_-(x)\phi(y) + \phi_-(y)\phi(x) + \sum_{(x)(y)}\phi_-(x'y')\phi(x''y'') \right. \\
 \llcorner + \sum_{(x)}&\left.\left[\phi_-(x')\phi(x''y) + \phi_-(x'y)\phi(x'')\right] + \sum_{(y)}\left[\phi_-(y')\phi(xy'') + \phi_-(xy')\phi(y'')\right]\right).
 \end{align*}
 We have rightfully used the induction hypothesis to get the last line since $|y'|,|y''|<|y|$ and $|x'|,|x''|<|x|$. On the other hand we have
 \begin{equation*}
  \phi_-(xy) = -\Pi\left(\phi(xy) + \sum_{(xy)}\phi_-((xy)')\phi((xy)'')\right).
 \end{equation*}
 Now we just have to clarifies the set over which this last summation is: the set of proper subgraphs of $x\amalg y$. And it is clear that
 \begin{equation*}
  \{\gamma\subset(x\amalg y)\} = \{x'\}\cup\{y'\}\cup\{(x'\amalg y\}\cup\{x\amalg y'\}\cup\{x'\amalg y'\}\cup\{x\}\cup\{y\}
 \end{equation*}
 with $x'$ and $y'$ proper subgraphs of $x$ and $y$ respectively. Moreover, since $(xy)'':=(x\amalg y)/(xy)'$ we have
 \begin{align*}
  \sum_{(xy)}\phi_-((xy)')\phi((xy)'') & = \sum_{(x)}\phi_-(x')\phi(x''y) + \sum_{(y)}\phi_-(y')\phi(xy'') + \sum_{(x)}\phi_-(x'y)\phi(x'') + \sum_{(y)}\phi_-(xy')\phi(y'') \\
                             \llcorner & + \sum_{(x)(y)}\phi_-(x'y')\phi(x''y'') + \phi_-(x)\phi(y) + \phi_-(y)\phi(x).
 \end{align*}
 Hence, $\phi_-$ is a character, since $A$ is commutative. Finally, $\phi_+=\phi_-*\phi$ is a character since the set of characters is a group for the convolution product.
\end{proof}
Before applying this theorem to the Feynman rules, let us precise that we did not give the most general form of the Birkhoff decomposition. We stuck to the form needed for our purpose. 
Moreover, the Birkhoff decomposition preserves the cocycle property, as shown in \cite{Ma04}: we did not either write all the possible results concerning the Birkhoff decomposition.

Finally, we see as advertised in the introduction of this thesis that the renormalized value of a Feynman graph is obtained by subtracting the divergences coming from its subgraphs. Hence, 
for the Green functions (which are series of graphs), the fact that they generate a Hopf subalgebra means that they are renormalizable: if there is one graph $\Gamma$ such that one of the 
graphs generated (say, $\gamma$) by the action of the coproduct on $\Gamma$ is not within this series, then the renormalized value of $\Gamma$ would need to evaluate $\gamma$. But the 
evaluation of $\gamma$ could not be deduced from the initial condition as, inductively, the evaluation of the other graphs 

\subsection{Regularized and renormalized Feynman rules}


As we already said, Feynman rules are morphisms from $H$ to a certain target algebra $A$. They are built from the lagrangian of the theory and the details of this construction can be found in any 
book (but \cite{Peskin} is a good point to start) of QFT and will not be discussed here.

For the example, let us assume that we work in a scalar theory, with cubic interaction, in $6$ space-time dimensions. Then the evaluation of the one-loop diagram $\Gamma$ with two external lines is 
naively the integral
\begin{equation}
 \int\d^6p\frac{1}{p^2(q-p)^2}.
\end{equation}
If we perform the angular integration and study the convergence of the above integral for $|p|\longrightarrow+\infty$ we are left with an integrand proportional to $p$ which is not integrable 
at infinity. This is where the renormalization process comes into the game. It relies on two pillars: first, we regularize the divergent integral, to make it dependent on a unphysical parameter. 
The original (divergent) integral is then found back when we take a certain limit in this parameter. Popular regularization schemes are (in physics) cut-off regularization and dimensional 
regularization. The first consists in integrating $|p|$ not over $\mathbb{R}_+$ but simply over $[0,\Lambda]$. The second is to perform the integration over $D$ dimensions, with $D$ a complex 
number. We will use this last method.

There exist a lot of other regularization schemes, used in some areas of physics (Pauli--Villar regularization, lattice regularization,\dots) and mathematics (analytic regularization, 
zeta regularization,\dots). We will not discuss them here and simply assume that the final result is independent of the choice of the regularization scheme.

The second point of the renormalization procedure is the renormalization itself. The idea it to redefine the physical quantities inside the lagrangian (masses, coupling constants and fields) such 
that the results of the computations are finite. It has already been said in the introduction that we have the right to put under the carpet the divergences of the Feynman integrals because it is a 
consequence of assumptions of analyticity of the Green functions of the theory.

%

Now, by additivity of the physical quantities, the Feynman rules have to be characters of $H$ evaluated in an algebra $A$. We will assume that the regularization has been performed, therefore we 
will take as before $A=\mathbb{C}[1/\varepsilon, \varepsilon]]$. This is a small abuse of language: the coefficients are actually germs of functions of masses, impulsions, and so on\dots When we 
have a loop, we have a power of $1/\varepsilon$ that comes into the game. But the question is: how do we deal with multiloop graphs? And, in particular, with overlapping divergences?

The Birkhoff decomposition gave an answer to this question. If $\phi$ is the unrenormalized Feynman rule then $\phi_+$ in its Birkhoff decomposition is the renormalized Feynman rule. Let us assume 
that we know the evaluation of the renormalized Feynman rule over all the graphs with $l$ loops. Then the evaluation of $\phi_+$ over a graph $\Gamma$ with $l+1$ loops is given by 
\begin{equation}
 \phi_+(\Gamma) = (I_A - \Pi)\left(\phi(\Gamma) + \sum_{(\Gamma)}\phi_-(\gamma)\phi(\Gamma/\gamma)\right).
\end{equation}
We see that the counterterm is $\phi_-(\Gamma)$ which can itself be seen as a twisted antipode. Therefore, the renormalization procedure is really hidden into the Hopf algebraic structure, and the 
way to take care of (maybe overlapping) subdivergences is given by the antipode.

From now on, we will focus our attention on renormalized characters, and we will write them $\left.\phi^R:=\phi_+\right|_{\varepsilon=0}$. Let us just precise that this Hopf algebraic approach does not only 
provide a conceptual framework to tame the combinatorics of renormalization: it also allows to rigorously prove some physically useful results, such as the locality property of the counterterms.

\section{Renormalization Group}

\subsection{Renormalization group equation}

This subsection will be based on two ideas. The first one is to make the Feynman rules depend on a parameter that we will call $L$. This is done by taking the algebra of (germs of) meromorphic 
functions of $L$ as the target algebra $A$\footnote{actually, the target algebra might be more complicated: it can be the space of (germs of) meromorphic functions of all the kinematical invariants and 
the fine structure constants. We will not need this level of sophistication here.}. $L$ will represents the reference impulsion at which the measurements are performed. More precisely, if the exterior impulsion of a two-point 
graph is fixed to $q$  we will have $L=\log(q^2/\mu^2)$. Then we want to know how the Feynman rules change when $L$ changes. Let us write $\phi_L$ the elements of the associated characters group, which 
will itself be called $G$. Then we renormalize with the procedure described above these characters. We will use the short-hand notation $\phi^R_L:=(\phi_L)^R$ for the renormalized characters. The 
second idea is needed to perform this task: we will not directly 
work with the $\phi_L^R$ but rather with functions $\phi_{\varepsilon}$ ($\varepsilon$ the parameter of dimensional regularization). We will closely follow the presentation of \cite{CoKr00}.

First, for $t\in\mathbb{R}$ let us define $\theta_t\in\text{Aut}(H)$ by
\begin{equation}
 \forall\Gamma\in H; \quad \theta_t(\Gamma) := e^{t|\Gamma|}\Gamma.
\end{equation}
Then $\theta_t$ has a natural action on the characters group of $H$ (in the following, we will call this group $G$):
\begin{equation}
 \forall\Gamma\in H;\forall\phi\in G; \quad \theta_t(\phi)(\Gamma) := \phi\left(\theta_t(\Gamma)\right).
\end{equation}
The technical part of the proof of the renormalization group equation is within the following lemma:
\begin{lemm}
 $\forall\phi\in G$, $L_1,L_2\in\mathbb{R}$ we have the following decomposition:
 \begin{equation}
  \theta_{L_1+L_2}(\phi) = \theta_{L_1}(\phi)*\theta_{L_1}\left(\phi^{-1}*\theta_{L_2}(\phi)\right).
 \end{equation}
\end{lemm}
\begin{proof}
 This proof is an interesting exercice to get used to Hopf algebra. It might exist in many textbook but since I have never seen it, I will write it. We brutally 
 compute, $\forall\Gamma\in H$:
 \begin{align*}
    & \left(\theta_{L_1}(\phi)*\theta_{L_1}\left(\phi^{-1}*\theta_{L_2}(\phi)\right)\right)(\Gamma) \\
  = & \sum_{\gamma\subseteq\Gamma}\phi(\theta_{L_1}(\gamma))\left(\phi^{-1}*\theta_{L_2}(\phi)\right)\left(e^{L_1|\Gamma/\gamma|}\Gamma/\gamma\right) \\
  = & e^{L_1|\Gamma|}\sum_{\gamma\subseteq\Gamma}\phi(\gamma)\left(\phi^{-1}*\theta_{L_2}(\phi)\right)(\Gamma/\gamma) \qquad \text{since }|\Gamma/\gamma| = |\Gamma|-|\gamma| \\
  = & e^{L_1|\Gamma|}\left[\phi *\left(\phi^{-1}*\theta_{L_2}(\phi)\right)\right](\Gamma) \\
  = & e^{L_1|\Gamma|}\theta_{L_2}(\phi)(\Gamma) \qquad \text{by associativity of the convolution product} \\
  = & e^{(L_1+L_2)|\Gamma|}\phi(\Gamma) \\
  = & \theta_{L_1+L_2}(\phi)(\Gamma)
 \end{align*}
\end{proof}
Now, let us assume that there is a map $\mathbb{C}\setminus\{0\}\longrightarrow G$ which associates $\phi_{\varepsilon}$ to $\varepsilon$ such that
\begin{equation}
 \phi_L^R = \lim_{\varepsilon\rightarrow0}\phi_{\varepsilon}^{-1}*\theta_{L\varepsilon}(\phi_{\varepsilon}).
\end{equation}
We will assume the existence of the $\phi_{\varepsilon}$, which was proven in \cite{CoKr00}. In this article, an explicit formula was given for $\phi_{\varepsilon}$ in term of $\phi^R_L$ and its derivatives with 
respect to $L$.
\begin{thm} \cite{CoKr00}
 With the above definitions we have
 \begin{equation} \label{rel_Feynm_rules}
  \phi_{L+L'}^R = \phi_L^R*\phi_{L'}^R.
 \end{equation}
\end{thm}
\begin{proof}
 Since we have assumed $\phi_{\varepsilon}\in G$ we simply use the previous lemma:
 \begin{equation*}
  \phi_{\varepsilon}^{-1}*\theta_{(L+L')\varepsilon}(\phi_{\varepsilon}) = \left(\phi_{\varepsilon}^{-1}*\theta_{L\varepsilon}(\phi_{\varepsilon})\right)*\theta_{L\varepsilon}\left(\phi_{\varepsilon}^{-1}*\theta_{L'\varepsilon}(\phi_{\varepsilon})\right)
 \end{equation*}
 by associativity of the convolution product. Using 
 \begin{equation*}
  \theta_{L\varepsilon}\left(\phi^{-1}*\theta_{L'\varepsilon}(\phi)\right)(\Gamma/\gamma) = e^{L\varepsilon|\Gamma/\gamma|}\left(\phi_{\varepsilon}^{-1}*\theta_{L'\varepsilon}(\phi_{\varepsilon})\right)(\Gamma/\gamma) = \left(\phi_{\varepsilon}^{-1}*\theta_{L'\varepsilon}(\phi_{\varepsilon})\right)(\Gamma/\gamma) + \mathcal{O}(\varepsilon)
 \end{equation*}
 Therefore we have
 \begin{equation}
  \phi_{\varepsilon}^{-1}*\theta_{(L+L')\varepsilon}(\phi_{\varepsilon}) = \left(\phi_{\varepsilon}^{-1}*\theta_{L\varepsilon}(\phi_{\varepsilon})\right)*\left(\phi_{\varepsilon}^{-1}*\theta_{L'\varepsilon}(\phi_{\varepsilon})\right)(\Gamma/\gamma) + \mathcal{O}(\varepsilon)
 \end{equation}
 and we obtain the theorem by taking the limit $\varepsilon\longrightarrow0$.
\end{proof}
We are almost at what we will call the renormalization group equation. Before deriving it as a simple corollary of the theorem, let us make a innocent remark: the theorem has to be corrected if we 
work at $\varepsilon\neq0$.
\begin{coro}
 The renormalized Feynman rules obey the differential equation
 \begin{equation}
  \partial_L\phi_L^R = \left.\partial_L\phi_L^R\right|_{L=0}*\phi_L^R
 \end{equation}
 which we will refer as the renormalization group equation.
\end{coro}
\begin{proof}
 We can rewrite \eqref{rel_Feynm_rules} as
 \begin{equation*}
  \phi_{L}^R = \phi_{L-L'}^R*\phi_{L'}^R.
 \end{equation*}
Then taking a derivative with respect to $L$ and evaluating the result at $L'=L$ give the renormalization group equation.
\end{proof}
We have used the Leibniz rule for the convolution product. This was allowed by the following lemma:
\begin{lemm}
 Let $\phi_L$ and $\psi_L$ two families of elements of $G$ depending continuously of the parameter $L$. Then $\partial_L$ obeys the Leibniz rule for 
 the convolution product $*$:
 \begin{equation}
  \partial_L(\phi_L*\psi_L) = (\partial_L\phi_L)*\psi_L + \phi_L*(\partial_L\psi_L).
 \end{equation}
\end{lemm}
\begin{proof}
 By direct computation we have
 \begin{equation*}
  \partial[(\phi_L*\psi_L)(x)] = \partial_L\left(\sum_{(x)}\phi_L(x')\psi_L(x'')\right)  = \sum_{(x)}\left((\partial_L\phi_L)(x')\psi_L(x'') + \phi_L(x')\partial_L\psi_L(x'')\right)
 \end{equation*}
 which is the desired result.
\end{proof}

\subsection{Correlation functions}

Our final goal for this first chapter is to observe the consequences of the renormalization group equation on the correlation functions of the theory. First, we will define those objects and state 
(often without proof) some results. We want to take the time to carefully define the correlation functions since in the physics literature there is quite a big mix-up between correlation 
functions, Green functions and even sometimes propagators.

But first, two useful precisions. From now on, we will work in a theory where there is no need to label the edges and the vertices. Typically this might be the massless Wess--Zumino model or the 
massless scalar model in six dimensions with a cubic interaction. Hence the residue of a graph will be given by a natural number. Moreover, we will work in $\mathcal{B}_{\tau}$ rather than in 
$H_{\tau}$. Indeed, in the details of the computations, we will need to distinguish between the various types of loopless graphs, typically between \vertx~ and \prop.

Now, let us recall that the loop number of a connected graph is
\begin{equation*}
 l(\Gamma) = |I(\Gamma)|-|V(\Gamma)|+1.
\end{equation*}
if $V(\Gamma)\neq\emptyset$ and $l(\Gamma)=k(\Gamma)-|E(\Gamma)|/2$ otherwise. Alternatively, we can use only the first of these formulae and still have the loop number of \prop~ to be zero by defining
%
\begin{equation*}
 I\left(\prop\right) = -1.
\end{equation*}
We will implicitly use this value in some computations below. With this we can define the Green functions as a sum over 1PI graphs
\begin{equation}
 \Gamma_r = \sum_{\substack{\res(\Gamma)=r\\l(\Gamma)\geq0}}\frac{\Gamma}{s(\Gamma)}
\end{equation}
with $s(\Gamma)$ the symmetry factor of $\Gamma$. The inverse of the Green function will be written $G_r:=(\Gamma_r)^{-1}$, where the inverse is the one of formal series. The $G_r$s are what we 
call the correlation functions. We will often expand $\Gamma_r$ (or $G_r$) in term of the loop number:
\begin{equation}
 \Gamma_r = \sum_{l\geq0}\Gamma_{r,l} = \sum_{l\geq0}\sum_{\substack{\res(\Gamma)=r\\l(\Gamma)=l}}\frac{\Gamma}{s(\Gamma)}.
\end{equation}
We will be working in theories with only a cubic interaction therefore, adapting \cite{CoKr00}, we define the effective coupling constant of the theory by
\begin{equation}
 a_{\text{eff}} = \frac{(\Gamma_3)^2}{(\Gamma_2)^3}.
\end{equation}
It is normalized to $1$ at the impulsion of reference $p^2=\mu^2\Leftrightarrow L=0$. This object has the right dimension: $a_{\text{eff}}\alpha g^2$ with $g$ the coefficient of the fully 
renormalized vertex.

Now, since we want to apply the renormalization group equation to the correlation and Green functions we will need to study how the coproduct acts on them. This is a complicated combinatorial 
question that has already been studied. Therefore, we will only state the results. One result of \cite{vSu06} is
\begin{equation}
 \sum_{\substack{\res{\Gamma}=r\\l(\Gamma)=l}}\frac{\Delta(\Gamma)}{s(\Gamma)} = \sum_{k=0}^l\sum_{\substack{\res(\Gamma)=r\\l(\Gamma)=l-k}}\sum_{\substack{\gamma\\l(\gamma)=k}}\frac{\Gamma|\gamma}{s(\gamma)s(\Gamma)}\gamma\otimes\Gamma
\end{equation}
with the sum over $\gamma$ being over 1PI graphs and $\Gamma|\gamma$ being the number of ways to insert $\gamma$ in $\Gamma$. Then, summing over $l$ we see that the sums in the right hand side reduce to a 
sum over $\Gamma$ and a sum over $\gamma$. Moreover, it was shown in \cite{vSu08} that the summed right hand side could be expressed in terms of the $2$ and $3$ points Green function:
\begin{equation}
 \Delta(\Gamma_r) = \sum_{\substack{\res(\Gamma)=r\\l(\Gamma)\geq0}}(\Gamma_3)^{V(\Gamma)}(\Gamma_2)^{-I(\Gamma)}\otimes\frac{\Gamma}{s(\Gamma)}.
\end{equation}
Now, in any graph, each leg is either an external half edge or an internal edge. Moreover, each vertex has valency $3$ and each internal edge is attached to $2$ vertices, hence the number of internal edges 
is linked to the number of vertices and to the residue of a graph through the following formula:
\begin{equation*}
 I(\Gamma) = \frac{1}{2}\left(3V(\Gamma)-r\right).
\end{equation*}
Using this together with the formula \eqref{loop_def} we get
\begin{equation} \label{coprod_Green}
 \Delta(\Gamma_r) = \sum_{\substack{\res(\Gamma)=r\\l(\Gamma)\geq0}}(\Gamma_3)^{r-2}(\Gamma_2)^{-r+3}a_{\text{eff}}^{l(\Gamma)}\otimes\frac{\Gamma}{s(\Gamma)}
\end{equation}
which is the equation needed to write the renormalization group equation on the correlation functions. Notice than this equation also holds for \prop.

\subsection{RGE for the correlation functions}

To derive the renormalization group equation for the correlation functions, we will use a simple property of the graded bialgebras.
\begin{propo} \cite{BeSc08}
 Let $A$ and $B$ be two elements of the completion of a graded bialgebra: 
 \begin{equation*}
  A = \sum_{l\geq0}A_l \qquad B=\sum_{l\geq0}B_l.
 \end{equation*}
 Then the following equivalence holds
 \begin{align}
  \Delta(AB) = \sum_{l\geq0}ABa^l\otimes (AB)_l \Leftrightarrow 
  \begin{cases}
   & \Delta(A) = \sum_{l\geq0}Aa^l\otimes A_l \\
   & \Delta(B) = \sum_{l\geq0}Ba^l\otimes B_l.
  \end{cases}
 \end{align}
for $a$ an element of the graded bialgebra.
\end{propo}
\begin{proof}
 First let us notice that $(AB)_l = \sum_{n=0}^lA_nB_{l-n}$. Now, let us prove the reverse assertion $\Leftarrow$. From the axiom of the bialgebra structure we have
 \begin{align*}
  \Delta(AB) & = (m\otimes m)\circ \tau_{23}\left(\Delta(A)\otimes\Delta(B)\right) \\
             & = (m\otimes m)\circ \tau_{23}\left(\sum_{l\geq0}\sum_{m\geq0}Aa^l\otimes A_l\otimes Ba^m\otimes B_m\right) \\
             & = \sum_{l\geq0}\sum_{m\geq0}ABa^{l+m}\otimes A_lB_m 
 \end{align*}
 We can reorganize the two sums to write
 \begin{equation*}
  \Delta(AB) = \sum_{n\geq0}\sum_{p=0}^n ABa^n\otimes A_pB_{n-p}.
 \end{equation*}
 Hence the reverse assertion is true. We prove the direct assertion $\Rightarrow$ by going the other way around. Then we have
 \begin{equation*}
  \Delta(AB) = (m\otimes m)\circ \tau_{23}\left(\left(\sum_{l\geq0}Aa^l\otimes A_l\right)\otimes \left(\sum_{m\geq0}Ba^m\otimes B_m\right)\right) = (m\otimes m)\circ \tau_{23}\left(\Delta(A)\otimes\Delta(B)\right)
 \end{equation*}
 and we find the desired result by identification.
\end{proof}
Now we can start to exploit the formula \eqref{coprod_Green} for $r=2$. Since the coproduct of the Hopf algebra preserves the loop number we can expand it with respect to the loop number to get:
\begin{equation}
 \Delta(\Gamma_2) = \sum_{l\geq0}\Gamma_2a_{\text{eff}}^l\otimes\Gamma_{2,l}.
\end{equation}
Since we have shown that the set of elements having the above property is stable by multiplication, we shall have
\begin{equation} \label{RGE_G2_prim}
 \Delta(G_2) = \sum_{l\geq0}G_2a_{\text{eff}}^l\otimes G_{2,l}.
\end{equation}
Notice that $1=G_2\Gamma_2$ has indeed this property. Using this coproduct on the renormalization group equation we get
\begin{align}
 \partial_L\phi_L^R(G_2) & = \sum_{l\geq0}\left.\partial_L\phi_L^R(G_2a_{\text{eff}}^l)\right|_{L=0}\phi_L^R(G_{2,l}) \nonumber \\
                         & = \sum_{l\geq0}\left.\partial_L\left(\phi_L^R(G_2)\phi_L^R(a_{\text{eff}}^l)\right)\right|_{L=0}\phi_L^R(G_{2,l}) \qquad \text{since $\phi_L^R$ is a character} \nonumber \\
                         & = \sum_{l\geq0}(\gamma+l\beta)\phi_L^R(G_{2,l})
\end{align}
with $\gamma:=\partial_L^R\phi_L^R(G_2)|_{L=0}$ and $\beta:=\partial_L\phi_L^R(a_{\text{eff})}|_{L=0}$ the usual gamma and beta functions of the theory. The last line was obtained using the Leibniz rule. 
Finally, let us make one more simplification: writing $a:=\phi_L(a_{\text{eff}})$ we get
\begin{equation*}
 \phi_L^R(G_{r,l}) = \tilde G_{r,l}a^l.
\end{equation*}
Indeed, $a$ comes from the evaluation of two vertices since $a\alpha g^2$ and therefore of a loop. In other words: every loop in a Feynman graph will bring a further power of $a$. This is the 
basis of the standard perturbative approach of quantum field theories. Thus
\begin{equation} \label{l_deriv_a}
 a\partial_a\phi_L^R(G_r) = \phi_L^R\left(\sum_{l\geq0}lG_{r,l}\right).
\end{equation}
and then, since $\sum_{l\geq0}\gamma\phi^R_L(G_{2,l}) = \gamma G_2$ we have
\begin{equation} \label{RGE_G2_gen}
 \partial_L\phi_L^R(G_2) = (\gamma+\beta a\partial_a)\phi_L^R(G_2).
\end{equation}
The procedure for the three-points correlation function is essentially the same. We use the relation \eqref{coprod_Green} with $r=3$. Then the same arguments than the ones for the two-points 
function lead to
\begin{equation*}
 \Delta(\Gamma_3) = \sum_{l\geq0}\Gamma_3a_{\text{eff}}^l\otimes\Gamma_{3,l}.
\end{equation*}
Therefore, as argued before, we shall have
\begin{equation*}
 \Delta(G_3) = \sum_{l\geq0}G_3a_{\text{eff}}^l\otimes G_{3,l}
\end{equation*}
and the renormalization group equation gives 
\begin{equation*}
 \partial_LG_3 = (\gamma_m + \beta a\partial_a)G_3
\end{equation*}
with $\gamma_m:=\partial_L\phi_L^R(G_3)|_{L=0}$, where we have used the relation \eqref{l_deriv_a} once again.

With this final result we have now enough material to start studying one of the most central object of this Ph.D.: the Schwinger--Dyson equation. Before doing so, let me just remind the reader 
that this presentation of the Hopf algebra of renormalization is by no means complete. We hope to have convinced the reader that it is a powerful tool to understand the process of renormalization 
and, more generally speaking, to understand quantum field theories. Moreover, there is more than what has been presented here: very powerful results that the lack of time and space forbade us to 
present here. 

%
%
\chapter{Linear Schwinger--Dyson equations}

 \noindent\hrulefill \\
 {\it
Agile et noble, avec sa jambe de statue. \\
Moi, je buvais, crispé comme un extravagant, \\
Dans son oeil, ciel livide où germe l'ouragan, \\
La douceur qui fascine et le plaisir qui tue. \\

Charles Baudelaire. A une passante.} 

 \noindent\hrulefill
 
 \vspace{1.5cm}

We start our exploration of Schwinger--Dyson equations with a class of equations named ``linear''. With this term, we mean that the equation for the two-point function is an integrodifferential 
equation with an integrand linear in the renormalized propagator. Through this chapter we will use the sign convention of \cite{Itzykson}.

\section{The massless Yukawa model} \label{Yuk_nm}

Let us start with a presentation of some results of \cite{BrKr99}. In this article, the first exact solution to a Schwinger--Dyson equation has been found, and understanding precisely how this 
was done is a crucial step before tackling more complicated equations. 

\subsection{The Schwinger--Dyson equation}

We will study a model containing only a massless fermion $\psi$ represented by a plain line and a massless scalar $\phi$ represented by a dashed line. They are interacting through the usual Yukawa 
lagrangian $\mathcal{L}_{Yuk}=g\bar\psi\phi\psi$: the only authorized vertex of the theory is therefore \vertxyuk. Moreover we will assume that only the fermion gets renormalized. Hence the Schwinger--Dyson equation of the model is
\begin{equation} \label{SDElinYuk}
 \left(
\tikz \node[prop]{} child[grow=east] child[grow=west];
\right)^{-1} = 1 - a \;\;
\begin{tikzpicture}
\draw (-1.2,0)--(-0.8,0);
\draw (0.8,0)--(1.2,0);
\draw[dashed] (-0.8,0) .. controls (-0.8,1) and (0.8,1) .. (0.8,0);
\node at (0,0) [circle,minimum size=6mm,draw,fill=green!30] {} child [grow=east] child[grow=west];
\end{tikzpicture}.
\end{equation}
This is a relevant equation to consider since it generates a sub-Hopf algebra of the Hopf algebra of renormalization. This sub-Hopf algebra is isomorphic to the Hopf algebra of rooted trees. The inverse of the propagator is 
$G(p^2)\slashed{p}$ and therefore we can write \eqref{SDElinYuk} as
\begin{equation} \label{SDE_Yuk_int}
 G(p^2)\slashed{p} = \slashed{p} -\frac{2a}{\pi^2}\int\d^4l\frac{1}{G(l^2)\slashed{l}(p+l)^2} + S\slashed{p} 
\end{equation}
with $S$ a counterterm that will be determined by the initial condition $G(\mu^2)=1$. Using $\slashed{l}^2=l^2$ we can rewrite the integrand in the above equation as
$\frac{\slashed{l}}{G(l^2)l^2(p+l)^2}$. Then, multiplying the two sides of \eqref{SDE_Yuk_int} by $\slashed{p}$, using $\slashed{p}.\slashed{l}=p.l$ and dividing by $p^2$ we arrive to
\begin{equation}
 G(p^2) = 1 - \frac{2a}{\pi^2}\int\frac{\d^4l}{p^2}\frac{1}{G(l^2)l^2}\frac{p.l}{(p+l)^2} + S.
\end{equation}
Now, as usual we want to write the four-dimensional integral as a one-dimensional one. We have to take care of the angular dependence of $\frac{p.l}{(p+l)^2}$. First, let us write
$p.l = [(p+l)^2-p^2-l^2]/2$: we will end up with three integrals. Then we can use the four-dimensional angular average
\begin{equation} \label{angint}
 q^2l^2\left\langle\frac{1}{(q+l)^2}\right\rangle_{d=4} = \text{min}(q^2,l^2).
\end{equation}
So, two of the three integrals that we had will be split into two pieces. Moreover, when performing the angular integration we will have to take care of the volume of the $d-1$-dimensional 
sphere, which is
\begin{equation}
 \text{vol}(S^{d-1}) = \frac{2\pi^{d/2}}{\Gamma(d/2)}.
\end{equation}
Then we have
\begin{align*}
 G(p^2) & = 1 - a\left[\frac{1}{p^2}\int_0^{+\infty}\frac{l\d l}{G(l^2)} - \frac{1}{p^2}\int_0^{p}\frac{l\d l}{G(l^2)} - \int_{p}^{+\infty}\frac{\d l}{lG(l^2)} - \frac{1}{p^4}\int_0^p\frac{l^3\d l}{G(l^2)} - \frac{1}{p^2}\int_p^{+\infty}\d l\frac{l}{G(l^2)}\right] + S \\
        & = 1 + a\left[\int_{p}^{+\infty}\d l\frac{1}{lG(l^2)} + \frac{1}{p^4}\int_0^p\d l\frac{l^3}{G(l^2)}\right] + S.
\end{align*}
We are very aware that we did not take care of the convergence of the integrals popping up in this computation. And, as a matter of fact, they are typically non-convergent. These divergences will 
be cured by the counterterm $S$. Later, we will take a derivative with respect to $p^2$, which would have canceled $S$ and thus all the other divergences if we had taken it earlier. The result 
does not depend on if one takes the derivative at the beginning or the end of the computation, as can be checked by direct computation. We have chosen this approach in order to be closer to the 
original work of \cite{BrKr99}.

Now is the right time to fix $S$. By imposing $G(\mu^2)=1$ we have
\begin{equation*}
 S = -a\left[\int_{\mu}^{+\infty}\d l\frac{1}{lG(l^2)} + \frac{1}{\mu^4}\int_0^{\mu}\d l\frac{l^3}{G(l^2)}\right]
\end{equation*}
and thus
\begin{equation*}
 G(p^2) = 1 + a\left[\frac{1}{p^4}\int_0^p\d l\frac{l^3}{G(l^2)} - \frac{1}{\mu^4}\int_0^{\mu}\d l\frac{l^3}{G(l^2)} - \int_{\mu}^{p}\d l\frac{1}{lG(l^2)}\right].
\end{equation*}
Let us now perform the change of integration variable $y=l^2$. Then we have
\begin{align}
 G(p^2) & = 1 + \frac{a}{2}\left[\frac{1}{p^4}\int_0^{p^2}\d y\frac{y}{G(y)} - \frac{1}{\mu^4}\int_0^{\mu^2}\d y\frac{y}{G(y)} - \int_{\mu^2}^{p^2}\d y\frac{1}{yG(y)}\right] \nonumber \\
        & = 1 - \frac{a}{2}\int_{\mu^2}^{p^2}\d y\frac{1}{yG(y)} + F(p^2) - F(\mu^2) \label{SDE_Yuk_F}
\end{align}
with
\begin{equation*}
 F(x) = \frac{a}{2}\int_{0}^{x}\d y\frac{1}{yG(y)}\left(\frac{y}{x}\right)^2.
\end{equation*}
At this point, a useful remark is
\begin{equation*}
 \frac{\d}{\d x}F(x) = \frac{a}{2}\frac{1}{xG(x)} -\frac{a}{x^3}\int_0^x\d y\frac{y}{G(y)}
\end{equation*}
therefore, taking a derivative with respect to $x:=p^2$ it comes
\begin{equation*}
 \frac{\d}{\d x}G(x) =  -\frac{a}{x^3}\int_0^x\d y\frac{y}{G(y)} \Leftrightarrow x^3\frac{\d}{\d x}G(x) =  -a\int_0^x\d y\frac{y}{G(y)}.
\end{equation*}
Now, defining the differential operator $D=x\frac{\d}{\d x}$ and taking one more derivative with respect to $x$ we arrive to
\begin{align*}
 & \frac{\d}{\d x}(x^2DG(x)) = -\frac{ax}{G(x)} \\
\Leftrightarrow & 2xDG(x) +x^2\frac{\d}{\d x}DG(x) = -\frac{ax}{G(x)}.
\end{align*}
After simplification we finally arrive to
\begin{equation} \label{SDE_Yuk_diff}
 G(x)D(D+2)G = -a.
\end{equation}
Hence, we have mapped an integral equation to a integrodifferential one to finally arrive to a simple differential one. The fact that this is possible is the first important step into 
finding an exact solution of \eqref{SDElinYuk}.

\subsection{Parametric solution}

Going from an integrodifferential equation a differential one is indeed a great simplification, but in this case it is merely a consequence 
of the linearity (in the sense defined in the header of this chapter) of \eqref{SDElinYuk}. The real magic is coming now: as a fact, we can solve \eqref{SDE_Yuk_diff}. Let us start by defining 
the dimensionless variable $z$ and the function $\tilde G$:
\begin{subequations}
 \begin{align}
  & z := \left(\frac{q^2}{\mu^2}\right)^2 \\
  & \tilde{G}(z) := \sqrt{\frac{2}{a}}zG(\mu^2\sqrt{z}).
\end{align}
\end{subequations}
Then the initial condition becomes $\tilde G(1)=\sqrt{2/a}$. Moreover, if we recall $x=q^2/\mu^2$, we have
\begin{equation*}
 D = q^2\frac{\d}{\d q^2} = \mu^2\sqrt{z}2\frac{q^2}{\mu^4}\frac{\d}{\d z} = 2z\frac{\d}{\d z}.
\end{equation*}
Thus \eqref{SDE_Yuk_diff} becomes
\begin{equation*}
 \frac{a}{2}\frac{\tilde G}{z}2z\frac{\d}{\d z}\left(2z\frac{\d}{\d z} + 2\right)\frac{\tilde G}{z} = -a \Leftrightarrow 2\tilde G\frac{\d}{\d z}\left(z\frac{\d}{\d z} + 1\right)\frac{\tilde G}{z} = -1.
\end{equation*}
but
\begin{equation*}
 \left(z\frac{\d}{\d z} + 1\right)\frac{\tilde G}{z} = z\left(\frac{\tilde G'(z)}{z}-\frac{\tilde G(z)}{z^2}\right)+\frac{\tilde G(z)}{z} = \tilde G'(z)
\end{equation*}
with the prime denoting a differentiation with respect to $z$. Then \eqref{SDE_Yuk_diff} is simply
\begin{equation} \label{avant_int}
 2\tilde G(z)\tilde G''(z) = -1.
\end{equation}
Now we are able to integrate this equation to
\begin{equation*}
 2\tilde G'(z)\tilde G''(z) = -\frac{\tilde G'(z)}{\tilde G(z)} \Leftrightarrow \left[\tilde G'(z)\right]^2 = \kappa -\log\tilde G(z)
\end{equation*}
with $\kappa$ a constant. Now, we can write $\tilde G(z)$ in term of the parameter $p:=\tilde G'(z)$:
\begin{equation} \label{sol_step1}
 p^2 = \kappa -\log\tilde G \Leftrightarrow \tilde G = \sqrt{\frac{2}{a}}\exp(p_0^2-p^2)
\end{equation}
where we have used the initial condition $\tilde G(1)=\sqrt{2/a}$ and the parameter $p$ at $z=1$: $p_0:=\tilde G(1)$. Hence we have the first half of a parametric solution: the unknown function is 
now written in term of the parameter $p$. The next step is to relate this parameter to $z$. This was done in \cite{BrKr99} by studying $\tilde\alpha(p) = z/\tilde G$. In order to find a 
differential equation satisfied by $\tilde\alpha$ we first have to find how $z$ and $\tilde G$ vary with $p$. Using directly \eqref{sol_step1} we have
\begin{equation*}
 \frac{\d\tilde G}{\d p} = -2p\tilde G.
\end{equation*}
But, using the chain rule, we can also write
\begin{equation*}
 \frac{\d\tilde G}{\d p} = \frac{\d z}{\d p}\frac{\d\tilde G}{\d z} = \frac{\d z}{\d p}p
\end{equation*}
where we have used the definition of $p$ for the last equality. Therefore, equalizing those two results for the derivative of $\tilde G$ with respect to $p$ one gets
\begin{equation*}
 \frac{\d z}{\d p} = -2\tilde G.
\end{equation*}
Now, we can simply compute
\begin{equation*}
 \frac{\d\tilde\alpha(p)}{\d p} = \frac{\d z}{\d p} -\frac{z}{\tilde G^2}\frac{\d\tilde G}{\d p} = -2 + 2p\tilde\alpha(p).
\end{equation*}
Inverting this relation, we obtain a simple differential equation satisfied by $\tilde\alpha$
\begin{equation} \label{eq_diff_alpha}
 \tilde\alpha = \frac{1}{p} + \frac{1}{2p}\frac{\d\tilde\alpha}{\d p}.
\end{equation}
This equation can be exactly solved to
\begin{equation}
 \tilde\alpha(p) =e^{p^2}\sqrt{\pi}\left(\lambda-\text{erf}(p)\right)
\end{equation}
with erf the error function defined by
\begin{equation*}
 \text{erf}(p) = \frac{2}{\sqrt{\pi}}\int_0^p\d s\exp(-s^2)
\end{equation*}
and $\lambda$ a constant. It can be determined with the initial condition but a simpler way to go is to notice that $\tilde\alpha$ has to be regular at infinity. Indeed, at 
$p\longrightarrow\infty$, the intermediate result \eqref{sol_step1} gives $\tilde G=0$, and therefore $z=0$ since we can assume that $G$ is vanishing nowhere, from the hypothesis of analycity of 
the propagator. Then, writing 
\begin{equation*}
 \tilde\alpha(p) = \sqrt{\frac{a}{2}}\frac{1}{G}
\end{equation*}
we have
\begin{equation*}
 \tilde\alpha \simall{p}{\infty} \sqrt{\frac{a(q=0)}{2}}\frac{1}{G(0)},
\end{equation*}
which is finite. Since $\text{erf}(p)\simall{p}{\infty}1$ we have to take $\lambda=1$ to solve this. Hence we obtain
\begin{equation} \label{sol_alpha}
 \tilde\alpha(p) =e^{p^2}\sqrt{\pi}\text{erfc}(p) = 2\int_p^{+\infty}\d s\exp(p^2-s^2)
\end{equation}
with erfc the complementary error function defined by
\begin{equation*}
 \text{erfc}(p) = 1-\text{erf}(p) = \frac{2}{\sqrt{\pi}}\int_p^{+\infty}\d s\exp(-s^2).
\end{equation*}
In \cite{BrKr99} this result was obtain from an expansion in powers of $1/p$ of $\tilde\alpha$. We have preferred here a slightly different approach in order to make clear that the constant 
stemming from the resolution of the differential equation \eqref{eq_diff_alpha} is not put under the carpet. This was the second step of the parametric solution, and we now just have to put the 
pieces together. First, if we recall
\begin{equation*}
 \tilde\alpha(p) = \frac{z}{\tilde G} = \sqrt{\pi}\exp(p^2)\text{erfc}(p)
\end{equation*}
and $z=(q^2/\mu^2)^2$ it comes
\begin{equation*}
 \left(\frac{q^2}{\mu^2}\right)^2 = \sqrt{\pi}\exp(p^2)\text{erfc}(p)\tilde{G} = \sqrt{\frac{2\pi}{a}}e^{p_0^2}\text{erfc}(p)
\end{equation*}
where we have used \eqref{sol_step1} for $\tilde G$. Now, we can simplify this expression if we remember that $p_0$ is defined as the value of $p$ at $z=1$ and thus
\begin{equation*}
 \tilde\alpha(p_0) = \frac{1}{\tilde G(1)} = \sqrt{\frac{a}{2}} = \sqrt{\pi}e^{p_0^2}\text{erfc}(p_0) \Leftrightarrow e^{p_0^2} = \sqrt{\frac{a}{2\pi}}\frac{1}{\text{erfc}(p_0)}.
\end{equation*}
Then we are left with the beautiful
\begin{equation}
 q^2 = \mu^2\sqrt{\frac{\text{erfc}(p)}{\text{erfc}(p_0)}}.
\end{equation}
To write $G$ in term of $p$ we use the definition of $\tilde G$ together with the previous result $z=\text{erfc}(p)/\text{erfc}(p_0)$ and we get
\begin{equation*}
 \frac{\text{erfc}(p)}{\text{erfc}(p_0)}G(q^2) = \exp(p_0^2-p^2).
\end{equation*}
Using once again the result $e^{p_0^2} = \sqrt{\frac{a}{2\pi}}\frac{1}{\text{erfc}(p_0)}$ found to express $z$ in term $p$ we arrive to the desired result. Altogether, let us write the full 
solution 
\begin{subequations}
 \begin{align}
  G(q^2) = \sqrt{\frac{a}{2\pi}}\frac{e^{-p^2}}{\text{erfc}(p)} \\
  q^2 = \mu^2\sqrt{\frac{\text{erfc}(p)}{\text{erfc}(p_0)}}.
 \end{align}
\end{subequations}
Let us conclude this subsection with the remark that this approach has more to offer than a analytic solution that might not be very convenient for practical purposes. For example, using the 
definition of $p_0$ and $\tilde G$ we have
\begin{equation*}
 p_0 = \tilde G'(1) = \sqrt{\frac{2}{a}}\left.\left(G + z\frac{\d G}{\d z}\right)\right|_{z=1}.
\end{equation*}
Using $z\frac{\d}{\d z} = \frac{1}{2}x\frac{\d}{\d x} = \frac{1}{2}q^2\frac{\d}{\d q^2}$ it comes
\begin{equation*}
 p_0 = \sqrt{\frac{2}{a}}\left.\left(G + \frac{1}{2}q^2\frac{\d G}{\d q^2}\right)\right|_{q^2=\mu^2}.
\end{equation*}
We recognize in this last expression the anomalous dimension $\tilde\gamma:=\left.q^2\frac{\d G}{\d q^2}\right|_{q^2=\mu^2}$ and hence
\begin{equation} \label{p0_Yuk_nm}
 p_0 = \frac{2+\tilde\gamma}{\sqrt{2a}}.
\end{equation}
If we use this result in the asymptotic series found for $\tilde\alpha$ from the equation\eqref{sol_alpha}
\begin{equation*}
 \tilde\alpha(p) \sim \frac{1}{p} + \frac{1}{p}\sum_{n=1}^{+\infty}\frac{(2n-1)!!}{(-2p^2)^n}
\end{equation*}
we obtain an asymptotic series for $\tilde\gamma$ around $a=0$, as stressed in \cite{BrKr99}, which allow efficient computation to high orders of perturbation theory.

\subsection{Lessons from the massless Yukawa model}

Having a exact solution for a given Schwinger--Dyson equation the natural next step is to ask for generalizations. This was done in \cite{Cl14}. Before presenting the results of this paper, let us 
briefly summarize the key points that have made possible the exact resolution of \eqref{SDElinYuk}.

The result of \cite{BrKr99} presented above rests upon three crucial steps. First, the integrodifferential Schwinger--Dyson equation is written as a differential one. As already stressed, this is 
made possible by the linearity of the Schwinger--Dyson equation. Therefore in the following sections of this chapter we will only study linear equations.

The second step is the integration, that is made when the Schwinger--Dyson equation is written in term of dimensionless parameters. In the previous subsection, this step was performed from the 
equation \eqref{avant_int}. In the \cite{Cl14} it was noticed that this step could be performed for more complicated theories than the massless Yukawa model.

The third step is to write the integrated equation with a parameter, and solve the equation when it is written in term of this parameter. This was done by finding an explicit expression for 
$\tilde\alpha$ in the previous subsection. Moreover, the chosen parameter has also to be expressible in term of the initial variables. Again, this step will still be doable for some theories 
distinct from the massless Yukawa model. Furthermore, in some cases, this last step will not be needed: an explicit solution (and not a parametric one) will be find directly after the integration 
(i.e. after step 2).

We will now turn our attention to the massive Yukawa model, which is the most natural generalization of the massless Yukawa model. It will be solved in the ultraviolet and infrared regime. Then 
we will turn our attention to a massive Wess--Zumino model with two renormalized superfields. A solution of this model will be found without restriction on its range of validity.

\section{The massive Yukawa model}

\subsection{The Schwinger--Dyson equation}

In this model, the inverse of the propagator is now
\begin{equation}
 P_{mY}^{-1} = G(q^2)\slashed{q}+M(q^2)m.
\end{equation}
The interaction lagrangian is still $g\bar\psi\phi\psi$ therefore the pictorial Schwinger--Dyson equation presents no difference in the massless and in the massive cases:
\begin{equation} \label{SDlinnSUSY}
 \left(
\tikz \node[prop]{} child[grow=east] child[grow=west];
\right)^{-1} = 1 - a \;\;
\begin{tikzpicture}
\draw (-1.2,0)--(-0.8,0);
\draw (0.8,0)--(1.2,0);
\draw[dashed] (-0.8,0) .. controls (-0.8,1) and (0.8,1) .. (0.8,0);
\node at (0,0) [circle,minimum size=6mm,draw,fill=green!30] {} child [grow=east] child[grow=west];
\end{tikzpicture}
\end{equation}
with $1$ denoting the free propagator as before. In order to find back the free propagator at a the impulsion of reference $\mu$, we have the initial conditions $G(\mu^2)=M(\mu^2)=1$. Now, the 
equation \eqref{SDlinnSUSY} can be written in term of Feynman integrals:
\begin{equation*}
 G(p^2)\slashed{p}+M(p^2)m = \slashed{p} + m -\frac{2a}{\pi^2}\int\d^4l\frac{1}{G(l^2)\slashed{l}+M(l^2)m}\frac{1}{(p+l)^2} + S_1\slashed{p}+S_2m.
\end{equation*}
with $S_1$ and $S_2$ two counterterms that will be fixed to fulfill the initial conditions as before. Multiplying the numerator and the denominator of the integrand by $G(l^2)\slashed{l}-M(l^2)m$ 
and using $\slashed l^2 =l^2$ we arrive to:
\begin{equation}
 G(p^2)\slashed{p}+M(p^2)m = \slashed{p} + m -\frac{2a}{\pi^2}\int\d^4l\frac{G(l^2)\slashed{l}-mM(l^2)}{G^2(l^2)l^2-M^2(l^2)m^2}\frac{1}{(p+l)^2} + S_1\slashed{p}+S_2m.
\end{equation}
The basic idea is that the $\slashed{p}$ and the $m$ parts do not talk to each other, so their coefficients should independently vanish. Multiplying the $\slashed p$ equation by $\slashed p$ and 
dividing by $p^2$ we are left with a system of two coupled integral 
equations:
\begin{subequations}
 \begin{align}
  & G(p^2) = 1 - \frac{2a}{\pi^2}\int\text{d}^4l\frac{G(l^2)}{G^2(l^2)l^2-m^2M^2(l^2)}\frac{1}{p^2}\frac{l.p}{(p+l)^2}+S_1 \label{SDEnSUSY1a} \\
  & M(p^2) = 1 + \frac{2a}{\pi^2}\int\text{d}^4l\frac{M(l^2)}{G^2(l^2)l^2-m^2M^2(l^2)}\frac{1}{(p+l)^2} + S_2. \label{SDEnSUSY1b}
\end{align}
\end{subequations}
From now on those two equations will be the ones referred as the Schwinger--Dyson equations. As in the previous section, the two integrals can be computed by separating their radial and angular 
parts and using the four-dimensional angular average \eqref{angint}. For the equation \eqref{SDEnSUSY1a} we have to use again $p.l = [(p+l)^2-p^2-l^2]/2$. After simplifications we are left with:
\begin{subequations}
 \begin{align}
  & G(p^2) = 1 + \frac{a}{p^2}\int_0^p\text{d}l\frac{l^3G(l^2)}{G^2(l^2)l^2 - m^2M^2(l^2)}\frac{l^2}{p^2} + a\int_{p}^{+\infty}\text{d}l\frac{lG(l^2)}{G^2(l^2)l^2-mM^2(l^2)} + S_1 \\
  & M(p^2) = 1 + \frac{2a}{p^2}\int_0^p\text{d}l\frac{l^3M(l^2)}{G^2(l^2)l^2-m^2M^2(l^2)} + 2a\int_p^{+\infty}\text{d}l\frac{lM(l^2)}{G^2(l^2)l^2-m^2M^2(l^2)} + S_2.
\end{align}
\end{subequations}
The reader may notice that we give less details here than in the previous section. Indeed, the computation is exactly similar, only more tedious: in this case, three terms have combined themselves 
together to cancel and we are left with such a simple equation for $G$. Now, using the initial conditions $G(\mu^2)=M(\mu^2)=1$ allows to fix the counterterms:
\begin{align*}
 & S_1 = -\frac{a}{\mu^2}\left[\int_0^{\mu}\d l\frac{l^3G(l^2)}{G^2(l^2)l^2-m^2M^2(l^2)}\frac{l^2}{\mu^2} + \mu^2\int_{\mu}^{+\infty}\d l\frac{lG(l^2)}{G^2(l^2)l^2-m^2M^2(l^2)}\right] \\
 & S_2 = -\frac{2a}{\mu^2}\left[\int_0^{\mu}\d l\frac{l^3M(l^2)}{G^2(l^2)l^2-m^2M^2(l^2)} + \mu^2\int_{\mu}^{+\infty}\d l\frac{lM(l^2)}{G^2(l^2)l^2-m^2M^2(l^2)}\right].
\end{align*}
Plugging those two counterterms into the Schwinger-Dyson equations we obtain 
\begin{align*}
    G(p^2) & = 1 - a\int_{\mu}^{p}\d l\frac{lG(l^2)}{G^2(l^2)l^2-m^2M^2(l^2)} \\
 \llcorner & + \frac{a}{p^2}\int_0^p\d l\frac{l^3G(l^2)}{G^2(l^2)l^2-m^2M^2(l^2)}\frac{l^2}{p^2} - \frac{a}{\mu^2}\int_0^{\mu}\d l\frac{l^3G(l^2)}{G^2(l^2)l^2-m^2M^2(l^2)}\frac{l^2}{\mu^2} \nonumber \\
    M(p^2) & = 1 - 2a\int_{\mu}^p\d l\frac{lM(l^2)}{G^2(l^2)l^2-m^2M^2(l^2)} \\ 
 \llcorner & + 2a\int_0^p\frac{lM(l^2)}{G^(l^2)l^2-m^2M^2(l^2)}\frac{l^2}{p^2} - 2a\int_0^{\mu}\d l\frac{lM(l^2)}{G^2(l^2)l^2 - m^2M^2(l^2)}\frac{l^2}{\mu^2}. \nonumber
\end{align*}
Performing the change of variable $y=l^2$ we can write the Schwinger--Dyson equations with $x:=p^2$
\begin{subequations}
 \begin{align}
  & G(x) = 1 - \frac{a}{2}\int_{\mu^2}^x\d y\frac{G(y)}{G^2(y)y-m^2M^2(y)}+F_m(x)-F_m(\mu^2) \label{eqG}\\
  & M(x) = 1 - a\int_{\mu^2}^x\d y\frac{M(y)}{G^2(y)y-m^2M^2(y)}+E_m(x)-E_m(\mu^2),
\end{align}
\end{subequations}
with:
\begin{subequations}
 \begin{align}
  & F_m(x) = \frac{a}{2}\int_0^x\d y\frac{G(y)}{G^2(y)y-m^2M^2(y)}\left(\frac{y}{x}\right)^2 \\
  & E_m(x) = a\int_0^x\d y\frac{M(y)}{G^2(y)y-m^2M^2(y)}\frac{y}{x}.
\end{align}
\end{subequations}
Taking a derivative with respect to $x$ in \eqref{eqG} brings
\begin{equation*}
 x^3\frac{\d G}{\d x} = - a\int_0^x\d y\frac{y^2G(y)}{G^2(y)y-m^2M^2(y)}.
\end{equation*}
Writing $D=x\frac{\d}{\d x}$ and taking another derivative with respect to $x$ brings
\begin{equation*}
 2xDG(x) + x^2\frac{\d}{\d x}DG(x)  = -\frac{ax^2G(x)}{G^2(x)x-m^2M^2(x)}.
\end{equation*}
For $M(x)$, a first derivative with respect to $x$ leads to
\begin{equation*}
 x^2\frac{\d M}{\d x} = -\frac{a}{x^2}\int_0^x\d y\frac{yM(y)}{G^2(y)y-m^2M^2(y)}.
\end{equation*}
Then, a second derivative allows to write the Schwinger--Dyson equation \eqref{SDlinnSUSY} as a system of two coupled ordinary differential equations:
\begin{subequations}
 \begin{align}
  & \left[G^2(x)x - m^2M^2(x)\right]D(D+2)G(x) = -axG(x) \label{SDEnSUSY2a} \\
  & \left[G^2(x)x - m^2M^2(x)\right]D(D+1)M(x) = -axM(x) \label{SDEnSUSY2b}
\end{align}
\end{subequations}
One can easily check that the massless limit $M(x)=0$ solves \eqref{SDEnSUSY2b} and brings \eqref{SDEnSUSY2a} back \eqref{SDE_Yuk_diff}. Let us notice however that the generalization of the 
massless case studied in \cite{BrKr99} is not trivial. In particular, the right-hand-side appears to be now dependent of the external momenta.

Now, when decoupling the equations (\ref{SDEnSUSY2a}-\ref{SDEnSUSY2b}), one ends up with complicated non-linear equations having a non-trivial denominator. It makes a rigorous analysis quite 
challenging, since we would have to take care that the denominator does not vanish. Instead, we will now rather tackle the equations (\ref{SDEnSUSY2a}-\ref{SDEnSUSY2b}) in the physically 
relevant cases of the ultraviolet and infrared limits.

\subsection{Solution in the ultraviolet limit}

In this limit the exterior impulsion is much higher than the mass of the fermion: $x>>m^2$. So, assuming that the functions $G(x)$ and $M(x)$ are regular the equations 
(\ref{SDEnSUSY2a}-\ref{SDEnSUSY2b}) become
\begin{subequations}
 \begin{align}
  & G(x)D(D+2)G(x) = -a \label{2a} \\
  & G^2(x)D(D+1)M(x) = -aM(x). \label{2b}
 \end{align}
\end{subequations}
The equation \eqref{2a} has been solved in \cite{BrKr99} and the previous section. It imposes to make the change of variables
 \begin{align*}
  & z = \left(\frac{x}{\mu^2}\right)^2 \\
  & \tilde{G} = \sqrt{\frac{2}{a}}zG(\mu^2\sqrt{z}).
\end{align*}
To tackle (\ref{2a}) let us define the analogous of $\tilde{G}$ for the mass function
\begin{equation}
 \tilde M(z) := M(\mu^2\sqrt{z}).
\end{equation}
Then \eqref{2b} reduces to
\begin{equation}
 \frac{\tilde G^2(z)}{z}\frac{\d}{\d z}\left(2z\frac{\d}{\d z}+1\right)\tilde M(z) = -\tilde M(z).
\end{equation}
Now, since we are working in the UV limit (i.e. in the limit $z>>1$), if we assume $\tilde M$ to be finite we have
\begin{equation*}
 \frac{\d}{\d z}\left(2z\frac{\d}{\d z}+1\right)\tilde M(z) = 2\tilde M'(z) + 2z\tilde M''(z)+\tilde M(z) \sim 2z\tilde M''(z).
\end{equation*}
Hence we have a relevant approximation of the \eqref{2b} in the UV limit:
\begin{equation} \label{approx_UV}
 2\tilde G^2(z)\tilde M''(z) = -\tilde M(z).
\end{equation}
Since we have an expression for $\tilde G$ in term of the parameter $p$, let write this equation as an differential equation for $h(p):=\tilde M(z(p))$. First we have to use
\begin{equation*}
 z = \frac{\text{erfc}(p)}{\text{erfc}(p_0)} \Leftrightarrow p = \text{erfc}^{-1}(\text{erfc}(p_0)z) \Rightarrow \frac{\d p}{\d z} = -\frac{\sqrt{\pi}}{2}e^{p^2}\text{erfc}(p_0).
\end{equation*}
Then, from the intermediate result
\begin{equation*}
 \tilde G = \sqrt{\frac{2}{a}}\exp(p_0^2-p^2)
\end{equation*}
we can rewrite \eqref{approx_UV} as
\begin{align*}
 & \frac{\pi}{a}\exp(2p_0-p^2)\text{erfc}(p_0)\frac{\d}{\d p}\left(e^{p^2}\text{erfc}(p_0)h'(p)\right) = -h(p) \\
\Leftrightarrow & \frac{1}{a}\left(\sqrt{\pi}e^{p_0^2}\text{erfc}(p_0)\right)^2\left[2ph'(p)+h''(p)\right] = -h(p).
\end{align*}
At this point, we recognize that the constant term in brackets is just $\tilde\alpha(p_0)=\sqrt{a/2}$, found in \eqref{sol_alpha}. Hence we end up with a very simple equation for $h$:
\begin{equation}
 h''(p)+2ph'(p)+2h(p) = 0.
\end{equation}
Its general solution is
\begin{equation}
 h(p) = e^{-p^2}\left[\lambda_1+\lambda_2\text{erfi}(p)\right].
\end{equation}
With erfi the imaginary error function defined by
\begin{equation*}
 \text{erfi}(x) = -i\text{erf}(ix).
\end{equation*}
The constants $\lambda_1$ and $\lambda_2$ can be determined from the initial condition $h(p_0)=1$ and another initial condition that we will now determine. We compute
\begin{equation*}
 h'(p) = \frac{\d z}{\d p}\frac{\d q^2}{\d z}\frac{\d M}{\d q^2} = \frac{\mu^2}{\sqrt{\pi}e^{p^2}\text{erfc}(p_0)}\frac{\d M}{\d q^2}.
\end{equation*}
Then using once again \eqref{sol_alpha}: $\tilde\alpha(p_0)=\sqrt{\pi}e^{p_0^2}\text{erfc}(p_0)=\sqrt{a/2}$ we end up with 
\begin{equation*}
 h'(p_0) = -\sqrt{\frac{2}{a}}\delta
\end{equation*}
with $\delta$ (the massive anomalous dimension) defined as
\begin{equation}
 \delta := q^2\left.\frac{\d M}{\d q^2}\right|_{q^2=\mu^2}.
\end{equation}
Then we can easily compute $\lambda_1$ and $\lambda_2$. Hence we obtain the solution for the massive Yukawa model in the ultraviolet limit:
\begin{subequations}
 \begin{align}
  & G(q^2) = \sqrt{\frac{a}{2\pi}}\frac{1}{\exp(p^2)\text{erfc}(p)} \\
  & M(q^2) = e^{-p^2}\left[e^{p_0^2} + \sqrt{\pi}\left(p_0-\frac{\delta}{\sqrt{2a}}\right)\left(\text{erfi}(p)-\text{erfi}(p_0)\right)\right] \\
  & q^2 = \mu^2\sqrt{\frac{\text{erfc(p)}}{\text{erfc}(p_0)}}.
 \end{align}
\end{subequations}
Now we can move on to the other limit: the infrared one.

\subsection{Solution in the infrared limit}

In this case we have $x<<m^2$. However, two possibilities have to be separated. Indeed either $a$, the coupling constant of the theory is big enough so the RHS of 
(\ref{SDEnSUSY2a}-\ref{SDEnSUSY2b}) is of the order of the LHS, either it is not. Let us start with the first case, that is:
\begin{equation} \label{hypa}
 a \sim \frac{m^2}{x}
\end{equation}
This case is called ``soft infrared'' since the external momenta is not small enough to have the coupling constant negligible. Then the equations (\ref{SDEnSUSY2a}-\ref{SDEnSUSY2b}) become:
\begin{subequations}
 \begin{align} 
  & m^2M(x)D(D+1)M(x) = ax  \label{16a} \\
  & m^2M^2(x)D(D+2)G(x) = axG(x) \label{16b}.
\end{align}
\end{subequations}
We will start by solving \eqref{16a} since there is only one unknown function in it. Using the reduced coupling constant $\tilde{a} := a\frac{\mu^2}{m^2}$ let us change the variables:
\begin{subequations}
 \begin{align}
  & z = \left(\frac{x}{\mu^2}\right)^2 \\
  & \tilde{M}(z) = \sqrt{\frac{z}{\tilde{a}}}M(\mu^2\sqrt{z}).
\end{align}
\end{subequations}
Then \eqref{16a} becomes:
\begin{equation} \label{eqMtilde}
 4\tilde{M}(z)\frac{\d}{\d z}\left[\sqrt{z}\frac{\d}{\d z}\tilde{M}(z)\right] = 1.
\end{equation}
The previous analysis can be done one more time to this equation, but a critical step will there be missing because of the $\sqrt{z}$ into the outer derivative: we cannot integrate the 
equation. Actually, there is no definition of $z$ and $\tilde{M}$ that would make the new version of \eqref{eqMtilde} integrable. Fortunately we are saved by noticing that it exists a simple 
solution to \eqref{eqMtilde}. Indeed, with $f(z)=\frac{2}{\sqrt{3}}z^{3/4}$ we have
\begin{equation*}
 \frac{\d}{\d z}\sqrt{z}\frac{\d}{\d z}f(z) =\frac{\sqrt{3}}{2}\frac{\d}{\d z}z^{1/4} = \frac{1}{4f(z)}.
\end{equation*}
Hence we have a particular solution
\begin{equation} \label{solM}
 \tilde{M}(z) = \frac{2}{\sqrt{3}}z^{3/4}.
\end{equation}
This solution satisfies the initial condition $\tilde{M}(1)=1/\sqrt{\tilde{a}}$ if, and only if, the reference scale is
\begin{equation}
 \frac{2}{\sqrt{3}}=\sqrt\frac{m^2}{a\mu^2} \Leftrightarrow \mu = \frac{m}{2}\sqrt{\frac{3}{a}}.
\end{equation}
This is coherent with the hypothesis \eqref{hypa} since the the impulsion $q$ should be of the same order than the impulsion of reference. Notice that, at this impulsion of reference, one gets 
$\tilde{a} = 3/4$. Now, plugging the solution \eqref{solM} into the equation \eqref{16b} one gets, for the function $\tilde{G}(z) = zG(\mu^2\sqrt{z})$, the very simple equation:
\begin{equation}
 \frac{16}{3}\tilde G''(z) =\frac{\tilde G(z)}{z^2}
\end{equation}
which has the simple solution
\begin{equation*}
 \tilde G(z) = Az^{1/2+\sqrt{7}/4} + Bz^{1/2-\sqrt{7}/4}.
\end{equation*}
We need to initial conditions to determine the coefficients $A$ and $B$. The first is obviously $\tilde G(1)=1$. and the second comes from the remark
\begin{equation*}
 \tilde G'(1) = \left.\frac{\d}{\d z}(zG)\right|_{z=1} = \left.G\right|_{q^2=\mu^2} + \left.z\frac{\d q^2}{\d z}\frac{\d G}{\d q^2}\right|_{q^2=\mu^2} = \left.1 + \frac{\mu^2}{2}\frac{\d G}{\d q^2}\right|_{q^2=\mu^2} = 1+\frac{\gamma}{2}.
\end{equation*}
Then we end up with 
\begin{align*}
 A & = \frac{1}{2} + \frac{1}{\sqrt{7}}(1+\gamma) \\
 B & = \frac{1}{2} - \frac{1}{\sqrt{7}}(1+\gamma).
\end{align*}
To summarize, we get a solution to the Schwinger-Dyson equation to the massive Yukawa model with a coupling constant of the order of $\frac{m^2}{q^2}$.
\begin{subequations}
 \begin{align}
  & M(q^2) = 2\sqrt{\frac{a}{3}}\frac{q}{m}  \\
  & G(q^2) = \left[\frac{1}{2} + \frac{1}{\sqrt{7}}(1+\gamma)\right]\left(\frac{q^2}{\mu^2}\right)^{\sqrt{7}/4-1/2} + \left[\frac{1}{2} - \frac{1}{\sqrt{7}}(1+\gamma)\right]\left(\frac{q^2}{\mu^2}\right)^{-\sqrt{7}/4-1/2}.
 \end{align}
\end{subequations}

Now, let us look at the case for $a$ is not big enough to cancel the fact that we are in the infrared regime. By opposition to the previous case, this one is called ``deep infrared''. This case 
is more interesting from a physical point of view since the small energies (at which $a<<\frac{m^2}{q^2}$) are easier to reach. Moreover, it is at those low energies that the coupling constant 
of QCD becomes too big to allow perturbative computations\footnote{i.e. of order of unit, which is still much lower than the ratio $m^2/q^2$.} making the Schwinger--Dyson equations of interest 
for physicists as a door to non-perturbative regimes. In this case, the RHS of equations (\ref{SDEnSUSY2a}-\ref{SDEnSUSY2b}) is negligible, so this system becomes:
\begin{align*}
  & m^2M^2(x)D(D+1)M(x) = 0 \\
  &  m^2M^2(x)D(D+2)G(x) = 0.
\end{align*}
Since we are not looking for vanishing solutions, this system actually decouples:
\begin{subequations}
 \begin{align}
 & D(D+1)M(x) = 0 \Leftrightarrow x^2M''(x) + 2xM'(x) = 0 \\
 & D(D+2)G(x) = 0 \Leftrightarrow x^2G''(x) + 3xG'(x) = 0
\end{align}
\end{subequations}
with now the prime being for a derivative with respect to $x=q^2$. Those equations are very easy to solve with the following initial conditions:
\begin{subequations}
 \begin{align}
  & M(\mu^2) = G(\mu^2) = 1 \\
  & q^2\frac{\d M(q^2)}{\d q^2}|_{q^2=\mu^2} = \delta \\
  & q^2\frac{\d G(q^2)}{\d q^2}|_{q^2=\mu^2} = \gamma.
\end{align}
\end{subequations}
Then one ends up with:
\begin{subequations}
 \begin{align}
  & M(q^2) = 1+\delta-\frac{\delta\mu^2}{q^2} \\
  & G(q^2) = 1+\frac{\gamma}{2}-\frac{\gamma}{2}\frac{\mu^4}{q^4}
 \end{align}
\end{subequations}
This solution being obviously for every scale of reference $\mu$.

We have unraveled some interesting results however we did not find any exact solution to the Schwinger--Dyson equation \eqref{SDlinnSUSY}. This was due to the mixing between the mass and 
wavefunction renormalization that forbids to perform some steps of the computation (namely the integration). We will see that in a supersymmetric model this restriction does not exist.

\section{Massive linear Wess--Zumino model}

We will work with a massive version of the Wess--Zumino-like model already studied in \cite{BeLoSc07}. This model has two superfields, one massive, $\Psi_i$, and one massless, $\Phi_{ij}$ 
($i,j=1,2,...,N)$. Each superfield represent a complex scalar ($A_i$ or $B_{ij}$), a Weyl fermion ($\chi_i$ or $\xi_{ij}$) and a complex auxiliary field ($F_i$ or $G_{ij}$):
\begin{align*}
    \Psi_i(y) & = A_i(y) + \sqrt{2}\theta\chi_i(y) + \theta\theta F_i(y) \\
 \Phi_{ij}(y) & = B_{ij}(y) + \sqrt{2}\theta\xi_{ij}(y) + \theta\theta G_{ij}(y).
\end{align*}
With $y$ the chiral coordinates defined by $y^{\mu} = x^{\mu}+i\theta\sigma^{\mu}\bar{\theta}$. This model has a cubic superfield interaction Lagrangian
\begin{align}
 L_{\text{int}} & = \frac{g}{\sqrt{N}}\sum_{i,j=1}^N\int \d^2\theta\Psi_i\Phi_{ij}\Psi_j \nonumber \\
   & = \frac{g}{\sqrt{N}}\sum_{i,j=1}^N(A_iG_{ij}A_j + 2A_iB_{ij}F_j - \chi_iB_{ij}\chi_j - \chi_i\xi_{ij}A_j - A_i\xi_{ij}\chi_j + \text{h.c.}). \label{L_SUSY}
\end{align}
A more detailed presentation of this model can be found in \cite{BeLoSc07}. We have interests in this model due to the non-renormalization theorems, which imply that we need only wavefunction 
renormalization. The simplest proof of non-renormalization theorems is usually credited to Seiberg \cite{Se93}, but a nice introduction of the subject can be found in \cite{De89}. 
Together with the supersymmetry of the theory, the non-renormalization theorem allows us to reduce the system of Schwinger--Dyson equations to only one differential equation.

\subsection{The Schwinger--Dyson equation}

From the lagrangian \eqref{L_SUSY} we have a system of Schwinger--Dyson equations. Graphically, it can be written as
\begin{subequations}
 \begin{eqnarray}
  \left(
\tikz \node[prop]{} child[grow=east]{edge from parent [photon]} child[grow=west]{edge from parent [photon]};
\right)^{-1} & = & 1 - a \;\;
\begin{tikzpicture}
\draw[photon] (-1.2,0)--(-0.8,0);
\draw[photon] (0.8,0)--(1.2,0);
\draw[dashed] (-0.8,0) .. controls (-0.8,1) and (0.8,1) .. (0.8,0);
\node at (0,0) [circle,minimum size=6mm,draw,fill=green!30] {} child [grow=east,dashed] child[grow=west,dashed];
\end{tikzpicture} \\
  \left(
\tikz \node[prop]{} child[grow=east] child[grow=west];
\right)^{-1} & = & 1 - a \;\;
\begin{tikzpicture}
\draw (-1.2,0)--(-0.8,0);
\draw (0.8,0)--(1.2,0);
\draw[dashed] (-0.8,0) .. controls (-0.8,1) and (0.8,1) .. (0.8,0);
\node at (0,0) [circle,minimum size=6mm,draw,fill=green!30] {} child [grow=east] child[grow=west];
\end{tikzpicture}
-a\;\;
\begin{tikzpicture}
\draw (-1.2,0)--(-0.8,0);
\draw (0.8,0)--(1.2,0);
\draw (-0.8,0) .. controls (-0.8,1) and (0.8,1) .. (0.8,0);
\node at (0,0) [circle,minimum size=6mm,draw,fill=green!30] {} child [grow=east,dashed] child[grow=west,dashed];
\end{tikzpicture} \\
  \left(
\tikz \node[prop]{} child[grow=east,dashed] child[grow=west,dashed];
\right)^{-1} & = & 1 - a \;\;
\begin{tikzpicture}
\draw[dashed] (-1.2,0)--(-0.8,0);
\draw[dashed] (0.8,0)--(1.2,0);
\draw[photon] (-0.8,0) .. controls (-0.8,1) and (0.8,1) .. (0.8,0);
\node at (0,0) [circle,minimum size=6mm,draw,fill=green!30] {} child [grow=east,dashed] child[grow=west,dashed];
\end{tikzpicture}
-a\;\;
\begin{tikzpicture}
\draw[dashed] (-1.2,0)--(-0.8,0);
\draw[dashed] (0.8,0)--(1.2,0);
\draw (-0.8,0) .. controls (-0.8,1) and (0.8,1) .. (0.8,0);
\node at (0,0) [circle,minimum size=6mm,draw,fill=green!30] {} child [grow=east] child[grow=west];
\end{tikzpicture} \\
   \llcorner & - & a \;\;
\begin{tikzpicture}
\draw[dashed] (-1.2,0)--(-0.8,0);
\draw[dashed] (0.8,0)--(1.2,0);
\draw[dashed] (-0.8,0) .. controls (-0.8,1) and (0.8,1) .. (0.8,0);
\node at (0,0) [circle,minimum size=6mm,draw,fill=green!30] {} child [grow=east]{edge from parent [photon]} child[grow=west]{edge from parent [photon]};
\end{tikzpicture} \nonumber
 \end{eqnarray}
\end{subequations}
with the plain lines being for the fermionic fields, the dashed lines for the scalar fields and the windy lines for the auxiliary fields. In the large $N$ limit the one-loop contributions to 
the dressed propagators are the only one to not be suppressed. This is why we consider a model with a vector superfield and a matrix one: a solution of the above system is more than a 
solution of a truncated Schwinger--Dyson equation; it is the full dressed propagators of the theory in the large $N$ limit. Now, the non-renormalization theorem allows us to write the dressed 
propagators as:
\begin{subequations}
 \begin{eqnarray}
 \Pi_{\chi}^{-1}(q) & = & \frac{q^2G_{\chi}(q^2)+m^2}{q_m\sigma^m} \\
    \Pi_{A}^{-1}(q) & = & q^2G_{A}(q^2)+m^2 \\
    \Pi_{F}^{-1}(q) & = & \frac{q^2G_{F}(q^2)+m^2}{q^2}
\end{eqnarray}
\end{subequations}
with the $\{\sigma^m\}$ the Pauli matrices and where we have dropped the subscript $i$ for simplicity. Now, we can write the above system of Schwinger--Dyson equations. Due to the length of the computations, we will not go into the full
details of the derivation. After some simplifications, we end up to
\begin{subequations}
 \begin{align}
 q^2G_{\chi}(q^2) & = q^2 - \frac{g^2}{4\pi^4}\int\d^4p\frac{q.p}{[p^2G_{\chi}(p^2)+m^2](q-p)^2} - \frac{g^2}{4\pi^4}\int\d^4p\frac{q^2-q.p}{[p^2G_{A}(p^2)+m^2](q-p)^2}  \label{SDE_SUSY1} \\
         G_F(q^2) & = 1 - \frac{g^2}{4\pi^4}\int\d^4p\frac{1}{[p^2G_{A}(p^2)+m^2](q-p)^2} \label{SDE_SUSY2} \\
      q^2G_A(q^2) & = q^2 - \frac{g^2}{4\pi^4}\int\d^4p\frac{1}{p^2G_{\chi}(p^2)+m^2} - \frac{g^2}{4\pi^4}\int\d^4p\frac{p^2}{[p^2G_{F}(p^2)+m^2](q-p)^2} \label{SDE_SUSY3} \\
        \llcorner & - \frac{g^2}{4\pi^4}\int\d^4p\frac{\text{Tr}(p_m\sigma^m(q_n-p_n)\sigma^n)}{[p^2G_{\chi}(p^2)+m^2](q-p)^2}. \nonumber
 \end{align}
\end{subequations}
We did not write explicitly the counter-terms in this system in order to keep it of reasonable size. Now, as already noticed in \cite{BeLoSc07}, a coherent ansatz to solve this system is
\begin{equation}
 G_{\chi}(q^2) = G_F(q^2) = G_A(q^2) = G(q^2).
\end{equation}
Indeed, with this ansatz, the first integral of \eqref{SDE_SUSY1} cancels the $q.p$ term of the second integral and therefore \eqref{SDE_SUSY1}$\Leftrightarrow$\eqref{SDE_SUSY2}. Moreover, using 
Tr$(\sigma^m\sigma^n) = 2\eta^{mn}$ and $-2q.p=(q-p)^2-p^2-q^2$ we end up with \eqref{SDE_SUSY3}$\Leftrightarrow$\eqref{SDE_SUSY2}. This fact is obviously a consequence of supersymmetry. 
Finally, within this ansatz, we only have one equation to solve as advertised
\begin{equation}
 G(p^2) = 1 - \frac{2a}{\pi^2}\int\d^4l\frac{1}{\left[G(l^2)l^2+m^2\right](l+p)^2} + S
\end{equation}
with $S$ a counter-term and, again, $a=\frac{g^2}{2\pi^2}$ the fine-structure constant of the theory. The Feynman integral above can be computed by performing the angular integral \eqref{angint} 
(with $l$ switched to $-l$) as in the Yukawa model. We end up with
\begin{equation*}
 G(p^2) = 1 - \frac{2a}{p^2}\left(\int_0^{p}\frac{l^3\d l}{G(l^2)l^2+m^2} + p^2\int_{p}^{+\infty}\frac{l\d l}{G(l^2)l^2+m^2}\right) + S.
\end{equation*}
The counterterm $S$ is fixed by the initial condition $G(\mu^2)=1$. This provides:
\begin{equation*}
 S = \frac{2a}{\mu^2}\left(\int_0^{\mu}\frac{l^3\d l}{G(l^2)l^2+m^2} + \mu^2\int_{\mu}^{+\infty}\frac{l\d l}{G(l^2)l^2+m^2}\right).
\end{equation*}
Thus the Schwinger-Dyson equation simply becomes:
\begin{equation}
 G(p^2) = 1-2a\left(\frac{1}{p^2}\int_0^p\frac{l^3\d l}{G(l^2)l^2+m^2} - \frac{1}{\mu^2}\int_0^{\mu}\frac{l^3\d l}{G(l^2)l^2+m^2} - \int_{\mu}^p\frac{l\d l}{G(l^2)l^2+m^2}\right).
\end{equation}
Let us rewrite the previous equation with $x=p^2$ and $y=l^2$
\begin{equation} \label{SDE1}
 G(x) = 1 + a\int_{\mu^2}^x\frac{\d y}{G(y)y + m^2} + F(\mu^2) - F(x)
\end{equation}
with
\begin{equation}
 F(x) = a\int_0^x\frac{\d y}{G(y)y+m^2}\frac{y}{x}.
\end{equation}
Taking two derivatives with respect to $x$ of (\ref{SDE1}) it comes:
\begin{equation} \label{SDE2}
 \left[G(x) + \frac{m^2}{x}\right]\text{D}(\text{D}+1)G(x) = a
\end{equation}
with once again $D=x\frac{\d}{\d x}$. Now, we will see how, following the footsteps of \cite{BrKr99}, we can solve this equation.

\subsection{Parametric solution}

As in the previous subsection, we will not write all the details of the computation since they are very similar to 
those detailed in section \ref{Yuk_nm}. First, we switch to dimensionless variables, and define the new parameters and functions:
\begin{subequations}
 \begin{align}
  & z = \frac{x}{\mu^2} \\
  & \tilde{m}^2 = \frac{m^2}{\mu^2}\frac{1}{\sqrt{2a}} \\
  & \tilde{G}(z) = \frac{1}{\sqrt{2a}}zG(\mu^2z).
\end{align}
\end{subequations}
Then our Schwinger--Dyson equation \eqref{SDE2} becomes
\begin{equation}
 2\left[\tilde{G}(z)+\tilde{m}^2\right]\tilde{G}''(z) = 1. \label{eqGtildeSUSY}
\end{equation}
Let us notice than $\tilde{m}$ is still constant with respect to $z$ since we take $\mu^2$ being fixed. Then the above equation can easily be integrated to:
\begin{equation*}
 \left(\tilde{G}'(z)\right)^2 = \ln\left(\tilde{G}(z)+\tilde{m}^2\right)+\text{cste}.
\end{equation*}
Now, let us use the parameter $p=\tilde{G}'(z)$. We can write $\tilde{G}(z)$ as a function of $p$:
\begin{equation} \label{formeG}
 \tilde{G}(z) = \left(\tilde{m}^2+\frac{1}{\sqrt{2a}}\right)\exp(p^2-p_0^2)-\tilde{m}^2
\end{equation}
with $p_0:=\tilde{G}'(1)$ and where we have used the initial condition $\tilde{G}(1)=\frac{1}{\sqrt{2a}} \Leftrightarrow G(\mu^2)=1$. We can then determine the parametric solution by 
studying:
\begin{equation} \label{alpha}
 \alpha = \frac{z}{2(\tilde{G}+\tilde{m}^2)}.
\end{equation}
Indeed, using the definition of $p$ and the equation \eqref{formeG} we obtain a differential equation for $\alpha$
\begin{equation}
 \alpha = \frac{1}{2p}-\frac{1}{2p}\frac{\d\alpha}{\d p}.
\end{equation}
We could explicitly solve this equation but the determination of the constant of integration would be less easy than in the case of the massless Yukawa model. Indeed, in the case of the massive 
linear Wess--Zumino model we have $\alpha\longrightarrow0$ as $p$ goes to infinity. Therefore we will follow more closely the approach of \cite{BrKr99} an develop the solution at infinity. It 
gives the asymptotic expansion:
\begin{equation} \label{expansionAlpha}
 \alpha(p) \simeq \frac{1}{2p}+\frac{1}{2p}\sum_{n=1}^{+\infty}\frac{(2n-1)!!}{(2p^2)^n}.
\end{equation}
Now, using the definition \eqref{alpha} of $\alpha$ and the one of the anomalous dimension
\begin{equation*}
 \gamma = q^2\frac{\d G(q^2)}{\d q^2}|_{q^2=\mu^2},
\end{equation*}
we get the relations:
\begin{subequations}
 \begin{align} 
  & \alpha(p_0) = \sqrt{\frac{a}{2}}\frac{1}{\frac{m^2}{\mu^2}+1} \label{refInita} \\
  & p_0 = \frac{1}{\sqrt{2a}}(1+\gamma) \label{refInitb}.
\end{align}
\end{subequations}
Such relations are the equivalent of the relations \eqref{p0_Yuk_nm} and $\tilde\alpha(p_0)$ found in \cite{BrKr99} and in the section \ref{Yuk_nm}. The second of the above relations comes from
\begin{equation*}
 p_0:=\frac{1}{\sqrt{2a}}\frac{\d}{\d z}(zG)|_{z=1}.
\end{equation*}
Then, we easily get the parametric solution by noticing that the expansion \eqref{expansionAlpha} is the one of:
\begin{eqnarray*}
 \alpha(p) & = & \frac{1}{2}\sqrt{\pi}e^{-p^2}\Re[\text{erfi}(p)] \\
           & = & \frac{1}{2}\left(\sqrt{\pi}e^{-p^2}\text{erfi}(p) + \text{i}\sqrt{\pi}e^{-p^2}\right) \\
           & = & \frac{\text{i}}{2}\sqrt{\pi}e^{-p^2}\text{erfc}(\text{i}p)
\end{eqnarray*}
With erfc the complementary error function and erfi the imaginary error function. In the following, we will simply write $\E(p):=\Re[\text{erfi}(p)]$ to emphasize that it is a real function. 
Hence we have:
\begin{equation} \label{formAlpha}
 \alpha(p) = \frac{1}{2}\sqrt{\pi}e^{-p^2}\E(p).
\end{equation}
Now, we have everything to write the parametric solution to the equation \eqref{SDE2}. First, we need to write $z$ as a function of $p$ (since we will write $G$ as a function of $p$). Using its 
definition and the form \eqref{formeG} of $\tilde{G}$ we get
\begin{equation*}
 1 = 2\exp(p^2-p_0^2)\left(\tilde{m}+\frac{1}{\sqrt{2a}}\right)\frac{\d p}{\d z}.
\end{equation*}
As in the non-supersymmetric case, using the equation \eqref{eqGtildeSUSY} and the definition \eqref{alpha} of $\alpha$ we obtained $\frac{\d p}{\d z}=\frac{\alpha}{z}$. Hence, using the above 
formula \eqref{formAlpha}, we end up with
\begin{equation*}
 z = \sqrt{\pi}\E(p)\left(\tilde{m}+\frac{1}{\sqrt{2a}}\right)\exp(-p_0^2).
\end{equation*}
This could be written on a much simpler form by using the relation \eqref{refInita} together with \eqref{formAlpha}. Then one obtains for $q^2$ the form written below, in the  complete solution. 
Finally, writing $G(q^2)=\sqrt{2a}\left(\frac{1}{2\alpha}-\frac{\tilde{m}^2}{z}\right)$ we get the parametric solution to the Schwinger-Dyson equation of our massive supersymmetric model:
\begin{subequations}
 \begin{align}
  & G(q^2) = \sqrt{\frac{2a}{\pi}}\frac{e^{p^2}}{\E(p)}-\frac{m^2}{\mu^2}\frac{\E(p_0)}{\E(p)} \\
  & q^2 = \mu^2\frac{\E(p)}{\E(p_0)}.
\end{align}
\end{subequations}
To our knowledge, this is the second (the first being in \cite{BrKr99}) known exact solution of a Schwinger--Dyson equation, and the first for a theory with a mass term. \\

This section was an interesting first ride in the land of Schwinger--Dyson equations. Let us now turn our attention to the one that lies at the heart of the work done during this Ph.D. and that 
will keep us busy for the two next chapters: the Schwinger--Dyson equation of the non-linear massless Wess-Zumino model.

%
%
\chapter{The massless Wess--Zumino model I: the physical plane} \label{chap3}

 \noindent\hrulefill \\
 {\it
Qu'as-tu fait, ô toi que voilà \\
Pleurant sans cesse, \\
Dis, qu'as-tu fait, toi que voilà, \\
De ta jeunesse ? \\

Paul Verlaine. Le ciel est par-dessus le toit.}

 \noindent\hrulefill
 
 \vspace{1.5cm}

This chapter, or at least its first section, is somehow the direct continuation of the first one. As in the previous chapter, we will deal with the evaluated correlation functions. Therefore, we will drop the $\phi_L^R$ for the 
sake of readability, and $\phi_L^R(G_2)$ will simply be written $G_2(L)$. Moreover, let us recall that $L=\ln(p^2/\mu^2)$.

\section{The equations}

\subsection{The renormalization group equation}

We will start with the equation \eqref{RGE_G2_gen}, that we will recall here:
\begin{equation*} 
 \partial_LG_2 = (\gamma+\beta a\partial_a)G_2
\end{equation*}
with $\gamma:=\partial_LG_2|_{L=0}$ and $\beta=\partial_La|_{L=0}$. $\gamma$ and $\beta$ are obviously functions of $a$, the fine structure constant of the theory. Now, we are working with the 
massless Wess--Zumino model. In this model, we have a non-renormalization theorem of vertex. This is essentially due to the fact that to any graph with $l_f$ fermionic loops and $l_b$ bosonic 
loops, we will have a graph with bosons and fermions exchanged. Due to supersymmetry, those two graphs will have the same evaluation with opposite signs if the number of external legs is odd and 
therefore their contributions will cancel each other in this case. Hence we have $\Gamma_3=1$ and $a_{\text{eff}}=(\Gamma_2)^{-3}=(G_2)^3$. Then \eqref{RGE_G2_prim} becomes
\begin{equation}
 \partial_LG_2 = \sum_{l\geq0}\gamma(1+3l)G_{2,l}.
\end{equation}
Using once again the remark \eqref{l_deriv_a} that the multiplication by the loop number could be seen as a derivative with respect to $a$ we get
\begin{equation} \label{RGE_G2_WZ}
 \partial_LG_2 = \gamma(1+3a\partial_a)G_2.
\end{equation}
This will be the standard renormalization group equation often used through this work. Moreover, from a simple comparison with \eqref{RGE_G2_gen} we get a straightforward proof of the classical 
result
\begin{equation}
 \beta = 3\gamma
\end{equation}
that is tediously shown by superspace techniques in \cite{Piguet}, among others.

Now, let us expand $G:=G_2$ in powers of $L=\ln(p^2/\mu^2)$:
\begin{equation*}
 G(L) = \sum_{k=0}^{+\infty}\gamma_k\frac{L^k}{k!},
\end{equation*}
with, a priori, all the $\gamma_k$ being functions of $a$. Since $G(0)=1$ (this will be clear from the Schwinger--Dyson equation) and from the definition of $\gamma$ we have $\gamma_0=1$ and 
$\gamma_1=\gamma$. Hence 
\begin{equation} \label{2ptB}
 G(L) = 1+\sum_{k=1}^{+\infty}\gamma_k\frac{L^k}{k!}.
\end{equation}
From this we straightforwardly obtain
\begin{align*}
 \partial_L G(L) & = \sum_{k\geq0}\gamma_{k+1}\frac{L^k}{k!} \\
 \gamma G(L) & = \gamma\sum_{k\geq1}\gamma_{k}\frac{L^k}{k!} \\
 \gamma a\partial_a G(L) & = \gamma\sum_{k\geq1}a\partial_a(\gamma_{k})\frac{L^k}{k!} .
\end{align*}
Hence the renormalization group equation \eqref{RGE_G2_WZ} can written for the $\gamma_k$s by equalizing the coefficients of the same powers of $L$ in both sides:
\begin{equation} \label{recursion_gamma}
 \gamma_{k+1} = \gamma_1(1+ 3a\partial_a)\gamma_k.
\end{equation}
In this model, all the coefficients of the expansion of $G(L)$ are therefore simple functions of the first coefficient. It will be enough to know $\gamma$ to fully know the two-points correlation 
function.

\subsection{The Schwinger--Dyson equation} \label{methods}

We did not specify the model we are working with. It is a massless Wess--Zumino model with only one superfield (instead of two for the model studied in the previous chapter). This superfield 
represents a complex scalar $\phi=A+iB$, a complex auxiliary field $\mathcal{F}=F-iG$ and a Majorana fermion $\psi$. The only interaction term that is invariant under the supersymmetry
transformations was shown in \cite{WeZu74a} to be
\begin{equation}
 \mathcal{L}_{Int} = g\left(F(A^2+B^2)+2GAB-i\bar\psi(A-\gamma_5B)\psi\right).
\end{equation}
This is the original model studied in \cite{WeZu74b}. For a detailed derivation of this lagrangian, the reader can be referred to \cite{Weinberg96c}, pages 55-59. A historical introduction to 
supersymmetry can be found in the same reference. This Wess--Zumino model is then presented pages 6-7.

The same procedure works in this case: all the components of the supermultiplet receive the same corrections. Therefore, as before, it is enough to solve a single equation (the one for the 
auxiliary field). That time, it is a non-linear equation (in the sense discussed in the header of the second chapter).
\begin{equation}\label{SDnlin}
\left(
\tikz \node[prop]{} child[grow=east] child[grow=west];
\right)^{-1} = 1 - a \;\;
\begin{tikzpicture}[level distance = 5mm, node distance= 10mm,baseline=(x.base)]
 \node (upnode) [style=prop]{};
 \node (downnode) [below of=upnode,style=prop]{}; 
 \draw (upnode) to[out=180,in=180]   
 	node[name=x,coordinate,midway] {} (downnode);
\draw	(x)	child[grow=west] ;
\draw (upnode) to[out=0,in=0] 
 	node[name=y,coordinate,midway] {} (downnode) ;
\draw	(y) child[grow=east]  ;
\end{tikzpicture}.
\end{equation}
We have indeed $G(0)=1$ as advertised in the previous subsection. Now, at a given loop order, the full propagator is the free propagator times the two-point function, which has a finite expansion 
in the logarithm of the momentum.
\begin{equation}
 P(p^2/\mu^2) = \frac{1}{p^2/\mu^2}\left(1+\sum_{k=1}^{+\infty}\gamma_k\frac{L^k}{k!}\right)
\end{equation}
Then the loop integral is an integral over a sum of powers of the logarithm of the external impulsion. The trick that allows to compute this integral is to take its Mellin transform. This is 
nothing but noticing:
\begin{equation}
 \left(\ln\frac{p^2}{\mu^2}\right)^k = \left.\left(\frac{\text{d}}{\text{dx}}\right)^k\left(\frac{p^2}{\mu^2}\right)^x\right|_{x=0}.
\end{equation}
Let us write this down carefully. We have to compute
\begin{align*}
 \text{I}(q^2) & = \frac{g^2}{8\pi^4}\int\frac{\d^4p}{p^2(q-p)^2}\left(1+\sum_{n=1}^{+\infty}\frac{\gamma_n}{n!}[\ln(p^2)]^n\right)\left(1+\sum_{m=1}^{+\infty}\frac{\gamma_m}{m!}[\ln((p-q)^2)]^m\right) \\
	       & = \frac{g^2}{8\pi^4}\int\frac{\d^4p}{p^2(q-p)^2}\left(1+\sum_{n=1}^{+\infty}\frac{\gamma_n}{n!}\left.\frac{\d^n}{\d x^n}(p^2)^x\right|_{x=0}\right)\left(1+\sum_{m=1}^{+\infty}\frac{\gamma_m}{m!}\left.\frac{\d^m}{\d x^m}((p-q)^2)^y\right|_{y=0}\right) \\
	       & = \frac{g^2}{8\pi^4}\left(1+\sum_{n=1}^{+\infty}\frac{\gamma_n}{n!}\frac{\d^n}{\d x^n}\right)\left(1+\sum_{m=1}^{+\infty}\frac{\gamma_m}{m!}\frac{\d^m}{\d y^m}\right)\int\d^4p\left.\frac{1}{\left(p^2\right)^{1-x}[(q-p)^2]^{1-y}}\right|_{x=0,y=0}
\end{align*}
where we have dropped all the $\mu^2$ in the integrands for the sake of readability. Hence we see that what we will need to compute is really
\begin{equation} \label{I_a_calculer}
 \text{I}(q^2/\mu^2,x,y) = \int\text{d}^4p\frac{1}{\left(p^2/\mu^2\right)^{1-x}[(q-p)^2/\mu^2]^{1-y}}.
\end{equation}
This will be an interesting task, and will be done in the next subsection. Hence the Schwinger--Dyson equation \eqref{SDnlin} can be written
\begin{equation}
 G(q^2) = 1 - \frac{g^2}{8\pi^4}\left(1+\sum_{n=1}^{+\infty}\frac{\gamma_n}{n!}\frac{\d^n}{\d x^n}\right)\left(1+\sum_{m=1}^{+\infty}\frac{\gamma_m}{m!}\frac{\d^m}{\d y^m}\right)\left.\text{I}(q^2/\mu^2,x,y)\right|_{x=0,y=0}.
\end{equation}
Now, the renormalization group equation \eqref{recursion_gamma} tells us that it is enough to know the value of the first derivative of $G(L)$ at $L=0$ to fully know $G$. Hence we will take a 
derivative with respect to $L$ and evaluate the result at $0$\footnote{let us recall that $L=0$ corresponds to $p^2=\mu^2$ and, if we do not write explicitly the $\mu$, to $p^2=1$.}. Following the footsteps of 
\cite{BeCl13} and \cite{BeCl14} we define the function
\begin{equation}
 H(x,y) = -\frac{1}{\pi^2}\left.\frac{\partial\text{I}(q^2/\mu^2,x,y)}{\partial L}\right|_{L=0}
\end{equation}
(the $1/\pi^2$ is there to kill a $\pi^2$ that will arise from the computation of I) and the differential operator
\begin{equation} \label{def_mathcal_I}
 \mathcal{I}\left(f(x,y)\right) = \left(1+\sum_{n=1}^{+\infty}\frac{\gamma_n}{n!}\frac{\text{d}^n}{\text{dx}^n}\right)\left(1+\sum_{m=1}^{+\infty}\frac{\gamma_m}{m!}\frac{\text{d}^m}{\text{dy}^m}\right)f(x,y)\bigg|_{x=y=0}.
\end{equation}
Then the Schwinger--Dyson equation \eqref{SDnlin} is simply
\begin{equation} \label{SDE}
 \gamma = a\mathcal{I}\left(H(x,y)\right)
\end{equation}
with $a=\frac{g^2}{8\pi^2}$ the fine structure constant of the theory.

\subsection{The one-loop Mellin transform} \label{one_loop}

This subsection will be devoted to the computation of the one-loop Mellin transform
\begin{equation}
 I_G(q^2,\beta_1,\beta_2) := \int\d^Dp\frac{1}{(p^2)^{\beta_1}((p-q)^2)^{\beta_2}}
\end{equation}
in $D=4$ space-time dimension. We can write this integral in the Schwinger representation
\begin{equation} \label{one_loop_Mellin}
 I_G(q^2,\beta_1,\beta_2) = \frac{\pi^{D/2}}{\Gamma(\beta_1)\Gamma(\beta_2)}\int\d\alpha_1\d\alpha_2(\alpha_1)^{\beta_1-1}(\alpha_2)^{\beta_2-1}U_G^{-D/2}\exp\left(-\frac{V_G}{U_G}\right)
\end{equation}
with $U_G$ and $V_G$ the first and second Symanzik polynomials of the graph $G$. Moreover, the integrations over the alpha parameters run from $0$ to $+\infty$. A detailed derivation of the 
Schwinger representation of Feynman integrals in the most general case is not so common in the literature, but a very careful one can be found in \cite{Go13}, appendix B.

Let us now define the Symanzik polynomials. We will follow the presentation of \cite{BoWe10}, but with the notation defined in the subsection \ref{graph_theory}. First, let us recall that a tree 
is a graph without loops, and that a forest is a graph that can be written has a disjoint union of trees (hence a forest is made of trees!). Moreover, if a forest $T$ can be written as the 
disjoint union of $n$ non-empty trees it is said to be a $n$-forest. 

Now, for $\gamma\subseteq\Gamma$\footnote{here $\gamma$ stands for a subgraph, it has nothing to do with the anomalous dimension.} we say that $\gamma$ is a spanning subgraph of $\gamma$ if, and 
only if, all the vertices of $\Gamma$ are vertices of $\gamma$. Hence we say that $T\subseteq\Gamma$ is a spanning tree of $\Gamma$ if, and only if, $T$ is tree and a spanning subgraph of 
$\Gamma$. Similarly, we define the spanning $n$-forests of $\Gamma$.

Now we have enough matter to define the Symanzik polynomials. Let $\mathcal{T}_{\Gamma}$ the set of all spanning trees of $\Gamma$, and $(e_i)_{i=1\dots|I_{\Gamma}|}$ a labeling of the 
internal edges of $\Gamma$. We attach the Schwinger parameter $\alpha_i$ to the edge $e_i$. Then the first Symanzik polynomial is
\begin{equation} \label{first_Symanzik}
 U_{\Gamma}(\{\alpha_i\}) := \sum_{T\in\mathcal{T}_{\Gamma}}\prod_{e_i\notin T}\alpha_i.
\end{equation}
Finally, let $P_T$ be the set of external momenta attached to the tree $T$ and $m_i$ be the mass of the particle pictured by the edge $e_i$. We will also write $\mathcal{F}^n_{\Gamma}$ the set of 
spanning $n$-forests of the graph $\Gamma$. Then the second Symanzik polynomial is
\begin{equation} \label{second_Symanzik}
 V_{\Gamma}(\{\alpha_i\},\{p_i\},\{m_i\}) = \sum_{(T_1,T_2)\in\mathcal{F}^2_{\Gamma}}\left(\prod_{e_i\notin (T_1,T_2)}\alpha_i\right)\left(\sum_{p_j\in P_{T_1}}\sum_{p_k\in P_{T_2}}\frac{p_j.p_k}{\mu^2}\right) + U_{\Gamma}(\{\alpha_i\})\sum_{i=1}^{|I(\Gamma)|}\alpha_i\frac{m^2_i}{\mu^2}
\end{equation}
In our case we are interested by the graph $G$
%
%
%
\begin{equation*}
 G = 
 \begin{tikzpicture}[anchor=base,baseline]
  \draw (0,0) circle (1);
  \draw (0,1) node[below]{\text{$\alpha_1$}};
  \draw (0,-1) node[above]{\text{$\alpha_2$}};
  \draw[*-] (0.91,0) -- (2,0);
  \draw[-*] (-2,0) -- (-0.91,0);
  \draw (-1.5,0) node[above]{\text{$p^2$}};
  \draw (1.5,0) node[above]{\text{$p^2$}};
 \end{tikzpicture}.
\end{equation*}
It is clear that it has only two spanning trees, each of them having only one edge: one of the two internal edges of $G$. Therefore
\begin{equation} \label{formula_U_G}
 U_G(\alpha_1\alpha_2) = \alpha_1+\alpha_2.
\end{equation}
Moreover, $G$ has only one spanning $2$-forest, made of the two trees with only one vertex and no edge. These two trees have both external impulsion $p$. Moreover, since we take $m_i=0$ we have
\begin{equation} \label{formula_V_G}
 V_G(\alpha_1,\alpha_2,q) = \alpha_1\alpha_2q^2
\end{equation}
if we drop the $\mu^2$, as usual. In order to compute \eqref{one_loop_Mellin} we insert $1$ in the integrand under the form
\begin{equation}
 1 = \int_0^{+\infty}\d x\delta(x-\sum_{i\in\mathcal{J}}\alpha_i)
\end{equation}
which is true for any non-empty subset $\mathcal{J}$  of the set labels of the internal edges of $G$ since the Schwinger parameters are positive. This is nothing but proving a very weak form of 
the famous theorem of Cheng and Wu (\cite{Cheng}, page 259). After an exchange in the order of the integrations we have
\begin{align}
 I_G(q^2,\beta_1,\beta_2) & = \\
\frac{\pi^{D/2}}{\Gamma(\beta_1)\Gamma(\beta_2)} & \int_0^{+\infty}\d x\int\d\alpha_1\d\alpha_2(\alpha_1)^{\beta_1-1}(\alpha_2)^{\beta_2-1}U_G^{-D/2}\exp\left(-\frac{V_G}{U_G}\right)\delta(x-\sum_{i\in\mathcal{J}}\alpha_i). \nonumber
\end{align}
Now, rescaling both of the Schwinger parameters by $x$, we have $U_G\longrightarrow xU_G$ and $V_G\longrightarrow x^2V_G$. Then 
\begin{align*}
 I_G(q^2,\beta_1,\beta_2) & = \\
\frac{\pi^{D/2}}{\Gamma(\beta_1)\Gamma(\beta_2)}\int_0^{+\infty} & \d xx^{\beta_1+\beta_2-2-D/2}\int\d\alpha_1\d\alpha_2(\alpha_1)^{\beta_1-1}(\alpha_2)^{\beta_2-1}U_G^{-D/2}\exp\left(-x\frac{V_G}{U_G}\right)\delta(x-x\sum_{i\in\mathcal{J}}\alpha_i). 
\end{align*}
Now we can see the delta function as acting on the Schwinger parameters. We use the well-known formula
\begin{equation*}
 \delta(f(X)) = \sum_{X_0|f(X_0)=0}\frac{\delta(X-X_0)}{|f'(X_0)|}
\end{equation*}
for a function $f$ having only simple zeros. In our case we take $f(X)=xX$ and $X=1-\sum_{i\in\mathcal{J}}\alpha_i$. Hence, after one more integration shift we arrive to
\begin{align*}
 I_G(q^2,\beta_1,\beta_2) & = \\
\frac{\pi^{D/2}}{\Gamma(\beta_1)\Gamma(\beta_2)}\int\d\alpha_1\d\alpha_2 & (\alpha_1)^{\beta_1-1}(\alpha_2)^{\beta_2-1}U_G^{-D/2}\delta(1-\sum_{i\in\mathcal{J}}\alpha_i)\int_0^{+\infty}\d xx^{\beta_1+\beta_2-D/2-1}\exp\left(-x\frac{V_G}{U_G}\right) 
\end{align*}
Now we can compute the integral over $x$ with one further change of variable: $X=xV_G/U_G$. Then the integral over $X$ is just the definition of Euler's $\Gamma$ function. Using the formula 
\eqref{formula_U_G} and \eqref{formula_V_G} for $U_G$ and $V_G$ and extracting the overall $q^2$ dependence of the integral we end up with 
\begin{equation} \label{I_G_interm}
 I_G(q^2,\beta_1,\beta_2) = (p^2)^{D/2-\beta_1-\beta_2}\frac{\Gamma(\beta_1+\beta_2-D/2)}{\Gamma(\beta_1)\Gamma(\beta_2)}\underbrace{\int\d\alpha_1\d\alpha_2\delta(1-\sum_{i\in\mathcal{J}}\alpha_i)\frac{(\alpha_1)^{D/2-\beta_2-1}(\alpha_2)^{D/2-\beta_1-1}}{(\alpha_1+\alpha_2)^{D-\beta_-\beta_2}}}_{:=I(\beta_1,\beta_2)}
\end{equation}
Choosing $\mathcal{J}=\{2\}$ we arrive to
\begin{align*}
 I(\beta_1,\beta_2) & = \int\d\alpha\frac{\alpha^{D/2-\beta_2-1}}{(\alpha+1)^{D-\beta_1-\beta_2}} \\
                    & = \int\d\alpha\frac{\alpha^{D/2-\beta_2-1}}{\Gamma(D-\beta_1-\beta_2)}\int_0^{+\infty}t^{D-\beta_1-\beta_2-1}e^{-t(\alpha+1)}\d t \\
		    & = \frac{1}{\Gamma(D-\beta_1-\beta_2)}\left(\int t^{D-\beta_1-\beta_2-1}e^{-t}\d t\right)\int\alpha^{D/2-\beta_2-1}e^{-\alpha_t}\d\alpha.
\end{align*}
We compute the last integral over $\alpha$ with one final change of variable $A=\alpha.t$ Then the integral over $A$ just gives a $\Gamma$ function and, after simplifications, we arrive to
\begin{align*}
 I(\beta_1,\beta_2) & = \frac{\Gamma(D/2-\beta_2)}{\Gamma(D-\beta_1-\beta_2)}\int t^{D/2-\beta_1-1}e^{-t}\d t \\
		    & = \frac{\Gamma(D/2-\beta_2)\Gamma(D/2-\beta_1)}{\Gamma(D-\beta_1-\beta_2)}
\end{align*}
which is symmetric under the exchange $\beta_1\leftrightarrow\beta_2$ as expected, although our choice of $\mathcal{J}$ has hidden this fact at some stage of the computation. Then, plugging 
$I(\beta_1,\beta_2)$ into \eqref{I_G_interm} we arrive to the famous result
\begin{equation}
 I_G(q^2,\beta_1,\beta_2) = (p^2)^{D/2-\beta_1-\beta_2}\pi^{D/2}\frac{\Gamma(D/2-\beta_2)\Gamma(D/2-\beta_1)\Gamma(D/2-\beta_0)}{\Gamma(\beta_1)\Gamma(\beta_2)\Gamma(\beta_0)}
\end{equation}
with $\beta_0:=D-\beta_1-\beta_2$. We can use this formula to explicitly write down the Schwinger--Dyson equation \eqref{SDnlin}. Evaluating the above integral at $\beta_1=1-x$, $\beta_2=1-y$ and 
$D=4$ we can compute the $I$ of \eqref{I_G_interm} to
\begin{equation*}
 \text{I}(q^2/\mu^2,x,y) = (q^2)^{x+y}\pi^{2}\frac{\Gamma(1+x)\Gamma(1+y)\Gamma(-x-y)}{\Gamma(1-x)\Gamma(1-y)\Gamma(2+x+y)}.
\end{equation*}
Then, using $(q^2)^{x+y}=e^{(x+y)L}$ and $-(x+y)\Gamma(-x-y)=\Gamma(1-x-y)$ we have the $H$ written in the Schwinger--Dyson equation \eqref{SDE}
\begin{equation} \label{def_H}
 H(x,y) = -(x+y)\frac{\Gamma(1+x)\Gamma(1+y)\Gamma(-x-y)}{\Gamma(1-x)\Gamma(1-y)\Gamma(2+x+y)} = \frac{\Gamma(1+x)\Gamma(1+y)\Gamma(1-x-y)}{\Gamma(1-x)\Gamma(1-y)\Gamma(2+x+y)}.
\end{equation}
Now, before studying the Schwinger--Dyson equation \eqref{SDE}, let us remark that we can write the function $H(x,y)$ as an exponential of a polynomial having odd Riemann's zeta as coefficients. 
Indeed, using the famous formula (which can be found, for example, in \cite{abramowitz})
\begin{equation*}
 \ln\Gamma(z+1) = -\gamma z+\sum_{k=2}^{+\infty}\frac{(-1)^k}{k}\zeta(k)z^k
\end{equation*}
where $\gamma$ is the Euler--Mascheroni constant, it is easy (but quite lengthy, thus we will not do it here) to show that
\begin{equation} \label{Hzeta}
 H(x,y) = \frac{1}{1+x+y}\exp\Bigl(2\sum_{k=1}^{+\infty}\frac{\zeta(2k+1)}{2k+1}\left((x+y)^{2k+1}-x^{2k+1}-y^{2k+1}\right)\Bigr).
\end{equation}
The Mellin transform was already written in this form in \cite{BeSc08}, so one expects only odd zeta values in the result. This will be a key coherency check in the next steps of our computation.

\section{Asymptotic solution}

\subsection{Contributions from individual poles} \label{contribution}

Solving brutally the equation \eqref{SDE} together with the renormalization group equation \eqref{recursion_gamma} seems hopeless due to the complicated structure of the one-loop Mellin transform 
$H(x,y)$. However, if we look carefully at the formula \eqref{def_H} we see that it has poles at $x;y=-k$, $k\in\mathbb{N}^*$ (we call these poles the simple ones) and at the lines $x+y=+k$, 
$k\in\mathbb{N}$ (the general poles). The simple poles are linked with the IR divergences of the loop integral while the general poles come from its UV divergences. Both kinds of pole
arise when, in the Mellin transform, a subgraph becomes scale invariant, as noticed in \cite{BeSc12}.

Let us justify these statements. First, the integral \eqref{I_a_calculer} can be written as integral over a projective space. Then, when a subgraph becomes scale invariant, the integrand of this 
projective integral is constant over a non-compact subspace, making the total integral divergent. About the link between the poles of the Mellin transform and IR/UV divergences, we just have to take $p<<q$ in 
\eqref{I_a_calculer}. Then
\begin{equation*}
 I \sim \int\frac{\d^4p}{(p^2)^{1-x}} \sim \int_0^{\epsilon}p^{1+2x}\d p
\end{equation*}
which is indeed divergent for $x=-k\in-\mathbb{N}^*$. For the UV divergences, let us take $p>>q$ in\eqref{I_a_calculer}. Then
\begin{equation*}
 I \sim \int\frac{\d^4p}{(p^2)^{2-x-y}} \sim \int_{\Lambda}^{+\infty}p^{2x+2y-1}\d p
\end{equation*}
which is indeed divergent if $x+y\in\mathbb{N}$. From this analysis, we also see that the first singularity ($x+y=0$) will be different from the others since it is the only logarithmic 
divergence. It will indeed call for a different treatment in the Borel approach of this problem, that will be detailed in the fourth chapter.

Now,  the simple poles can be expanded as:
\begin{equation}
 \frac{1}{k+x} = \frac{1}{k}\sum_{i=0}^{+\infty}\left(-\frac{x}{k}\right)^i,
\end{equation}
so that by \eqref{SDE}, the contribution of such a pole of the Mellin transform to the anomalous dimension transform is:
\begin{align} 
 F_k & := \left.\left(1+\sum_{n=1}^{+\infty}\frac{\gamma_n}{n!}\frac{\text{d}^n}{\text{d}x^n}\right)\frac{1}{k}\sum_{i=0}^{+\infty}\left(-\frac{x}{k}\right)^i\right|_{x=0} \nonumber \\
     & = \frac{1}{k}\left.\left(1+\sum_{n=1}^{+\infty}\sum_{i=1}^{+\infty}\frac{\gamma_n}{n!}\frac{\text{d}^n}{\text{d}x^n}\left(-\frac{x}{k}\right)^i\right|_{x=0}\right) \nonumber \\
     & = \frac{1}{k}\sum_{n=0}^{+\infty}\left(-\frac{1}{k}\right)^n\gamma_n \label{form_F}
\end{align}
with the  convention $\gamma_0=1$. Now we can use the recurrence relation \eqref{recursion_gamma} between \(\gamma_n\) and \(\gamma_{n+1}\):
\begin{align*}
 \gamma(1 + 3a\partial_a) F_k & = \gamma/k + \frac{1}{k}\sum_{n=1}^{+\infty}\left(-\frac{1}{k}\right)^n\gamma_{n+1} \\
			      & = \gamma/k - \sum_{n=2}^{+\infty}\left(-\frac{1}{k}\right)^n\gamma_{n} \\
			      & = \gamma/k -(F_k-1+\gamma/k).
\end{align*}
Hence we obtain a nice equation for $F_k$
\begin{equation} \label{equa_F}
 \gamma(1 + 3a\partial_a) F_k = -k F_k + 1.
\end{equation}

For the other poles, the situation is much more subtle. The idea is given in \cite{Be10a} but basically the numerators, i.e., the residues of $H(x,y)$ at these poles, must be taken in 
account from the start. Let us call $Q_k(x,y)$ the residue of $H(x,y)$  at $x+y=k$. Then the numerator at this pole is:
\begin{equation}
 N_k(\partial_{L_1},\partial_{L_2}) = Q_k(\partial_{L_1},\partial_{L_2}).
\end{equation}
The contribution $L_k$ of this pole to the anomalous dimension is then, according to \eqref{SDE}
\begin{equation} \label{def_Lk}
 L_k:=\sum_{n,m=0}^{+\infty}\frac{\gamma_n\gamma_m}{n!m!}\left.\frac{\text{d}^n}{\text{d}x^n}\frac{\text{d}^m}{\text{d}y^m}\frac{Q_k(x,y)}{k-x-y}\right|_{x=0,y=0}.
\end{equation}
This looks quite complicated. However, we are saved by noticing that $\gamma_n/n!=\partial_L^nG(L)/n!$ with $\partial_{L}$ the operator that differentiates with respect to $L$ and evaluates the 
result to zero. We define the same operator for $x$, written $\partial_x$. Then, we can use that, for two functions analytical in a neighborhood of the origin,
\begin{equation}
 f(\partial_L)g(L) = \sum_{n=0}^{+\infty}\frac{\partial_x^nf(x)}{n!}\partial_L^ng(L).
\end{equation}
This formula is a direct consequence of the Taylor expansion of $f$, and the function of a differential operator has to be read as a (Taylor) series of the differential operator. Now we can use 
twice this formula to get
\begin{align*}
 L_k & = \sum_{n=0}^{+\infty}\partial_{L_1}^nG(L_1)\frac{\partial_x^n}{n!}\frac{Q_k(x,\partial_{L_2})}{k-x-\partial_{L_2}}G(L_2) \\
     & = \frac{Q_k(\partial_{L_1},\partial_{L_2})}{k-\partial_{L_1}-\partial_{L_2}}G(L_1)G(L_2).
\end{align*}
Therefore, evaluating
\begin{equation*}
 \left(\frac{\partial}{\partial L_1}+\frac{\partial}{\partial L_2}\right)f(L_1)g(L_2) = g(L_2)\frac{\partial f(L_1)}{\partial L_1} + f(L_1)\frac{\partial g(L_2)}{\partial L_2}
\end{equation*}
at $L_1=L_2=L$ we deduce
\begin{equation*}
 \left.\left(\frac{\partial}{\partial L_1}+\frac{\partial}{\partial L_2}\right)f(L_1)g(L_2)\right|_{L_1=L_2=L} = \frac{\partial}{\partial L}(f(L)G(L)).
\end{equation*}
We can use this since we are indeed evaluating the R.H.S. of \eqref{def_Lk} at $L_1=L_2(=0)$. Hence we have
\begin{equation*}
 \frac{1}{k-\partial_{L_1}-\partial_{L_2}}G(L_1)G(L_2) = \frac{1}{k}\sum_{n=0}^{+\infty}\frac{1}{k^n}\partial_L^nG(L)^2.
\end{equation*}
Now, we can finally use the renormalization group equation \eqref{RGE_G2_WZ}:
\begin{align*}
 \partial_LG(L)^2 & = 2G(L)\partial_LG(L) \\
		  & = 2G(L)\gamma(1+3a\partial_a)G(L)|_{L=0} \\
		  & = \gamma(2+3a\partial_a)G(L)^2|_{L=0}.
\end{align*}
Thus, by induction, 
\begin{equation*}
 \frac{1}{k-\partial_{L_1}-\partial_{L_2}}G(L_1)G(L_2) = \frac{1}{k}\sum_{n=0}^{+\infty}\frac{1}{k^n}[\gamma(2+3a\partial_a)]^nG(L)^2|_{L=0} = \frac{1}{k-\gamma(2+3a\partial_a)}G(L_1)G(L_2)|_{L_1=L_2=0}.
\end{equation*}
Moving the operator $k-\gamma(2+3a\partial_a)$ on the L.H.S. of \eqref{def_Lk} we obtain a nice renormalization group-like equation for $L_k$:
\begin{equation} \label{equa_H}
 (k-2\gamma - 3\gamma a\partial_a)L_k = N_k(\partial_{L_1},\partial_{L_2})G(L_1)G(L_2)|_{L_1=L_2=0}.
\end{equation}
This was the last subtle part of this section. For the following computations, we will not write everything down explicitly, in order to keep the size of this chapter within a reasonable size.

\subsection{Derivation of the asymptotic solution}

The idea of \cite{BeSc08} is to approximate the function $H(x,y)$ by its first singularities. The equation \eqref{SDE} was then studied numerically. The same idea allows the author of \cite{Be10a} 
to analytically derive the asymptotic behavior of $\gamma$. We will here present the method of this last article, since it was generalized in \cite{BeCl13}, which is the goal of this chapter.

Now, let us approximate the function $H(x,y)$ by its first poles at $x=-1$, $y=-1$ and $x+y=+1$. This gives the following approximating function to use in \eqref{SDE}:
\begin{equation} \label{appr}
 h(x,y) = (1+xy)\left(\frac{1}{1+x} + \frac{1}{1+y}-1\right) + \frac{1}{2}\frac{xy}{1-x-y} + \frac{1}{2}xy.
\end{equation}
This means that we only use the contributions $F\equiv F_1$ of the poles $1/(1+x)$ and $1/(1+y)$ and $L\equiv L_1$ of the pole $xy/(1-x-y)$ to compute \(\gamma\). Then the renormalisation group 
equation \eqref{equa_F} gives
\begin{equation*}
  F = 1 - \gamma(3a\partial_a+1)F 
\end{equation*}
while from \eqref{equa_H} and $\partial_{L_1}\partial_{L_2}G(L_1)G(L_2)=\gamma^2$ we deduce
\begin{equation*}
 L = \gamma^2 + \gamma(3a\partial_a+2)L.
\end{equation*}
For the Schwinger--Dyson equation, we have to be slightly more careful. Using the differential operator $\mathcal{I}$ defined in \eqref{def_mathcal_I} we have
\begin{align*}
 & \mathcal{I}\left(\frac{1}{1+x} + \frac{1}{1+y}-1\right) = 2F-1 \\
 & \mathcal{I}\left(\frac{xy}{1-x-y}\right) = L \\
 & \mathcal{I}(xy) = \gamma^2.
\end{align*}
The most intricate case is obviously $\mathcal{I}\left(\frac{xy}{1+x}\right)$. Let us already write the most general formula 
\begin{equation} \label{mathcal_I_useful}
 \mathcal{I}\left(\frac{(xy)^n}{k+x}\right) = (-k)^n\gamma_n\left[F_k-\frac{1}{k}\sum_{i=0}^{n-1}\left(-\frac{1}{k}\right)^i\gamma_i\right]
\end{equation}
which can be directly read from
\begin{equation*}
 \frac{(xy)^n}{k+x} = (-k)^ny^n\left[\frac{1}{x+k}-\frac{1}{k}\sum_{i=0}^{n-1}\left(-\frac{x}{k}\right)^i\right].
\end{equation*}
Hence, all in all, from the approximation \eqref{appr}, the renormalization group equation and the Schwinger--Dyson equation reduce to a system of three coupled non-linear differential 
equations:
\begin{subequations}
 \begin{align}
  & F = 1 - \gamma(3a\partial_a+1)F , \label{SDE_simples_a} \\
  & L = \gamma^2 + \gamma(3a\partial_a+2)L  ,\\
  & \gamma = 2a F -a -2a\gamma( F-1) + \frac{1}{2}a(L-\gamma^2). \label{SDE_simples_c}.
 \end{align}
\end{subequations}
We look for a perturbative solution of these equations, and expand $F$, $L$ and $\gamma$ in powers of $a$: $F=\sum f_na^n$, $L=\sum l_na^n$ and $\gamma=\sum c_na^n$. Then we easily get the first 
coefficients:
\begin{align*}
 c_0=&0 \qquad c_1=1 \qquad c_2=-2 \\
 &f_0=1 \qquad f_1=-1 \\
 &~l_0=0 \qquad l_1=1.
\end{align*}
We will make the assumption that the $\{f_n\}$, the $\{l_n\}$ and the $\{c_n\}$ have a fast growth and keep only the dominant contributions. Then the equation for $\gamma$ is simply
\begin{equation*}
 c_{n+1} \simeq 2f_n + \frac{1}{2}l_n.
\end{equation*}
Hence we see that at least one of the two sequences $\{f_n\}$ and $\{l_n\}$ has a growth faster than the one of $\{c_n\}$. At order zero we have 
\begin{align*}
 f_{n+1} & \simeq -3nf_n \\
 l_{n+1} & \simeq 3nl_n. 
\end{align*}
Thus we see that the two series have a very different behavior and therefore do not talk to each other. So, at order one we have
$f_{n+1} \simeq -3nf_n + 6(n-1)f_{n-1} - f_n -c_{n+1} \simeq -(3n+5)f_n$ and $l_{n+1} \simeq 3nl_n - 6(n-1)l_{n-1} + 2l_n \simeq 3nl_n$. From this it is clear that the dominant contribution to 
$c_{n+1}$ is $2f_n$. Moreover, since $l_0=l_1=0$, we know even from the zero order result that we will observe that $l_n$ will be at most of the order of $f_{n-2}$, if there is no too special 
transient regime (this being quite a reasonnable assumption). Hence we find $c_{n+1}\simeq 2f_n \simeq-2(3(n-1)+5)f_{n-1}\simeq-(3n+2)c_n$. To summarize, we have the asymptotic behavior of the solution of the Schwinger--Dyson equation 
\eqref{SDE}
\begin{subequations}
\begin{align}
  & f_{n+1} \simeq -(3n+5) f_n, \\
  & l_{n+1} \simeq 3n l_n ,\\
  & c_{n+1} \simeq -(3n+2) c_n. 
\end{align}
\end{subequations}
And we will now turn our attention to the procedure detailed in \cite{BeCl13}, and see how to go further than this asymptotic behavior.

\section{Corrections to the asymptotics}

\subsection{Change of variables} \label{FirstOrder}

There is a priori no reason that forbids us to compute the $1/n$ corrections to the asymptotic solution above. However, two technical details make this analysis intractable in practice. First, 
we will have to include contributions from the others poles of $H(x,y)$. Moreover, the subdominant contributions are quite convoluted and is will be much less simple to find the next terms in the 
recursions. In practice, the $1/n$ order is doable, but the followings are very painful. If we want to do precise computations, this problem cries for another method.

Such a method was devised in \cite{BeCl13}: to simplify our calculations, we separate the alternating contributions to $(c_n)$ from the ones with a constant sign. We will define two symbols for 
this, one to encode the asymptotic behavior coming from $l_n$ and the other one for the asymptotic behavior coming from $f_n$. They will be defined through two series:
\begin{subequations}
 \begin{align}
  & A_{n+1} = -(3n+5)A_n \label{eq_serie_symbolesa} \\
  & B_{n+1} = 3nB_n \label{eq_serie_symbolesb}
 \end{align}
\end{subequations}
but we will only use the following two symbols corresponding to the formal series:
\begin{subequations}
 \begin{align}
  & A = \sum A_n a^n \\
  & B = \sum B_n a^n.
\end{align}
\end{subequations}
The relations \eqref{eq_serie_symbolesa}-\eqref{eq_serie_symbolesb} can be expressed as differential equations on the symbols $A$ and $B$:
\begin{subequations}
 \begin{align}
  & 3a^2\partial_a A = -A-5aA \label{relation_symbolesa} \\
  & 3a^2\partial_a B = B \label{relation_symbolesb}
 \end{align}
\end{subequations}
In fact, \eqref{eq_serie_symbolesa} and \eqref{eq_serie_symbolesb} do not entirely determine $A$ and $B$. We must have an initial condition at an order $n_0$. Then 
\eqref{relation_symbolesa}-\eqref{relation_symbolesb} are true up to a term of degree $n_0$. Therefore, $n_0$ shall be chosen higher than the degree to which we will compute the gamma function. 

Before using the relations \eqref{relation_symbolesa}-\eqref{relation_symbolesb}, let us derive them explicitly.
\begin{align*}
 3a^2\partial_a A & = 3a^2\sum_{n\geq n_0}nA_na^{n-1} \\
		  & = a\sum_{n\geq n_0}(-A_{n+1}-5A_n)a^n \qquad \text{thanks to \eqref{eq_serie_symbolesa}} \\
		  & = -(A-A_{n_0}a^{n_0}) - 5aA.
\end{align*}
Similarly we have for $B$
\begin{equation*}
 3a^2\partial_a B = \sum_{n\geq n_0}B_{n+1}a^{n+1} = B - B_{n_0}a^{n_0}.
\end{equation*}
In the Borel plane approach that will be presented in the next chapter, discarding the $a^{n_0}$ terms will be understood as neglecting the analytic parts of the Borel transform of $\gamma$ near 
one of its singularities, and we will see that it is actually the correct thing to do.

Although these relations are not totally exact they will drastically reduce the complexity of the computation of the corrections to the asymptotic behavior, allowing us to go up to the fifth order,
with computations of the same degree of complexity that those of the second order with the previous method. To do that, we will define nine unknown functions, which are the  
coefficients of \(A\) and \(B\) for the $F$, $L$ and $\gamma$ functions.
\begin{subequations}
 \begin{align}
  & F = f +Ag+Bh \label{redefinitiona} \\
  & L = l+Am+Bn \\
  & \gamma = a(c+Ad+Be) \label{redefinitionc}
 \end{align}
\end{subequations}
The $a$ in the ansatz for $\gamma$ comes from the $a$ in front of the function $H(x,y)$ in \eqref{SDE}. In what follows, $c=c(a)$ will be called the low order part of $\gamma$.

Now, one can rewrite the Schwinger--Dyson equations \eqref{SDE_simples_a}-\eqref{SDE_simples_c} in the language of this new set of functions, with the derivatives of the $A$ and $B$ symbols
being removed, thanks to the relation \eqref{relation_symbolesa} and \eqref{relation_symbolesb}. The new equations are found by saying that the coefficient of the symbol~$A$ and the one of the 
symbol~$B$ shall independently vanish, since the two divergences of the Mellin transform are of different nature and thereof do not talk to each other. Similarly, the term without any symbol 
should also independently vanish. Hence, this change of variables allows to efficiently separate the alternating part and the part of constant sign of the Mellin transform. It will drastically 
simplify the equations to solve since we will have three equations for each of the previous ones, with the overall system being of the same complexity than the one with the old formalism. 

Although separating the $A$ and $B$ terms sounds quite natural, we are not yet able to prove that it is the right thing to do. Once again, we will understand that in the Borel plane approach of 
the Wess--Zumino model. Then, separating the $A$ and $B$ terms will be understood as working in the neighborhood of different singularities of the Borel transform of $\gamma$.

For the time being, with this procedure, we end up with nine equations:
\begin{subequations}
\begin{align}
 & f+c(3a^2\partial_a+a)f=1 \label{eq_debut} \\
 & g+d(3a^2\partial_a+a)f+c(-1-4a+3a^2\partial_a)g=0 \label{eq_g}\\
 & h+e(3a^2\partial_a+a)f+c(1+a+3a^2\partial_a)h=0 \\
 l& =a^2c^2+ac(3a\partial_a+2)l \\
 m& =2a^2dc+c(3a^2\partial_a-1-3a)m+ad(3a\partial_a+2)l \\
 n& =2a^2ec+c(3a^2\partial_a+1+2a)n+ae(3a\partial_a+2)l \label{eq_n}\\
 c& =2f-1-2ac(f-1)-\frac{1}{2}(l+a^2c^2) \\
 d& =2g+\frac{1}{2}am+2ad(1-f)-ac(2g+ad) \\
 e& =2h+\frac{1}{2}an+2ae(1-f)-ae(2f+ac) \label{eq_fin}
\end{align}
\end{subequations}
There are obviously more terms into the expansion of the equations \eqref{SDE_simples_a}-\eqref{SDE_simples_c}, proportional to  $A^2$, $AB$,~$B^2$, but they will not be considered: if \(A\) and 
\(B\) begin by a large number of vanishing coefficients, they correspond to corrections of very high order.

The perturbative solution of these equations goes as follows. We write:
\begin{equation}
 c(a)=\sum_{n=0}^{+\infty}c_na^n,
\end{equation}
and similarly for all the other functions. At each order, the equations should be solved in the right order: one shall first solve the equations for $f$ and $l$, then for $c$, then for $g$, $m$ 
and $h$, $n$, and finally for $d$ and $e$. Following this procedure, one ends up with the solution up to the order $a^1$.
\begin{align*}
  & f(a) = 1-a\\
  & g(a) = g_0+g_1a\\
  & h(a) = -\frac{1}{4}an_0\\
  & l(a) = 0\\
  & m(a) = 0\\
  & n(a) = n_0+n_1a\\
  & c(a) = 1-2a\\
  & d(a) = 2g_0+(-2g_0+2g_1)a\\
  & e(a) = \frac{1}{2}n_0+\frac{1}{2}(n_1-n_0)a  
\end{align*}
Here the assumption of fast growth of the series which was done in \cite{Be10a} is not necessary since the symbols $A$ and $B$ take care of the necessary properties.

The coefficients $g_1$ and $n_1$ are not specified at this stage. It is a general feature of this parametrization that one needs to go at the order $a^{p+1}$ to 
fix the parameters $g_p$ and $n_p$. Indeed, since $c_0=1$, the $a^n$ order in the equation \eqref{eq_g} is:
\begin{equation*}
 g_n + (...)-c_0g_n + (...) = 0.
\end{equation*}
therefore does not depend on $g_n$, with a similar phenomenon appearing in \eqref{eq_n}. However, the next order of equations \eqref{eq_g} and \eqref{eq_n} is not hard and does not involve 
higher coefficients of \(d\) or \(e\), so that we obtain the solution (up to the order $a^1$) to the equations \eqref{SDE_simples_a}-\eqref{SDE_simples_c} with only two unconstrained parameters. 
We however need the values of the next order for \(c\).
\begin{subequations}
 \begin{align}
  & f(a) = 1-a\\
  & g(a) = g_0\left(1+\frac{16}{3}a\right)\\
  & h(a) = -\frac{1}{4}an_0\\
  & l(a) = 0\\
  & m(a) = 0\\
  & n(a) = n_0\left(1-\frac{11}{3}a\right)\\
  & c(a) = 1-2a\\
  & d(a) = 2g_0\left(1+\frac{13}{3}a\right)\\
  & e(a) = \frac{1}{2}n_0\left(1-\frac{14}{3}a\right)
 \end{align}
\end{subequations}
The fact that there remain two unconstrained parameters, $g_0$ and $n_0$, is not really surprising since they  were already present in the former formalism, where the asymptotic behavior was 
inferred from the ratio of successive coefficients of the Taylor series. Since only ratios could be computed, the overall factors in the asymptotic behavior of the series for \(F\) and \(L\) 
are unconstrained. In this new formalism, equations stemming from the part linear in \(A\) are linear in the coefficients \(d\), \(g\) and \(m\) of \(A\)  in the unknown functions: if there is 
any non trivial solution, all its multiples are also solutions. Hence, although we used here the analysis of the previous subsection (which is from \cite{Be10a}), the right induction to use for 
the coefficients $A_n$ could have been found from the requirement of the existence of a non-trivial solution. This also justifies \`a posteriori the analysis of the previous subsection: it indeed 
gives the right induction relations, the ones that give non-trivial solutions to \eqref{SDE_simples_a}-\eqref{SDE_simples_c}. Analogous statements hold for the terms proportional to \(B\).

Up to now, we have only obtained the first two coefficients of each function, since the other poles of the $ H(x,y)$ function will contribute to the next terms. This is the subject of the next 
subsections.

\subsection{Taking care of every poles: a guideline}

To go further we must include the contributions of all the poles of the Mellin transform $H(x,y)$ to the $\gamma$ function. The equations for $F$ and $L$ do not change. Then we have to compute the 
residues of $H(x,y)$ at its various poles. They are not difficult to derive, especially since Euler's $\Gamma$ function has only simple poles at the negative integers, with residue $(-1)^k/k!$ at 
the pole in $z=-k$. However, their derivation is quite lengthy and is left as an exercise for the reader's Ph.D. student. The residues are given for $k\ge2$ by
\begin{align*}
 & \text{Res}(H,x=-k) = \frac{-y}{k-1}\prod_{i=1}^k\left(1-\frac{y}{i}\right)\prod_{i=1}^{k-2}\left(1-\frac{y}{i}\right) ,\\
 & \text{Res}(H,x=k-y) = \frac{(k-y)y}{k(k+1)}\prod_{i=1}^{k-1}\left(1-\frac{(k-y)y}{i(k-i)}\right),
\end{align*}
with the convention $\prod_{i=1}^{k-2}=1$ for $k=2$. In order to simplify the computations, we use the fact that the first polynomial is defined at $x=-k$ and the 
second at $x=k-y$ to make the numerators symmetric in $x$ and $y$. One gets:
\begin{align}
 & \text{Res}(H,x=-k) = \frac{xy}{k(k-1)}\prod_{i=1}^k\left(1+\frac{xy}{ki}\right)\prod_{i=1}^{k-2}\left(1+\frac{xy}{ki}\right) = P_k(xy) \label{residusPk}\\
 & \text{Res}(H,x=k-y) = \frac{xy}{k(k+1)}\prod_{i=1}^{k-1}\left(1-\frac{xy}{i(k-i)}\right) = Q_k(xy) \label{residusQk}.
\end{align}
The coefficients of those polynomials will be of interest. Hence we define them in the following way:
\begin{align*}
 & P_k(X) = \sum_{n=1}^{2k-1}p_{k,n}X^n \\
 & Q_k(X) = \sum_{n=1}^{k}q_{k,n}X^n.
\end{align*}
Notice that the residues at the poles in $x=-k$ and $y=-k$ are exactly the same since $H(x,y)$ is symmetric under the exchange of $x$ and~$y$. We therefore write:
\begin{eqnarray} \label{Mellin_tsfo_complete}
 H(x,y) &=& (1+xy)\left(\frac{1}{1+x} + \frac{1}{1+y}-1\right) + \frac{1}{2}\frac{xy}{1-x-y} \\
  &&+\sum_{k=2}^{+\infty}\left(\frac{1}{k+x} + \frac{1}{k+y} -\frac{1}{k}\right)P_k(xy) + \sum_{k=2}^{+\infty}\frac{Q_k(xy)}{k-x-y} + \tilde{H}(x,y). \nonumber
\end{eqnarray}
The \(-1/k\) term coming with the poles at \(x=-k\) and \(y=-k\) does not contribute to the singularities but appears necessary to obtain the exact Taylor expansion of \(H(x,y)\) around the 
origin. Moreover, $\tilde{H}(x,y)$ shall be a holomorphic function and is the difference between $H(x,y)$ such as written in (\ref{def_H}) and the above expansion of sums over the poles. We have 
checked that $\tilde{H}(x,y)$ shall be of degree at least 10. We have also verified that some infinite families of derivatives of $\tilde{H}(x,y)$ vanish at the origin. Hence we will make the 
conjecture $\tilde{H}(x,y)=0$ in the following, which is the raison d'\^etre of the $-1/k$. We are fairly confident that this conjecture is true, but it would be pleasant to have a proof of it, 
which would give a stronger ground to these computations.

To obtain the anomalous dimension of the theory at a given order $p$, one must include additional terms of the Schwinger--Dyson equations to deal with all the contributions at this order. Indeed, 
we have seen in subsection \ref{contribution} that the equations for $L_k$ and $\gamma$ depend on the residues of $H(x,y)$. Since those residues are polynomial, and because of the definition of 
the transformation $\mathcal{I}$, we can truncate those equations to take care only of the terms which will contribute to a given order. In other words: at a given order $p$, we will have to 
replace the function $H(x,y)$ by a more precise approximation than $h(x,y)$. 

Let us write $P_{k,p}(X)$ for the polynomial $P_k(X)$ truncated to a degree less or equal to \(p\), and  $Q_{k,p}(X)$ for the polynomial $Q_k(X)$ similarly truncated. Then the equation 
\eqref{equa_H} for \(L_k\) becomes, for the order $p$
\begin{equation}
  (k-2\gamma -3\gamma a \partial_a ) L_k=Q_{k,p-1}(\partial_{L_1}\partial_{L_2})G(L_1)G(L_2)|_{L_1=L_2=0},
\end{equation}
and the right hand side of the equation \eqref{SDE_simples_c} gets an additional part in its right hand side,
\begin{equation}
 a\sum_{k=2}^{+\infty}L_k +a\sum_{k=2}^{+\infty} {\I}\left(P_{k,p-1}(xy)\left[\frac{2}{k+x}-\frac{1}{k}\right]\right)
\end{equation}
The $2$ in the equation for $\gamma$ is there because for each $k$, $H(x,y)$ has a pole $x=-k$ and at $y=-k$ which give equal contributions.  
${\I}$ is the linear transform defined in subsection \ref{methods}, which reduced in this pole expansion to:
\begin{subequations}
 \begin{align}
 & {\I}\left(\frac{(xy)^n}{k+x}\right) = (-k)^n\gamma_n\left[F_k-\frac{1}{k}\sum_{i=0}^{n-1}\left(-\frac{1}{k}\right)^i\gamma_i\right], \\
 & {\I}\left((xy)^n\right) = \gamma_n^2
 \end{align}
\end{subequations}
Where we have used \eqref{mathcal_I_useful} once again.

With this definition of ${\I}$ we see that only the term $(xy)^{p-1}$ is needed for the solution at the order $a^{2p+1}$ of \(\gamma\) since the leading term of $a {\I}\left((xy)^p\right)$ which 
is $a\gamma_n^2$ is of order $2p+1$. However, when looking at the coefficients of the symbols \(A\) and~\(B\) in \(\gamma_n\), they still are proportional to \(a\). The term \( a(xy)^p \) will 
therefore contribute terms of order \(a^{p+2}\).

Last, but not least, the equation for $F_k$ is similar to the one for $F$, and is given in \eqref{equa_F}. This equation will never change, whatever the order one needs, simply because any new 
term which might affect $F_k$ will come through changes in~$\gamma$. The modifications to the $L_k$ functions instead come also from changes to the equation of $L_k$. 

\subsection{Solutions up to the fifth order}

To start our study of the effect of the infinitude of poles, we will take only the $(xy)$ contribution to each pole. Hence the Schwinger--Dyson equation comes from the following approximate Mellin 
transform:
\begin{eqnarray}
 \frac{1}{a}h(x,y) &=& (1+xy)\left(\frac{1}{1+x} + \frac{1}{1+y}-1\right) + \frac{1}{2}\frac{xy}{1-x-y} \\
 	&& +\sum_{k=2}^{+\infty}\left(\frac{1}{k+x} + \frac{1}{k+y} -\frac{1}{k}\right) \frac{xy}{k(k-1)} + \sum_{k=2}^{+\infty}\frac{1}{k-x-y}\frac{xy}{k(k+1)}. \nonumber
\end{eqnarray}
In the equation for $L_k$ we use only the linear term for the numerator in \eqref{equa_H}:
\begin{equation}
 [k-\gamma(2+3a\partial_a)]L_k = \frac{1}{k(k+1)}\partial_{L_1}\partial_{L_2}G(L_1)G(L_2)|_{L_1=L_2=0} =  \frac{1}{k(k+1)}\gamma^2
\end{equation}
For the Schwinger--Dyson equation, one has simply to apply \eqref{SDE}. Some series arise, which are easily computable.  So we end up with five coupled non-linear partial 
differential equations to solve,
\begin{subequations}
\begin{align}
 & F = 1 - \gamma(3a\partial_a+1)F  \label{equations1} \\
 & L = \gamma^2 + \gamma(3a\partial_a+2)L \\
 & kF_k = 1 - \gamma(1+3a\partial_a)F_k \\
 & kL_k = \frac{1}{k(k+1)}\gamma^2 + \gamma(2+3a\partial_a)L_k \\
 & \gamma = 2aF - a -2a\gamma(F-1) + \frac{1}{2}aL +2a\gamma -a\gamma^2[3-\zeta(2)] + a\sum_{k=2}^{+\infty}\left(L_k-2\gamma\frac{F_k}{k-1}\right) \label{equation5}
\end{align}
\end{subequations}
with $\zeta$ being Riemann's zeta function. One may be worried by the $\zeta(2)$ in the last equation, due to the remark made at the end of section \ref{one_loop}. However, the sum will give 
compensating terms and this $\zeta(2)$ will not appear any more in the result. This provides a check that the calculations are correct.

Now, as in Section \ref{FirstOrder}, we can define the functions $f_k$, $g_k$ and $h_k$, and $l_k$, $m_k$ and $n_k$ for the functions $F_k$ and $L_k$. Then the system of equations
\eqref{equations1}-\eqref{equation5} shall be rewritten for those functions. One ends up with fifteen coupled partial non-linear differential equations for fifteen 
functions, that we will not write down explicitly. 

Solving those equations should be done in the same order than in the section \ref{methods}, with the equations for $f_k$ and $l_k$ solved with the equations for
$f$ and $l$, and similarly those for $g_k$, $m_k$, $h_k$ and~$n_k$ with \(h\) and~\(m\). As in the previous case, the order two terms of $n(a)$ and~$g(a)$ are not fixed by the $a^2$ equations. However we only need the order three terms of the equations for \(g\) and \(n\) to fix this,  while the most tedious equations to 
solve at a given order are those for $d$ and~$e$ which involve sums over \(k\). Fixing those two last coefficients thus does not add much complexity. Moreover, we do not need to add more terms in the $\gamma$ equation, since we are looking for the equation on the coefficient $c$ of $\gamma$, and the higher order 
(such as the $(xy)^2$ term) will act on the $d$ and $e$ terms only, thanks to the relations \eqref{relation_symbolesa} and \eqref{relation_symbolesb}. The already computed orders $a^0$ and $a^1$ 
are unchanged by the addition of the new terms as expected and, all computations being done, we end up with the solution to the equations \eqref{equations1}-\eqref{equation5} up to the order 
$a^2$.
\begin{subequations}
\begin{align}
  & f(a) = 1-a+6a^2\\
  & g(a) = g_0\left(1+\frac{16}{3}a+\frac{2}{9}\left[-65+12\zeta(3)\right]a^2\right)\\
  & h(a) = n_0\left(-\frac{1}{4}a+\frac{29}{12}a^2\right)\\
  & l(a) = a^2\\
  & m(a) = 2a^2g_0\\
  & n(a) = n_0\left(1-\frac{11}{3}a+\frac{8}{9}\left[28-3\zeta(3)\right] a^2\right)\\
  & f_k(a) = \frac{1}{k}-\frac{1}{k^2}a+\frac{2(2+k)}{k^3}a^2 \\
  & g_k(a) = -\frac{2g_0}{k(k-1)}\left(a+\frac{12-28k+13k^2}{3k(k-1)}a^2\right) \\
  & h_k(a) = \frac{n_0}{k(k+1)}\left(-\frac{a}{2}+\frac{6+16k+7k^2}{3k(k+1)}a^2\right) \\
  & l_k(a) = \frac{a^2}{k^2(k+1)} \\
  & m_k(a) = \frac{4g_0}{k(k+1)^2}a^2 \\
  & n_k(a) = \frac{n_0}{(k-1)k(k+1)}a^2 \\
  & c(a) = 1-2a+14a^2\\
  & d(a) = 2g_0\left(1+\frac{13}{3}a+\frac{2}{9}\left[-71+12\zeta(3)\right]a^2\right)\\
  & e(a) = n_0\left(\frac{1}{2}-\frac{7}{3}a+\frac{1}{6}\left[\frac{329}{3}-8\zeta(3)\right]a^2\right)
 \end{align}
\end{subequations}
with the $n_2$ and $g_2$ being fixed by a computation at the $a^3$ order.\footnote{We use the same notations to denote the functions \(g_k(a)\), appearing as factors of \(A\) in \(F_k\) and the 
coefficients \(g_i\) of the function \(g(a)\), in order to keep a strong parallelism between the expansion of \(F\) and \(F_k\), and similarly for \(n_k(a)\) and~\(n_i\). We hope that the context 
make the two different usages clear.}
\begin{align*}
 & g_2 = \frac{2}{9}\left[-65+12\zeta(3)\right]g_0 \\
 & n_2 = \frac{8}{9}\left[28-3\zeta(3)\right]n_0
\end{align*}
So this order is a nice check of our procedure since the two first order are unchanged and $\zeta(2)$ disappears everywhere as expected. The next orders are not more difficult to reach, just more 
tedious: we will be quite sketchy.

For the $a^3$ terms in $c$, $d$ and $e$ (so the fourth order of $\gamma$) we first have to determine the coefficients of \(P_{k,2}(X)\), which appears in the equation for $\gamma$. It is simply:
\begin{equation}
 P_{k,2}(X) = \frac X {k(k-1)} +  \frac{X^2}{k^2(k-1)}\left(H_k+H_{k-2}\right).
\end{equation}
This is true for all values of $k$ with the convention that $H_k$, the $k^{th}$ harmonic number, is defined by $H_0=0$, $H_{k}=H_{k-1}+1/k$. Then the equation for $\gamma$ becomes:
\begin{align}
 \gamma & = 2aF - a -2a\gamma(F-1) + \frac{1}{2}aL  + a\sum_{k=2}^{+\infty}\left(L_k-2\gamma\frac{F_k}{k-1} + 2\gamma_2F_k\frac{H_k+H_{k-2}}{k-1}\right) +2a\gamma \nonumber \\
\llcorner &  -a\gamma^2[3-\zeta(2)] -6a \gamma_2+  2a\gamma_2\gamma\bigl[6-\zeta(2)-3\zeta(3)\bigr] - a\gamma_2^2 
	\bigl [10-2\zeta(2)-\textstyle{\frac{3}{2}}\zeta(4)-4\zeta(3)\bigr].
\end{align}
The only other equation to be changed is the one for $L_k$ which gets a new term
\begin{equation*}
 Q_{k,2}(X)=\frac1{k(k+1)}X -\frac{2H_{k-1}}{k^2(k+1)}X^2
\end{equation*}
and so the equation for $L_k$ is now:
\begin{equation}
 kL_k = \gamma(3a\partial_a+2)L_k+\frac{1}{k(k+1)} \gamma^2 -\frac{2H_{k-1}}{k^2(k+1)}\gamma_2^2.
\end{equation}
These equations (together with the three equations for the other functions) can now be solved at the third order. For the sake of readability, we will not write the fifteen functions at this 
order, but only the functions which are a part of $\gamma$. One ends up with a solution without any even zetas,
\begin{subequations}
\begin{align} \label{ordre3}
  & c(a) = 1-2a+14a^2+16\left[\zeta(3)-10\right]a^3 \\
  & d(a) = g_0\left(2+\frac{26}{3}a+\left[-\frac{284}{9}+\frac{16}{3}\zeta(3)\right]a^2+\frac{4}{9}\left[\frac{7873}{9}-134\zeta(3)\right]a^3\right) \\
  & e(a) = n_0\left(\frac{1}{2}-\frac{7}{3}a+\frac{1}{6}\left[\frac{329}{3}-8\zeta(3)\right]a^2+\frac{1}{9}\left[-\frac{33889}{18}+188\zeta(3)\right]a^3\right)
\end{align}
\end{subequations}
We have already included the  coefficients $g_3$ and $n_3$.
\begin{align*}
 & g_3 = g_0\frac{8}{9}\left[\frac{1687}{9}-8\zeta(3)\right] \\
 & n_3 = \frac{n_0}{81}\left[-22207+3168\zeta(3)\right]
\end{align*}
Again, the disappearance of every even zeta values from the final result is a very useful fact to detect mistakes in the computations.

For the $a^4$ order, we need to compute the coefficient of degree two in a product of linear terms. We use:
\begin{equation*}
 \prod_{i=1}^n\left(1+\alpha_iX\right) = 1+X\sum_{i=1}^n\alpha_i+\frac{X^2}{2}\left(\Bigl[\sum_{i=1}^n\alpha_i\Bigr]^2-\sum_{i=1}^n\alpha_i^2\right)+\mathcal{O}(X^3).
\end{equation*}
One finds easily the coefficient of the cubic term of the $P_k(X)$ polynomials \eqref{residusPk}
\begin{equation} \label{ordre3}
	p_{k,3} =  \frac{1}{k^3(k-1)}\biggl(H_kH_{k-2}+\sum_{1\leq i<j\leq k}\frac{1}{ij}+\sum_{1\leq i<j\leq k-2}\frac{1}{ij}\biggr)
\end{equation}
The last sum is not defined for $k=2$ and $k=3$ and we will  write those two cases separately, with the values
\begin{align*}
  p_{2,3} &= \frac{1}{16} \\ 
  p_{3,3} &= \frac{17}{324}
\end{align*}
This phenomenon of a general term undefined for the first coefficients will appear for the coefficient of any term $(xy)^n$, and we would have to deal with it at any further order.
\footnote{However, one might notice that we find the right values of the cubic terms of those two first polynomials if we simply set to zero the undefined term in (\ref{ordre3}). Since we don't 
have a proof of this effect being true at any order, it seemed to be simpler to separate the first terms off the others.}

Now, using
\begin{equation}
  \frac{(xy)^3}{k+x} = -k^3y^3\left(k\frac{1}{k+x}-1+\frac{x}{k}-\frac{x^2}{k^2}\right)  \longrightarrow -k^2\gamma_3\left(kF_k-1+\frac{\gamma}{k}-\frac{\gamma_2}{k^2}\right)
\end{equation}
we end up with the following equation for $\gamma$:
\begin{equation}
 \gamma=\text{[orders 0, 1, 2 and 3]}+a\gamma_3\sum_{k=2}^{+\infty}\left[-2k^3F_k + 2k^2 - 2k\gamma + 2\gamma_2 - \frac{\gamma_3}{k}\right]p_{k,3}+a\sum_{k=2}^{+\infty}L_k.
\end{equation}
We can use \eqref{ordre3} in there and the values of the \(p_{k,3}\) in this equation. Many series will arise, which could all be computed in terms of zetas, multizetas and rational numbers. These 
computations have some interesting features, justifying working them out. However, such a computation is complex and it is a better strategy to not separately sum each series, but rather to 
combine the generic terms of the series. We need the following expansion:
\begin{equation*}
 -2k^3F_k + 2k^2 - 2k\gamma + 2\gamma_2 - \frac{\gamma_3}{k} = S_k + T_kA + U_kB.
\end{equation*}
The series $S_k$, $T_k$ and $U_k$ are not so complicated anymore, since we only need their dominant terms:
\begin{subequations}
\begin{align}
  & S_k = \frac{28}{k}a^3+\mathcal{O}(a^4) \\ 
  & T_k = ag_0\left[4\frac{k^2}{k-1}-4k-4-\frac{2}{k}\right]+\mathcal{O}(a^2) = ag_0\left[\frac{4}{k-1}-\frac{2}{k}\right]+\mathcal{O}(a^2) \\
  & U_k = an_0\left[\frac{k^2}{k+1}-k+1-\frac{1}{2k}\right]+\mathcal{O}(a^2) = an_0\left[\frac{1}{k+1}-\frac{1}{2k}\right]+\mathcal{O}(a^2)
\end{align}
\end{subequations}
The higher order terms are not needed here since there is a $\gamma_3$ in front of those sums which starts at the $a^3$ order for its low order part and at the $a^1$  one for the other parts:
\begin{equation*}
 \gamma_3 = 28a^3+\mathcal{O}(a^4) + 2g_0A\left((a+\mathcal{O}(a^2)\right)+\frac{1}{2}n_0B\left(a+\mathcal{O}(a^2)\right).
\end{equation*}
These results only depend on the lowest order values of \(\gamma\) together with the renormalisation group equation \eqref{recursion_gamma} and the relations between the symbols $A$, $B$, 
and their derivatives \eqref{relation_symbolesa}-\eqref{relation_symbolesb}.

The contribution of the $(xy)^3$ term for $c_4$ vanishes since there is no $a^4$ term without $A$ and $B$. We are simply left with the following series:
\begin{align}
 R_1 & = a\sum_{k=2}^{+\infty}\left[2g_0S_k+28a^3T_k\right]p_{k,3} \nonumber \\
     & = 112g_0a^4\left[\frac{115}{1296}+\sum_{k=4}^{+\infty}\frac{p_{k,3}}{k-1}\right] \\
 R_2 & = a\sum_{k=2}^{+\infty}\left[\frac{1}{2}n_0S_k+28a^3U_k\right]p_{k,3}\nonumber \\
     & = 28n_0a^4\left[\frac{71}{2592}+\sum_{k=4}^{+\infty}\frac{p_{k,3}}{k+1}\right].
\end{align}
Hence we got $\gamma_3\sum_{k=2}^{+\infty}\left[-2k^3F_k + 2k^2 - 2k\gamma + 2\gamma_2 - \frac{\gamma_3}{k}\right]p_{k,3}=R_1A+R_2B+\mathcal{O}(a^5)$. Those sums are still not very simple, but 
much simpler than the ones we had before. The Schwinger--Dyson equation is now written in a very compact form:
\begin{equation} \label{SDE_fin}
 \gamma=\text{[orders 0, 1, 2 and 3]}+a(R_1A+R_2B)+a\sum_{k=2}^{+\infty}L_k
\end{equation}
One still has to add the $(xy)^3$ term into the equation of $L_k$, which depends on the coefficient cubic \(q_{k,3}\) of \(Q_k(X)\):
\begin{equation}
 q_{k,3}=\frac{1}{k^3(k+1)}\left(2H_{k-1}^2-H_{k-1,2}-2\frac{H_{k-1}}{k}\right).
\end{equation}
Here, $H_{k,n}$ denote the generalized harmonic numbers defined by $H_{0,n} = 0$ and $H_{k,n} = H_{k-1,n} + 1/k^n$.
Hence, at this order, the equation for $L_k$ becomes:
\begin{equation}
 kL_k = \gamma(3a\partial_a+2)L_k+\frac{1}{k(k+1)}\gamma^2-\frac{2H_{k-1}}{k^2(k+1)}\gamma_2^2
 +q_{k,3}\left(\gamma_3\right)^2.
\end{equation}
The equations for $F$, $L$ and $F_k$ are unchanged and we end up, after these simplifications, with a system of five coupled equations.
\begin{subequations}
\begin{align} \label{equations12}
 & F = 1 - \gamma(3a\partial_a+1)F \\
 & L = \gamma^2 + \gamma(3a\partial_a+2)L \\
 & kF_k = 1 - (1+3a\partial_a)F_k \\
 & kL_k = \gamma(3a\partial_a+2)L_k+\gamma^2\frac{1}{k(k+1)}-\frac{2H_{k-1}}{k^2(k+1)}\gamma_2^2
 	+q_{k,3}\left(\gamma_3\right)^2\\
 & \gamma=\text{[orders 0, 1, 2 and 3]}+a(R_1A+R_2B)+a\sum_{k=2}^{+\infty}L_k \label{equation52}
\end{align}
\end{subequations}
We can solve them in order to get the anomalous dimension of the massless Wess--Zumino model up to the fourth order. For the sake of readability, we will write only this fourth order:
\begin{subequations}
 \begin{align}
  & c_4 = 2444-328\zeta(3) \label{resultat41} \\
  & d_4 = 2g_4 + \frac{1}{81}g_0\left[-73720 +65952\zeta(3)\right] \\ 
  & e_4 = \frac{1}{2}n_4 + n_0\left[\frac{384227}{324}-\frac{974}{9}\zeta(3)\right] \label{resultat42}
 \end{align} 
\end{subequations}
One striking observation about our result \eqref{resultat41} - \eqref{resultat42} is that they contain only rational numbers and $\zeta(3)$, when the summation over \(k\) gives multizetas of 
weight 5 for \(d_4\) and~\(e_4\). However the highest weight terms cancel each others, so that the weight is not higher than the one for \(c_4\), where the \((xy)^3\) does not contribute and 
every sum is of weight smaller than 4.

Now, let us look at the final solution. The coefficients $g_4$ and $n_4$ could be fixed by going to the fifth order, thanks to the equations of $F$ and $L$ respectively.
\begin{subequations}
 \begin{align}
  & g_4 = g_0\left[-\frac{652516}{243} + \frac{5212}{27}\zeta(3) + 168\zeta(5) + \frac{32}{9}\zeta(3)^2\right] \\
  & n_4 = n_0\left[\frac{3741119}{1944} - \frac{8522}{27}\zeta(3) - 84\zeta(5) + \frac{16}{9}\zeta(3)^2\right]
 \end{align}
\end{subequations}
So we end up with the final values for the fourth order of the anomalous dimension of the massless Wess--Zumino model.
\begin{subequations}
 \begin{align} \label{sol_a4}
  & c_4 = 2444-328\zeta(3) \\
  & d_4 = -g_0\left[\frac{1526192}{243} +\frac{32408}{27}\zeta(3)+ 336\zeta(5) + \frac{64}{9}\zeta(3)^2\right] \\ 
  & e_4 = n_0\left[\frac{6046481}{1944}-\frac{11444}{27}\zeta(3)-84\zeta(5)+\frac{16}{9}\zeta(3)^2\right] \label{sol_a4_fin}
 \end{align} 
\end{subequations}
So $d_4$ and $e_4$ finally involve some $\zeta(5)$ and $\zeta(3)^2$. This highest weight terms stem from the order~5 in the equation for \(g(a)\) (resp.\ \(n(a)\)), which contains \(\zeta(5)\) 
through \(g_0 c_5\) (resp.\ \(n_0 c_5\)) and \(\zeta(3)^2\) through \(g_2 c_3\) (resp.\ \(n_2 c_3\)).  The similarity of their origin explains that this highest weight terms differ simply by a 
factor \(\pm4\). 

Now, let us go back to the weight of the zetas in the low order part. We want to prove that the weight in \(\zeta\) of~\(c_p\), which is the coefficient of \(a^{p+1}\) in~\(\gamma\), is \(p\) or 
less. If we suppose  this property true, the renormalisation group equation \eqref{recursion_gamma} allows to show that the coefficient of \(a^{n+p}\) in $\gamma_n$ is of maximal weight \(p\). 
From the expression \eqref{Hzeta} of $H(x,y)$, it is clear that the derivative of total order \(k\) of \(H(x,y)\) has maximal weight \(k\). The same upper bound on the weight can be deduced from 
its expression as a sum over the poles, but it is however highly non trivial in this case that only products of zeta values at odd integers appear. It then follows that for every term 
\(h_{n,m}\gamma_n\gamma_m\) in Eq. \eqref{SDE}, the terms of degree \(p+1\) in \(a\) is of maximal weight less or equal to \(p\). Our hypothesis on the weight of the zetas appearing in 
\(\gamma\) can therefore be proved by induction.

Now, for the parts of $\gamma$ proportional to $A$ and $B$ the reasoning made for the low order part does not hold. The highest weight terms in the sums over the poles cancel in the terms we have 
studied, so that we do not have the weights \(2p\) for the coefficient of \(a^p\). One can try to guess what happens in the following orders, but it highly technical (if not purely impossible) to 
prove anything with this technique. We will see that the Borel plane approach does not suffer from the same problem. Hence, before turning our attention the this topic, let us check that this 
analysis is coherent with the numerical results of \cite{BeSc08}.

\begin{figure}[h!]
\caption{$2g_0$ from different fits.} \label{fig_g0}
 \includegraphics[scale=0.5]{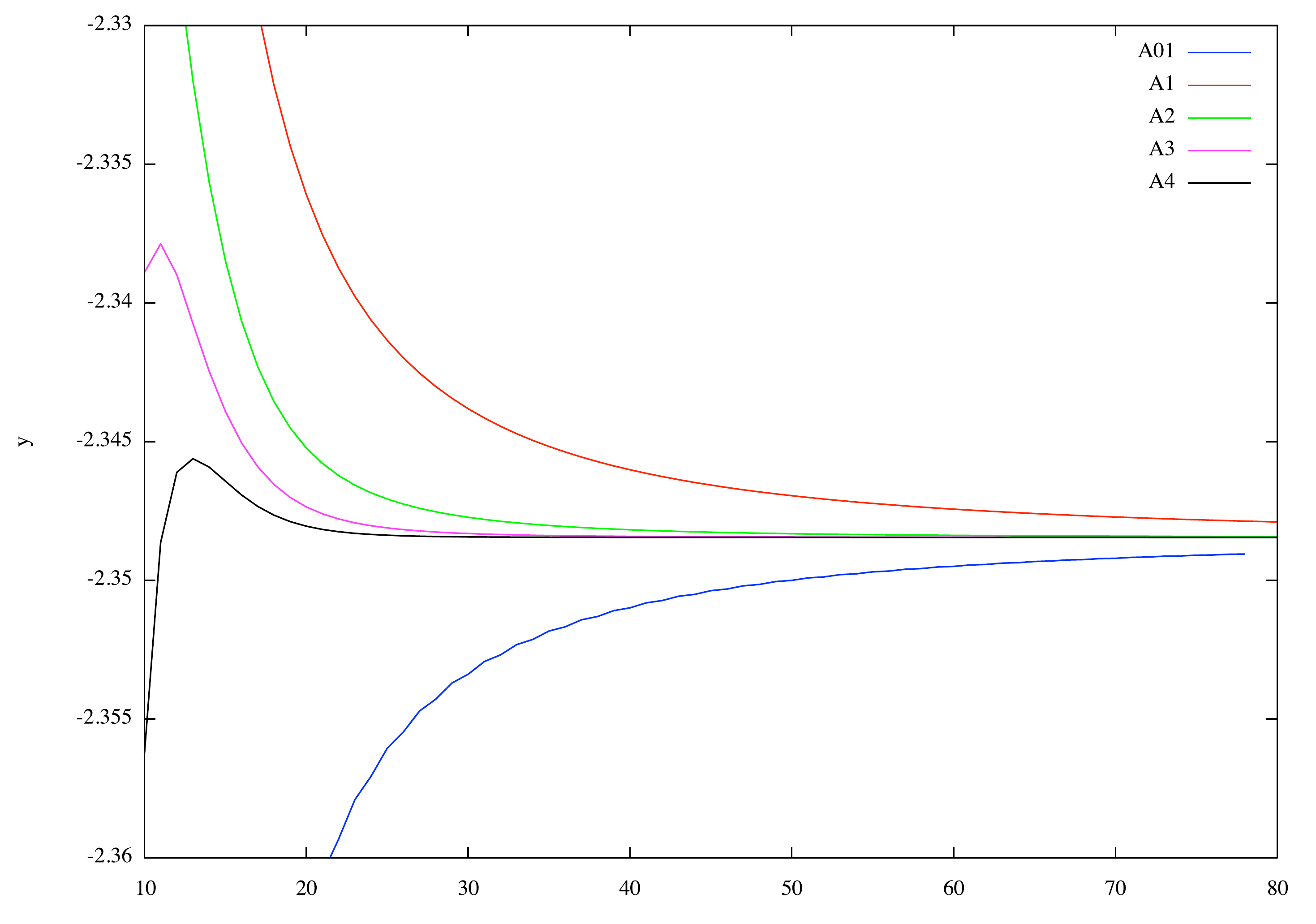}
\end{figure}

\subsection{Comparison to numerical results}

This numerical study will be a useful check that the previous results are correct. Indeed, we are expecting the convergence of the fit to become faster when we take into 
account higher orders of our solution. So looking for speed of convergence and how it changes when we include higher orders will be a check of our computations. We had emphasized in the previous 
subsections that the cancellations of even zetas are an analytical check, and this new one will allow now to verify the rational coefficients. Both were needed when we performed these computations.

On the other hand, we want to find $g_0$ and $n_0$ since they remain as free parameters in our analytical study. We will fit $g_0$ and $n_0$ to make our computed asymptotic behavior match the data 
of \cite{BeSc08} at two consecutive orders. Hence this numerical aspect of our problem will unravel the data unreachable by purely analytical means.


The obtained values of $g_0$ are presented in figure~\ref{fig_g0} as a function of the order at which the fit is done for different approximations of the asymptotic behavior. Likewise, the 
values for \(n_0\) are presented in figure~\ref{fig_n0}.
\begin{figure}[h!]
\caption{$\frac{1}{2}n_0$ from different fits.} \label{fig_n0}
 \includegraphics[scale=0.5]{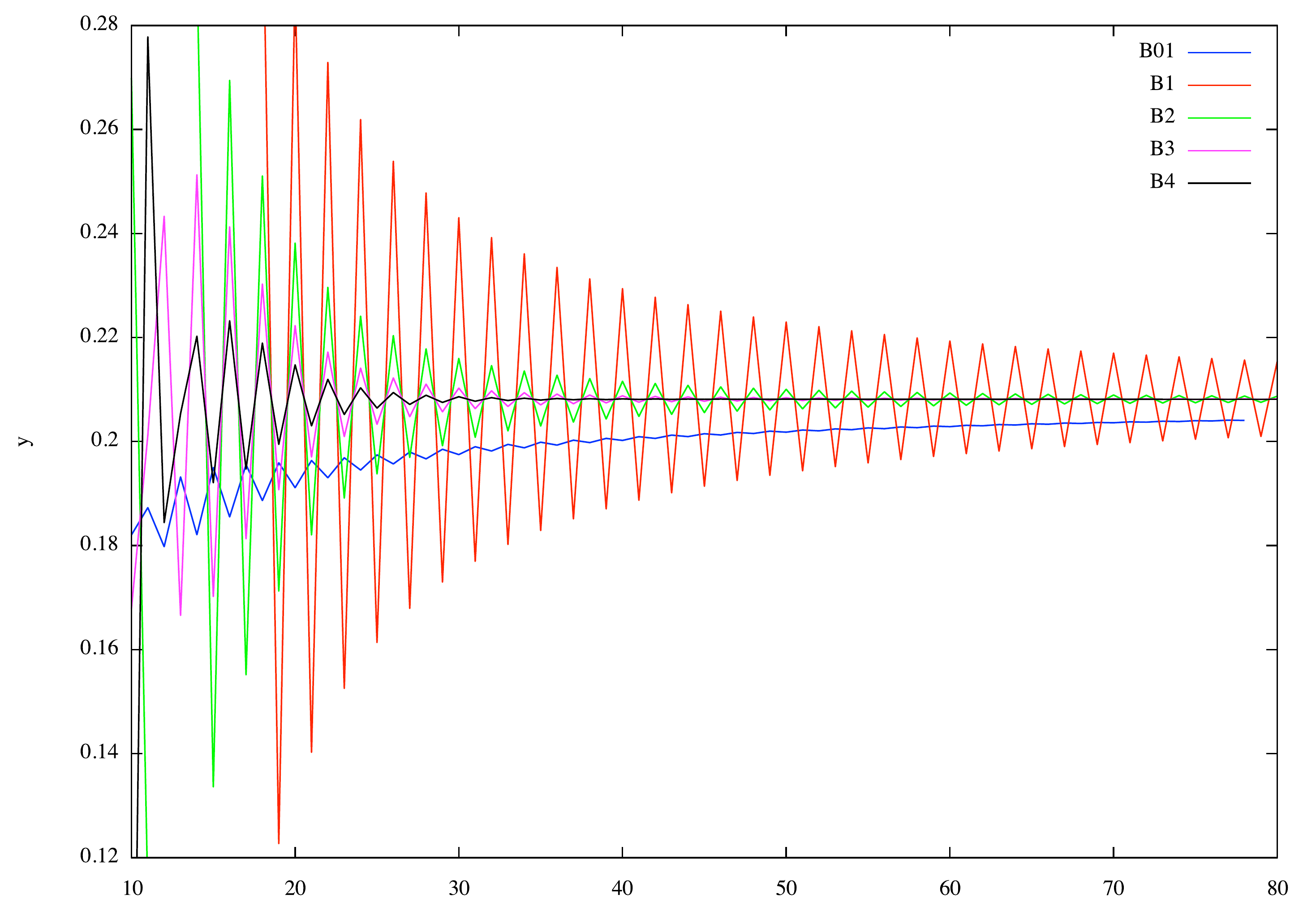}
\end{figure}

The curves $Ai$ or $Bi$ are obtained when one approximates the asymptotic behavior by including terms of up to order \(i\) in $d(a)$ and~$e(a)$. Since the case without any correction has very 
poor convergence, we plot $A01$ and $B01$, which correspond to fits on three values on a combination of \(A\), \(B\) and~\(aA\), without imposing the relation we deduced between the two terms 
proportional to \(A\). One clearly see the convergence improvement when using more terms of the asymptotic series. This can be seen as a check of our computations by numerical experiment.

We numerically get $2g_0\simeq-2.3484556$ and $\frac{1}{2}n_0\simeq0.208143(4)$. The relative precision is better on $g_0$ than on \(n_0\), which was expected since the $A_n$ sequence grows 
faster than the $B_n$ one. To improve the precision on \(g_0\) and \(n_0\), one can either go to higher order in \(a\) or compute additional terms in the asymptotic expansion. Had we not have the 
numerical results of \cite{BeSc08}, we probably could obtain the same precision on \(g_0\) and \(n_0\) with fewer low order terms of \(\gamma\) and some additional terms of the asymptotic 
behavior, for a smaller total computational cost. This is not so important here where computations remain manageable, but could be of serious interest when adding higher loop corrections to the 
Schwinger--Dyson equation.

%
%
\chapter{The massless Wess--Zumino model II: the Borel plane}

 \noindent\hrulefill \\
 {\it

-Que peux-tu donc, sinon t'ensanglanter encore les ongles et te faire prendre~?

-Rien d'autre que cela, je le sais. Mais cela, du moins, je le peux. Et il faut faire ce que l'on peut. \\

Jean Anouilh, Antigone (Cr\'eon et Antigone).}

 \noindent\hrulefill
 
 \vspace{1.5cm}

\section{Elements of Borel summation theory}

\subsection{The need of Borel summation theory}

In the previous chapter, we have presented the procedure of \cite{BeCl13} that allows precise computations, far beyond the asymptotics of the solution. Moreover, this approach was quite tailored to be 
implemented on a computer. Hence, using a formal computing software, we have been able to compute the $a^5$ corrections of the asymptotics of the anomalous dimension of the Wess--Zumino model. Moreover, the results were 
found to be in excellent agreement with the numerical results of \cite{BeSc08}.

Notwithstanding its successes, the analysis of \cite{BeCl13} suffers of some issues that have been underlined in the previous chapter. Let us summarize them here:
\begin{itemize}
 \item[$\bullet$] The symbols $A$ and $B$ have an unclear meaning.
 \item[$\bullet$] The expansion of the Mellin transform could not be proven to be exact.
 \item[$\bullet$] We had to drop the $a^{n_0}$ terms in \eqref{relation_symbolesa} and \eqref{relation_symbolesb}.
 \item[$\bullet$] We separated the $A$ and $B$ contributions in the Schwinger--Dyson equations, without a proper justification (saying ``they have a very different behavior'' does not qualify as 
a proper justification).
 \item[$\bullet$] Why were we allowed to drop the crossed terms $AB$, $A^2$ and so on? This was justified by saying that it would correspond to corrections at very high orders, but it is not a fully 
satisfactory answer.
 \item[$\bullet$] An analytic analysis of the number-theoretical contents of the expansion seems intractable within this approach.
\end{itemize}
These features are quite unsatisfactory and call for a more rigorous analysis. Indeed, a better understanding of the method of \cite{BeCl13} is needed before its use becomes possible in more 
physically relevant models. From the definitions \eqref{eq_serie_symbolesa}-\eqref{eq_serie_symbolesb} we see that the Borel transform of the formal series $A$ and $B$ is a well-defined function. 

This strongly suggests to study the Borel transform of our problem (i.e. to take the Borel transform of the renormalization group equation and of the Schwinger--Dyson equation). Within this new 
approach, the divergent series will be interpreted as markers of the simplest singularities of the Borel transform of the anomalous dimension. This will cure some of the issues listed above, and 
the others will come as by-products.

Before mapping our problem to the Borel plane, let us give a short presentation of the theory of Borel summation.

\subsection{The Borel transform}

There are many introductions to the Borel transform, and we do not intend to make a new one. We will only state without proof some useful facts and follow the presentation of \cite{Bo11}. A more 
rigorous introduction is \cite{Sa14}.

The Borel transform might be seen as a ring morphism between two rings of formal series:
\begin{eqnarray}
             \mathcal{B}: a\mathbb{C}[[a]] & \longrightarrow & \mathbb{C}[[\xi]] \\
\tilde{f}(a) = a\sum_{n=0}^{+\infty}c_na^n & \longrightarrow & \hat{f}(\xi) = \sum_{n=0}^{+\infty}\frac{c_n}{n!}\xi^n \nonumber
\end{eqnarray}
The idea is that even if $\tilde{f}$ is a purely formal series (that is, has a null radius of convergence), $\hat{f}$ might be convergent. There is an inverse Borel transform, the Laplace 
transform, defined through an integral along the line of complex numbers of argument $\theta$:
\begin{equation}
 \mathcal{L}^{\theta}(\phi)(z) := \int_0^{e^{i\theta\infty}}\phi(\zeta)e^{-\zeta/z}\d\zeta.
\end{equation}
It is easy to see that the Laplace transform of a Borel transform matches the original series since 
\begin{equation}
 \int_0^{+\infty}\frac{\zeta^n}{n!}e^{-\zeta/z}\d\zeta = z^{n+1}.
\end{equation}
The Laplace transform of the Borel transform (when it exists) is called the Borel sum of $\tilde f$. Hence, whenever the original series $\tilde{f}$ is convergent, the Borel sum matches the usual 
sum. When it is not the case, we can still have $\hat{f}$ analytic in some domain, such that the Laplace transform exists. Finally, if $\hat{f}$ has the right analyticity properties, we can vary 
$\theta$ in the definition of the Laplace transform so that all the Borel sums are analytically prolonged into each other. This procedure is the Borel resummation method.

However, this resummation has to be done in sectors of the complex plane, bounded by the lines of singularities of the Borel transform. One speaks of sectorial resummation.  When one crosses such 
a line of singularities of the Borel transform between two different sectors, the result of the summation changes. This is known as the Stokes phenomenon. Methods have been devised to compute 
these changes, and their study is very active, especially in the field of dynamical systems.

The essential properties of the Borel transform that we will use are: first, it is a linear transformation. Secondly, the Borel transform of a pointwise product of functions is the convolution 
product of the Borel transforms:
\begin{eqnarray} \label{prodB}
 \mathcal{B}(\tilde{f}\tilde{g})(\xi) & = & \hat{f}\star\hat{g}(\xi) \\
                                      & = & \int_0^{\xi}\hat{f}(\xi-\eta)\hat{g}(\eta)\d\eta. \nonumber
\end{eqnarray}
The last line being well-defined if and only if $\hat{f}$ and $\hat{g}$ have analytic continuations along a suitable path between \(0\) and~ \(\xi\). A consequence of this relation is that the 
Borel transform of $a.f$ is the primitive of $\hat{f}$.
\begin{equation} \label{primB}
 \mathcal{B}(a.f)(\xi) = \int_0^{\xi}\hat{f}(\eta)\d\eta
\end{equation}
Another very useful relation is
\begin{equation} \label{derivB}
 \mathcal{B}\left(a\partial_a\tilde{f}(a)\right)(\xi) = \partial_{\xi}\left(\xi\hat{f}(\xi)\right).
\end{equation}
It can be proved through the manipulation of formal series:
\begin{equation*}
 a\partial_a\tilde f(a) = a\sum_{n=0}^{+\infty}c_n(n+1)a^n.
\end{equation*}
Then
\begin{equation*}
 \mathcal{B}(a\partial_a\tilde f)(\xi) = \sum_{n=0}^{+\infty}(n+1)\frac{c_n}{n!}\xi^n = \partial_{\xi}\sum_{n=0}^{+\infty}\frac{c_n}{n!}\xi^{n+1} = \partial_{\xi}\left(\xi\hat{f}(\xi)\right).
\end{equation*}
Finally, we will refer in the following to the plane of $a$ as the physical plane, and the plane of $\xi$ as the Borel plane. Let us now move from the former to the later.

\section{Mapping to the Borel plane}

\subsection{Renormalization group equation}

We start from the renormalization group equation \eqref{RGE_G2_WZ}. Let us recall it here:
\begin{equation*}
 \partial_L G(a,L) = \gamma\left(1 + 3a\partial_a\right)G(a,L).
\end{equation*}
However, we have not defined the Borel transform of a constant function. This could be done through the formal unity of the convolution product: the Dirac $\delta$ ``function''. Since we want to 
deal only with analytic quantities, we have rather chosen to omit the $1$ in the Borel transform. Hence we will consider the function \(\tilde G= G -1 \). Then \eqref{RGE_G2_WZ} leads to the 
following renormalization group equation for $\tilde{G}$:
\begin{equation}
 \partial_L \tilde{G}(a,L) = \gamma\left(1 + 3a\partial_a\right)\tilde{G}(a,L)+\gamma.
\end{equation}
This equation is easily mapped into the Borel plane by using the rules \eqref{prodB} and \eqref{derivB} since $\tilde{G}$ has no constant part and has therefore its Borel transform well-defined. 
As shown in the previous chapter this is also true for $\gamma$.
\begin{equation*}
 \partial_L\hat{G}(\xi,L) = \hat{\gamma}(\xi) + \int_0^{\xi}\hat{\gamma}(\xi-\eta)\hat{G}(\eta,L)\d\eta + 3\int_0^{\xi}\hat{\gamma}(\xi-\eta)\partial_{\eta}\left(\eta\hat{G}(\eta,L)\right)\d\eta.
\end{equation*}
Treating the convolution product as a perturbation in this equation, one obtains terms which are proportional to \(L^n\). However, the resultant power series in \(L\) is not really informative and 
is not suitable for a study of the singularities of the Borel transform (which, as we will see later, are really the crucial objects of interest). Also, due  to the presence of the derivative with 
respect to \(\xi\) of \(\hat G\), one cannot expect to find \(\hat G\) as a fixed point.

Integrating by parts the last integral and using $\hat{\gamma}(0)=1$ (which is true since $\gamma(a) = a +\mathcal{O}(a^2)$) leads to an equation which will prove itself much more convenient.
\begin{equation}
 \partial_L\hat{G}(\xi,L) - 3\xi\;\hat{G}(\xi,L) = \hat{\gamma}(\xi) + \int_0^{\xi}\hat{\gamma}(\xi-\eta)\hat{G}(\eta,L)\d\eta + 3\int_0^{\xi}\hat{\gamma}'(\xi-\eta)\eta\hat{G}(\eta,L)\d\eta
\end{equation}
Here, if we neglect the convolution parts, we have the order zero equation
\begin{equation*}
 \partial_L\hat{G}(\xi,L) - 3\xi\;\hat{G}(\xi,L) = \hat{\gamma}(\xi)
\end{equation*}
which, with the condition \(\hat G(\xi,0)=0\) coming from $G(1)=1$, has the zero order solution:
\begin{equation}
\hat G_0(\xi,L) = \frac 1 {3\xi} \hat\gamma(\xi)( e^{3\xi L} - 1).
\end{equation}
Introducing this order zero solution in the convolution products suggests that \(\hat G\) for fixed Borel parameter \(\xi\) can be represented as a superposition of exponentials of \(L\) with 
parameters between \(0\) and \(3\xi\). Since \(L\) is the logarithm of \(p^2\), it means that we simply have a general power of the impulsion squared.  However, we would like to have a 
representation which does not depend on the path joining \(0\) and \(\xi\) and which easily deals with the singularities we expect to have at the ends of the path, since the order 0 solution has 
Dirac masses at these points.

We therefore parametrize $\hat{G}$ as a contour integral,
\begin{equation} \label{param_G}
 \hat{G}(\xi,L) = \oint_{\mathcal{C}_{\xi}}\frac{f(\xi,\zeta)}{\zeta}e^{3\zeta L}\d\zeta
\end{equation}
with $\mathcal{C}_{\xi}$ any contour enclosing $0$ and $\xi$. On a contour minimally including the endpoints,  the jump of \(f\) along a cut from \(0\) to \(\xi\) gives a smooth integral, while 
the singularities at the end points will contribute to singular terms. The condition that \(\hat G(\xi,0)\) is zero is also easily obtained in this formalism, since the exponential becomes 1 for 
\(L=0\) and the contour can be expanded to infinity.  It is therefore sufficient that \(f\) have limit 0 at infinity. The renormalization group equation for $\hat{G}$ becomes an equation on $f$, 
since one can use the same contour for the computation of \(\hat G\) for all the necessary values of \(\eta\) and then, switching the order of the contour integral and the other operations, one 
can write everything as a contour integral on a common path. The \(L\) independent term can also be given the same form, noting 
\begin{equation*}
 1 = \oint_{\mathcal{C}_{\xi}}e^{3\zeta L}\frac{\d\zeta}{\zeta}.
\end{equation*}
(A factor \(1/(2\pi i)\) has been included in the definition of the contour integral \(\oint\) to simplify notations.) Using 
\begin{equation*}
 \partial_L\hat G(\xi,L) = \oint_{\mathcal{C}_{\xi}}(3\zeta)\frac{f(\xi,\zeta)}{\zeta}e^{3\zeta L}\d\zeta
\end{equation*}
we ends up with the following equation for $f$:
\begin{equation} \label{renorm_f}
 3(\zeta-\xi)f(\xi,\zeta) = \hat{\gamma}(\xi) + \int_0^{\xi}\hat{\gamma}(\xi-\eta)f(\eta,\zeta)\d\eta + 3\int_0^{\xi}\hat{\gamma}'(\xi-\eta)\eta f(\eta,\zeta)\d\eta.
\end{equation}
We will see later that this equation is the right one to study the singularities of $\hat{\gamma}$.

\subsection{Schwinger--Dyson equation}

We start with the Schwinger--Dyson equation in the physical plane \eqref{SDnlin}. We know that we only need its derivative with respect to \(L\) at the renormalization point, which defines 
\(\gamma\),
\begin{equation}
 \gamma(a) = \left.-\frac{a}{\pi^2}\frac{\partial}{\partial L}\int\d^4qP\left(q^2\right)P\left((p-q)^2\right)\right|_{L=0}
\end{equation}
with $P$ the fully renormalized propagator:
\begin{equation}
 P(p^2) = \frac{1}{p^2}\left(1+\tilde{G}(L(p^2))\right).
\end{equation}
In the following, as before, we will denote simply by $\partial_L$ the operator taking the partial derivative with respect to $L$ and evaluating at zero. The integral naturally splits in three 
parts, according to the number of \(\tilde G\) factors,
\begin{equation}
 \gamma(a) = -\frac{a}{\pi^2}\partial_L\left[I_1(L)+2I_2(L)+I_3(L)\right],
\end{equation}
with:
\begin{eqnarray*}
 I_1(L) & = & \int\d^4q\frac{1}{q^2(p-q)^2} + S_1(\mu^2) \\
 I_2(L) & = & \int\d^4q\frac{\tilde{G}\left(q^2,a\right)}{q^2(p-q)^2} + S_2(\mu^2) \\
 I_3(L) & = & \int\d^4q\frac{\tilde{G}\left(q^2,a\right)\tilde{G}\left((p-q)^2,a\right)}{q^2(p-q)^2}+ S_3(\mu^2).
\end{eqnarray*}
The $S_i$'s are the formally infinite counterterms of kinematical renormalization, which ensure that \(\tilde G\) is zero at the reference impulsion $\mu$, exactly like in the second chapter. 
They disappear when deriving with respect to \(L\).

Now,  $I_1$ gives the term proportional to \(a\) in \(\gamma\), and \(a\) was normalized so that $\gamma(a) = a\left(1+\mathcal{O}(a)\right)$. Otherwise $\tilde{G}$ is 0 for \(a=0\), hence 
$a\partial_LI_2$ (resp. $a\partial_LI_3$) starts by $a^2$ (resp. $a^3$), so that we have:
\begin{equation}
 \partial_LI_1(L) = -\pi^2.
\end{equation}
The Schwinger--Dyson equation is therefore written as follows:
\begin{equation}
 \gamma(a) = \frac{a}{\pi^2}\left(\pi^2-2\partial_L\int\d^4q\frac{\tilde{G}\left(q^2,a\right)}{q^2(p-q)^2} - \partial_L\int\d^4q\frac{\tilde{G}\left(q^2,a\right)\tilde{G}\left((p-q)^2,a\right)}{q^2(p-q)^2}\right)
\end{equation}
This equation can be mapped to the Borel plane, using the relation \eqref{primB} to express the multiplication by \(a\). We end up with
\begin{equation} \label{SDE_B}
 \hat{\gamma}(\xi) = 1 - \frac{2}{\pi^2}\int_0^{\xi}\d\eta\;\partial_L\int\d^4q\frac{\hat{G}(q^2,\eta)}{q^2(p-q)^2}
  - \frac{1}{\pi^2}\int_0^{\xi}\d\eta\;\partial_L\int\d^4q
  \frac{\hat{G}\left(q^2,\eta\right)\star\hat{G}\left((p-q)^2,\eta\right)}{q^2(p-q)^2}.
\end{equation}
The convolution product in the last integral shall be read as
\begin{equation*}
 \hat{G}\left(q^2,\eta\right)\star\hat{G}\left((p-q)^2,\eta\right) = \int_0^{\eta}\hat{G}\left(q^2,\eta-\sigma\right)\hat{G}\left((p-q)^2,\sigma\right)\d\sigma.
\end{equation*}
Now,  using the parametrization \eqref{param_G} of $\hat{G}$ within the Schwinger--Dyson equation \eqref{SDE_B}, we get
\begin{equation*}
 \partial_L\int\d^4q\frac{\hat{G}(q^2,\eta)}{q^2(p-q)^2} = \oint_{\mathcal{C}_{\eta}}\d\zeta\frac{f(\eta,\zeta)}{\zeta}\partial_L\int\d^4q\frac{e^{3\zeta L(q^2)}}{q^2(p-q)^2}
\end{equation*}
for the first non-trivial term in \eqref{SDE_B}. If we recall that $e^{3\zeta L(q^2)}=(q^2)^{3\zeta}$ we can write
\begin{equation*}
 \partial_L\int\d^4q\frac{e^{3\zeta L(q^2)}}{q^2(p-q)^2} = \partial_L\int\d^4q\frac{1}{(q^2)^{3\zeta}(p-q)^2} = -\pi^2H(3\zeta,0) = -\frac{\pi^2}{1+3\zeta}
\end{equation*}
with $H$ the Mellin transform \eqref{def_H}. The loop integral can therefore be computed with the Mellin transform, pointing to the interesting properties of the parametrization of 
\(\hat G\) \eqref{param_G}.

For the second integral the situation is essentially the same, but slightly more complicated. Using an obvious notation we have
\begin{equation*}
 \partial_L\int\d^4q\frac{\hat{G}\star\hat{G}}{q^2(p-q)^2} 
 = \int_0^{\eta}\d\sigma\oint_{\mathcal{C}_{\eta-\sigma}}
 \oint_{\mathcal{C}_{\sigma}}\d\zeta\d\zeta'\frac{f(\xi-\sigma,\zeta)f(\sigma,\zeta')}{\zeta\,\zeta'}\partial_L\int\d^4q\frac{e^{3\zeta L\left(q^2\right)}e^{3\zeta'L\left((p-q)^2\right)}}{q^2(p-q)^2}.
\end{equation*}
And, similarly to the physical plane, the derivative of the last integral can be evaluated to $-\pi^2H(3\zeta,3\zeta')$. Hence we end up with the Schwinger--Dyson equation in the Borel plane written in term of $f$:
\begin{equation} \label{SDE_gamma_f}
 \hat{\gamma}(\xi) = 1 + 2\int_0^{\xi}\d\eta\oint_{\mathcal{C}_{\eta}}\d\zeta\frac{f(\eta,\zeta)}{\zeta(1+3\zeta)} + \int_0^{\xi}\d\eta\int_0^{\eta}\d\sigma\oint_{\mathcal{C}_{\eta-\sigma}}\d\zeta\frac{f(\eta-\sigma,\zeta)}{\zeta}\left(\oint_{\mathcal{C}_{\sigma}}\d\zeta'\frac{f(\sigma,\zeta')}{\zeta'}H(3\zeta,3\zeta')\right).
\end{equation}
In these expressions, care must be taken that the Mellin transform \(H\) is not holomorphic, but meromorphic: when trying to use these formulas for the analytic continuation of \(\hat\gamma\), the different contours should not go through the poles of \(H\).

\subsection{Formal series vs. singularities} \label{backto}

We would like to link the Borel plane computation and the ones made in \cite{BeCl13}, which have been presented in the previous chapter. We will not detail all the computations here, especially 
they are quite similar to one another. So, roughly speaking, only the first will be completely made. The main difficulty is to carry a combinatorial factor.

Modulo this caveat, we will show that the perturbative study made using the formal series $A$ and $B$ is equivalent to a well-defined computation in the Borel plane. Let $f$ and $g$ be two functions of 
the structure constant $a$ involving a formal series $A$:
\begin{eqnarray*}
 f(a) & = & a^n + a^mA \\
 g(a) & = & a^p + a^qA.
\end{eqnarray*}
To simplify the notations, we only take one power of \(a\) for each possible terms, but the computations of the previous chapter involve sums of such terms with varying exponents \(n\), \(m\), \(p\) 
and~\(q\). \(A\) is encoding the asymptotic behavior of the functions, or equivalently a singularity of the Borel transform:
\begin{equation}
 A = \sum A_n a^n \qquad \frac{A_{n+1}}{A_n}=\alpha n -\beta
\end{equation}
with $\alpha\neq0$. This is a formal series but is Borel summable. Without loss of generality, we can assume $\alpha=1$ since we can make an expansion in $\tilde{a}=a/\alpha$. This is nothing 
but mapping the singularity of the Borel transform to $\xi=1$. When doing our perturbative analysis, we assumed that the product of the functions $f$ and $g$ was
\begin{equation} \label{prod_utile}
 f(a)g(a) = a^{n+p} + \left[a^{m+p}+a^{q+n}\right]A.
\end{equation}
We will check that this is coherent with the map into the Borel plane, that is, compute $\widehat{fg}$ and $\hat{f}\star\hat{g}$ and check that they coincide in the right limits. First, the Borel 
transform of the formal series $A$ is
\begin{equation}
 \hat{A} = \sum\frac{A_{n+1}}{n!}\xi^n = \sum\hat{A}_n\xi^n.
\end{equation}
Thus we get the recurrence relation for the $\hat{A}_n$ coefficients:
\begin{equation} \label{rec_aB}
 \frac{\hat{A}_n}{\hat{A}_{n-1}} = 1-\frac{\beta}{n}.
\end{equation}
The most natural $\hat{A}_n$ coefficients satisfying the above recurrence relations are
\begin{equation} \label{aB1}
 \hat{A}_n = c\prod_{i=1}^n\left(1-\frac{\beta}{i}\right).
\end{equation}
However, this can only be the right form for the $\hat{A}_n$'s if $\beta\notin\mathbb{N}$. Indeed, if $\beta\in\mathbb{N}$, we would have $\hat{A}_n=0$, for large enough \(n\). Since 
\eqref{rec_aB} has to be asymptotically true (it encodes the asymptotic behavior of $f$ and $g$), for $\beta\in\mathbb{N}$, we must take a product beginning at \(\beta + 1\) in the formula for the 
$\hat{A}_n$'s. For generic \(\beta\), we have an explicit formula for $\hat{A}$ (up to an overall multiplicative factor):
\begin{equation}\label{hatA}
 \hat{A} = (1-\xi)^{\beta-1}.
\end{equation}
Indeed, this has the right Taylor expansion around $0$:
\begin{align*}
 \left.\frac{\d^n}{\d\xi^n}\hat A\right|_{\xi=0} & = \left.(-1)^n\left[\prod_{i=0}^n(\beta-1-i)\right](1-\xi)^{\beta-n-i}\right|_{\xi=0} \\
						 & = \prod_{i=0}^{n-1}(i+1-\beta) \\
						 & = n!\prod_{i=1}^n\left(1-\frac{\beta}{i}\right).
\end{align*}
An other way to prove this result is to use the differential equation satisfied formally by \(A\), Eq. \eqref{relation_symbolesa}, convert it to a differential equation for \(\hat A\) and see that 
the equation \eqref{hatA} gives its solution up to an overall factor. Then, by induction, it is easy to prove
\begin{equation}
 \widehat{a^nA} \simall{\xi}{1} \frac{(-1)^n}{(\beta)_n}(1-\xi)^{\beta+n-1}
\end{equation}
since multiplication by \(a\) corresponds to taking the primitive of the Borel transform and since the case $n=0$ correspond to $\hat A$ itself. Here $(x)_n$  is the Pochhammer symbol defined by
\begin{equation*}
 (x)_n = \frac{\Gamma(x+n)}{\Gamma(x)} = x(x+1)...(x+n-1).
\end{equation*}
Now, it is easy to compute the following equivalence relations
\begin{subequations}
 \begin{eqnarray} 
  \hat{f}(\xi) & \simall{\xi}{0} & \frac{\xi^{n-1}}{(n-1)!} \label{f0N}, \\
  \hat{f}(\xi) & \simall{\xi}{1} & \frac{(-1)^m}{(\beta)_m}(1-\xi)^{\beta+m-1}, \\
  \hat{g}(\xi) & \simall{\xi}{0} & \frac{\xi^{p-1}}{(p-1)!}, \\
  \hat{g}(\xi) & \simall{\xi}{1} & \frac{(-1)^q}{(\beta)_q}(1-\xi)^{\beta+q-1}. \label{g0N}
 \end{eqnarray}
\end{subequations}
The equivalences around \(1\) are taken modulo functions holomorphic in the neighborhood of \(1\), since any such term would either be subdominant in the asymptotic behavior of the coefficients 
\(f_n\) or captured by a different symbol.  One way of getting rid of these holomorphic terms is to take the difference between the analytic continuation of the Borel transform by either side of 
\(1\). Any holomorphic function is killed, while the non-integer powers are multiplied by \(\sin(\pi\beta)/\pi\) (for convenience, the difference is divided by \(2\pi i\)). First, it is trivial 
to check
\begin{equation}
 \hat{f}\star\hat{g}(\xi) \simall{\xi}{0} \frac{\xi^{n+p-1}}{(n+p-1)!} = \mathcal{B}(a^{n+p}).
\end{equation}
Hence the $a^{n+p}$ term of (\ref{prod_utile}) is justified: it is just the correspondence between ordinary product and the convolution product of the Borel transform.

Now, using that $\hat{f}$ and $\hat{g}$ have only one singularity in $\xi=1$, we have
\begin{equation*}
 \hat{f}\star\hat{g}(\xi) \simall{\xi}{1} \underbrace{\int_0^\frac12 \hat{f}(t)\hat{g}(\xi-t)\d t}_{I_1(\xi)} + \underbrace{\int_\frac12^\xi\hat{f}(t)\hat{g}(\xi-t)\d t}_{I_2(\xi)}.
\end{equation*}
Let us start by $I_1$.
\begin{equation*}
 I_1(\xi) = \frac{(-1)^{q}}{(n-1)!(\beta)_q}\int_0^\frac12 t^{n-1}(1-\xi+t)^{\beta+q-1}\d t
\end{equation*}
Performing $n-1$ integration by parts in order to get rid of the $t^{n-1}$ in the integrand and taking care of the combinatorial factors we end up with
\begin{equation} \label{I2betaNotInN}
 I_1(\xi) \simall{\xi}{1}\frac{(-1)^{q+n}}{(\beta)_{q+n}}(1-\xi)^{\beta+q+n-1} = \mathcal{B}(a^{q+n}A)
\end{equation}
The contribution from the other end point is holomorphic for \(\xi\) in the neighborhood of \(1\) and is therefore negligible.
For $I_2$ we have
\begin{equation*}
 I_2(\xi) = \frac{(-1)^m}{(p-1)!(\beta)_m}\int_{\frac12}^\xi (1-t)^{\beta+m-1}(\xi-t)^{p-1}\d t.
\end{equation*}
Using the transformation $x=\xi-t$ we transform this integral into an integral similar to $I_1$, and the same integrations by part give us
\begin{equation} \label{I1betaNotInN}
 I_2(\xi) \simall{\xi}{1}\frac{(-1)^{p+m}}{(\beta)_{p+m}}(1-\xi)^{\beta+p+m-1} = \mathcal{B}(a^{p+m}A).
\end{equation}
Hence, (\ref{I2betaNotInN}) and (\ref{I1betaNotInN}) justify the $\left[a^{m+p}+a^{q+n}\right]A$ term in (\ref{prod_utile}) for $\beta\notin\mathbb{N}$ through the correspondence between the asymptotic behavior of the perturbative series and the singularities of the Borel transform.

For $\beta\in\mathbb{N}^*$ we have to take another form for the $\hat{A}_n$'s. We will take
\begin{equation}
 \hat{A}_n = \frac{1}{n(n-1)...(n-\beta+1)}
\end{equation}
and start the sum within $\hat{A}$ at $n=\beta$. Then:
\begin{eqnarray}
 \hat{A}(\xi) & = & \sum_{n\geq\beta}\underbrace{\int^{\xi}...\int_0}_{\beta\text{ times}}t^{n-\beta}\d t \nonumber \\
              & = & \int^{\xi}...\int_0\frac{\d t}{1-t} \nonumber \\
              & \simall{\xi}{1} & \frac{(-1)^\beta}{(\beta-1)!} (1-\xi)^{\beta-1}\ln(1-\xi).
\end{eqnarray}
Then, by induction, it is easy to prove
\begin{equation}
 \widehat{a^nA} \simall{\xi}{1}\frac{(-1)^{n+\beta}}{(\beta+n-1)!}(1-\xi)^{\beta+n-1}\ln(1-\xi).
\end{equation}
For $\beta=0$ no integration has to be performed when computing $\hat{A}$ and hence $\hat{A}(\xi) \simall{\xi}{1}(1-\xi)^{-1}$. Nevertheless, the above formula includes the case $\beta=0$. The 
equivalence relations (\ref{f0N})--(\ref{g0N}) become now
\begin{subequations}
 \begin{eqnarray}
  \hat{f}(\xi) & \simall{\xi}{0} & \frac{\xi^{n-1}}{(n-1)!}, \\
  \hat{f}(\xi) & \simall{\xi}{1} & \frac{(-1)^{m+\beta}}{(\beta+m-1)!}(1-\xi)^{\beta+m-1}\ln(1-\xi), \\
  \hat{g}(\xi) & \simall{\xi}{0} & \frac{\xi^{p-1}}{(p-1)!}, \\
  \hat{g}(\xi) & \simall{\xi}{1} & \frac{(-1)^{q+\beta}}{(\beta+q-1)!}(1-\xi)^{\beta+q-1}\ln(1-\xi). \\
 \end{eqnarray}
\end{subequations}
Following the same strategy than for the case $\beta\notin\mathbb{N}$ we find that the combinatorial factors nicely combine such that
\begin{eqnarray}
 \hat{f}\star\hat{g}(\xi) & \simall{\xi}{1} & - \frac{(\xi -1)^{\beta+q+n-1}}{(\beta+q+n-1)!}\ln(1-\xi) - \frac{(\xi -1)^{\beta+m+p-1}}{(\beta+m+p-1)!}\ln(1-\xi) \nonumber \\
                          & = & \mathcal{B}(a^{q+n}A) + \mathcal{B}(a^{m+p}A).
\end{eqnarray}
Thus our perturbative computations are strictly equivalent to computations around the singularities of the Borel transform. Here we see that the Borel transform approach to the Schwinger--Dyson equation 
allows for a more natural interpretation of our results. 

Moreover, let us notice that neither $A$ nor $B$ can appear alone in
$\gamma$. The lowest order terms are $aA$ and $aB$. Hence they correspond
in the Borel plane to singularities at $\xi=\pm1/3$ and 
\begin{eqnarray*}
 \widehat{aA} & \simall{\xi}{-1/3} & \left(\xi+\frac{1}{3}\right)^{-5/3} \\
 \widehat{aB} & \simall{\xi}{1/3} & \ln\left(\xi-\frac{1}{3}\right)
\end{eqnarray*}
as stated in \cite{BeCl13}.

In fact, these computations are but the first steps in a general approach to the singularities of the Borel transform initiated some time ago by Jean \'Ecalle, the Alien calculus \cite{Ecalle81}, 
an introduction of which can be found in \cite{Sa14}. In our case, it just means that we extract the singular part of a function around \(\xi\) by taking the difference of the two analytic 
continuations around \(1\) and shifting to have an expansion around 0.  The coefficients of \(A\) which describe the asymptotic properties of the formal power series \(f\) and~\(g\) are therefore 
a description of the singularity of the Borel transforms, which can be extracted by an operator \(\Delta_1\):
\begin{subequations}
 \begin{align}
    \hat{f}(\xi)  \simall{\xi}{1} \frac{(-1)^{m+\beta}}{(\beta+m-1)!}(1-\xi)^{\beta+m-1}\ln(1-\xi) & \Longrightarrow \Delta_1 \hat{f} = - \frac{ \xi^{\beta + m -1}}{(\beta + m - 1)!}, \\
    \hat{f}(\xi)  \simall{\xi}{1} \frac{(-1)^m}{(\beta)_m}(1-\xi)^{\beta+m-1} & \Longrightarrow \Delta_1 \hat{f} = \frac{-\sin(\pi\beta)}{\pi} \frac{\xi^{\beta + m -1} }{ (\beta)_m} 
 \end{align}
\end{subequations}
%
The first line corresponds to the case where \(\beta\) is a positive integer, the second one to non integer \(\beta\).  The computations we just
made tell us that \(\Delta_1\) is a derivation with respect to the convolution product of the functions in the Borel plane.  

The whole story is subtler, because our computation was limited to singularities of the Borel transform on the limit of the disk of convergence. In many cases, one expects that there will be 
singularities for any integer multiple of a given singularity.  Then the singularity of the convolution product receives contributions from the pinching of the integration contour between 
singularities of \(\tilde f(\eta)\) and \(\tilde g(\xi - \eta)\). However \'Ecalle has shown that, by summing the singularities of the \(2^k\) differing analytic continuations of a function along 
paths going above or under the \(k\) singularities between the origin and a potential singularity with suitable weights, one obtains a derivation with respect to the convolution product, that he 
named an {\em alien}  derivation. Such derivations can then be used to compute the relation between the sums defined by integrating the Borel transform in different sectors.

Applying an alien derivation \(\Delta_\xi\) to a system of equations for the Borel transforms, one obtains a system of equations which is linear in the alien derivatives of the indeterminate 
functions: for generic values of the parameter \(\xi\), the only solution of this system will be zero, and we can conclude that the solutions in the Borel plane have no new singularity at this 
point (it is still possible to have a singularity if \(\xi\) is the sum of the positions of other singularities). At other points, there will be a one dimensional space of solutions, which will 
determine the singularities at this point up to a single scale.

For finite order computations, it is much easier to use formal series in the physical plane, which are easily multiplied by computer algebra
systems, exactly how we have done in the previous chapter. At this stage, alien calculus is just giving us a nice interpretation beyond formal series.

\section{Truncated Schwinger--Dyson equation}


\subsection{Justification of the truncation}

The goal of this section is to study the asymptotic behavior of $\hat\gamma$. We will start by justifying that in the Schwinger--Dyson equation \eqref{SDE_gamma_f}, the quadratic term in $f$ can 
be dropped without modifying this asymptotic behavior.

Indeed, it will be shown in the next section that the singularities of $\hat{\gamma}$ all lie on the real line\footnote{Therefore, it would have been meaningful to put the next section before this 
one. I have chosen to not do so since the next section contains the most important results of this \'etude.}. Thus we will take $\xi\notin\mathbb{R}$. Now, the truncation comes from a rather 
trivial fact: the asymptotics of $\gamma$ (in the physical plane) was given in \cite{Be10a} by the first pole of the one loop Mellin transform. Here, this corresponds to the pole in $\zeta=-1/3$ 
in \eqref{SDE_gamma_f}, for which the integral linear in $f$ is the dominant contribution.

To justify more formally this truncation, let us write
\begin{equation*}
 H(\zeta,\zeta') = \frac{1}{1+\zeta+\zeta'}\frac{\Gamma(1-\zeta-\zeta')\Gamma(1+\zeta)\Gamma(1+\zeta')}{\Gamma(1+\zeta+\zeta')\Gamma(1-\zeta)\Gamma(1-\zeta')}.
\end{equation*}
Then, the Stirling approximation $\Gamma(1+x)\sim \sqrt{2\pi}x^{x+1/2}e^{-x}$, valid for any complex \(x\) except in the immediate vicinity of
the negative real axis, leads to
\begin{align*}
 H(\zeta,\zeta') & \sim \frac{1}{1+\zeta+\zeta'}\frac{(-\zeta-\zeta')^{-\zeta-\zeta'+1/2}\zeta^{\zeta+1/2}\zeta'^{\zeta'+1/2}}{(\zeta+\zeta')^{\zeta+\zeta'+1/2}(-\zeta)^{-\zeta+1/2}(-\zeta')^{-\zeta'+1/2}} \\
                 & = \frac{-i}{1+\zeta+\zeta'}\frac{\zeta^{2\zeta}\zeta'^{2\zeta'}}{(\zeta+\zeta')^{2\zeta+2\zeta'}}
\end{align*}
if the imaginary parts of \(\zeta\) and \(\zeta'\) are both positive. Now, let us write $\zeta'=\alpha\zeta$. Since $\xi\notin\mathbb{R}$, we can always deform the contour of integration 
$\mathcal{C}_{\xi}$ so that it does not cross the real axis in the vicinity of $\xi$. Moreover, since we are interested by the asymptotic behavior of $\hat{\gamma}$, we can suppose than $\zeta$ 
and $\zeta'$ are not near the origin. Hence we can assume that $\zeta$ and $\zeta'$ are in the same quadrant of the complex plane. Thus we we arrive to
\begin{equation*}
 H(\zeta,\zeta') \sim \frac{-i}{1+\zeta(1+\alpha)}\left(\frac{\alpha^{\alpha}}{(1+\alpha)^{1+\alpha}}\right)^{2\zeta} .
\end{equation*}
Using the definition of complex power $z^{z'}=\left(|z|^2\right)^{z'/2}e^{iz'\text{arg}(z)}$ we end up with
\begin{equation*}
 \left|\frac{\alpha^{\alpha}}{(1+\alpha)^{1+\alpha}}\right| < 1 \Leftrightarrow \alpha_1\ln\left|\frac{\alpha}{1+\alpha}\right| - \ln|1+\alpha| - \frac{\alpha_2}{2}\left[\text{atan}\left(\frac{\alpha_2}{\alpha_1}\right) - \text{atan}\left(\frac{\alpha_2}{1+\alpha_1}\right)\right] < 0
\end{equation*}
with \(\alpha_1 = \Re(\alpha)\) and $\alpha_2=\Im(\alpha)$.
Since $\alpha_1>0$, $\alpha_1\ln\left|\frac{\alpha}{1+\alpha}\right|<0$ and since the function \(\text{atan}\) is monotonic, increasing over 
$\mathbb{R}$,
$ \frac{\alpha_2}{2}\left[\text{atan}\left(\frac{\alpha_2}{\alpha_1}\right) - \text{atan}\left(\frac{\alpha_2}{1+\alpha_1}\right)\right]>0$. Therefore, $H(\zeta,\zeta')$ is exponentially small at infinity for $\Re(\zeta)>0$. 

The conclusion of this subsection is that, in a sector with positive real and imaginary values of \(\xi\),  the term quadratic in $f$ in the Schwinger--Dyson equation \eqref{SDE_gamma_f} will 
involve an exponentially small \(H(\zeta,\zeta')\), except when one of the argument is in the vicinity of 0.  It is therefore plausible that the contribution of this quadratic part remains 
subdominant and can be ignored without any dramatic change of the asymptotic behavior of the solution for $\Re(\xi)>0$ and $\xi$ far enough of the real line.

Another argument will be given in the next subsection, that will come from the analysis of \cite{Be10}.

\subsection{Truncated equation}

First, let us notice than we can solve a specialization of the renormalization group equation \eqref{renorm_f}. Defining $g(\xi):=f(\xi,-1/3)$ and specializing \eqref{renorm_f} to $\zeta=-1/3$ 
leads to
\begin{equation} \label{eq_g}
 -(1+3\xi)g(\xi)=\hat{\gamma}(\xi)+\int_0^{\xi}\hat{\gamma}(\xi-\eta)g(\eta)\d\eta +3\int_0^{\xi}\hat{\gamma}'(\xi-\eta)\eta g(\eta)\d\eta.
\end{equation}
This equation can be solved by adding a parameter $\lambda$ and writing $g$ as a series in this parameter.
\begin{equation*}
 g(\xi) = \sum_{n\geq0} \lambda^n g_n(\xi)|_{\lambda=1}
\end{equation*}
Then \eqref{eq_g} gives the recurrence relations amongst the $g_n$'s.
\begin{eqnarray*}
               g_0(\xi) & = & -\frac{\hat{\gamma}(\xi)}{3\xi+1} \\
  -(1+3\xi)g_{n+1}(\xi) & = & \underbrace{\int_0^{\xi}\hat{\gamma}(\xi-\eta)g_n(\eta)\d\eta}_{=I_1^{n+1}(\xi)} + \underbrace{3\int_0^{\xi}\hat{\gamma}'(\xi-\eta)\eta g_n(\eta)\d\eta}_{=I_2^{n+1}(\xi)}
\end{eqnarray*}
Now, for a given $g_n$ the $I_1^n$ can either come from $I_1^{n-1}$ or from $I_2^{n-1}$. Hence we can write the recurrence relations for the $I$'s as well:
\begin{align*}
 & I_1^{n+1}(\xi) = -\frac{1}{3}\int_0^{\xi}\frac{\hat{\gamma}(\xi-\eta)}{\eta+1/3}\left[I_1^{n}(\eta)+3I_2^{n}(\eta)\right]\d\eta \\
 & I_2^{n+1}(\xi) = -\int_0^{\xi}\frac{\hat{\gamma}'(\xi-\eta)}{\eta+1/3}\eta\left[I_1^{n}(\eta)+3I_2^{n}(\eta)\right]\d\eta.
\end{align*}
Now we can solve the induction for the $I$'s and thus solve \eqref{eq_g}. The solution will be a Chen integral. First, let us define two functions:
\begin{align*}
 \begin{cases}
  f_0(\xi,\eta) = - \frac{\eta}{\eta+1/3} \hat{\gamma}'(\xi-\eta)\\
  f_1(\xi,\eta) = - \frac{1}{3\eta+1}\hat{\gamma}(\xi-\eta).
 \end{cases}
\end{align*}
Moreover, let $I_n=(i_1,...,i_n)$ be a string of integers, with $i_k\in\{0,1\}$. Now we define the iterated integrals
\begin{equation}
 F^{I_{n}}_{0,\xi} = \int_{0\leq x_n\leq...\leq x_1\leq x_0 := \xi}\left[\prod_{k=1}^{n}f_{i_k}(x_{k-1},x_{k})\right]\hat{\gamma}(x_n)\d x_n...\d x_1.
\end{equation}
Then, with the same argument than the one used to find the recurrence relations for the $I$'s, we see that a $g_n$ is given by the sums of the $F^{I_n}$ build from all the possible strings $I_n$ 
of length $n$. Hence, the solution of \eqref{eq_g} is:
\begin{equation} \label{sol_g}
 g(\xi) = \frac{-1}{3\xi+1}\left(\hat{\gamma}(\xi)+\sum_{n\geq1}\sum_{\{I_{n}\}}F^{I_{n}}_{0,\xi}\right).
\end{equation}
Now, according to our analysis of the previous subsection, we can neglect the term $\hat{G}\star\hat{G}$ in the Schwinger--Dyson equation when looking for the asymptotic behavior of 
$\hat{\gamma}$. Hence, from \eqref{SDE_gamma_f}, we get the following equation for $\hat{\gamma}$:
\begin{equation} \label{eq_gamma}
 \partial_{\xi}\hat{\gamma}(\xi) = 2\oint_{\mathcal{C}_{\xi}}f(\xi,\zeta)\frac{1}{\zeta(1+3\zeta)}\d\zeta.
\end{equation}
Deforming the integration contour $\mathcal{C}_{\xi}$ to a circle of infinite radius, the loop integral vanishes, thanks to Jordan's lemma, and differs from the integral above only by the opposite 
of the residue at $\zeta=-1/3$ (since $0$ is enlaced by $\mathcal{C}_{\xi}$). We have to take care of the sign here, so let us write everything down carefully:
\begin{equation*}
 \text{Res}\left(\frac{f(\xi,\zeta)}{\zeta(1+3\zeta)},\zeta=-1/3\right) = \text{Res}\left(\frac{f(\xi,\zeta)}{3\zeta(\zeta+1)},\zeta=-1/3\right) = \lim_{\zeta\rightarrow-1/3}\frac{f(\xi,\zeta)}{3\zeta} = -g(\xi).
\end{equation*}
Hence, all in all, we get
\begin{equation}
 \partial_{\xi}\hat{\gamma}(\xi) = + 2g(\xi).
\end{equation}
And, with the solution \eqref{sol_g}, we have an equation for $\hat{\gamma}$.
\begin{equation} \label{SDE_tronquee}
 \partial_{\xi}\hat{\gamma}(\xi) = \frac{-2}{3\xi+1}\left(\hat{\gamma}(\xi)+\sum_{n\geq1}\sum_{\{I_{n}\}}F^{I_{n}}_{0,\xi}\right).
\end{equation}
This equation is coherent with $\hat{\gamma}(0)=1$ and $\hat{\gamma}'(0)=-2$.

Before we go further, let us emphasize that the relation $\hat{\gamma}'=2g$ can be used to justify our truncation scheme, as advertised in the previous subsection. Indeed, if we plug it into the 
renormalization group equation \ref{renorm_f} specialized to $\zeta=-1/3$, we end up with an integrodifferential equation for $\hat{\gamma}$. Taking the inverse Borel transform of this equation we 
end up with a differential equation on $\gamma$:
\begin{equation}
 -(1+3\xi)\hat\gamma'(\xi) = 2\hat\gamma(\xi) +( \hat\gamma\star\hat\gamma')(\xi) + 3\left(\hat\gamma'\star(Id.\hat\gamma')\right)(\xi).
\end{equation}
We can re-do the integration by part on the last term: $\left(\hat\gamma'\star(Id.\hat\gamma')\right)(\xi) = -\xi\hat\gamma'(\xi)+[\hat\gamma\star\partial_{\eta}(\eta\hat\gamma'(\eta))](\xi)$ to 
get 
\begin{equation}
 -\hat\gamma = 2\hat\gamma+\hat\gamma\star\hat\gamma'+3\hat\gamma\star\partial_{\eta}(\eta\hat\gamma'(\eta)).
\end{equation}
Noticing $\frac{\d}{\d\xi}(\hat\gamma\star\hat\gamma) = \hat\gamma(0)+\hat\gamma'\star\gamma=\hat\gamma +\hat\gamma\star\hat\gamma'$ we can integrate the above equation to
\begin{equation}
 -\hat\gamma = -1 + \int\hat\gamma + \hat\gamma\star\hat\gamma + 3\int\hat\gamma\star\partial_{\eta}(\eta\hat\gamma'(\eta))
\end{equation}
with all the integrations from $0$ to $\xi$. Using the known formula \eqref{primB} and \eqref{derivB}, we can take the inverse Borel transform: 
\begin{equation}
 \gamma = a - a\gamma -\gamma^2 - 3a\gamma\partial_af
\end{equation}
with $f$ a function such that $\mathcal{B}(f)=\hat\gamma'$. It is easy to see that $f=\gamma/a-1$. Indeed
\begin{equation*}
 \gamma(a)=a\left(1+\sum_{n=1}^{+\infty}c_na^n\right) \Rightarrow \hat\gamma(\xi)=1+\sum_{n\geq1}\frac{c_n}{n!}\xi^n \Rightarrow \hat\gamma'(\xi)=\sum_{n\geq1}\frac{c_n}{(n-1)!}\xi^{n-1}.
\end{equation*}
Moreover
\begin{equation*}
 \mathcal{B}\left(\frac{\gamma}{a}-1\right) = \mathcal{B}\left(a\sum_{n=0}^{+\infty}c_{n+1}a^n\right) = \sum_{n=0}^{+\infty}\frac{c_{n+1}}{n!}\xi^n = \hat\gamma'(\xi).  
\end{equation*}
Hence, after simplifications, we arrive to
\begin{equation}
 \gamma = a - a\gamma + 2\gamma^2 - 3a\gamma\gamma'
\end{equation}
which is exactly the equation for $\gamma$ found in \cite{Be10} (equation (17)), up to terms that do not contribute to the asymptotics of $\gamma$. Thus, this is a nice check that a solution 
$s(\xi)$ to \eqref{SDE_tronquee} has the right asymptotic behavior.

Now, the equation \eqref{SDE_tronquee} appears as a fixed point equation. By defining a suitable metric on the space of functions, the integral operator could become contracting, proving the 
existence of a solution. Defining such a contracting metric is a non-trivial task that is left for further studies. Here, we will only highlight a link between the equation \eqref{SDE_tronquee} 
and the Multi-Zeta Values (MZVs) before numerically study the asymptotic behavior of the solution of \eqref{SDE_tronquee}.

\subsection{Link to MZVs}

We will explain in this subsection how the equation \eqref{SDE_tronquee} can be mapped to an equation on the algebra of Multi-Zeta Values. Having such a structure arising in our problem is 
interesting per se. Moreover, since the algebra of MZVs is fairly complicated, this explains why the study of \eqref{SDnlin} done in \cite{BeCl13} (presented in the previous chapter) seems to show 
some non-trivial combinatorics.

Since there are two possible choices for each functions in $F^{I_{n}}_{0,\eta}$, we can map them to the MZVs. Let $(a_1,...,a_r)$ be a string of positive integers and 
$(\varepsilon_1,...,\varepsilon_n)$ be its binary representation. That is:
\begin{equation}
 (\varepsilon_1,...,\varepsilon_n) = (\underbrace{0,...0}_{a_1-1},1,...,1,\underbrace{0,...,0}_{a_r-1},1).
\end{equation}
Then Kontsevitch's formula is:
\begin{equation}
 \zeta(a_1,...,a_r) = \int_{0\leq t_1<...<t_n\leq1}\omega_{\varepsilon_1}(t_1)...\omega_{\varepsilon_n}(t_n)
\end{equation}
with $\omega_0(t)=dt/t$ and $\omega_1(t)=dt/(1-t)$. To get such iterated integrals in our problem, we have to invert the order of integration. To simplify the notation, define $\Delta^n_x\subset\mathbb{R}^n$ by
\begin{equation}
 \Delta^n_x = \{(x_1,...,x_n)\in\mathbb{R}^n|0\leq x_1<...<x_n\leq x\}
\end{equation}
Then we can reverse the order of integration in $F^{I_n}_{0,\xi}$ and take strict inegalities since it will only modify a negligible subset of points.
\begin{equation}
 F^{I_n}_{0,\xi} = \int_{\Delta^n_{\xi}}\left[\prod_{k=0}^{n-1}f_{i_{k+1}}(x_{n+1-k},x_{n-k})\right]\hat{\gamma}(x_1)\d x_1...\d x_n
\end{equation}
with $x_{n+1}:=\xi$. Then it is easy to define a morphism of functionals $\alpha$ which will recast our sums over $\{I_n\}$ as sums over MZVs.
\begin{align}
                \alpha_{\xi}: \mathcal{F}(\mathcal{C}^{\infty}(\mathbb{R}^2)) & \longrightarrow \mathcal{F}(\mathcal{C}^{\infty}(\mathbb{R})) \nonumber \\
 \int_{\Delta^n_{xi}}T(x_1,...,x_{n})\hat{\gamma}(x_{n-1}-x_n)\d x_n...\d x_1 & \longrightarrow \int_{\Delta_{1}^{n+1}}\beta\left[T(x_1,...,x_{n})\right]\omega_1(x_{n+1})\d x_{n+1}...\d x_1
\end{align}
with $\beta:\mathcal{C}^{\infty}(\mathbb{R}^2)\longrightarrow\mathcal{C}^{\infty}(\mathbb{R})$ defined by $\beta(ab) = \beta(b)\beta(a)$ and
\begin{equation}
 \beta\left[f_i(x_{k+1},x_k)\right] = \omega_i(x_k).
\end{equation}
Since we have ordered our sum by the length of $I_{n}$, we end with a sum over the weights of MZVs.
\begin{equation}
 \partial_{\xi}\hat{\gamma}(\xi) = \frac{-2}{3\xi+1}\left(\hat{\gamma}(\xi)+\sum_{n\geq1}\sum_{w(\zeta)=n+1}\alpha^{-1}_{\xi}[\zeta(a_1,...,a_r)]\right)
\end{equation}
With $(a_1,...,a_r)$ the string of integers having $(\tau(I_{n}),1)$ as its binary representation. $\tau$ is the operator over the semigroup $X$ of words written in the alphabet $\{0,1\}$ which reverses 
the order: $\tau(0)=0$, $\tau(1)=1$ and $\tau(ab)=\tau(b)\tau(a)\quad\forall a,b\in X$.

An interesting fact is that while the summation is over the weight of the MZVs, the contribution to a particular term depends on the number of $f_0$ and $f_1$ it had initially. Hence the contributions 
to these sums will depends on the depth of the MZVs. However, studying precisely how this works is equivalent to trying to (asymptotically) solve (\ref{eq_gamma}). On the other hand, this construction 
unravels some relations amongst the $F^{I_n}_{0,\xi}$.

Also, since we allow for $\zeta(1,...)$, the duality theorem breaks down. However, we can restore it by defining such elements to be self-dual. Another way to deal with this issue would be to modify the definition of $\beta$ to include a $\omega_0$.

\subsection{Numerical analysis}

Now, to study the solution $\hat{\gamma}$ numerically, we have to fix a $\xi$ and compute $\hat{\gamma}(\eta)$ for $\eta$ on the line between the origin of the complex plane and $\xi$ with 
$\hat{\gamma}(0)=1$ and $\hat{\gamma}'(0)=-2$ as initial data. We have to take $\xi$ big enough, i.e. big with respect to the periodicity of the singularities of $\hat{\gamma}$, that is $1/3$. 
$\xi$ should also not be too close to the real line for our analysis to not be spoiled by the singularities of $\hat{\gamma}$ that are known to lie on the real line. This is why we have done our 
computations with $\xi=40+35i$, which is not too big so that the algorithm runs in a reasonable time.

The difficulties of the numerical analysis come from the fact that we have to compute convolution integrals that are very sensitive to numerical instabilities. Therefore, standard tools do not 
work for them. We have used Simpson's rule to get the following results from \eqref{SDE_gamma_f} without the $\hat{G}\star\hat{G}$ term.
\begin{figure}
    \includegraphics{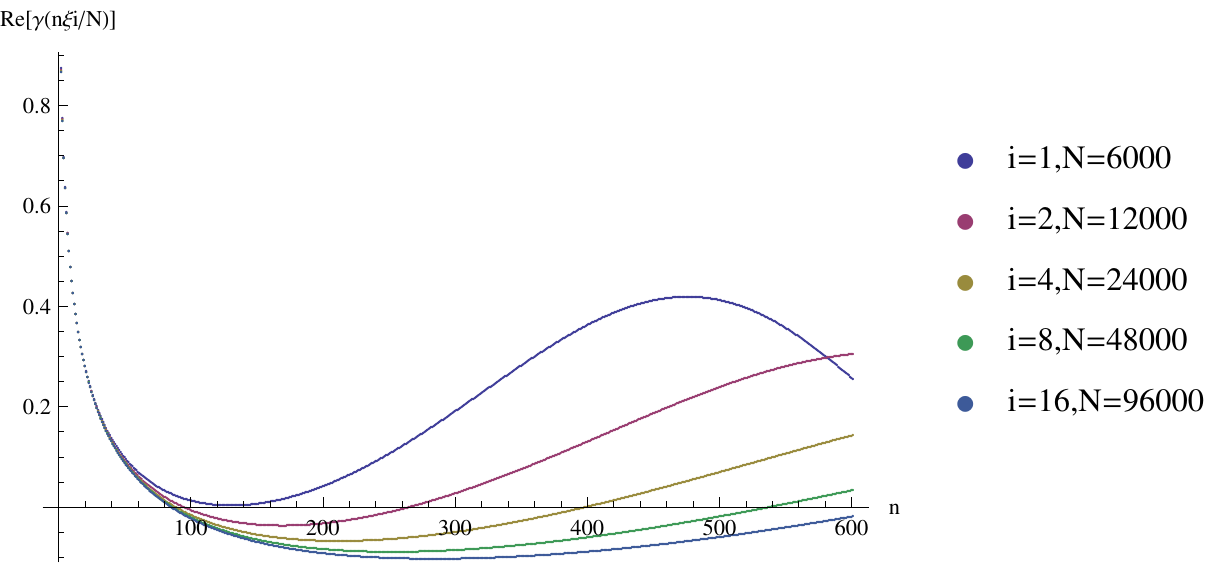}
\caption{Real part of $\hat{\gamma}$ for various precisions.} \label{ReGamma}
\end{figure}
It is clear from the above picture that a very small interval is needed in order to avoid numerical instabilities that we can see for the least precise case $i=1, N=6000$. Moreover, the minimum 
of the other curves seems to be a computational artifact since its position varies as the number of points taken increases. Although numerical methods are probably not the best way to tackle 
convolution integrals, we already see that the asymptotic behavior of the real part of $\hat{\gamma}$ seems to be a constant, maybe zero.

For the imaginary part, the same features are found, but the amplitudes are smaller (since the imaginary part of $\hat{\gamma}(0)$ is $0$), making the results harder to read.
\begin{figure}
    \includegraphics{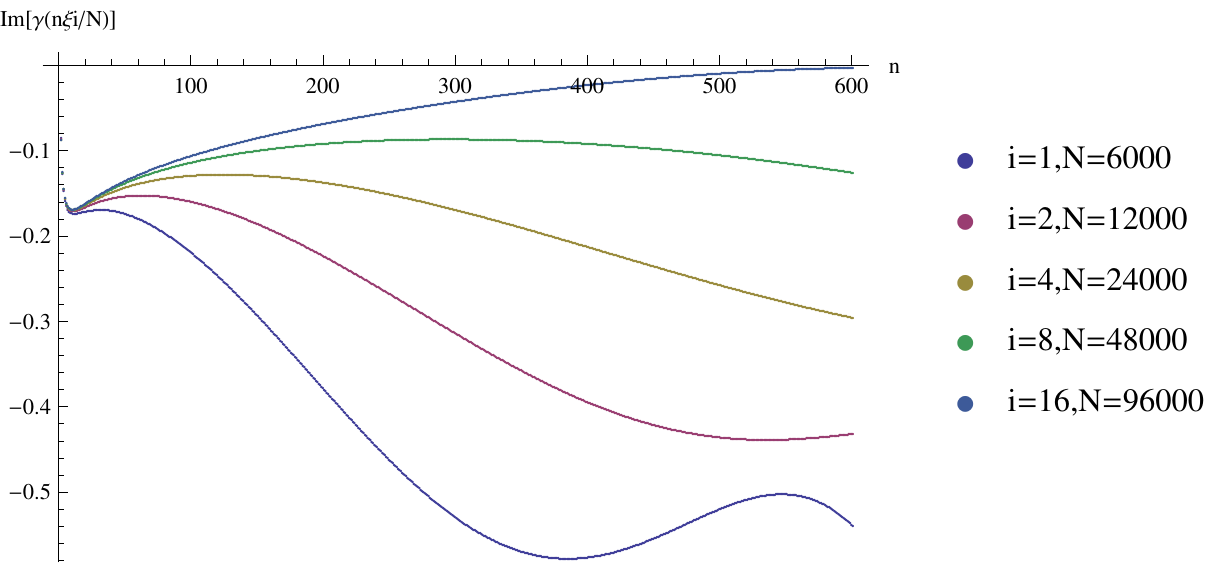}
\caption{Imaginary part of $\hat{\gamma}$ for various precisions.} \label{ImGamma}
\end{figure}
Hence, this numerical study suggests that $|\hat{\gamma}|$ is asymptotically bounded by a constant (for a non-real infinity). More precise results would require more sophisticated tools. Since we 
are mainly interested by analytical results, such study was not performed. 

\section{Singularities of the solution}

The time has finally come to turn our attention to the main point of this chapter: the singularities of the Borel transform of the Schwinger--Dyson equation \eqref{SDnlin}. We will start by 
showing that the singularities of $\hat\gamma$ all lies on the real axis, as stated in the previous section. Then, their precise behavior and the transcendental contents of their expansion around them 
will be studied.

\subsection{Localization of the singularities} \label{localization}

Here we will prove that any singularity of $\hat{\gamma}$ is linked to a singularity of $H$. 

First, we need to make an assumption on the singularities of $\hat{\gamma}$. We will assume that they are of the type studied in section \ref{backto}. We call such singularities algebraic and they 
are characterized by the exponent \(\beta\) which we call their order. This assumption is quite natural since the singularities studied in \cite{BeCl13} are indeed algebraic in this sense. For 
now, we will prove that any algebraic singularity of $\hat{\gamma}$ has to correspond to a singularity of $H$. Hence, if $\xi_0$ is a (algebraic) singularity of $\hat{\gamma}$ we will write:
\begin{equation}
 \hat{\gamma}(\xi) \underset{\xi\rightarrow\xi_0}{\sim} c(\xi-\xi_0)^{\beta}
\end{equation}
with $c$ a constant. Strictly speaking, this is not an equivalence in the usual meaning of the symbol: if $\beta$ is a non negative integer there is a
logarithmic factor. Furthermore, for positive real part of \(\beta\), the difference between the two terms can be any function holomorphic in the neighborhood
of \(\xi_0\). Moreover, the derivative of \(\hat\gamma\) will be equivalent in the same sense to \(c\beta(\xi-\xi_0)^{\beta-1}\), except in the case
\(\beta=0\) where we forget the factor \(\beta\). The virtue of our definition of an algebraic singularity is that one has not to take care if there are logarithms or not at the singularity.

We can deduce many things from the equation \eqref{renorm_f}. Indeed, let us assume that $\hat\gamma$ has a singularity of order $\beta$ at $\xi=\xi_0$. Then, $\forall\zeta\neq\xi_0$, the RHS of 
\eqref{renorm_f} vanishes if $\xi\longrightarrow f(\xi,\zeta)$ is regular at $\xi=\zeta$ while the LHS has no reason to do so. Hence, $\xi\longrightarrow f(\xi,\zeta)$ is singular in $\xi=\zeta$.
Now, let us define $f_{\zeta}:=\xi\longrightarrow f(\xi,\zeta)$ and take a derivative with respect to $\xi$ in \eqref{renorm_f}. It comes
\begin{equation}
 -3f_{\zeta}(\xi) + 3(\zeta-\xi)f_{\zeta}'(\xi) = \hat\gamma'(\xi) + f_{\zeta}(\xi) - 6\xi f_{\zeta}(\xi) + \int_0^{\xi}\hat\gamma'(\xi-\eta)f(\eta)\d\eta + 3\int_0^{\xi}\hat\gamma''(\xi-\eta)\eta f_{\zeta}(\eta)\d\eta.
\end{equation}
Since $\hat\gamma'$ and $\hat\gamma''$ are regular in zero, the two last terms have a singularity of order $\alpha+1$ and are therefore subdominant. Moreover $\hat\gamma$ is regular at 
$\xi=\zeta\neq\xi_0$ by hypothesis. Thus
\begin{equation}
 -3f_{\zeta}(\xi) + 3(\zeta-\xi)f_{\zeta}'(\xi) \simall{\xi}{\zeta} f_{\zeta}(\xi) - 6\xi f_{\zeta}(\xi).
\end{equation}
Hence, writing the singularity of $f_{\zeta}$
\begin{equation}
 f_{\zeta}(\xi) \simall{\xi}{\zeta} A(\xi-\zeta)^{\alpha} 
\end{equation}
we obtain
\begin{equation}
 3-3\alpha = 1 - 6\xi \Leftrightarrow \alpha = -\frac{4}{3} + 2\zeta
\end{equation}
since we work at $\xi\sim\zeta$.

The second case is simpler: $\forall\zeta\neq\xi_0$ $\xi\longrightarrow f(\xi,\zeta)$ has a singularity at $\xi=\xi_0$ of order $\beta$ (i.e.: like $\hat{\gamma}$) since only $\hat\gamma$ is 
singular in the RHS of \eqref{renorm_f}. The third possibility is a combination of the two first ones, with $\zeta=\xi_0$. The two exponents  \(\beta-1\) and \(-4/3+2\xi_0\) appear possible, but 
such a situation requires a case by case study.

As a function of its second argument and for any value of \(\xi\) which is not singular for \(\hat\gamma\), the function $f(\xi,\zeta)$ has a singularity of order $-4/3+2\xi$ at $\zeta=\xi$, since 
at the singularity this function coincides with $f_{\zeta}$. Let us emphasize that the function $\zeta\longrightarrow f(\xi,\zeta)$ has only singularities at \(0\) and \(\xi\) and is in 
particular regular at $\zeta=\xi_0$.

Now, let us assume that $\xi_0$ is an algebraic singularity of $\hat{\gamma}$ and that $H(3\xi_0,0)$ is not singular. Then the integral over $\eta$ of the first integral of 
\eqref{SDE_gamma_f} does not have to cross any singularity of $H$. We can deform its integration contour, then Jordan's lemma gives
\begin{equation*}
 \oint_{\mathcal{C}_{\eta}}\d\zeta\frac{f(\eta,\zeta)}{\zeta(1+3\zeta)} = -\text{Res}\left(\frac{f(\eta,\zeta)}{\zeta(1+3\zeta)},\zeta=-1/3\right).
\end{equation*}
$\xi_0\neq-1/3$ (since $(-1,0)$ is a singularity of $H$). Furthermore $\eta$ is running from $0$ to $\xi$, and $\xi\rightarrow\xi_0$. Then $\zeta\longrightarrow f(\eta,\zeta)$ is regular at 
$\zeta=-1/3$. Hence
\begin{equation} \label{contm1}
 \oint_{\mathcal{C}_{\eta}}\d\zeta\frac{f(\eta,\zeta)}{\zeta(1+3\zeta)} = f(\eta,-1/3).
\end{equation}
According to \eqref{SDE_gamma_f}, $\eta\longrightarrow f(\eta,-1/3)$ has a singularity in $\xi_0$, but that singularity is of the same order than the singularity of $\hat{\gamma}(\xi)$. Since we 
have assumed this singularity to be algebraic, $\int\d\eta f(\eta,-1/3)$ is less singular than $\hat{\gamma}(\xi)$. Hence the first integral \eqref{SDE_gamma_f} is not
sufficient to allow a singularity of $\hat{\gamma}(\xi)$ at \(\xi_0\). However, let us notice that this construction tells us that this integral will
give a dominant contribution to the singularity at $\xi=-1/3$ of $\hat{\gamma}(\xi)$.

For the second integral, using the fact that the alien derivative is a derivative with respect to the convolution product we get a relation between the singular part of $\hat{\gamma}$ and of $f$:
\begin{equation}
 \Delta_{\xi_0}\hat{\gamma}(\xi) \sim \int_0^{\xi}\d\eta\oint_{\mathcal{C}_0}\frac{\d\zeta}{\zeta}\oint_{\mathcal{C}_{\xi_0}}\frac{\d\zeta'}{\zeta'}H(3\zeta,3\zeta')f(-,\zeta)\star\Delta_{\xi_0}f(-,\zeta').
\end{equation}
Since in this equation, we are only interested in the behavior of \(\Delta_{\xi_0} \hat\gamma(\xi)\) in the vicinity of the origin, the integration
contour for \(\zeta\) can be a fixed one around 0 and the one for \(\zeta'\) a fixed contour enlacing 0 and \(\xi_0\).  
In the last loop integral, if \(H(0,3\xi)\) is not singular for any value of \(\xi\) on the straight line from \(0\) to \(\xi_0\), the
contour can be freely deformed to one contour \(\mathcal{C}_{\xi_0}\) which does not touch $\xi_0$.  Therefore, in the convolution integral,
\(\Delta_{\xi_0}f(\xi,\zeta')\) is of order \(\beta\) for all \(\zeta'\) on the contour. Then at least two 
integrals are taken from the convolution product and the explicit integration and since the loop integrals do not modify the singularity we end up
with a singularity of order $\beta-2$. The hypothesis that \(\hat\gamma\) has a singularity of order \(\beta\) is therefore incoherent, since we have
shown that it is equal to the sum of two terms which are less singular.

In the case where \(H(0,3\xi_0)\) is singular,  this argument does not hold:  we cannot deform the contour to include $\xi_0$ without modifying the value of 
the integral. Hence, when $\xi\rightarrow\xi_0$, the contour is pinched between $\xi$ and $\xi_0$ and there is a contribution from
$f(\xi,\zeta=\xi_0)$. 

Hence we have proven half of the claim made in the previous section: the singularities of $\hat\gamma$ can only be at a pole of $H(3\zeta,3\zeta')$. For the simple poles\footnote{we use the 
terminology of the subsection \ref{contribution}}, we have that $-\mathbb{N}^*/3$ is the singular locus of $\hat\gamma$ because, according to the analysis of \cite{BeCl13}, $\hat\gamma$ has a singularity in $-1/3$, 
which creates a singularity in $-2/3$, and so on. For the general poles of $H(3\zeta,3\zeta')$, the situation is subtler and will be fully performed in subsection \ref{PosSing}.

\subsection{Study of the negative singularities} \label{negative}

We will now study the behavior of $\hat{\gamma}$ near the singularities on the negative real axis. We will use the equations \eqref{renorm_f} and \eqref{SDE_gamma_f}. First, let us show that the 
function $\xi\longrightarrow f(\xi,\zeta)$ can be expressed near $0$ as
\begin{equation} \label{f_zero}
 f(\xi,\zeta)\simall{\xi}{0}\sum_{p=1}^{+\infty}\frac{\hat{\gamma}_p(\xi)}{(3\zeta)^p}.
\end{equation}
First, we have
\begin{equation*}
 \oint_{\mathcal{C}_{\xi}}\frac{e^{3\zeta L}}{(3\zeta)^p}\frac{\d\zeta}{\zeta} = \sum_{n=0}^{+\infty}\frac{(3L)^n}{n!}\oint_{\mathcal{C}_{\xi}}\frac{\zeta^{n-p-1}\d\zeta}{3^p} = \sum_{n=0}^p\frac{(3L)^n}{n!}\text{Res}\left[\zeta^{n-p-1}/3^p,\zeta=0\right] = \frac{L^p}{p!}
\end{equation*}
we get, when using \eqref{f_zero} in \eqref{param_G}
\begin{equation}
 \hat{G}(\xi,L) = \sum_{p=1}^{+\infty}\hat{\gamma}_p(\xi)\frac{L^p}{p!}.
\end{equation}
And this is exactly the Borel transform of $\tilde{G}=G-1$ when $G$ is defined by \eqref{2ptB}, i.e. around zero. Let us remark that the above expression is well-defined, as a formal series in 
$\xi$, only for $\xi$ near the origin of the complex plane. When $\xi$ goes to a vicinity of a singularity of $\gamma$ all the $\hat{\gamma}_p$ have the same kind of singularity and therefore the 
above expression is no longer clearly convergent.

From equation \eqref{contm1}, the term linear in $\hat{G}$ in \eqref{SDE_gamma_f} will give a contribution proportional to \(f(\xi,-1/3)\) and will make the case \(\xi_0=-1/3\) special.  
For now on, we will focus on the cases $\xi_0\neq-1/3$. To study the contribution of the term $\hat{G}\star\hat{G}$ of equation \eqref{SDE_gamma_f} to a negative singularity of $\hat{\gamma}$, 
let us split $H(3\zeta,3\zeta')$ between a regular and a singular part:
\begin{equation} \label{H_split}
  H(3\zeta,3\zeta') = \tilde{H}_k(3\zeta,3\zeta') + \sum_{l=1}^k\biggl(\frac{P_l(3\zeta)}{3\zeta'+l} + \frac{P_l(3\zeta')}{3\zeta+l} \biggr)                                                                                                 
\end{equation}
with the same $P_k$ than the ones mentioned in the previous chapter: they are just the residues of $H$. Since the term $\tilde{H}_k$ is regular up to \(3\zeta = -k-1\), the integration contours can be deformed and it will not give dominant contributions to the singularity of $\hat{\gamma}$: it is 
the same analysis than the one done to localize the singularities of $\hat{\gamma}$ in the previous subsection. The singular term being simple rational functions, we can once again compute the 
integral on the pole part by using Jordan's lemma.
\begin{equation} \label{pole}
 \oint_{\mathcal{C}_{\sigma}}\d\zeta'\frac{f(\sigma,\zeta')}{\zeta'}\frac1{3\zeta'+k} = \frac{f(\sigma,-k/3)}{k}
\end{equation}
This equality is established for \(\sigma\) in the vicinity of 0, but can be extended by analytic continuation.  Similarly, the integration for a monomial \( (3\zeta)^m\) can be easily established 
for \(\sigma\) in the vicinity of 0 using \eqref{f_zero} and extended to the whole Borel plane:
\begin{equation} \label{monom}
 \oint_{\mathcal{C}_{\sigma}}\d\zeta\frac{f(\sigma,\zeta)}{\zeta} (3\zeta)^m = \hat\gamma_m(\sigma).
\end{equation}
Now, we only want the most singular part of the quadratic in \(f\) term. This cannot come from the regular part of \(H\) and the contributions of the
poles can be written using \eqref{pole}) and \eqref{monom} as a sum of terms \(\hat\gamma_m\star f(-,-l/3)\).  For the singularity in \(-k/3\), the
most singular of these terms is \(\hat\gamma\star f(-,-k/3)\) and we obtain:
\begin{equation}
 \hat{\gamma}(\xi)\simall{\xi}{\xi_0}-2\int\d\eta\int\d\sigma\; \hat{\gamma}(\eta-\sigma)f(\sigma,-k/3)\frac1{k(k-1)}
\end{equation}
since the linear part of \(P_k\) is \(-x/(k-1)\).
Hence, if $\hat{\gamma}$ has a singularity of order $\beta_k$ at $\xi=-k/3$, 
then $f(\xi,-k/3)$ has to have a singularity of order $\beta_k-2$. 
We then get a relation between $c_k$, the leading coefficient of $\hat{\gamma}$ and $f_k$, the leading coefficient of $f(\xi,-k/3)$:
\begin{equation} \label{rel_c_f2_int}
 c_k = \frac{2}{k(k-1)}\frac{f_k}{\beta_k(\beta_k-1)}.
\end{equation}
To find $\beta_k$, we are a priori in the complicated case where \(\hat\gamma\) has a singularity for the value of \(\zeta\).  However, since
the order of \(\hat\gamma\) is small enough,  the renormalization group equation (\ref{renorm_f}) at its most singular order \(\beta_k-1\) takes the simple form:
\begin{equation*}
 3 f_k = \frac{- f_k}{\beta_k-1} + \frac{6\xi_0 f_k}{\beta_k-1}
\end{equation*}
 where we have used $\hat{\gamma}(0)=1$ and $\hat{\gamma}'(0)=-2$. Using $\xi_0=-k/3$ we get
\begin{equation}
 \beta_k = -\frac{2}{3}(k-1).
\end{equation}
Hence the relation \eqref{rel_c_f2_int} becomes
\begin{equation}
 c_k = \frac{9}{k(k-1)^2(2k+1)}f_k.
\end{equation}
Now, let us go back to the case $\xi_0=-1/3$. From the previous analysis, the most singular term in the Schwinger--Dyson equation \eqref{SDE_gamma_f}
is the one linear in \(f\), so that \(f(\xi,-1/3)\) must be of order \(\beta_1-1\) and we have the following relation between the leading coefficients
around \(\xi_0\) of \(\hat\gamma\) and \(f(\xi,-1/3)\):
\begin{equation} \label{rel_coef_premier}
 \beta_1 c_1 = 2 f_1.
\end{equation}
In the renormalization group equation \eqref{renorm_f}, the leading singularity is now of order \(\beta_1\) and its coefficient includes a contribution from $\hat{\gamma}$:
\begin{equation*}
 -3f_1 = c_1 + \frac{f_1}{\beta_1}-6\left(-1/3\right)\frac{f_1}{\beta_1}.
\end{equation*}
Then using \eqref{rel_coef_premier} we get
\begin{equation}
 \beta_1 = -5/3, \quad \quad c_1 = -\frac{6}{5}f_1,
\end{equation}
in conformity with the result found in \cite{BeCl13} (and implied by the previous chapter).

\subsection{Study of the positive singularities} \label{PosSing}

For the positive singularities, the previous analysis has to be modified. Indeed, the denominators in the poles are of the form \(k - 3\zeta
-3\zeta'\) and the residues as a function of \(\zeta'\) would involve $f(\xi,k/3-\zeta)$. When performing the 
second contour integral, the relation \eqref{f_zero} then tells us that we would have to take derivatives of $f$ with respect to its second argument, and the renormalization group equation 
\eqref{renorm_f} implies that those derivatives are as singular as the first term, so all of them would need to be taken into account. This would make the analysis intractable in practice.

In the previous chapter, we have determined a renormalization group like equation satisfied by the contribution \(L_k\) stemming from a pole term
in the Mellin transform \(H\). We will simply translate this equation in the Borel plane. The positive pole of order \(k\) was written 
\begin{equation*}
 \frac{Q_k(xy)}{k-x-y},
\end{equation*}
with the residue written in term of a polynomial \(Q_k\) of degree \(k\):
\begin{equation*}
 Q_k(X) = \sum_{i=1}^{k}q_{k,i}X^i.
\end{equation*}
Then the equations for the $L_k$ functions are:
\begin{equation}
 (k-2\gamma-3\gamma a\partial_a)L_k = \sum_{i=1}^kq_{k,i}\gamma_i^2.
\end{equation}
Using the rules of the Borel transform, we map this equation into the Borel plane.
\begin{equation*}
 k\hat{L}_k-2\hat{\gamma}\star\hat{L}_k-3\hat{\gamma}\star\partial_{\xi}\left(\xi\hat{L}_k\right) = \sum_{i=1}^kq_{k,i}\hat{\gamma}_i\star\hat{\gamma}_i
\end{equation*}
As in the renormalization group equation, we integrate by parts the second convolution integral, using once again $\hat{\gamma}(0)=1$ to get 
\begin{equation} \label{eqLk}
 (k-3\xi)\hat{L}_k(\xi) = 2\hat{\gamma}\star\hat{L}_k(\xi) + 3\hat{\gamma}'(\xi)\star\left(\text{Id}.\hat{L}_k\right)(\xi) 
 	+ \sum_{i=1}^k q_{k,i}\hat{\gamma}_i\star\hat{\gamma}_i
\end{equation}
where the `$.$' in the second convolution integral has to be read as the pointwise product over functions.

If \(\hat L_k\) has a singularity in a point \(\xi_0\), the Schwinger--Dyson equation implies that \(\hat\gamma\) has also a singularity, but with the
order of a primitive of \(\hat L_k\).  Now, near a singularity at $\xi=\xi_0$, let us parametrize the singularity of $\hat{\gamma}$ and 
$\hat{L}_k$ by:
\begin{eqnarray*}
 \hat{\gamma}(\xi) & \simall{\xi}{\xi_0} & \frac{c_k}{\alpha_k}\left(\xi-\xi_0\right)^{\alpha_k} \\
    \hat{L}_k(\xi) & \simall{\xi}{\xi_0} & c_k\left(\xi-\xi_0\right)^{\alpha_k-1}.
\end{eqnarray*}
Now, the question is whether equation \eqref{eqLk} allows a singular \(\hat L_k\).  The right hand side terms are less singular than \(\hat L_k\), so
that the only possibility is when the factor \(k-3\xi\) vanishes: we then have that \(\xi_0 = k/3\). From the recursion relation for the 
$\hat{\gamma}_i$s, it is easy to see that no $\hat{\gamma}_i\star\hat{\gamma}_i$ will contribute. Indeed, the most singular term is for $i=1$ and $\hat{\gamma}\star\hat{\gamma}$ is singular as the second primitive of $\hat{L}_k$.
The most singular terms in \eqref{eqLk} can therefore be written and give:
\begin{equation*}
 c_k(k-3\xi)\left(\xi-\frac{k}{3}\right)^{\alpha_k-1} = 2\frac{\hat{\gamma}(0)}{\alpha_k}c_k\left(\xi-\frac{k}{3}\right)^{\alpha_k} +
 3c_k\frac{\hat{\gamma}'(0)\xi}{\alpha_k}\left(\xi-\frac{k}{3}\right)^{\alpha_k}
\end{equation*}
Now, simplifying this relation, using $\hat{\gamma}(0)=1$ and $\hat{\gamma}'(0)=-2$ and evaluating the remaining at $\xi=k/3$ we end up with a very simple 
formula for $\alpha_k$. 
\begin{equation}
 \alpha_k = \frac{2}{3}(k-1).
\end{equation}
Notice that for the positive singularities, no singularity has to be treated separately. Moreover, for $k=1$, we find $\alpha_k=0$, that is, a logarithmic singularity, as we found in our 
previous work and in section \ref{backto}.

We have shown that for the general poles of $H(3\zeta,3\zeta')$, the induced singularity of $\hat\gamma$ is localized in $\xi\in\mathbb{N}/3$. Hence, we arrive to a beautiful result that we can 
state as a theorem
\begin{thm} \cite{BeCl14}
 The only singularities of $\hat\gamma$ are for $\xi\in\mathbb{Z}^*/3$. They have the orders 
 \begin{subequations}
  \begin{align}
   \beta_k & = -\frac{2}{3}(k-1) \quad \text{for }\xi=-k/3, k\geq2 \\ 
   \alpha_k & = \frac{2}{3}(k-1) \quad \text{for }\xi=+k/3, k\geq1 \\
   \beta_1 & = -5/3 \quad \text{for }\xi=-1/3.
  \end{align}
 \end{subequations}
\end{thm}

\subsection{Transcendental Content of the Borel transform}

Now, a very natural question to ask is what the number-theoretical contents of $\hat{\gamma}$ near its singularities is. However, the equations for the singular parts are linear, so that these 
singular parts are only determined up to a global constant which will be determined by matching with numerical determination of the singularity. Therefore, whenever we speak of the number 
theoretical content or the weight of a coefficient in the expansion of a singularity, we really speak of the ratio of this coefficient with respect to this global constant. In the study of 
\cite{BeCl13} (and of the previous chapter) the first orders were computed in the physical plane around the two first singularities of $\hat{\gamma}$ (i.e. around $\xi_0=\pm1/3$). It was found that the expansion of 
$\hat{\gamma}$ around those poles were rational linear combinations of products of odd zeta values.

Even this simple fact was very technical to prove in the physical plane because it involved the computation of complicated series and identities among MultiZeta Values to show the annulation of 
the other terms (typically: MZVs and Euler sums with higher weights). We will see that it is much simpler to show this result in the Borel plane. We are quite proud of this result, so let us write it down 
properly
\begin{thm}
 The coefficients of the expansion of $\hat\gamma$ around any of his singularities can always be written as rational linear combinations of products of odd Riemann zeta values.
\end{thm}
\begin{proof}
Throughout this proof we will use the splitting \eqref{H_split} and replace $H$ by the relevant $\tilde{H}_k$ or its equivalent for the positive singularities, since the polar parts do not 
change the transcendental content of the equation. Moreover, getting rid of the polar parts allows to evaluate $\tilde{H}$ (which will denote the properly subtracted \(H\) in each case) at the 
singular point of $H$. Now, from the renormalization group equation \eqref{renorm_f} with $\xi$ near a singularity, we see that one can expand $f(\xi,\zeta)$ near a singularity:
\begin{equation} \label{f_all}
 f(\xi,\zeta)\sim \sum_{\substack{r\geq0 \\ s\geq1}} \frac{\psi_{r,s}(\xi)} {\zeta^r(\zeta-\xi_0)^s}
 \sim \sum_{\substack{r\geq0 \\ s\geq1}}\sum_{n\geq0}\frac{\psi_{r,s}^{(n)}} {\zeta^r(\zeta-\xi_0)^s} (\xi-\xi_0)^{\alpha_k+r+s-1+n}
\end{equation}
with $\psi_{r,s}^{(n)}\in\mathbb{C}$. This comes from writing the L.H.S. of \eqref{renorm_f} as $3(\xi_0-\xi+(\zeta-\xi_0))$. The $1/\zeta^r$ terms
come from the expansion \eqref{f_zero} of $f(\eta,\zeta)$ with $\eta$ near $0$, which get multiplied by the singular part of \(\hat\gamma\) or \(\hat\gamma'\).
Using this in the Schwinger--Dyson equation \eqref{SDE_gamma_f} for $\xi\rightarrow\xi_0$, we see that the  loop integral in the factor where \(f\)
is singular in \(\xi_0\) will give derivatives of $\tilde{H}$, evaluated at 
$(3\zeta,0)$ and $(3\zeta,3\xi_0)$. The other \(f\) has only to be taken in the vicinity of 0 so that the expansion \eqref{f_zero} can be used, and
the second contour integral will ensure that we only have to evaluate \(\tilde H_k\) together with its derivatives at the points \((0,0)\) and
\((0,3\xi_0)\). Using the representation \eqref{Hzeta} for $H(x,y)$ we see that
one can rewrite once again the Mellin transform as:
\begin{equation} \label{H_zeta}
 H(3\zeta,3\zeta') = \frac{1}{1+3\zeta+3\zeta'}\exp\left(2\sum_{k=1}^{+\infty}\frac{\zeta(2k+1)}{2k+1}\left((3\zeta+3\zeta')^{2k+1}-(3\zeta)^{2k+1}-(3\zeta')^{2k+1}\right)\right)
\end{equation}
and $\tilde{H}$ only differs from $H$ by rational terms around \((0,0)\), so that its derivatives have the same transcendental contents as the above expression.
When taking values around \((0,3\xi_0)\), we can use the functional relation on \(\Gamma\) and the fact that \(3\xi_0\) is an integer to show that \(\tilde H(x,k+y)\) is a rational multiple,
with rational coefficients, of \(H(x,y)\), again up to the addition of a rational fraction with rational coefficients. Therefore, in every cases, the
only transcendental numbers which can appear are the odd zeta values, with a total weight which is bounded by the total number of derivatives.  
Using this information in a recurrent determination of the higher order correction to the singular behavior of \(\hat\gamma\) and all the coefficients
\(\psi_{r,s}^{(n)}\), we see that only these transcendental numbers can appear.  Hence we have proved that the expansions of $\hat{\gamma}$ around its singularities have no even zeta values, nor 
MultiZeta Values that cannot be expressed as $\mathbb{Q}$-linear combinations of products of odd zetas.
\end{proof}

I hope that the reader agrees with me when I am saying that this is already a nice result. However, a quite striking remark made in the physical plane is that the weights of those odd zetas were 
lower than expected at a given order. Now, let us show that our study in the Borel plane allows us to put a bound on those weights that is saturated by the weights found in \cite{BeCl13}.

\subsection{Weight of the odd Zetas}

Now, let us try to be more specific and get a bound on the weights of the different coefficients. To study the expansion of $\hat{\gamma}$, let us expand it around a singularity:
\begin{equation}
 \hat{\gamma}(\xi) \simall{\xi}{k/3}\sum_{p=0}^{+\infty}c_k^{(p)}(\xi-k/3)^{\alpha_k+p}.
\end{equation}
We will also need the expansion of $\hat{\gamma}$ around $0$
\begin{equation}
 \hat{\gamma}(\xi) \simall{\xi}{0} \sum_{p=0}^{+\infty} c_p\xi^p.
\end{equation}
Let us define the usual weight function defined by $w\left(\zeta(n)\right)=n,w(a.b)=w(a)+w(b), w(0)=-\infty$ and $w(a+b)=\max\{w(a),w(b)\}$. Then we have a simple lemma
\begin{lemm} \cite{BeSc08}
 We have $w(c_1) = w(c_2)= 0$ and $w(c_4)=3$. For all the other $p$ we have $w(c_p)=p$.
\end{lemm}
\begin{proof}
 The proof is trivial from the Borel plane version of the Schwinger--Dyson equation \eqref{SDE} and of the renormalization group equation \eqref{recursion_gamma}: we just have to notice that the 
 $n^{\text{th}}$ derivative of $H$ will bring a $\zeta(n)$, thanks to \eqref{H_zeta}. The cases $p=1,2,4$ have to be corrected due to the lack of odd zetas of weight $0, 1$ or $4$.
\end{proof}
In the general cases, computations must be done with alien derivatives since the expansion around the other singularities of any particular analytic continuation will involve terms stemming from 
iterated alien derivatives. This does not really change the computations, but the present formulation becomes inexact. Therefore, we will work here only for the two first singularities of 
$\hat{\gamma}$ and work out explicitly the case $k=+1$.

Since $\hat L_1(\xi)$ carries the most singular contribution to $\hat{\gamma}(\xi)$ for $\xi\sim \xi_0=1/3$, it is natural to assume that it will also carry the zetas of highest weight. So let us 
expand it around this singularity:
\begin{equation}
 \hat L_1(\xi) \simall{\xi}{\xi_0} \sum_{n=0}^{\infty}L_1^{(n)}(\xi-\xi_0)^{\alpha_1+n-1}.
\end{equation}
The study of \ref{PosSing} implies that the singular term of order $\alpha_1+n-1$ in $\hat L_1$ contributes to the singular term of order $\alpha_1+n$ in $\hat{\gamma}$. Our aim is to show that other contributions of this order to \(\hat\gamma\) are of lower weight so that the weights in \(\hat L_1\) determine those in \(\hat\gamma\).
To study the weight of $L_1^{(n)}$ we use the renormalization 
group equation (\ref{eqLk}).  In the neighborhood of \( \xi_0 \), the pointwise multiplication by \( \xi \) does not lower the order of the singularity, so that the most singular part of the RHS of eq.~(\ref{eqLk}) comes from the term convoluted with $\hat{\gamma}'$, with \(\hat L_1\) the singular factor.  Indeed, \(L_1\) is of order 2 at the origin, so that \(\hat L_1\) vanishes at \(\xi=0\). Multiplication by \(k - 3\xi\) in the LHS lower the order by one so that we end up with
\begin{equation} \label{recurrence_poids}
 w(L_1^{(n)}) \leq \text{max}_{p\in[2,n+1]}\{w(c_p)+w(L_1^{(n-p+1)})\}
\end{equation}
from the term proportional to \( (\xi-\xi_0)^{\alpha_1+n} \).
Since the \(c_p\) appears in the relation between the coefficients of order differing by \(p-1\), the weight of \(L_1^{(n)} \) cannot be simply \(n\). However, a weight like \(3n/2\) allows terms which are not possible.  For example, at level \(2n\), \(\zeta(3)^n \) is the only term of weight \(3n\).

In fact, there is a way to describe exactly the terms which can appear in \(L_1^{(n)}\). 
We define a modified weight system \(W\) such that $W(\zeta(2n+1)) = 2n$. With this modified weight, \(c_p\) is of weight \(p-1\) and  equation (\ref{recurrence_poids}) shows that \(L_1^{(n)}\) is of maximal weight \(n\). In fact, since the weight of the odd zetas is even, all weights are even and additional terms can only appear for even orders.

Now, let us check that the contributions from all other terms have a smaller weight. In order to simplify notations and computations, we will extend the weight \(W\) to formal series by defining:
\begin{equation}\label{WeiSeries}
	W( \sum_{p=0}^\infty a_p \xi^p ) = \sup_p \bigl( W(a_p) - p \bigr).
\end{equation}
It is now easy to show that the weight of a convolution product is bounded by the weights of its factors:
\begin{equation} \label{WeiPro}
	W ( \hat f  \star \hat g ) \leq W(\hat f) + W(\hat g) -1.
\end{equation}
Let us remark that a negative weight implies that  the first terms in the series are zero.
We will also need a similar definition around a singularity \(\xi_0\), defining the weight function \(W_{\xi_0} \) from the weights of the expansion of a function around \(\xi_0\). Here the definition will depend on the reference exponents \(\alpha_k\).  For example, we will have that 
\begin{equation}
W_{k/3}(\hat \gamma) = \sup_p (c_k^{(p)} - p ).
\end{equation}
Using the properties of the singular part of a convolution product, we can generalize formula~(\ref{WeiPro}) to
\begin{equation}\label{WeiProS}
  W_{\xi_0} ( \hat f  \star \hat g ) \leq \max(W_{\xi_0}(\hat f) + W(\hat g)-1, W(\hat f) + W_{\xi_0}(\hat g) -1) 
\end{equation}
The hypothesis we want to prove take the simple form 
\begin{equation}
 	W_{1/3} ( \hat \gamma ) = 0.
\end{equation}

Let us suppose that this is the case. Using the weights of the convolution products, Eq.~(\ref{eqLk}) shows that \(W(\hat L_1) = -1\), and then that, with our hypothesis, \(W_{1/3}(\hat L_1) = +1\) since \(W(\hat\gamma') =0\) and \(W_{1/3}(\hat\gamma') = + 1\). The difficult part is to show that the additional terms in the Schwinger--Dyson equation~(\ref{SDE_gamma_f}) depending on the subtracted Mellin transform \(\tilde H_1\) are really subdominant. 

We will need the weight of $\hat{\gamma}_n$, which can be easily deduced from its recursive definition and the relation~(\ref{WeiPro})
\begin{equation}
 W(\hat{\gamma}_n) = 1-n.
\end{equation}
Now, using this expansion, the representations (\ref{f_zero}) and (\ref{f_all}) of $f(\xi,\zeta)$ and the splitting (\ref{H_split}) of $H$ in the Schwinger--Dyson equation 
(\ref{SDE_gamma_f})\footnote{more precisely, in the term of (\ref{SDE_gamma_f}) quadratic in $\hat{G}$ since the linear one will not 
bring any new zeta.}. We end up with
\begin{equation}
 \partial_{\xi}\hat{\gamma}(\xi) \simall{\xi}{\xi_0} \sum_{p\geq1}\sum_{r\geq0}\sum_{s\geq1} \bigl(h_r^p+\tilde{h}_s^p \bigr) \hat\gamma_p \star \psi_{r,s}
\end{equation}
with the equivalence sign meaning here up to rational terms. The quantities $h_r^p$ and $\tilde{h}_s^p$ are defined by:
\begin{subequations}
 \begin{eqnarray}
          h_r^p & := & \left.\frac{\d^p}{\d\zeta^p}\left(\sum_{i=0}^r q_i^r\frac{\d^i}{\d\zeta'^i}\tilde{H}(3\zeta,3\zeta')|_{\zeta'=0}\right)\right|_{\zeta=0} \\
  \tilde{h}_s^p & := & \left.\frac{\d^p}{\d\zeta^p}\left(\sum_{i=0}^{s-1} \tilde{q}_i^s\frac{\d^i}{\d\zeta'^i}\tilde{H}(3\zeta,3\zeta')|_{\zeta'=\xi_0}\right)\right|_{\zeta=0}
 \end{eqnarray}
\end{subequations}
with $q_i^r,\tilde{q}_i^s\in\mathbb{Q}$. \(h_r^p\) (resp.\ \(\tilde h_s^p\)) have therefore a weight bounded by \(p+r\) (resp.\ \(p+s-1\)). 

The only thing that is left to find is $W_{1/3}(\psi_{r,s})$ from the renormalization group equation~(\ref{renorm_f}).  One readily obtains that this weight is bounded by \(1-r-s\).  Using that \(s\) is bounded below by 1 and the law for the convolution products, we find that every terms in the sum have weight less than or equal to 0. The weight in \(\xi_0\) of \(f(\xi,-1/3)\) is also bounded by 0 so we verify that all these terms have subdominant weights with respect to \(\hat L_1\).

The case with \(\xi_0 =-1/3\) is quite similar: the only real difference is that, due to the presence of \(f(\xi,-1/3)\) in the right-hand side of the Schwinger--Dyson equation, each successive coefficient in the expansion of \(\hat\gamma\) comes from a system of equations derived from this Schwinger--Dyson equation and the renormalization group equation for \(f\). 

Hence, in practice, we have proven (in a quite sketchy way, admittedly but the detailed computations are seriously ugly and not much more informative than our discussion above)
\begin{equation}
 W_{1/3}(\hat\gamma) = W_{-1/3}(\hat\gamma) = 0.
\end{equation}
When translated into a nice theorem, we have
\begin{thm} \cite{BeCl14}
 With the modified weights system $W(\zeta(2n+1))=2n$ and all the other usual properties of a weight we have
 \begin{equation}
  W(c_{\pm1}^{(p)})\leq p.
 \end{equation}
\end{thm}
These weight limits are exactly the ones observed in the previous chapter: in the numbers $d_4$ and $e_4$ (which correspond to $c_{-1}^{(4)}$ and $c_{+1}^{(4)}$ respectively) we found the numbers 
$\zeta(3)^2$ and $\zeta(5)$, which are the only products of zetas of modified weight $4$.

Using the formalism of alien derivations, these results should generalize to the other singularities. This study is left for the next Ph.D. student. Indeed, it is now the time to end our 
journey in the land of the Schwinger--Dyson equations, and to move to the scary country of gauge theories.

%
%
\chapter{BV formalism}

 \noindent\hrulefill \\
 {\it
Ses purs ongles très-haut dédiant leur onyx, \\
L'Angoisse, ce minuit, soutient, lampadophore, \\
Maint rêve vespéral brûlé par le Phénix \\
Que ne recueille pas de cinéraire amphore\\

Stéphane Mallarmé. Le sonnet en X.}

 \noindent\hrulefill
 
 \vspace{1.5cm}

This chapter aims to present the second important topic of this Ph.D.: the Batalin--Vilkovisky (BV) formalism. This formalism allows to treat (more) rigorously theories with gauge symmetries, 
even with open gauge symmetries. The canonical example of a physical theory with open symmetries is supergravity without auxiliary fields, and this feature seems to have been noticed for the 
first time in \cite{FrNiFe76}.

The following presentation of the BV formalism is an incomplete account of works done with Nguyen Viet Dang, Christian Brouder, Fr\'ed\'eric H\'elein, Camille Laurent-Genoux, Serguei Barannikov... I 
am very thankful to all of them for all the hours of discussions we have had, and I wish to thank them once again for having motivated me to study this wonderful subject.

\section{Foreword}

The traditional approach when we write about a subject is to define all the needed objects and then to move to the presentation. Here, due to the length of the definitions, we will assume that our reader has 
some knowledge of topology, differential geometry, homology and cohomology. Some definitions will still be given, but rather to fix the notations: I will not start from scratch. \cite{Nash} is a 
remarkable introduction, aimed to physicists, to these difficult subjects and the definitions below will often follow its lines. But first, we will see a few definitions concerning symmetries.

\subsection{Open and closed symmetries}

One of the key ingredients in the BV theory (as in the BRST formalism) is to represent symmetries as cohomologies of some
complexes. The goal of this subsection is to give a definition of symmetries of an action functional as
the cohomology of a certain complex and also define the notions of open and closed symmetries.
We will here follow the presentation of \cite{Henneaux}, section 3.1. Let $S_0$ be the action functional of our theory, and $y^i$ the fields. Let us perform a transformation of the fields
\begin{equation}
 y^i\longrightarrow y^i +\delta_{\varepsilon}y^i = y^i + R^i_{\alpha}\varepsilon^{\alpha}
\end{equation}
where repeated indices indicates also a space-time integration. Now, let us assume that the action is invariant under the above transformations
\begin{equation}
 \delta_{\varepsilon}S_0 = \frac{\delta S_0}{\delta y^i}\delta_{\varepsilon}y^i = 0.
\end{equation}
Then this implies the existence of the Noether identities
\begin{equation}
 \frac{\delta S_0}{\delta y^i}R^i_{\alpha} = 0.
\end{equation}
Now, as noted in \cite{Henneaux}, the set of invertible gauge transformations leaving something (the action) invariant fulfills by definition the axioms of a Lie group, that we will denote by 
$\bar{\mathcal{G}}$. Its Lie algebra is then the set of infinitesimal gauge transformations. Let $\mathcal{N}$ be the set of trivial gauge transformations, i.e. of 
gauge symmetries derived from the equations of motion (i.e. from the equation $dS_0=0$). These transformations form an ideal and carry no physical information (one can check that the constants of motion derived 
from them are identically vanishing functions). This is why we are interested by the group $\mathcal{G}=\bar{\mathcal{G}}/\mathcal{N}$ which is obtained by quotienting gauge symmetries by 
trivial gauge transformations.

In mathematical terms: the non trivial gauge symmetries are represented as the cohomology in degree $1$ of some Koszul
complex. Indeed, we have a function $S_0$ on some vector space $V$ with coordinates $y^i$ and we consider
the kernel of the contraction map
\begin{eqnarray*}
\iota_{dS_0}: X\in \Gamma(V,TV) \mapsto dS_0(X)\in C^\infty(V). 
\end{eqnarray*}
The Lie algebra $\ker(\iota_{dS_0})$ turns out to be exactly the
infinitesimal gauge symmetries.
We assume the Lie algebra $\ker(\iota_{dS_0})$ is finitely 
generated by the family of vector fields $ (R^i_\alpha\frac{\partial}{\partial y^i})_\alpha $ as a $C^\infty(V)$
module (this happens to be true whenever $S_0$ is Morse--Bott i.e. a smooth function on a manifold whose critical set is a closed submanifold and whose Hessian is non-degenerate in the normal 
direction). In the algebraic case where $S_0$ is a polynomial function on $V$, 
$\overline{\mathcal{G}}=\ker(\iota_{dS_0})$ will be finitely generated as
a $C^\infty(V)$-module by the Noether Theorem.

To define rigorously the trivial gauge transformations,
we can extend the above contraction map as follows
\begin{eqnarray*}
\Gamma(V,\Lambda^2TV) \overset{\iota_{dS_0}}{\mapsto} \Gamma(V,TV) 
\end{eqnarray*}
whose image is the set of all vector fields of the form
$R^{ij}\frac{\partial S_0}{\partial x^i}\frac{\partial}{\partial x^j}$ where
$R^{ij}\in C^\infty(V)$. Hence we see that $\mathcal{N}=\text{im}\left(\iota_{dS_0}\right)$ belongs to the Lie algebra of infinitesimal gauge symmetries 
$\overline{\mathcal{G}}=\ker\left(\iota_{dS_0}\right)$ since $\iota^2=0$, but elements of $\mathcal{N}$ really come from the equations of motion. 
This means that the non trivial gauge 
symmetries are represented by the cohomology in degree $1$ of the complex
\begin{eqnarray*}
\Gamma(V,\Lambda^2TV) \overset{i_{dS_0}}{\mapsto} \Gamma(V,TV) \mapsto dS_0(X)\in C^\infty(V). 
\end{eqnarray*}
Now, $\mathcal{G}$ is really huge: it has as many elements as functionals of the fields. In more mathematical terms:
$\mathcal{G}$ is infinite dimensional and is a module over $C^\infty(V)$.
This explains the redundancies in $\mathcal{G}$. For instance if
$(X_\alpha)_\alpha$ is a family of vector fields that kills $S_0$
then so does the vector field $f^\alpha X_\alpha$ where
$(f^\alpha)_\alpha$ is a family of smooth functions on $V$.
In physical terms, the two gauge transformations
\begin{subequations}
 \begin{align}
  & \delta_{\varepsilon}y^i = R^i_{\alpha}\varepsilon^{\alpha} \\
  & \delta_{\eta}y^i = (R^i_{\alpha}M^{\alpha}_{\beta})\eta^{\beta}
 \end{align}
\end{subequations}
(with $M^{\alpha}_{\beta}$ a matrix that might depend on the fields) will bring two equivalent Noether identities. Therefore, it is coherent to restrict ourselves to a set of transformations that 
carry all the information about the Noether identities. We say that $G=\{\delta_{\varepsilon}y^i=R^i_{\alpha}\varepsilon^{\alpha}\}$ is a generating set if all the gauge transformations 
can be expressed as
\begin{equation}
 \delta_{\eta}y^i = \left(\mu^{\alpha}_AR^i_{\alpha} + M^{ij}_A\frac{\delta S_0}{\delta y^j}\right)\eta^A
\end{equation}
for some $\mu^{\alpha}_A$ and $M^{ij}_A=-M^{ji}_A$ that may depend on the fields. This precisely reflects the the fact that the first cohomology group 
$H^1(\Lambda TM,\iota_{dS_0})$ is finitely generated as a
$C^\infty(V)$ module

Then the commutator of two elements of the generating set, since it is a gauge transformation, can be expressed in the above form as well:
\begin{equation}
 R_{\alpha}^j\frac{\delta R_{\beta}^i}{\delta y^j} - R_{\beta}^j\frac{\delta R_{\alpha}^i}{\delta y^j} = C_{\alpha\beta}^{\gamma}R_{\gamma}^i + M_{\alpha\beta}^{ij}\frac{\delta S_0}{\delta y^j}
\end{equation}
with $C_{\alpha\beta}^{\gamma}$ and $M_{\alpha\beta}^{ij}$ that may depend on the fields. Then, even if the set of all the gauge transformations has a Lie group structure, the generating set 
might not have such a structure. This shall not arrive as a surprise since the generating set is not a basis for the Lie group structure. Finally, we will say that the symmetry of our theory is closed if 
$M_{\alpha\beta}^{ij}=0$ and open otherwise. Mathematically, the Lie algebra of symmetries $\overline{\mathcal{G}}$ is closed if it can be written as a direct sum of Lie
algebras $\overline{\mathcal{G}}=\mathcal{N}\oplus \mathcal{H}$ where $\mathcal{H}$ is isomorphic to $\mathcal{G}=\bar{\mathcal{G}}/\mathcal{N}$. Moreover, the generating set is much smaller than 
the set of all gauge transformations. This is a crucial point since we will have to add as many ghosts as there is elements in the generating set. Actually, we will have a ghost for each point of 
the space-time and each generator of the local gauge algebra. This allows us to speak of the ghosts as fields, one for each generator of the local gauge algebra.


We took the time to precise what we mean by open and closed symmetries because a remarkable feature of the BV formalism is that it allows to treat on an equal footing open and closed symmetries. We 
will later qualitatively explain why this is not true for the Faddeev--Popov formalism. 

\subsection{Some geometric concepts}

We will not give here the rigorous definition of a fiber bundle. This would be quite meaningless since we did not (and neither will we) give the rigorous definition of a manifold (it is a smooth space that 
locally looks like $\mathbb{R}^n$, for a given $n$). These are very 
standard, but I used the introduction of \cite{Nash}.

So, a fiber bundle is a space that locally looks like a Cartesian product of spaces. Let $E$ be the fiber bundle (the total space), $B$ be its base and $F$ its fiber. The key point is that there exists 
a projection
\begin{equation}
 \pi:E\longrightarrow B
\end{equation}
that sends any point of $E$ into a point of the base. So we say that the fiber $F$ is above the base $B$. Some examples of fiber bundles are the tangent and cotangent bundles. We will discuss them later, 
but first let us define one more object related to the concept of fiber bundle. A section of a fiber bundle is a continuous map 
\begin{equation}
 s:B\longrightarrow E
\end{equation}
such that $\pi\circ s= Id_E$. In other words: given a point $b\in B$, the section gives a point $x\in E$ such that $x$ would be projected on $b$. In a sense, the section makes us choose a point 
on the total space $E$. The set of sections of $E$ with base space $B$ will be written $\Gamma(B,E)$.

Now, let $M$ be a manifold, let $p(t)\subset M$ be a curve in $M$ parametrized by the parameter $t$ and $f:M\longrightarrow\mathbb{R}$ a smooth function. Then the rate of change of $f$ along the 
direction of $p(t)$ is
\begin{equation}
 \frac{\d}{\d t}f(p(t)) = Xf.
\end{equation}
$X$, in a local coordinates system reads
\begin{equation}
 X = \frac{\d x^i(p(t))}{\d t}\frac{\partial}{\partial x^i}.
\end{equation}
$X$ is a differential operator and an element of the tangent space of $M$ at $p$. This space is written $T_pM$ and $\{\frac{\partial}{\partial x^i}\}$ is one of its bases. Its dual space is the 
cotangent space of $M$ at $p$, denoted $T_p^*M$. A basis for this space is $\{dx^i\}$, the dual basis of $\{\frac{\partial}{\partial x^i}\}$. Elements of the tangent space are called vectors and 
elements of the cotangent space are called covectors. The tangent and cotangent spaces have a natural structure of vector spaces. Now we define the wedge product for the covectors
\begin{equation}
 \wedge:(dx^i,dx^j)\longrightarrow dx^i\wedge dx^j
\end{equation}
with the properties:
\begin{enumerate}
 \item $dx^i\wedge dx^j=0 \Leftrightarrow i=j$,
 \item $dx^i\wedge dx^j=-dx^j\wedge dx^i$,
 \item $(adx^{i_1}+bdx^{i_2})\wedge dx^{i_3} = adx^{i_1}\wedge dx^{i_3} + bdx^{i_2}\wedge dx^{i_3}$.
\end{enumerate}
If we iterate this procedure $r$ times ($r\leq\text{dim}(M)$) we obtain the $r$-forms of $M$ at $p$. $\Lambda^r(T_p^*M)$ will be the set of $r$-forms of $M$ at $p$. The famous exterior algebra of 
$M$ at $p$, which we will write $\Lambda(T_p^*M)$, is defined as the direct sum of the $\Lambda^r(T_p^*M)$ for $r\in\mathbb{N}$ (but for $r\geq\dim(M)+1$ we have $\Lambda^r(T_p^*M)=\{0\}$).

We can apply exactly the same procedure to vectors. We will also denote by $\wedge$ the wedge product for them. The obtained objects will be called polyvectors and the set of polyvectors of 
$M$ at $p$ will be written $\Lambda(T_pM)$. Up to now, we have built objects ``above'' a point $p\in M$. If we authorize $p$ to vary in all $M$ then we define the tangent space of $M$ as
\begin{equation}
 TM :=\{(p,v)|p\in M;v\in T_pM\},
\end{equation}
which has a natural bundle structure. Similarly we define the cotangent space of $M$, the polyvectors on $M$ and the forms on $M$. Let us notice that the tangent and cotangent bundles on $M$ can 
be restricted on any submanifold of $M$. They are subbundles of the tangent and cotangent bundles of the initial manifold.

From what has been said above, it is clear that $TM$ has the structure of a fiber bundle with $M$ as its base space. We define the polyvector fields as sections of this bundle. Intuitively, this is 
exactly what one sees when we speak about a field: a vector field is something that, at each point of $M$, associates a vector. The same is true with polyvectors. Similarly, form fields are 
sections of the fiber bundle $T^*M$ but, in an abuse of language, we will still call them forms.

Given a (super)manifold $M$ (super means that some coordinates can be anticommuting), we will write $\Pi TM$ the shifted tangent bundle. It is the bundle with reversed parity on the fibers. By 
this we mean that the coordinates on the fiber of $TM$ which were originally bosonic (commuting) become fermionic (anticommuting), and vice-versa. Then it is clear that the 
smooth functions on the shifted tangent bundle are elements of $\Lambda T^*M$, i.e. forms. This comes from the fact that a smooth function of anticommuting 
variables has to be a polynomial of degree at most the number of independent variables. Similarly, smooth functions on the shifted cotangent bundle are polyvector fields. Let us notice that this 
isomorphism (called the decalage isomorphism) survives when we pass to sections, i.e. to the fields.

Finally, let us give some definitions of symplectic geometry. They are from the lecture notes \cite{Cannas}. First, a symplectic manifold is a doublet $(M,\omega)$ such that $M$ is a 
manifold and $\omega$ a symplectic two-form. This means the following: let $\omega|_p$ be the restriction of $\omega$ to a point $p\in M$. Then , by definition
\begin{equation}
 \omega|_p:T_pM\times T_pM\longrightarrow\mathbb{R}
\end{equation}
and $\omega|_p$ is said to be symplectic if $\omega$ is closed (i.e. $d\omega=0$) and non-degenerate:
\begin{equation}
 \forall v\in T_pM,\omega|_p(u,v) = 0 \Leftrightarrow u=0.
\end{equation}
Then $\omega$ is said to be symplectic if $\omega|_p$ is symplectic for all $p$. It can be shown that a symplectic manifold is always of even dimension.

It shall be clear to the reader that $T^*M$ (and therefore $\Pi T^*M$) has a structure of symplectic space. Indeed, in the local coordinate system 
$\{x_1;\dots; x_n;\xi_1;\dots\xi_n\}$ of $T^*M$  one can check that the two-form
\begin{equation}
 \omega = dx^i\wedge d\xi_i
\end{equation}
does the job.

Let $(M,\omega)$ be a symplectic manifold. Then the submanifold $Y\subset M$ is a lagrangian submanifold if $\omega$ is identically vanishing on $T_pY$ and the dimension of $T_pY$ is one half 
of the dimension of $T_pM$.
\begin{equation}
 \forall p\in Y,\omega|_p\equiv0 \quad \text{and} \qquad 2\text{dim}(T_pY) = \text{dim}(T_pM)
\end{equation}
It can be shown that if $\omega$ is identically vanishing on a submanifold $X$ then the dimension of $T_pX$ is at most one half of the dimension of $T_pM$: lagrangian submanifolds are maximally 
dimensional submanifolds such that $\omega$ vanishes when restricted on them. They will be crucial later on. 

Now, let $X$ be a submanifold of $M$. Then the conormal space of $X$ at $p\in X$ is the subspace of the cotangent space of $M$ that vanishes on $T_pX$:
\begin{equation}
 N^*_pX = \{\xi\in T^*_pX|\forall v\in T_pX,\xi(v)=0\}.
\end{equation}
Then $\cup_{p\in M} N_p^*X$ is the conormal bundle of $X$, which will be written $N^*X$. Then we can show that $N^*X$ is a lagrangian submanifold of $T^*M$. We also have shifted conormal spaces 
$\Pi N^*X$, that are lagrangian submanifolds of $\Pi T^*M$.

Now, we define the Schouten--Nijenhuis brackets on $\Pi T^*M$ by
\begin{subequations}
 \begin{align}
  \{x^i,x^*_j\} & = \delta^i_j = -\{x^*_j,x^i\}, \\
  \{x^i,x^j\} & = \{x^*_i,x^*_j\} = 0
 \end{align}
\end{subequations}
extended by linearity. Here $(x^i)_{i=1\dots n}$ is a basis of $M$ and $(x^*_i)_{i=1\dots n}$ the corresponding basis on the fiber. These brackets are extended to brackets on polyvector fields 
(i.e. on $\mathcal{C}^{\infty}(\Pi T^*M)$ by
\begin{equation}\label{Sch-Nij}
 \{F,G\} = \frac{\derG\partial F}{\partial x^i}\frac{\derD\partial G}{\partial x^*_i} - \frac{\derG\partial F}{\partial x^*_i}\frac{\derD\partial G}{\partial x^i}
\end{equation}
with $\derD\partial$ the usual derivative and
\begin{equation} \label{der_gauche}
 \frac{\derG\partial F}{\partial y} = (-1)^{|y|(|F|+1)}\frac{\partial F}{\partial y}
\end{equation}
for a given grading that will be precised later.


%

\subsection{A motivational introduction} \label{motiv}

Let $X$ be a space with some structure (typically the configuration space) and $\mathcal{O}(X)$ the set of real valued functions on $X$ having some smoothness
properties. The elements of $\mathcal{O}(X)$ are the observables of our theory and we want to build an evaluation map
\begin{equation*}
 <>:\mathcal{O}(X)\longrightarrow\mathbb{R}.
\end{equation*}
The usual way of doing such a thing is to choose a well-behaved volume form $\Omega$ on $X$, and define
\begin{equation}
 <f>:=\int_X\mathcal{F}_{\Omega}(f)
\end{equation}
with $\mathcal{F}_{\Omega}(f)=f\Omega$. Indeed, in the path integral approach we integrate over a space of configurations and not only over the subspace that satisfies the equations of motion, which is 
the critical locus of the action functional. 
Now, from this, we can build the de Rham complex of forms. This complex is useful for two reasons. First, it gives a simple criteria to spot vanishing 
observables: if $f\Omega=d\omega$, then $<f>=0$. Therefore, the physical observables are rather elements of the de Rham cohomology group. Secondly, in the presence of gauge symmetries, the 
functions shall not be integrated over all the space $X$ but rather over a subspace of codimension $p$, which is a submanifold that we imagine as being transverse to the gauge orbits. Then we 
need our observables to be $N-p$ form, with $N=\text{dim}(X)$.

The tricky part is that in the cases of physical interest, the configuration space is infinite dimensional. Moreover, the number of dimensions is typically uncountable therefore the notion of 
top forms is not clear. The idea is then to define a dual complex to the de Rham complex: the homology complex of polyvector fields. The boundary operator of this complex is defined through 
the de Rham differential for the finite dimensional case. Then a formula for this operator is found, still in the finite dimensional case, and this formula is used as a definition for the 
evaluation map $<>$ in the infinite dimensional case.

Moreover, since we have seen that the polyvector fields can be viewed as functions on the shifted cotangent bundle, it will be relevant to think of the BV formalism as a theory of 
integration. Let us start with some notational points before building this theory of integration.

\subsection{From configuration space to odd symplectic manifolds}

Let us now be more precise. We shall discuss first the case when $X$ is finite dimensional and being acted on by a non compact group $G$ leaving the critical locus $\{dS_0=0\}$ invariant. We 
shall denote by $\mathfrak{g}$ the Lie algebra of $G$. We also assume that the action of $G$ on $X$ is such that the quotient space $X/G$ is a well defined manifold and $\pi:X\mapsto X/G$ is a 
fibration. In the physics language this means that everything goes well.

We denote by $(e_\alpha)_\alpha$ a basis of the Lie algebra $\mathfrak{g}$, $Vect(X)$ is the Lie algebra of vector fields on $X$ and $\rho:\mathfrak{g}\longmapsto Vect(X)$ is a homomorphism of Lie 
algebra which represents the action of $\mathfrak{g}$ on $X$.
\begin{eqnarray}
[e_\alpha,e_\beta]=C_{\alpha\beta}^\gamma e_\gamma\\
\rho(e_\alpha)=\rho_\alpha=\rho^i_\alpha\frac{\partial}{\partial x^{i}}.
\end{eqnarray}

$S_0$ is assumed to be $\mathfrak{g}$ invariant which means that
\begin{eqnarray}
\forall\alpha, \rho(e_\alpha)S_0=0\implies \rho_\alpha S_0=0.
\end{eqnarray}

Denote by $(e^\alpha)_\alpha$ the dual basis of $\mathfrak{g}^*$, $e^\alpha(e_\beta)=\delta^\alpha_\beta$. Then $(e^\alpha)_\alpha$ are linear coordinates on $\mathfrak{g}$. We denote by 
$\pi\mathfrak{g}$ the vector space $\mathfrak{g}$ shifted by $1$. Here shifted means that we reverse the parity: fermionic coordinates are now bosonic and vice-versa. Therefore the coordinates on 
$\pi\mathfrak{g}$, denoted by $(c^\alpha)_\alpha$, anticommute. We call them the ghosts. Actually, as already noticed superfonctions in anticommuting variables are polynomials in those anticommuting variables and we 
have an isomorphism of vector space $\mathcal{O}(\pi\mathfrak{g})\simeq\wedge \mathfrak{g}^*$ given in terms of a local basis by:
\begin{eqnarray*}
\left(c^{\alpha_1}\dots c^{\alpha_k}\right)\mapsto e^{\alpha_1}\wedge\dots \wedge e^{\alpha_k}, 
\end{eqnarray*}
(extended by linearity).
 
Consider now the supermanifold $V=X\times \pi\mathfrak{g}$. It has coordinates $(x^i,c^\alpha)_{i,\alpha}$. We will be interested by $\Pi T^*V$, the cotangent space of $V$ with shifted fiber. Then 
$\Pi T^*V$ has the coordinates $(x^i,c^\alpha, x_i^*,c_\alpha^*)$, where $(x^i, x_i^*)$ are coordinates on $\Pi T^*X$, so $x^i$ (the coordinate of $X$) is even and $x_i^*$ (the coordinate on the 
shifted fiber) is odd and where $(c^\alpha,c_\alpha^*)$ are coordinates in $\Pi T^*\mathfrak{g}^*= \pi\mathfrak{g}\times\mathfrak{g}^*$: $c_\alpha^*$ is even and $c^\alpha$ is odd.

So  $\Pi T^*V$ is an odd symplectic manifold isomorphic to $\Pi T^*\left(X\times \mathfrak{g}^* \right)$ with symplectic form $\omega=dx^i\wedge dx_i^*+dc_\alpha^*\wedge dc^\alpha$.

\section{BV formalism as a theory of integration}

Now, it is time to dive into our approach to the BV formalism. We will follow the program discussed in subsection \ref{motiv} for the space $X\times \pi\mathfrak{g}$.

\subsection{BV laplacian}

We are going to define a sort of divergence operator $\Delta$ on polyvector fields. We first give the definition of the interior product $\alpha\lrcorner \omega$ of a polyvector
field $\alpha\in \Gamma(M,\Lambda^p(TM))$ with a form $\omega\in \Gamma(M,\Lambda^q(T^*M))$ ($q\geq p$).

We have
\begin{equation}
 \lrcorner: \Gamma(M,\Lambda^p(TM))\times\Gamma(M,\Lambda^q(T^*M))\to\Gamma(M,\Lambda^{q-p}(T^*M).
\end{equation}
It is defined for $p=1$ by $(X\lrcorner\omega)(X_1,\dots,X_{q-1}):=(\iota_X\omega)(X_1,\dots,X_{q-1})=\omega(X,X_1,\dots,X_{q-1})$ and is then extended by 
$\left(\alpha\wedge\beta\right)\lrcorner\omega=\alpha\lrcorner(\beta\lrcorner\omega)$. The idea is now to define a map
\begin{equation}
 \Delta: \Gamma(M,\Lambda^p(TM))\to \Gamma(M,\Lambda^{p-1}(TM))
\end{equation}
as discussed in subsection \ref{motiv}.

Let $\alpha$ be a polyvector and $\Omega$ a well-behaved volume form. Then $\Delta$ is the unique operator from $\Gamma(M,\Lambda^p(TM))$ to $\Gamma(M,\Lambda^{p-1}(TM)$ such that the following 
equation holds true~:
\begin{equation}
(\Delta\alpha)\lrcorner\Omega=d(\alpha\lrcorner\Omega).
\end{equation}
If we denote by $\mathcal{F}_{\Omega}$ the isomorphism
\begin{equation}
\mathcal{F}_{\Omega}:\alpha \in \Gamma(M,\Lambda TM) \longmapsto \alpha\lrcorner \Omega \in \Gamma(M,\Lambda T^*M) 
\end{equation}
then 
\begin{equation} \label{BV_lapl_def}
\Delta = \mathcal{F}_{\Omega}^{-1}\circ d\circ \mathcal{F}_{\Omega}
\end{equation}
with $d$ the de Rham differential. Let us check that the operator $\Delta$ defined above does the job. If $\alpha\in \Gamma(M,\Lambda^p(TM))$, then 
$\alpha\lrcorner\Omega\in \Gamma(M,\Lambda^{\mathrm{top}-p}(T^*M))$, $d(\alpha\lrcorner\Omega)\in \Gamma(M,\Lambda^{\mathrm{top}+1-p}(T^*M))$. Therefore, 
$\Delta\alpha\in\Gamma(M,\Lambda^{p-1}(TM))$ as expected. We will show many properties of $\Delta$, but for now we will be happy with a cute proposition.
\begin{propo}
The operator $\Delta$ defined above satisfies $\Delta^2=0$.
\end{propo}
\begin{proof}
$\Delta^2= (\mathcal{F}_{\Omega}^{-1}\circ d\circ \mathcal{F}_{\Omega})(\mathcal{F}_{\Omega}^{-1}\circ d\circ \mathcal{F}_{\Omega})=(\mathcal{F}_{\Omega}^{-1}\circ d^2\circ \mathcal{F}_{\Omega})=0$.
\end{proof}
This is absolutely crucial: we have argued that $\Delta$ acts on the polyvectors like the de Rham derivative acts on forms, but in the opposite direction. This strongly suggests that we will 
construct a homology complex, which needs $\Delta^2=0$.

\subsection{BV integral and a Stokes theorem}

Before giving some properties of the BV laplacian, we will extend the metaphor between this operator and its homological complex (that has still to be described) and the de Rham complex. We will 
now construct the notion of an integral for which $\Delta$ plays a role similar to the one of the de Rham differential for the usual integral. Hence, we want an integration for a polyvector. A 
natural way to get a polyvector under in an integrand is to contract the polyvector with a volume form. This is still not what we want since the contraction will not be well-defined in the 
infinite dimensional case. It is why we have to define a BV integral.

We choose a volume form $\Omega\in \Gamma(M,\Lambda^{\mathrm{top}}T^*M)$ on $M$. Let $\alpha$ be an element of $C^\infty(\Pi T^*M)$ and $\Pi N^*\Sigma$ be the conormal of some submanifold then
the BV integral $\int_{\Pi N^*\Sigma}^{BV} \alpha$ is defined as follows:
\begin{eqnarray} \label{BV_int_def}
\int_{\Pi N^*\Sigma}^{BV} \alpha 
&=& \int_\Sigma \alpha\lrcorner \Omega.
\end{eqnarray}
This definition will allow us to get many beautiful results that will elucidate what are the observables of our theory, However, at this stage, it may look quite arbitrary, so let us 
give a motivation for the BV integral when $\Sigma$ is given by a set of equations. This will also allows us to derive the measure carried by $\Pi N^*\Sigma$ in this case.

So, let us assume that $\Sigma$ is defined by $k$ equations $\Sigma=\{x\in M| f^1(x)=\dots=f^k(x)=0\}$. Let us calculate
\begin{equation*}
I = \int_{\{f^1=\dots=f^k=0\}} \alpha\lrcorner \Omega = \int_M\delta(\{f^i\})(\alpha\lrcorner df)\Omega  = \int_{\Sigma}(\alpha\lrcorner df)\frac{\Omega}{df} 
\end{equation*}
with $df := df^1\wedge\dots\wedge df^k$. Then we have
\begin{align*}
 \alpha\lrcorner df & = \alpha^{i_1\cdots i_k}\varepsilon_{j_1\cdots j_k}\frac{\partial f^{j_1}}{\partial x^{i_1}}\cdots\frac{\partial f^{j_k}}{\partial x^{i_k}} \\
		    & = \int_{\Pi\mathbb{R}^k}\alpha^{i_1\cdots i_k}\theta_{j_1}\cdots\theta_{j_k}\frac{\partial f^{j_1}}{\partial x^{i_1}}\cdots\frac{\partial f^{j_k}}{\partial x^{i_k}} [\mathcal{D}\theta]
\end{align*}
with $[\mathcal{D}\theta]$ the Berezin measure. Then
\begin{align*}
 I & = \int_{\Sigma\times\Pi\mathbb{R}^k}\alpha^{i_1\cdots i_k}\theta_{j_1}\frac{\partial f^{j_1}}{\partial x^{i_1}}\cdots\theta_{j_k}\frac{\partial f^{j_k}}{\partial x^{i_k}} [\mathcal{D}\theta]\frac{\Omega}{df} \\
   & = \int_{\mathcal{L}\subset\Pi T^*M}\left.\alpha^{i_1\cdots i_k}x^*_{i_1}\cdots x^*_{i_k}\right|_{x^*_j=\theta_i\frac{\partial f^i}{\partial x^j}}[\mathcal{D}\theta]\frac{\Omega}{df}.
\end{align*}
The second line being obtained by applying the decalage isomorphism. Notice that in this second line, $[\mathcal{D}\theta]$ is the measure on the shifted fiber. It is quite easy to see that 
\begin{equation}
 \mathcal{L} = \Pi N^*\Sigma.
\end{equation}
Indeed, using the decalage isomorphism we have
\begin{align*}
 & \alpha^{i_1\cdots i_k}x^*_{i_1}\cdots x^*_{i_k} \in \mathcal{C}^{\infty}(\Pi N^*\Sigma) \\
 \Leftrightarrow & \alpha = \alpha^{i_1\cdots i_k}(x)\partial_{i_1}\wedge\cdots\partial_{i_1} \in \Lambda N\Sigma \\
 \Leftrightarrow & x\in\Sigma \qquad \text{and} \qquad \partial_i = \frac{\partial f^j}{\partial x^i}\frac{\partial}{\partial x^j}.
\end{align*}
Hence, using once again the decalage isomorphism, we are left with $x^*_i = \frac{\partial f^j}{\partial x^i}\theta_j$. Then we have arrived to the formula used for the definition of the BV integral. 
The above discussion has no meaning in the infinite dimensional case, nor when $\Sigma$ is not given by a set of equations, but our definition of the BV integral still does.

We already see that this definition is the right one since we have the equivalent to the Stokes theorem.
\begin{thm}
Let $\Sigma$ be a smooth submanifold with smooth boundary $\partial \Sigma$.
Then for any polyvector $\alpha\in C^\infty(\Pi T^*M)$:
$$\int^{BV}_{\Pi N^*\partial\Sigma}\alpha=-\int^{BV}_{\Pi N^*\Sigma}\Delta\alpha.$$
\end{thm}
\begin{proof}
\begin{eqnarray*}
\int^{BV}_{\Pi N^*\partial\Sigma}\alpha & = & \int_{\partial\Sigma} \alpha\lrcorner\Omega \qquad \text{by definition of the BV integral} \\
                                        & = & - \int_{\Sigma} d\left(\alpha\lrcorner\Omega\right) \qquad \text{by the usual Stokes theorem} \\
					& = & - \int_{\Sigma} \left(\Delta\alpha\right)\lrcorner\Omega \qquad \text{by definition of the BV laplacian} \\
					& = & - \int^{BV}_{\Pi N^*\Sigma}\Delta\alpha.
\end{eqnarray*}
\end{proof}
Now, a key point is that this theorem implies that the BV integrals depend only on the homology class of the integration domain if the integrand is $\Delta$-closed. As for forms, a polyvector 
$\alpha$ is said to by $\Delta$-closed if $\Delta\alpha=0$ and $\Delta$-exact if $\alpha=\Delta\beta$. 
\begin{coro} \label{gauge_fix}
 Let $\Sigma_1$, $\Sigma_2$ be two smooth submanifolds belonging to the same homology class. Then for any $\Delta$-closed polyvector $\alpha\in C^\infty(\Pi T^*M)$:
\begin{equation*}
 \int^{BV}_{\Pi N^*\Sigma_1}\alpha = \int^{BV}_{\Pi N^*\Sigma_2}\alpha.
\end{equation*}
\end{coro}
\begin{proof}
\begin{eqnarray*}
 \int^{BV}_{\Pi N^*\Sigma_1}\alpha - \int^{BV}_{\Pi N^*\Sigma_2}\alpha & = & \int_{\Sigma_1}\alpha\lrcorner\Omega - \int_{\Sigma_2}\alpha\lrcorner\Omega \\
								       & = & \int_{\partial\Sigma_1}\alpha\lrcorner\Omega \qquad \text{since $\Sigma_1$ and $\Sigma_2$ are in the same homology class} \\
								       & = &  \int^{BV}_{\Pi N^*\partial\Sigma_1}\alpha \quad \text{from the definition of the BV integral} \\
								       & = &  -\int^{BV}_{\Pi N^*\tilde\Sigma_1}\Delta \alpha \\
								       & = & 0
\end{eqnarray*}
\end{proof}
Finally, there is also a very simple, but important, lemma coming from the previous definitions.
\begin{lemm} \label{lemm_def_obs}
 If $\alpha$ is $\Delta$-exact, then 
\begin{equation*}
 \int_{\Pi N^*\Sigma}^{BV} \alpha = 0
\end{equation*}
\end{lemm}
\begin{proof}
 From the definitions of the BV integral and the BV differential we have, if $\alpha=\Delta\beta$:
 \begin{equation*}
  \int_{\Pi N^*\Sigma}^{BV}\alpha = \int_{\Sigma}(\Delta\beta)\lrcorner\Omega = \int_{\Sigma}d(\beta\lrcorner\Omega) = 0.
 \end{equation*}
\end{proof}
We see that all the classical results of the theory of integration over manifolds translate into results for the BV integration. Therefore, in the future, we will be able to compute with the 
BV integral without any reference to the L.H.S. in its definition \eqref{BV_int_def}. Its was our goal: the L.H.S. of \eqref{BV_int_def} does not admit a natural generalization to the uncountably 
infinite dimensional case.

\subsection{On the advantage of the BV formalism}

Now we know enough about the BV formalism to answer a simple question: why shall we bother with it? In other words, how is the BV formalism more powerful than the Faddeev--Popov 
approach, or than the cohomological construction of the BRST formalism?

If we are working in a theory with open symmetries, we might obtain quartic terms in the lagrangian after the gauge-fixing procedure. This mean a four-valent vertex with only ghosts. This 
feature was noticed in \cite{Ka78}. In the Faddeev--Popov formalism, the ghosts are interpreted as a consequence of the restriction of the domain of integration. We write under the functional integral 
delta functions that prevent us from integrating over the orbits of the gauge group. This procedure brings a determinant which is written as an integral over fermionic variables, which are 
interpreted as scalar fermionic particles: the ghosts. Therefore there is no liberty in the Faddeev--Popov procedure for the terms in which the ghosts appear and a quartic term is thus absolutely 
meaningless.

Now, if we treat a theory with open symmetries with the BRST formalism, we will find that the BRST operator no longer squares to zero. More precisely, it will become proportional to 
$\frac{\delta S}{\delta y^i}$ (see \cite{Weinberg96b} and references therein) . Thus the BRST cohomology will only exists on-shell. This is quite awkward since we do no longer have a complex 
hence the concept of cohomology itself becomes ill-defined.

On the other hand, in the BV formalism, the quartic terms in the ghosts fields will come from terms $\Phi^*\Phi^*\Phi\Phi$ in the action since the gauge fixing will be to impose 
$\phi^*=\frac{\delta\Psi}{\delta\Phi}$ (with $\Psi$ the so-called gauge-fixing fermionic function). But since the BV formalism is a theory of integration, it is absolutely right to integrate 
quartic functions! Therefore, the BV formalism allows us to treat open symmetries. Actually, we do not really need to specify which kind of symmetries we are dealing with.

We have talked in this discussion of gauge-fixing in BV formalism, a subject that has not been covered yet. So let us move to this problem.

%

\section{Gauge-fixing in BV formalism}

In order to clarify the discussion on the gauge-fixing procedure in the BV formalism, let us start by studying how to translate the usual BRST formalism in supergeometric language.

\subsection{An informative example} \label{example}

We are working on the supermanifold $V=X\times\mathfrak{g}$. Let $(x^i)=(\phi^i,c^{\alpha})$ be a basis of this manifold, with $\phi^i$ the even coordinates and $c^{\alpha}$ the odd ones. We 
have to define the Schouten--Nijenhuis brackets on this extended space. It is quite simple: the formula \eqref{Sch-Nij} will give us the right brackets on our supermanifold. Let us recall the 
formula:
\begin{subequations}
 \begin{align}
  & \{F,G\} = \frac{\derG\partial F}{\partial x^i}\frac{\derD\partial G}{\partial x^*_i} - \frac{\derG\partial F}{\partial x^*_i}\frac{\derD\partial G}{\partial x^i} \label{Sch-Nij2} \\
  & \frac{\derG\partial F}{\partial y} = (-1)^{|y|(|F|+1)}\frac{\partial F}{\partial y}.
 \end{align}
\end{subequations}
The grading has now to be precised: we will take
\begin{subequations}
 \begin{align}
  & |\phi^i| = 0 \\
  & |c^{\alpha}| = +1 \\
  & |\phi^*_i| = -1 \\
  & |c^*_{\alpha}| = -2
 \end{align}
\end{subequations}
and $|f.g|=|f|+|g|$. Hence $\phi^i$ and $c^*_{\alpha}$ are even, $\phi_i^*$ and $c^{\alpha}$ are odd. These brackets (that will also be referred as the Schouten--Nijenhuis brackets) have many 
properties (graded derivation, graded Jacobi\dots) that are straightforward but tedious to show (the devil is in the powers of $(-1)$). We will essentially not need them here but a quite 
crucial property\footnote{crucial because it is nearly the only one every authors agreed on.} is the following:
\begin{lemm}
 $\{,\}$ is graded antisymmetric. That is, for $f$, $g$ two polyvector fields
 \begin{equation} \label{Sch-Nij_AS}
  \{f,g\} = -(-1)^{(|f|+1)(|g|+1)}\{g,f\}.
 \end{equation}
\end{lemm}
\begin{proof}
 We start by extending the brackets:
 \begin{align*}
  \{f,g\} & = \frac{\partial f}{\partial\phi^i}\frac{\partial g}{\partial\phi^*_i} + (-1)^{|f|+1}\frac{\partial f}{\partial c^{\alpha}}\frac{\partial g}{\partial c^*_{\alpha}} - (-1)^{|f|+1}\frac{\partial f}{\partial\phi^*_i}\frac{\partial g}{\partial\phi^i} - \frac{\partial f}{\partial c^*_{\alpha}}\frac{\partial g}{\partial c^{\alpha}} \\
          & = (-1)^{|f|(|g|+1)}\frac{\partial g}{\partial\phi^*_i}\frac{\partial f}{\partial\phi^i} - (-1)^{|f|}(-1)^{|g|(|f|+1)}\frac{\partial g}{\partial c^*_{\alpha}}\frac{\partial f}{\partial c^{\alpha}} + (-1)^{|f|}(-1)^{|g|(|f|+1)}\frac{\partial g}{\partial\phi^i}\frac{\partial f}{\partial\phi^*_i} \\
\llcorner & - (-1)^{|f|(|g|+1)}\frac{\partial g}{\partial c^{\alpha}}\frac{\partial f}{\partial c^*_{\alpha}}.
 \end{align*}
Then using $(-1)^{|f|(|g|+1)} = (-1)^{(|g|+1)(|f|+1)}(-1)^{|g|+1}$ and $(-1)^{|g|(|f|+1)} = (-1)^{(|g|+1)(|f|+1)}(-1)^{|f|+1}$ we get
 \begin{align*}
  \{f,g\} & = (-1)^{(|g|+1)(|f|+1)}\left[(-1)^{|g|+1}\frac{\partial g}{\partial\phi^*_i}\frac{\partial f}{\partial\phi^i} + \frac{\partial g}{\partial c^*_{\alpha}}\frac{\partial f}{\partial c^{\alpha}} - \frac{\partial g}{\partial\phi^i}\frac{\partial f}{\partial\phi^*_i} \right. \\
\llcorner & \left.- (-1)^{|g|+1}\frac{\partial g}{\partial c^{\alpha}}\frac{\partial f}{\partial c^*_{\alpha}}\right] \\
	  & = -(-1)^{(|g|+1)(|f|+1)}\{g,f\}.
 \end{align*}
%
\end{proof}
Moreover, a useful lemma is
\begin{lemm} \label{useful_lemm}
 Let $f$ be a polyvector field. Then $\{f,-\}$ is a graded derivative of degree $|f|+1$, i.e. for any polyvector fields $g,h$ we have
 \begin{equation}
  \{f,gh\} = \{f,g\}h+(-1)^{|g|(|f|+1)}g\{f,h\}.
 \end{equation}
\end{lemm}
\begin{proof}
 We can use the Leibniz rule in the right derivatives coming from the definition \eqref{Sch-Nij2}. We get
 \begin{align*}
  \frac{\partial}{\partial x^*_i}(gh) & = \left(\frac{\partial g}{\partial\phi^*_i}\right)h + (-1)^{|g|}g\left(\frac{\partial h}{\partial\phi^*_i}\right) + \left(\frac{\partial g}{\partial c^*_{\alpha}}\right)h + g\left(\frac{\partial g}{\partial c^*_{\alpha}}\right) \\
    \frac{\partial}{\partial x^i}(gh) & = \left(\frac{\partial g}{\partial\phi^i}\right)h + g\left(\frac{\partial h}{\partial\phi^i}\right) + \left(\frac{\partial g}{\partial c^{\alpha}}\right)h + (-1)^{|g|}g\left(\frac{\partial g}{\partial c^{\alpha}}\right).
 \end{align*}
 Then we have
 \begin{align*}
  \{f,gh\} & = \{f,g\}h + \frac{\derG\partial f}{\partial x^i}(-1)^{|g||x^*_i|}g\frac{\derD\partial h}{\partial x^*_i} - \frac{\derG\partial f}{\partial x^*_i}(-1)^{|x^i||g|}\frac{\derD\partial h}{\partial x^i} \\
           & = \{f,g\}h + (-1)^{|g|(|x^*_i|+|x^i|+|f|)}g\{f,h\}.
 \end{align*}
 Using $|x_i^*|+|x^i|\equiv1[2]$ we end the demonstration of the lemma.
\end{proof}

Now, we have to define the de Rham derivative of a superfunction, i.e. of of function over $X\times\mathfrak{g}$. This definition as to take care of the fermionic character of the ghost coordinates and 
therefore has to be
\begin{equation}
 df :=\frac{\derG\partial f}{\partial x^i}dx^i.
\end{equation}
Then we have the following useful lemma
\begin{lemm} \label{deriv_superfonction}
 For $\alpha$ a polyvector field, $f$ a superfunction and $\Omega$ a well-behaved volume form we have
 \begin{equation} 
  df\wedge(\alpha\lrcorner\Omega) :=\{f,\alpha\}\lrcorner\Omega.
 \end{equation}
\end{lemm}
\begin{proof}
 First, we can take for $\Omega$ the Lebesgue measure since it has to be this measure up to a conformal factor which will be in the two side of the formula. Then since $f$ is a superfunction, it does 
 not depends on the coordinates on the fiber. Hence we have
 \begin{equation*}
  \{f,\alpha\} = \frac{\derG\partial f}{\partial x^i}\frac{\derD\partial \alpha}{\partial x^*_i}.
 \end{equation*}
 Using the decalage isomorphism and the rules of derivation of grassmannian variables we write this as
 \begin{equation*}
  \{f,\alpha\} = \sum_{n=1}^k\frac{\derG\partial f}{\partial x^{i_n}}\alpha^{i_1\cdots i_k}\partial_{i_1}\wedge\cdots\wedge\widehat{\partial_{i_n}}\wedge\cdots\wedge\partial_{i_k}.
 \end{equation*}
 Hence writing
 \begin{equation*}
  \Omega_0 = (-1)^{-(k-1)-\sum_{\substack{j=1 \\ j\neq n}}^k i_j}dx^{i_k}\wedge\cdots\wedge\widehat{dx^{i_n}}\wedge\cdots\wedge dx^{i_1}\wedge dx^1\wedge\cdots\wedge\widehat{dx^{i_1}}\wedge\cdots\wedge dx^{i_n}\wedge\cdots\wedge\widehat{dx^{i_k}}\wedge\cdots\wedge dx^N
 \end{equation*}
 we arrive to
 \begin{equation*}
  \{f,\alpha\} \lrcorner \Omega_0 = \sum_{n=1}^k\frac{\derG\partial f}{\partial x^{i_n}}\alpha^{i_1\cdots i_k}(-1)^{-(k-1)-\sum_{\substack{j=1 \\ j\neq n}}^k i_j}dx^1\wedge\cdots\wedge\widehat{dx^{i_1}}\wedge\cdots\wedge dx^{i_n}\wedge\cdots\wedge\widehat{dx^{i_k}}\wedge\cdots\wedge dx^N.
 \end{equation*}
 On the other hand we similarly compute
 \begin{equation*}
  df\wedge(\alpha\lrcorner\Omega) = \frac{\derG\partial f}{\partial x^{i}}dx^i\wedge(-1)^{-k-\sum_{j=1}^k i_j}\alpha^{i_1\cdots i_k}dx^1\wedge\cdots\wedge\widehat{dx^{i_1}}\wedge\cdots\wedge\widehat{dx^{i_k}}\wedge\cdots\wedge dx^N.
 \end{equation*}
 In this last expression, we see that the only non-vanishing terms are these with $i=i_j$ for some $j$. Replacing in all of these terms the $dx^i$ at its place in the wedge products on the left we find 
 the expression for $\{f,\alpha\} \lrcorner \Omega_0$ written above (with the right power of $-1$).

\end{proof}


Now, in the same way that in the BRST formalism we replace $S_0$ by $S_0+s\Psi$ ($s$ the BRST operator), we will perturbate $S_0$ in $\Pi T^*V$. Set a new action functional
\begin{eqnarray}\label{BVaction}
S=S_0+S_E+S_R=S_0+\underset{S_E}{\underbrace{c^\alpha c^\beta C_{\alpha\beta}^\gamma c^*_\gamma}} + 
\underset{S_R}{\underbrace{\rho_\alpha^i c^\alpha x_i^*}}.
\end{eqnarray}
$S_E+S_R$ is just the symbol of the BRST differential in $\Pi T^*V$. Let us clarify this assertion, and recall some aspects of the BRST formalism.

To the representation $\rho$ of the Lie algebra $\mathfrak{g}$ acting on a module $C^\infty(X)$, one can always associate the Chevalley--Eilenberg complex
\begin{eqnarray*}
\dots \longmapsto\Lambda^k \mathfrak{g}^*\otimes C^\infty(X) \overset{d_{CE}}{\longmapsto} \Lambda^{k+1} \mathfrak{g}^*\otimes C^\infty(X)\longmapsto \dots \\
d_{CE}=C_{\alpha\beta}^\gamma e^\alpha\wedge e^\beta i_{e_\gamma}\otimes 1 + e^\alpha\wedge \otimes \rho_\alpha^i\frac{\partial}{\partial x^i}.  
\end{eqnarray*}
The following trivial Lemma explains the relevance of the cocycles of the Chevalley Eilenberg complex.
\begin{lemm}\label{LemmatrivialCE}
Let $f\in C^\infty(X)$ be invariant under the action of $\mathfrak{g}$. Then $f\in \Lambda^0 \mathfrak{g}^*\otimes C^\infty(X)$ is a cocycle.
\end{lemm}
\begin{proof}
$f$ is $\mathfrak{g}$ invariant implies that $\rho_\alpha f=0,\forall\alpha$ hence $d_{CE}f=c^\alpha \rho_\alpha f=0$ and we are done.
\end{proof}
Hence the gauge invariant objects will be cocycles of the Chevalley--Eilenberg complex. Now we can formulate this complex in the context of supergeometry.

Recall that $\Lambda\mathfrak{g}^*=\mathcal{O}(\Pi\mathfrak{g})$ which implies that elements of the Chevalley--Eilenberg complex can be interpreted as superfunctions in 
$\mathcal{O}(\Pi\mathfrak{g}\times X)$ on the supermanifold $(\Pi\mathfrak{g}\times X)$. So there is an isomorphism 
$i:\Lambda^{\bullet} \mathfrak{g}^*\otimes C^\infty(X)\longmapsto \mathcal{O}(\Pi\mathfrak{g}\times X)$ given on generators by:
\begin{eqnarray*}
e^\alpha\otimes 1 \longmapsto c^\alpha\\
1\otimes f \longmapsto f 
\end{eqnarray*}
The operator $Q_{CE}=i\circ d_{CE}\circ i^{-1}$ acts as a vector field 
on $\mathcal{O}(\Pi\mathfrak{g}\times X)$:
\begin{equation}
Q_{CE}=C_{\alpha\beta}^\gamma c^\alpha c^\beta \frac{\partial}{\partial c^\gamma} + c^\alpha \rho_\alpha^i(x)\frac{\partial}{\partial x^i}
\end{equation}
and $Q_{CE}$ squares to zero:
\begin{equation}
Q_{CE}^2=0.
\end{equation}
We say that $Q_{CE}$ is a homological vector field. This result is quite lengthy to derive, but not really difficult, and we will not write a proof here. Let us just precise that one has to use 
the antisymmetry property and the Jacobi identity on the structure constants $C_{\alpha\beta}^\gamma$.

In classical symplectic geometry we have the following correspondance: let $V$ be a smooth manifold and $\pi:T^*V\mapsto V$ the cotangent bundle of $V$. To every vector field $Q$ on $V$, we 
associate a function $\sigma(Q)\in T^*V$ called its symbol in such a way that the Hamiltonian vector field $\{\sigma(Q),.\}$ is a vector field on $T^*V$ whose horizontal part is just $Q$, in other words
$\forall f\in C^\infty(V),\{\sigma(Q),\pi^*f\}=Qf$ (with $\pi^*f$ the pull-back of $f$ to $T^*V$) when they are calculated along the zero section of $T^*V$.

Here we apply the same technique for the vector field $Q_{CE}$ on the supermanifold $V$ and we find that:
\begin{eqnarray}
\sigma(Q_{CE})=c^\alpha c^\beta C_{\alpha\beta}^\gamma c^*_\gamma + 
\rho_\alpha^i c^\alpha x_i^*.
\end{eqnarray}
Furthermore, we find that $Q$ is homological implies that
\begin{equation*}
0= Q^2f = \{\sigma(Q),\pi^*(Qf)\} = \{\sigma(Q),\{\sigma(Q),\pi^*f\}.
\end{equation*}
Then using the graded Jacobi relation we find that the last relation implies
\begin{equation}
 \{\sigma(Q),\sigma(Q)\} = 0.
\end{equation}
Let $f\in\mathcal{O}(\Pi T^*V)$, then $f$ satisfies the classical master equation (CME) if $\{S,S\}=0$. From the above, it follows that $\sigma(Q_{CE})=S_E+S_R$ satisfies the CME.

The point is that instead of calculating a quantity like $d_{CE}f$, $f\in \Lambda^{\bullet} \mathfrak{g}^*\otimes C^\infty(X)$ we will choose the corresponding superfunction 
$i_*f\in \mathcal{O}(\Pi\mathfrak{g}\times X)$ and note that
\begin{equation}\label{symplecticCE}
d_{CE}f=0 \Leftrightarrow \{\sigma(Q_{CE}),\pi^*(i_*f)\}=0
\end{equation}
therefore we have a purely symplectic interpretation for $f$ to be a cocycle in the Chevalley--Eilenberg complex. 
 
\begin{propo} \label{propo531}
The quantity $S=S_0+S_E+S_R$ defined in \eqref{BVaction} solves the CME.
\end{propo} 
\begin{proof}
First note that $S=S_0+\sigma(Q_{CE})$ therefore $\{S,S\}=\{S_0,S_0\}+2\{\sigma(Q_{CE}),S_0\}+\{\sigma(Q_{CE}),\sigma(Q_{CE})\}$. We already proved that the last term 
$\{\sigma(Q_{CE}),\sigma(Q_{CE})\}=0$. Furthermore, $S_0$ is a function on $X$ and is identified with the same element $\pi^*S_0\in \mathcal{O}(T^*X)$ therefore it is obvious that $\{S_0,S_0\}=0$.
Finally identity \eqref{symplecticCE} implies that $2\{\sigma(Q_{CE}),S_0\}=2Q_{CE}S_0=d_{CE}S_0=0$ since $S_0$ is $\mathfrak{g}$ invariant, thus we can use Lemma \ref{LemmatrivialCE}.
\end{proof}
This example will clarify what is the gauge-fixing procedure in the BV formalism, and will also make clear that we have not yet all the needed tools...

\subsection{Lagrangian submanifold} \label{lagrangian}

For the moment, the integral $\int_{\Pi T^*V} e^{iS} $ should be understood as a BV integral $\int^{BV}_{\underline{0}\subset T^*\left(X\times\mathfrak{g}\right)}e^{iS}$
over the zero section of the symplectic manifold $\Pi T^*V=\Pi T^*\left(X\times\mathfrak{g}^*\right)$ i.e. $\underline{0}\simeq \left(X\times\mathfrak{g}\right)$. The zero section $\underline{0}$ 
is also a Lagrangian submanifold of $\Pi T^*V$.

First, if $S_0$ has a non compact critical manifold (the critical manifold of a function $f$ is the set of points such that $df=0$), the choice of the zero section as a Lagrangian domain of integration is a very bad one since the integral might diverge. However, thanks to 
lemma \ref{gauge_fix}, we know that we can change the domain of integration (if we stay in the same homology class) provided that the integrand is $\Delta$-closed. Thus, we will choose as domain 
of integration a Lagrangian submanifold that is the conormal of some submanifold $\Sigma$ which is transverse to the orbits of the gauge group. 

Hence we see that we may interpret the gauge fixing in the BV formalism as the choice of a lagrangian submanifold $\Pi N^*\Sigma$ to integrate over. The gauge fixing surface $\Sigma$ has to be 
transverse to the orbits of $\mathfrak{g}$. Then asking for an observable to be gauge invariant will be translated into asking for the $\Delta$-closeness of the integrand of the path 
integral. The reverse is obviously true: if the integral of a polyvector field does not depends on the lagrangian submanifold we are integrating over (providing that we stay in the same homology 
class), then we can take two infinitesimally close submanifolds. Then, from the proof of the lemma \ref{gauge_fix} we see that the integral of the BV laplacian over the infinitesimal volume that 
distinguish these two submanifold vanishes. Hence the BV laplacian vanishes in this infinitesimal volume. We can do the same analysis everywhere: the BV laplacian vanishes globally. In other 
worlds: the (BV) integral of a polyvector field $\alpha$ is gauge invariant if, and only if, $\Delta\alpha=0$.

The missing point is: over which set of lagrangian submanifolds can we integrate? We will give a partial answer, found in \cite{Fi04}, although the idea is already present in the 
seminal work of Batalin and Vilkovisky \cite{BaVi81b}. Take $\Psi_1$, $\Psi_2$ two smooth functions of $V$. Then the submanifold $\mathcal{L}_{\Psi_1}$ of $\Pi T^*V$ defined by
\begin{equation}
 x_i^* = \frac{\partial\psi_1}{\partial x^i}
\end{equation}
is a lagrangian submanifold. Moreover, since we can build the homotopy
\begin{equation}
 \Psi_t = t\Psi_1 +(1-t)\Psi_2,
\end{equation}
we find that $\mathcal{L}_{\Psi_1}$ and $\mathcal{L}_{\Psi_2}$ are in the same homology class. Finally, since the zero section is defined by $x_i^*=0$ we see that, for $\Psi$ is a smooth function 
$\mathcal{L}_{\Psi}$ is an admissible lagrangian submanifold to be integrate over: it has the same homology class as the zero section.

Let us see how all of this work in the example build in \ref{example}. If in the manifold $X$ (the configuration space of fields) the gauge surface is given by equations $Y=\{f_\alpha=0\}$ then 
in the space of fields plus antighosts $\left(X\times\mathfrak{g}^*\right)$ we choose the surface $\Sigma=\{ f_\alpha=0, c^*_\alpha=0 \}$ which corresponds to the surface
$Y\times \{0\}\subset \left(X\times\mathfrak{g}^*\right)$. Hence the conormal $\Pi N^*\Sigma \subset \Pi T^*\left(X\times\mathfrak{g}^*\right)$ is parametrized by the Lagrangian immersion:
\begin{equation}
(x^i,c_\alpha;\eta^i)\subset \left(Y\times \{0\}\right)\times \mathbb{R}^{(0|k)} \longmapsto (x^i,c_\alpha;\eta^i\frac{\partial f_\alpha}{\partial x^i},0 )\in \Pi N^*\left(Y\times \{0\}\right)\subset \Pi T^*\left(X\times\mathfrak{g}^*\right). 
\end{equation}
Once we choose the Lagrangian chain of integration then we define the gauge fixed BV integral as:
\begin{subequations}
 \begin{eqnarray}
  \int^{BV}_{\Pi N^*\Sigma} e^{i\frac{S}{\hbar}}\\
  S=S_0+S_E+S_R.
 \end{eqnarray}
\end{subequations}
The above integral should be invariant if we make compactly supported deformations of $Y$ in the configuration space of fields $X$. The only thing to check is that the integrand  
$e^{i\frac{S}{\hbar}}$, where $S$ defined by \eqref{BVaction}, should be $\Delta$-closed.

\begin{lemm}
Let $\Omega$ be a fixed volume form and $\Delta_{\Omega}$ the corresponding BV Laplacian. If the Lie algebra $\mathfrak{g}$ action preserves $\Omega$ then $\Delta_{\Omega} c^\alpha\rho_{\alpha}^ix_i^*=0$. 
\end{lemm}
\begin{proof}
 On the one hand, the fact that $\mathfrak{g}$ action preserves $\Omega$ is written $\rho_{\alpha}\Omega=0$. On the other hand, by definition of the BV laplacian we have
\begin{equation*}
 \Delta_{\Omega} c^\alpha\rho_{\alpha}^ix_i^* = c^\alpha\rho_{\alpha}^ix_i^*\lrcorner\Omega = c^\alpha\rho_{\alpha}^i\frac{\partial}{\partial x^i}\Omega = c^{\alpha}\rho_{\alpha}\Omega = 0.
\end{equation*}
\end{proof}

\begin{thm} \label{to_be_proved}
Let $S=S_0+S_E+S_R$ as defined in \ref{BVaction} and $\Delta$ the BV Laplacian corresponding to the Lebesgue measure $\Omega_0$, if
\begin{itemize}
\item[$\bullet$] the Lie algebra acts on $X$ in such a way that it preserves the measure $\Omega_0$
\item[$\bullet$] the Lie algebra $\mathfrak{g}$ is unimodular,
\end{itemize}
then 
\begin{itemize}
\item[$\bullet$] $e^{i\frac{S}{\hbar}}$ is $\Delta$ closed
\item[$\bullet$] $\Delta S=0$ and $S$ satisfies the quantum master equation.
\end{itemize}
\end{thm}
To our great shame, we are not yet able to prove this theorem. We are missing an explicit formula for the BV laplacian, which will be given in the next section. We have stated the theorem to 
motivate this next section: explicit calculations need such a formula. Before moving to this key section, let us make a (maybe not so) small digression on a technicality that we did not mention 
before.


\subsection{Gribov ambiguities}

There is a technical detail, noticed by Gribov in \cite{Gr78}: the choice of a gauge fixing surface $\Sigma$ can lead to over-counting (or, a priori, under-counting) some 
configurations. Indeed, let $[A]$ be the set of all the potentials $A'$ obtained from $A$ by a gauge transformation. The Faddeev--Popov procedure is then to integrate over a surface defined 
by, say, a fixed value of the divergence of $A$: $\partial^{\mu}A_{\mu}=f$. But the set $[A]$ defines a surface in the configuration space and this surface shall intersect $\Sigma$ exactly once. 
In general, there is no such $\Sigma$.

More precisely, two cases have to be taken in consideration. First, the case where for some $A$, $[A]$ does not intersect the gauge surface $\Sigma$. 
We will see later that this case does not occur. But for some $A$, it is known that $[A]$ can intersect 
$\Sigma$ strictly more than once. This is typical of non-abelian gauge theory in Euclidean space and occurs even for abelian cases in minkowskian space.

In the case of an abelian theory in minkowskian space, one can perform a Wick rotation and use index theory and spectral flow to interpret the multiple 
intersection between $\Sigma$ and $[A]$, as done in \cite{Es00} and in the references therein.

Let us now present Gribov's approach to this problem. We do not aim to give an exhaustive description of Gribov ambiguities, and we refer the reader to 
\cite{EsPeZa04} for a pedagogical presentation. The basic idea is that if $A$ and $A'$ are liked by an infinitesimal gauge transformation and have the same divergence, it implies
\begin{equation} \label{condition_tsfo_gauge}
 \mathcal{F}(A)\psi = 0
\end{equation}
with $\mathcal{F}(A)$ an operator whose determinant is the Faddeev--Popov determinant appearing in the path integral. This equation can therefore be seen as an equation for the eigenfunction of 
$\mathcal{F}(A)$ with eigenvalue zero. Therefore, $A$ can have a close gauge equivalent only if the Faddeev--Popov determinant is zero.

The next step is to study the eigenfunction equation
\begin{equation}
 \mathcal{F}(A)\psi = \varepsilon\psi.
\end{equation}
Close to $A=0$ it can be solved for $\varepsilon>0$ only. Indeed, we have $\mathcal{F}(0)=\partial_{\mu}\partial^{\mu}$ and therefore the solution to the above equation with $A=0$ is
\begin{equation*}
 \psi = \lambda_1e^{-\varepsilon_{\mu}x^{\mu}} + \lambda_2e^{\varepsilon_{\mu}x^{\mu}}
\end{equation*}
with $\varepsilon_{\mu}\varepsilon^{\mu}=\varepsilon$. Now, if we want $\psi$ vanishing at infinity, we have to take one of the two constant egal to zero (say, $\lambda_2$). Then, for $x^{\mu}$ light-like,
the condition $\varepsilon_{\mu}x^{\mu}>0$ implies $\varepsilon >0$.

As we get away from the point $A=0$, solutions for values of $\varepsilon$ closer to zero exist. At some point, a solution for $\varepsilon=0$ will exist, and then a region with solution for 
$\varepsilon>0$, again a solution for $\varepsilon=0$ and so on... In \cite{BaVi81} a nice geometrical picture of this feature was drawn that we will briefly present below.

Let $P$ be the principal bundle of the theory and $\mathscr{C}$ the set of connections on this bundle. Let $\mathcal{G}$ be the group of gauge transformations. The physical configuration space is then the 
quotient space 
\begin{equation*}
 \mathscr{M} = \mathscr{C}/\mathscr{G}.
\end{equation*} 
Let $p:\mathscr{C}\longrightarrow\mathscr{M}$ be the projection operator. Now, a metric can be defined on $\mathscr{C}$ through a scalar product $(,)$ on the tangent space $T(\mathscr{C})$. The 
key point is that $T(\mathscr{C})$ is locally a direct sum:
\begin{equation}
 T_{\omega}(\mathscr{C}) = H_{\omega}\oplus V_{\omega}.
\end{equation}
$V_{\omega}$ (the vertical space) encodes the gauge degrees of freedom while $H_{\omega}$ (the horizontal space) encodes the physical degree of freedom. From this, on can define a metric on 
$\mathscr{M}$ by 
\begin{equation*}
 g(X,Y) = (\tau_X,\tau_Y)
\end{equation*}
with $\tau_X$ and $\tau_Y$ the projection of the pull-back of $X$ and $Y$ to $T(\mathscr{C})$ on the horizontal space. Hence the quotient space is a manifold (if we 
assume some further technical restrictions, see the appendix of \cite{BaVi81} and the references therein for the details) and therefore the usual tool of Riemannian calculus can be used.

Now, let $\omega_0\in\mathscr{C}$ be a point of the orbit $p^{-1}(a_0)$ of $a_0\in\mathscr{M}$. Define $\mathscr{S}_0$ the affine subspace of $\mathscr{C}$ as the space generated by $H_{\omega_0}$. Then 
it was shown in \cite{BaVi81} that the orbit $p^{-1}(a)$ intersects $\mathscr{S}_0$ if $a$ is not to far away from $a_0$, and $p^{-1}(a)$ becomes tangent to $\mathscr{S}_0$ when the Faddeev--Popov 
determinant vanishes, i.e. when the equation \eqref{condition_tsfo_gauge} has a solution.

It was also shown in \cite{BaVi81} that the Gribov horizon (that is, the set of points of $\mathscr{C}$ where the equation \eqref{condition_tsfo_gauge} has a solution) is the set of points where at least 
two geodesics coming from $\omega_0$ cross each other. In other terms: the Gribov horizon is the set of focal points of $\omega_0$. From now on, we will take $a_0$ as the origin, i.e. the point where $A=0$.

Let us call $C_n$, the $n$-th Gribov region, the set of points reached from the origin $A=0$ by a line along which $A$ is strictly increasing and which has crossed $n$ 
vanishing Faddeev--Popov determinants.

The key point is that one has to restrict the integration domain to $C_0$. Indeed, the previous analysis shows that two fields linked by an infinitesimal gauge transformation have to live near 
the border of their Gribov region. It was also shown in \cite{Gr78} that two infinitesimally close $A$ and $A'$ can not live in the same Gribov region:
\begin{thm}
 (Gribov, 1978) Let $A$ be a field close to the boundary of the Gribov region it lives in. Let $A'\in[A]$, with $A'$ close but different to $A$ (that 
 is, $A'$ is obtained from $A$ by a infinitesimal but non-zero gauge transformation). Then $A'$ does not live in the same Gribov region than $A$.
\end{thm}
Hence, this theorem implies that it is needed to integrate at most over $C_0$. However, it says nothing about finite gauge transformation. It turns out that a further restriction is 
needed to avoid over-counting some fields: it was shown in \cite{vBa92} that one has to restrict the integration domain to the so-called fundamental modular domain as shown in \cite{vBa92}. First, let us define the 
function 
\begin{equation}
 F_A(g) = ||A[g]||^2 = \int_M\text{Tr}\left[\left(A_{\mu} + ig^{-1}\partial_{\mu}g\right)^2\right].
\end{equation}
It is a Morse function (i.e. a smooth non-degenerate function on a 
differential manifold) on $\mathcal{M}/G$, with $\mathcal{M}$ the set of maps from $M$ to $G$. Then it can be shown that the 
hessian of $F_A$ at its critical points is the Faddeev--Popov determinant. Let us noticed that, in this framework, the Gribov region $C_0$ is the set of gauge field having a strictly positive 
Faddeev--Popov determinant. Now, define the set 
\begin{equation}
 \Lambda = \{A|F_A(g)\geq F_A(1)\forall g\in\mathcal{G}\}.
\end{equation}
Clearly, $\Lambda\subset C_0$. It was shown in \cite{SeFr82} and \cite{DeZw91} that $\Lambda$ intersect all the $[A]$ at least once, hence one has not to worry to under-counting some configurations. 
In \cite{vBa92}, it was shown that the interior of $\Lambda$ does not contain any Gribov copies. Moreover, $\Lambda$, when its boundary points are properly identify, is a fundamental modular domain.

In the following, we will assume that we are always working in this fundamental domain and we will still denote $V$ the integration domain. Finally, let us conclude with a physical remark: the restriction of 
the integration domain does not change the perturbative approach of quantum field theory, but will change the non-perturbative sector. In particular, it is a strong candidate as the mechanism responsible 
of the observed quark confinement in QCD.


%

\section{Master equations}

\subsection{BV laplacians in coordinate}

We start our analysis by studying the case without ghosts: we will hence build our BV laplacian on $\Pi T^*X$. Then the same analysis will allows us to find out the 
BV laplacian on $\Pi T^*V$.

At first sight, there is an ambiguity in the definition of the BV integral since it depends of the choice of the volume form $\Omega$. Let us take for now the Lebesgue measure 
$dx^1\wedge\dots\wedge dx^N$ as a volume form, written $\Omega_0$. Then, there is an important proposition.
\begin{propo}
 The BV laplacian associated to the Lebesgue measure $\Omega_0$ is
 \begin{equation} \label{BV_lapl_comm}
  \Delta_{\Omega_0} = \frac{\partial}{\partial x^i}\frac{\partial}{\partial x^*_i}.
 \end{equation}
\end{propo}
\begin{proof}
 First, for a polyvector field 
 $\alpha = f^{i_1\dots i_n}x^*_{i_1}\dots x^*_{i_n} \simeq f^{i_1\dots i_n}\partial_{i_1}\wedge\dots\wedge\partial_{i_n}$, we want to compute first $\alpha\lrcorner\Omega_0$. From the definition of the 
 contraction product, we see that we shall rewrite $\Omega_0$ as
 \begin{equation*}
  \Omega_0 = (-1)^{-n-\sum_{j=1}^ni_j}dx^{i_n}\wedge\dots\wedge dx^{i_1}\wedge dx^1\wedge\dots\wedge\widehat{dx^{i_1}}\wedge\dots\wedge\widehat{dx^{i_n}}\wedge\dots\wedge dx^N.
 \end{equation*}
 Therefore, we have
 \begin{equation} \label{contraction_lebesgue}
  \alpha\lrcorner\Omega_0 = (-1)^{-n-\sum_{j=1}^ni_j}f^{i_1\dots i_n}dx^1\wedge\dots\wedge\widehat{dx^{i_1}}\wedge\dots\widehat{dx^{i_n}}\wedge\dots\wedge dx^N
 \end{equation}
 with the terms under the hats absent of the summation. Thus
 \begin{equation*}
  d(\alpha\lrcorner\Omega_0) = (-1)^{-n-\sum_{j=1}^ni_j}(\partial_if^{i_1\dots i_n})dx^i\wedge dx^1\wedge\dots\wedge\widehat{dx^{i_1}}\wedge\dots\wedge\widehat{dx^{i_n}}\wedge\dots\wedge dx^N.
 \end{equation*}
 The only non-vanishing terms in the above formula are those such that $i=i_j$ for $j\in[[1,n]]$. Therefore we have
 \begin{equation*}
  d(\alpha\lrcorner\Omega_0) =  (-1)^{-n-\sum_{j=1}^ni_j}\sum_{k=1}^n(\partial_{i_k}f^{i_1\dots i_n})dx^{i_k}\wedge dx^1\wedge\dots\wedge\widehat{dx^{i_1}}\wedge\dots\wedge\widehat{dx^{i_n}}\wedge\dots\wedge dx^N.
 \end{equation*}
 Now, we want to replace the $dx^{i_k}$ to its rightful place in the wedge products. How many powers of $-1$ will this bring? We will have $i_k-1$ powers of $-1$ coming from 
 $dx^1,dx^2,\dots dx^{i_k-1}$. But $dx^{i_1},dx^{i_2},\dots dx^{i_{k-1}}$ were not there, so we have overcounted $k-1$ powers of $-1$. Hence we arrive to
 \begin{equation*}
  d(\alpha\lrcorner\Omega_0) = \sum_{k=1}^n(\partial_{i_k}f^{i_1\dots i_n})(-1)^{X_{n,k}}dx^1\wedge\dots\wedge\widehat{dx^{i_1}}\wedge\dots\wedge dx^{i_k}\wedge\dots\wedge\widehat{dx^{i_n}}\wedge\dots\wedge dx^N
 \end{equation*}
 with 
 \begin{equation}
  X_{n,k} :=k-n-\sum_{\substack{j=1\\j\neq k}}^ni_j.
 \end{equation}
 Then, from a comparison with formula \eqref{contraction_lebesgue}, the definition \eqref{BV_lapl_def} gives
 \begin{equation*}
  \Delta_{\Omega_0}\alpha = \sum_{k=1}^n(\partial_{i_k}f^{i_1\dots i_n})(-1)^{k-1}\partial_{i_1}\wedge\dots\wedge\widehat{\partial_{i_k}}\wedge\dots\wedge\partial_{i_n}.
 \end{equation*}
 If we recognize that $(-1)^{k-1}\partial_{i_1}\wedge\dots\wedge\widehat{\partial_{i_k}}\wedge\dots\wedge\partial_{i_n}=\frac{\partial}{\partial x^*_{i_k}}x^*_{i_1}\dots x^*_{i_n}$ (with the right 
 sign!) we have proven the proposition.
\end{proof}
This laplacian can now be extended to polyvector fields. We can perform the same analysis on $\Pi T^*\mathfrak{g}$ (in which case the coordinates on the fiber will be commuting, and the 
coordinates on the base space anticommuting), with the same result. The only remaining question is: how to merge the 
two? It turns out that the right way to do it is
\begin{equation}
 \Delta = \frac{\partial}{\partial\phi^i}\frac{\partial}{\partial\phi^*_i} - \frac{\partial}{\partial c^{\alpha}}\frac{\partial}{\partial c^*_{\alpha}}.
\end{equation}
In the following, this $\Delta$ will be called ``the BV laplacian'' and the letter $\Delta$ will always refer to it. $\Delta$ corresponds to a Lebesgue measure on the superspace $\Pi T^*M$. 
For a quantum field theory, such a measure does not exist, but since we are working at the level of the BV laplacian, we will still be able to define observables! 

Now, we can start our study of this object with one simple lemma
\begin{lemm}
Let $\mathfrak{g}$ be a unimodular Lie algebra (i.e.: the measure on the exterior algebra $\Lambda\mathfrak{g}$ is invariant under the adjoint action of $G$) then $\Delta c^\alpha c^\beta C_{\alpha\beta}^\gamma c^*_{\gamma}=0$. 
\end{lemm}
\begin{proof}
The unimodularity implies that the adjoint action of $\mathfrak{g}$ on $\mathfrak{g}^*$ preserves the measure $\wedge dc^\alpha$. This means that for all $\alpha$, the vector field
$c^\beta C_{\alpha\beta}^\gamma\frac{\partial}{\partial c^\gamma}$ is divergent free which implies that $c^\beta C_{\alpha\beta}^{\alpha}=0$. To conclude this proof, we just have to notice that 
$\Delta c^\alpha c^\beta C_{\alpha\beta}^\gamma c^*_{\gamma}=c^\beta C_{\alpha\beta}^{\alpha}$. On the other hand, $\Delta c^\alpha c^\beta C_{\alpha\beta}^\gamma c^*_{\gamma} = -2c^\beta C_{\alpha\beta}^{\alpha}$.
\end{proof}
Moreover,  crucial property of this lagrangian is the following:
\begin{lemm}
 The Schouten--Nijenhuis brackets \eqref{Sch-Nij2} measure the obstruction of $\Delta$ to be a derivative. More precisely, for $f$, $g$, two polyvector fields
 \begin{equation} \label{BV_prod}
  \Delta(fg) = (\Delta f)g + (-1)^{|f|}f(\Delta g) + (-1)^{|f|}\{f,g\}.
 \end{equation}
\end{lemm}
\begin{proof}
 By direct computation we have
 \begin{align*}
  \Delta(fg) & = \frac{\partial}{\partial\phi^i}\left(\frac{\partial f}{\partial\phi^*_i}g + (-1)^{|f|}f\frac{\partial g}{\partial\phi^*_i}\right) - \frac{\partial}{\partial c^{\alpha}}\left(\frac{\partial f}{\partial c^*_{\alpha}}g + f\frac{\partial g}{\partial c^*_{\alpha}}\right) \\
	     & = (\Delta f)g+ (-1)^{|f|}f(\Delta g) + (-1)^{|f|}\frac{\partial f}{\partial\phi^i}\frac{\partial g}{\partial\phi^*_i} - \frac{\partial f}{\partial c^{\alpha}}\frac{\partial g}{\partial c^*_{\alpha}} + \frac{\partial f}{\partial\phi^*_i}\frac{\partial g}{\partial\phi^i} - (-1)^{|f|}\frac{\partial f}{\partial c^*_{\alpha}}\frac{\partial g}{\partial c^{\alpha}}.
 \end{align*}
 Then, using $\frac{\derG{\partial} f}{\partial\phi^i} = \frac{\partial f}{\partial\phi^i}$, $\frac{\derG{\partial} f}{\partial c^{\alpha}}=(-1)^{|f|+1}\frac{\partial f}{\partial c^{\alpha}}$, 
 $\frac{\derG{\partial} f}{\partial\phi^*_i}=(-1)^{|f|+1}\frac{\partial f}{\partial\phi^*_i}$ and 
 $\frac{\derG{\partial} f}{\partial c^*_{\alpha}}=(-1)^{|f|+1}\frac{\partial f}{\partial c^*_{\alpha}}$ from the definition \eqref{der_gauche} and the grading defined in \ref{example} we arrive to 
 the result.
\end{proof}
This property can actually be used (as, for example, in \cite{Gw12}) to find the representation \eqref{BV_lapl_comm} of the BV laplacian on $\Pi T^*X$ (up to an overall constant) if we define it 
axiomatically from the equation \eqref{BV_prod}. It is more elegant in the sense that everything derives from very simple first principles (an abstract algebraic structure), but the 
meaning of this laplacian is (at least in my point of view) darkened with this other approach. So we rather go the other way around and formulate this result as a small lemma.
\begin{lemm} (Gwilliam \cite{Gw12})
 $\Delta$ is the unique translation invariant second order differential operator that decreases the degree of the polyvectors it acts on of exactly $1$ and satisfies \eqref{BV_prod} for 
 $\Pi T^*X$. 
\end{lemm}
\begin{proof}
 The most general second order differential operator that decreases the degree of the polyvectors it acts on of exactly $1$ is
 \begin{equation*}
  \Delta = a^i_{~j}(x)\frac{\partial}{\partial x^i}\frac{\partial}{\partial x^*_j} + b_i(x)\frac{\partial}{\partial x^*_i} + c_{ij}^{~~k}(x)x^*_k\frac{\partial}{\partial x^*_i}\frac{\partial}{\partial x^*_j}.
 \end{equation*}
 If $\Delta$ is translation invariant then the $a$, $b$ and $c$ coefficients are constant. Then \eqref{BV_prod} implies
 \begin{equation*}
  \Delta(x^ix^*_j) = a^i_{~j} = x^ib_j+\delta^i_{~j}.
 \end{equation*}
 Hence $a^i_{~j}=\delta^i_{~j}$ and $b_j=0$. Now, using $c_{ij}^{~~k}=-c_{ji}^{~~k}$ and once again \eqref{BV_prod} we get
 \begin{equation*}
  \Delta(x^*_ix^*_j) = -2x^*_kc_{ij}^{~~k} = 0
 \end{equation*}
 and we end up with the announced $\Delta$.
\end{proof}
We have now enough knowledge to prove the theorem \ref{to_be_proved}. However, this will be our first encounter with the celebrated quantum master equation, and we will therefore postpone this 
task to the next subsection. Here, we did not yet solve our problem: we have a BV laplacian for the Lebesgue measure, but what if we do not like the Lebesgue measure? We want to be able to take 
another measure $\Omega = e^{f}\Omega_0$ with $|f|=0$. Then, from the definition \eqref{BV_lapl_def} we will have to compute
\begin{equation*}
 d(e^{f}\alpha\lrcorner\Omega_0)=e^{f}\left(d(\alpha\lrcorner\Omega_0) + df\wedge(\alpha\lrcorner\Omega_0)\right).
\end{equation*}
The first term is simply $e^f(\Delta\alpha)\lrcorner\Omega_0=(\Delta\alpha)\lrcorner\Omega$ while the second is just, thanks to the lemma \ref{deriv_superfonction}, 
$e^{f}\{f,\alpha\}\lrcorner\Omega_0=\{f,\alpha\}\lrcorner\Omega$. Hence we have proven a very nice proposition
\begin{propo} \label{cle_propo}
 The BV laplacian corresponding to the rescaled Lebesgue measure $\Omega = e^{f}\Omega_0$ is linked to the regular BV laplacian through the formula 
 \begin{equation}
  \Delta_{\Omega} = \Delta + \{f,-\}.
 \end{equation}
\end{propo}
That was the last needed details. Now we can fully study the quantum master equation.

\subsection{Quantum master equations}

As we have already mentioned, the corollary \ref{gauge_fix} and the discussion of the subsection \ref{lagrangian} imply that the gauge-invariant quantities shall be $\Delta$-closed (and even 
$\Delta_{\Omega}$-closed). Thus we see the importance of the theorem \ref{to_be_proved}: it says that the zero-point function is a gauge-invariant quantity. We will now prove it.
\begin{proof} \emph{of Theorem }\ref{to_be_proved}.
We will need to compute $\Delta \left(e^{i\frac{S}{\hbar}}\right)=\sum\frac{1}{n!}\left(\frac{i}{\hbar}\right)^n\Delta S^n$ since $\Delta$ is a differential operator.. We start by showing
\begin{equation} \label{result_interm1}
 \{S,S^n\} = nS^{n-1}\{S,S\}
\end{equation}
by induction on $n$. \eqref{result_interm1} is trivial for $n=1$. If it is true for $n$, using the lemma \ref{useful_lemm} on $S^{n+1}=S.S^n$ we have
\begin{align*}
 \{S,S^{n+1}\} & = \{S,S\}S^n + (-1)^{|S|(|S|+1)}S\{S,S^n\} \\
	       & = S^n\{S,S\} + SnS^{n-1}\{S,S\} \qquad \text{since } n|S|=0=|S| \\
	       & = (n+1)S^{n+1}\{S,S\}.
\end{align*}
Then, still by induction on $n$ we demonstrate
\begin{equation} \label{result_interm2}
 \Delta S^n = nS^{n-1}\Delta S + \frac{n(n-1)}{2}S^{n-2}\{S,S\}.
\end{equation}
Once again, \eqref{result_interm2} is trivial for $n=1$. If it is true for $n$ then the lemma \ref{BV_prod} gives
\begin{align*}
 \Delta S^{n+1} & = (\Delta S)S^n + (-1)^{|S|}S(\Delta S^n) + (-1)^{|S|}\{S,S^n\} \\
		& = S^n(\Delta S) + S\left(nS^{n-1}\Delta S + \frac{n(n-1)}{2}S^{n-2}\{S,S\}\right) + nS^{n-1}\{S,S\} \qquad \text{from \eqref{result_interm1}} \\
                & = (n+1)S^n\Delta S + \frac{(n+1)n}{2}S^{n-1}\{S,S\}.
\end{align*}
With this last equality we have
\begin{align*}
\Delta \left(e^{i\frac{S}{\hbar}}\right) & = \sum_{n=0}^{+\infty}\frac{1}{n!}\left(\frac{i}{\hbar}\right)^n\left(nS^{n-1}\Delta S + \frac{n(n-1)}{2}S^{n-2}\{S,S\}\right) \\
					 & = \Delta S\sum_{n=1}^{+\infty}\frac{1}{(n-1)!}\left(\frac{i}{\hbar}\right)^nS^{n-1} + \frac{1}{2}\{S,S\}\sum_{n=2}^{+\infty}\frac{1}{(n-2)!}\left(\frac{i}{\hbar}\right)^nS^{n-2} \\
					 & = \left(e^{i\frac{S}{\hbar}}\right)\left(\frac{i\Delta S}{\hbar} - \frac{1}{\hbar^2}\{S,S\} \right) \\
					 & = \left(e^{i\frac{S}{\hbar}}\right)\frac{i\Delta S}{\hbar}\text{ since $S$ solves the CME according to the proposition \ref{propo531}}\\
					 & = 0
\end{align*}
since $\Delta S_0=0$ because $S_0$ depends only on $x$, $\Delta S_R=c^\alpha \frac{\partial \rho^i_\alpha(x)}{\partial x^i}=0$ since every vector field $\rho_\alpha$ is divergence free and
$\Delta S_E=c^\beta C_{\alpha\beta}^\alpha=0$ by unimodilarity. 
\end{proof}
In the statement of the theorem \ref{to_be_proved} we have said ``$S$ satisfies the quantum master equation''. We have now to say what is this quantum master equation. Well, it is a equation 
found as an intermediate step in the proof of the theorem \ref{to_be_proved}:
\begin{equation} \label{QME}
 \{S,S\} - i\hbar\Delta S=0.
\end{equation}
This equation (the Quantum Master Equation, QME) is maybe the most important of the BV formalism. It tells us which theories (i.e.: which action functionals) with a given gauge invariance can be 
built. Solutions of the QME have been studied in \cite{IgItSo07} and the search for other solutions is still an active field.

Now, if $S$ satisfies the QME, what are the gauge-invariant observables? Once again, from the corollary \ref{gauge_fix} and the discussion of the subsection \ref{lagrangian}, they have to be
(loosely speaking for now) $\Delta$-closed polyvector fields. 
\begin{propo}
 Given a theory and its gauge invariant action functional $S$, the integral $\int_{\Pi N^*\Sigma}^{BV}\mathcal{F}e^{iS/\hbar}$ is gauge invariant if, and only if, $\mathcal{F}$ is solution of
 \begin{equation} \label{QME_obs}
  \{S,\mathcal{F}\} -i\hbar\Delta\mathcal{F} = 0.
 \end{equation}
\end{propo}
\begin{proof}
 We will need a further intermediate result:
 \begin{equation} \label{result_interm3}
  \{\mathcal{F},S^n\} = nS^{n-1}\{\mathcal{F},S\}.
 \end{equation}
 The proof is by induction on $n$ and exactly similar to the one of \eqref{result_interm1}. The case $n=1$ is trivial, and if \eqref{result_interm3} is true for a given $n$ then
 \begin{equation*}
  \{\mathcal{F},S^{n+1}\} = \{\mathcal{F},S\}S^{n} + S\{\mathcal{F},S^n\} = S^n\{\mathcal{F},S\} + nS^n\{\mathcal{F},S\} = (n+1)S^n\{\mathcal{F},S\}.
 \end{equation*}
 Then we know that $\int_{\Pi N^*\Sigma}^{BV}\mathcal{F}e^{iS/\hbar}$ is gauge-invariant if, and only if $\Delta(\mathcal{F}e^{iS/\hbar})=0$. Furthermore
 \begin{align*}
  \Delta(\mathcal{F}e^{iS/\hbar}) & = \Delta\mathcal{F} + \sum_{n=1}^{+\infty}\frac{1}{n!}\left(\frac{i}{\hbar}\right)^n\left[(\Delta\mathcal{F})S^n+(-1)^{|\mathcal{F}|}\mathcal{F}\Delta S+(-1)^{|\mathcal{F}|}\{\mathcal{F},S^n\}\right] \\
				  & = (\Delta\mathcal{F})e^{iS/\hbar} + (-1)^{|\mathcal{F}|}\mathcal{F}\Delta e^{iS/\hbar} + (-1)^{|\mathcal{F}|}\{\mathcal{F},S\}\sum_{n=1}^{+\infty}\frac{S^{n-1}}{(n-1)!}\left(\frac{i}{\hbar}\right)^n \qquad \text{from \eqref{result_interm3}.} \\
 \end{align*}
 Now $\Delta e^{iS/\hbar}=0$ since $S$ solves the QME. Hence we arrive to
 \begin{equation*}
  \Delta(\mathcal{F}e^{iS/\hbar}) = 0 \Leftrightarrow \Delta\mathcal{F} + (-1)^{|\mathcal{F}|}\frac{i}{\hbar}\{\mathcal{F},S\} = 0.
 \end{equation*}
 We find the equation \eqref{QME_obs} by multiplying this last line by $i\hbar$ and using the antisymmetry of the Schouten--Nijenhuis bracket that gives 
 $\{\mathcal{F},S\} = (-1)^{|\mathcal{F}|}\{S,\mathcal{F}\}$.
\end{proof}
We are nearly done now, but not completely. We can notice that we have performed all our computations with the Lebesgue measure. If we had chosen another measure, would we have found the same 
results? The answer is yes, but it has to be shown. Since the measure is a top form, another measure can only be the Lebesgue measure up to a conformal factor. If we notice that
\begin{equation}
 \int_{\Sigma}\mathcal{F}e^{iS/\hbar}\lrcorner\Omega_0 = \int_{\Sigma}\mathcal{F}\lrcorner\Omega
\end{equation}
with
\begin{equation}
 \Omega = e^{iS/\hbar}\Omega_0
\end{equation}
we deduce that the observables can also be defined as the $\Delta_{\Omega}$-closed polyvector field. The two definitions coincide, thanks to the proposition \ref{cle_propo}.
\begin{equation}
 \Delta_{\Omega}\mathcal{F} = 0 \Leftrightarrow (\Delta + \{iS/\hbar,-\})\mathcal{F} = 0 \Leftrightarrow \{S,\mathcal{F}\} -i\hbar\Delta\mathcal{F} = 0.
\end{equation}
This is quite simple, the only remarkable thing is that we end up with the right sign! On the first try!

Finally, using lemma \ref{lemm_def_obs} we see that two integrands whose difference is $\Delta$-closed will give the same BV integral, i.e. have the same measured values. This means that the 
observables of our theory are the elements of the homology group of the operator $\mathcal{O}$ defined by
\begin{equation*}
 \mathcal{O}(F) = \Delta(Fe^{iS/\hbar})
\end{equation*}
rather than polyvector fields. From the above explicit computations, we see that we have actually to compute the homology group of the operator
\begin{equation}
 \{S,-\} - i\hbar\Delta 
\end{equation}
as stated in the BV literature.

\subsection{Classical limit}

With the procedure described above, we concluded that, once we have solved the QME \eqref{QME}, the observables are the elements of the homology group of 
\begin{equation*}
 \{S,-\} - i\hbar\Delta. 
\end{equation*}
This is not a quantization procedure in the sense that we did not start from a classical theory that we would have quantized. We have rather built a quantum theory \emph{ab initio}. So, a good think to 
do now, as a final check, is to verify that we indeed find the classical theory if we take the limit $\hbar\longrightarrow0$.

The first remark is that if we take this limit in the QME we end up with the classical master equation. This shall not come as a surprise: if it is was not the case, we would not have called the CME 
``classical''. More importantly, if we take $\alpha$ a polyvector field
\begin{equation}
 \alpha = f^{i_1\dots i_p}\partial_1\wedge\dots\wedge\partial_p
\end{equation}
then $\alpha$ (or more precisely its homology class $[\alpha]$) is a classical observable if, and only if
\begin{align}
                 & \{S,\alpha\} = 0 \nonumber \\
 \Leftrightarrow & \sum_{j=1}^p\left[(-1)^{|x^{i_j}|}\frac{\partial S}{\partial x^{i_j}}f^{i_1\dots i_p}\partial_1\wedge\dots\wedge\widehat{\partial_{i_j}}\wedge\dots\wedge\partial_p - (-1)^{|x^*_{i_j}|}\frac{\partial S}{\partial x^*_{i_j}}\frac{\partial f^{i_1\dots i_p}}{\partial x^{i_j}}\partial_1\wedge\dots\wedge\partial_p\right] = 0
\end{align}
Hence, since all the terms of the above sum are linearly independent, the CME implies that $\alpha$ is non-vanishing only where $\frac{\partial S(x)}{\partial x^i}=0$ in the directions 
$x^i$, and where $\frac{\partial S}{\partial x^*_{i_j}}=0$ in the directions $x^*_i$, (but there is no freedom over those points since $S$ has typically a determined contents in the antifields): 
in the critical locus of $S$. That is the usual result of classical field theory: the solutions of the equations of motion (i.e. the contributions to the path integral) are the critical locus of 
the action.

To conclude, let us notice that the above makes clear how the BV formalism allows to treat the cases (that are the physically relevant situations) where the action is degenerated. This 
procedure still works in the quantum BV: it is one of the reasons why the BV formalism is so interesting.

\chapter*{Conclusion \markboth{{\scshape Conclusion}}{{\scshape Conclusion}}}

\addcontentsline{toc}{chapter}{Conclusion}

 \noindent\hrulefill \\
 {\it
Que dites-vous ?... C'est inutile ?... Je le sais ! \\
Mais on ne se bat pas dans l'espoir du succès ! \\
Non ! non, c'est bien plus beau lorsque c'est inutile ! \\

Edmond Rostant, Cyrano de Bergerac (Cyrano, Acte 5, Scène 6).}

 \noindent\hrulefill
 
 \vspace{1.5cm}

We have now been through the main text of this thesis. Starting from general definitions concerning algebra, bialgebra and Hopf algebra, we have been 
able to construct the Hopf algebra of renormalization. Then its Birkhoff decomposition has been presented and was used to define the renormalized 
Feynman rules of the theory. The end of the first chapter was devoted to the renormalization group equation studied in the framework of the Hopf 
algebra of renormalization.

The second chapter dealt with linear Schwinger--Dyson equations. After having presented a method to find exact solution to such equations, we 
applied it to various systems, with more-or-less successes: for the massive Yukawa model, solutions were found only in some limits, while for a 
supersymmetric model we found a parametric solution without taking any limits.

The two following chapters were devoted to the Schwinger--Dyson equation of the massless Wess--Zumino model, in the physical and the Borel plane respectively.
In the physical plane, our aim was to present the results of \cite{BeCl13}: the computations of the asymptotics of the solution. These computations 
were in excellent agreement with earlier numerical results.

In the Borel plane we were able to localize the singularities of the solution, and to characterize them. The most striking results concerned the number 
theoretical aspects of the singularities. We have shown that the transcendental numbers arising in the developments around these singularities could 
always be written as rational products of odd zetas. Moreover a bound for their weights was set for the two first singularities.

Finally, the fifth chapter gave a presentation of an interpretation of the BV formalism of gauge (and more general) theories as a theory of 
integration for polyvector fields. We have build an homology complex dual (in some sense) to the de Rham complex. Then the usual theorems of integration 
(Stokes theorem...) could be rephrased for polyvector fields. The gauge fixing have been shown to simply be the choice of the lagrangian submanifold 
to integrate over. Moreover the observables of the theory were shown to be elements of the homology group of an operator, as usually stated in the 
literature on the BV formalism. \\

Now, after this conclusion the reader can found three appendices. Each of these has a different reason not to be in the main text. 
The first one present a geometric approach of the BRST formalism. Although this is a very interesting topic, it is by no mean new, so has no place above. 

This introduction starts by a presentation of the objects we want to construct: the gauge invariant elements of a space of functions over the configuration 
space of the theory. We explain first why this can naturally be constructed from quotienting the initial space of functions. Then we briefly present the 
Koszul--Tate resolution and the Chevalley--Eilenberg complex. Together, they will give a bicomplex (the BRST bicomplex) which will have the right 
cohomology group.

Then we show that the BRST differential is actually an inner Poisson differential of degree one, but before this we give the essential definitions of Poisson superalgebra 
theory. This is particularly important since this structure is the one being deformed in the process of geometric quantization. Finally, in the case of QED, we give a classical 
argument showing that the physical states have to be elements of the cohomology group of the BRST operator, and we study some physical properties of the theory.

The second appendix presents some work done by Serguei Barannikov over the past few years. I did not put it into the main text because I only give the general pattern of this other way of seeing 
the BV formalism. It starts by giving a definition of the notion of degeneracy that unifies the usual ones for 
functions and Poisson bivectors. This allows to define degenerate polyvectors when they are solution of the classical master equation. Then the so called Darboux--Morse theorem is presented. It is a generalization to the supersymmetric case of the Morse 
lemma and of the Darboux theorem. This theorem is essential to an alternative approach of the BV formalism which makes clear the relation between BRST and BV formalisms.

The third and last appendix is about Feynman integrals. We will need these integrals in our future work, which concerns the scalar cubic model. However, since this project is still at an early 
stage, it seemed abusive to write these computations in the main text. First, we present an alternative way to compute the Feynman integrals of 
graphs of valence two. The idea is to link the integral of the two-points graph to the complete graph, which is a vacuum graph. This last one is 
easier to compute due to its lack of the exponential of a relation of the two Symanzik polynomials (in the massless case).

We check the method on the integral already computed (the one loop Mellin transform) and find the same result. Then we use it on an integral with a numerator. 
This integral is needed in the next system of equation of Schwinger--Dyson equations that we want to study: a truncated version of the massless scalar 
cubic model in six dimensions.

Then, we go thoroughly through a derivation of a three-points scalar integral. A result for this integral has already be given in \cite{BoDa86} but this 
article is quite laconic. Here we reduce the $D$-dimensional integral to a simple one. Making use relations for the Gauss' hypergeometric function and 
its Mellin--Barnes representation, we manage to express this integral as six series, which can be resumed to the fourth Appell's hypergeometric 
function. We clearly specify which relations are being used at each step. \\

This being said, what will come next? Well, many paths are in front of us. First, as already said, we have started to explore a truncated version of the 
Schwinger--Dyson equation of the massless cubic scalar model in six dimensions. Our goal is to find the singularities of the Borel transform of the solution 
and to see if they can be linked to a mass generation mechanism. This model is a good one for such a mechanism since it has asymptotic freedom and is 
therefore a good toy model for the strong interaction, where such a mass generating mechanism has to exist, and is still quite mysterious.

Moreover, the computations needed to study this model are rather similar to some work done by physicists interested in condensed matter system, and in 
particular for the QED in three dimensions. Hence we have started to discuss with some researchers in this area and hope to be able to further collaborate.

Concerning the Schwinger--Dyson equation of the massless Wess--Zumino model, some results have still to be proven. Our bound for the weigh of the odd 
zetas in the developments around the singularities of the solution of the Schwinger--Dyson equation hold only for the two first singularities. Indeed 
for the higher singularities an ambiguity occurs due to the many possible paths to reach them. Ecalle has shown that a suitable average over 
all the possible paths defines a derivation for the convolution product: the alien derivative. Taking an alien derivative of our system of equation 
and Ecalle's theory of résurgence will allow us to study the transcendental content of the higher singularities.

The presentation of the Hopf algebra of renormalization in the first chapter was not complete. Many powerful results in mathematics and physics have 
been reached through the Hopf algebraic approach. However, it seems that a folkloric theorem of physics can be  rigorously proved with Hopf algebraic 
techniques and has not been. It concerns the linearity in $\varepsilon$ of the $\beta$ function if we work in dimensional regularization. I am 
planning to turn my attention to this task in the near future.

Finally, there is plenty of room for future work in our approach of the BV formalism. For now, we are studying the QED in the Coulomb gauge with this 
formalism. Our goal is to make clear how the gauge-fixing procedure in the BV formalism can be linked to the usual notion of gauge-fixing. Moreover we 
want to understand how to get Feynman graphs (and more importantly, Feynman rules) out of this approach.

Also, in our presentation of the BV formalism, we assumed that they were no ghosts of ghosts and so on. Such fields can be added without drastically 
changing the main ideas but many cumbersome details have to be carried out. In particular, the grading on the superspace contains many more 
independent elements. Neither did we talk of renormalization in our approach. Recent progresses have been made toward a rigorous renormalization program 
\cite{Ng13}, with a point of view not to far from ours. Investigating the generalization of this work to gauge theories is an upcoming task.

In the main text, a brief presentation of the Gribov ambiguities was done. It would be interesting to thoroughly go through this subject with our BV formalism. First, 
that would allow to define and study the Gribov ambiguities in theories more general than gauge theories. One could also hope to solve some of the remaining open 
questions of this topic.

Our first motivation to study the BV formalism was to get a handle on the work of Brunetti, Fredenhagen and Rejzner \cite{BrFrRe13}. Although we are still 
quite far away from this goal (we still have to understand how the categorical framework of their formalism fits within the BV approach of gauge theories 
and how the BV formalism is expressed in the general landscape of algebraic quantum field theory) we are getting closer, and it is one of our drivers.

As a concluding remark, I would like to say that, although I have in the title of this thesis quite strongly opposed the analytical and geometric 
aspects of this work, I do not believe that they are that far apart. In the long run, I would like to see the two of them united, and use the 
Borel approach for the Schwinger--Dyson equations of gauge theories written with the BV formalism.




\renewcommand{\chaptername}{Appendix}
\appendix
\clearpage
\phantomsection
\addcontentsline{toc}{chapter}{Appendices}
%
%
%
%
%
%
%
%

\chapter{Geometric BRST formalism}

This geometric approach of the BRST formalism can be found in many textbooks, but the conventions are somehow moving. We quote in the text the articles and books from which we have worked. Moreover 
this study was done in collaboration with many others, in particular Christian Brouder.

\section{BRST differential}

We start by building the BRST differential and showing that its zeroth cohomology group is the thing that we want to compute. We roughly follow the presentation of \cite{KoSt86} and some others 
that will be specify when appearing.

\subsection{Set-up}

Let $M$ a symplectic manifold and $F(M)$ the smooth maps of $M$. $F(M)$ has the usual Poisson structure with brackets $\{,\}$. Let $C$ a coisotropic submanifold of $M$, that 
is that it exists $k$ functions over $M$ such that
\begin{equation}
 C = \{x\in M|f_i(x)=0\forall i\in[1,k]\}.
\end{equation}
Moreover, we will assume that the $f_i$'s are linearly independent and that they obey
\begin{equation}
 \{f_i,f_j\} = c_{ij}^{~~k}f_k
\end{equation}
with $c_{ij}^{~~k}$ elements of $F(M)$. In the physics language, $M$ is the configuration space of some phase space $N$ (hence $M=T^*N$) and the $f_i$'s are first class constraints in Dirac's 
nomenclature. Now let $G$ be a (connected) Lie group acting as symmetries of $F(M)$ and $\mathfrak{g}$ its Lie algebra. That is, there is a representation $\rho$ of $\mathfrak{g}$ over $F(M)$. 
Let $B$ be the quotient space
\begin{equation}
 B = C/G.
\end{equation}
Then we want to quantify functions of $B$ as operators acting on some Hilbert space. To do that, we want to express $F(B)$ in (co)homological terms. First, let us notice that any function of $C$ 
can be thought as a class of functions of $M$ coinciding on $C$. Defining the equivalence relation on $\mathcal{I}$ on $F(M)$ by
\begin{equation} \label{equiv_F(M)}
 f_1\sim f_2 \Leftrightarrow f_1(x) = f_2(x)\quad \forall x\in C
\end{equation}
we end up with
\begin{equation}
 F(C) = F(M)/\mathcal{I}.
\end{equation}
Similarly, one can express $F(B)$ in terms of $F(C)$: any function on $B$ can be thought as a function on $C$ constant along the orbits of $G$, and any such function on $C$ defines a function on 
$B$. Therefore
\begin{equation}
 F(B) = \{f\in F(C)|\xi f=0:=\rho_{\xi}(f)~\forall\xi\in\mathfrak{g}\}.
\end{equation}

\subsection{Koszul--Tate resolution}

Let $\delta$ be a linear map
\begin{equation}
 \delta:\mathfrak{g}\longrightarrow F(M)
\end{equation}
such that $\{\delta(\xi),f\} = \xi f$. Notice 
that up to now, no links were done between the Poisson structure and the Lie symmetries. It can be done either by assuming that $C=\Phi^{-1}(0)$ for some smooth map 
$\Phi:M\longrightarrow\mathfrak{g}$ known as the momentum map. This is the so-called Marsden-Weinberg reduction \cite{MaWe74}. In general cases, such map does not exist and the link has 
to be done by writing down an explicit formula for $\delta$.

In order to do so, let us notice than two elements of $F(M)$ are equivalent in the sens of \eqref{equiv_F(M)} if and only if their difference is vanishing on $C$, that is, if and only if their 
difference is spun by the $f_i$'s. Therefore, we can re-write
\begin{equation}
 F(C) = F(M)/(f_1F(M)+...+f_kF(M)).
\end{equation}
Now, from $\delta$ we define an odd derivation (that we will also write $\delta$) over $\mathfrak{g}\otimes F(M)$ by
\begin{subequations}
 \begin{eqnarray}
  \delta(\xi\otimes1) & = & 1\otimes \delta(\xi) \\
   \delta(1\otimes f) & = & 0.
 \end{eqnarray}
\end{subequations}
This extended $\delta$ being an odd derivation, it defines an odd derivation (still written $\delta$) over $\Lambda\mathfrak{g}\otimes F(M)$. It is an easy exercise to check that if $d$ is an 
odd derivation over $X$, then $d^2$ is an even derivation over $X$. Then $\delta^2$ is an even derivation and therefore is defined by its action over the generators only. Since 
$\delta^2(\xi\otimes1) = \delta^2(1\otimes f) = 0$ we end up with
\begin{equation}
 \delta^2 = 0.
\end{equation}
Hence we can define the homology of $\delta$. Now, it will be clear later that we will need this homology to be a resolution (i.e.: the only non-vanishing homology group is the 
zeroth). Hence we use Tate's procedure \cite{Ta56}. We will not detail it here but the idea is to add elements to the boundaries groups so that every cycle is a boundary. In the 
following, we will assume that the procedure has been performed and that we are working with the Koszul--Tate resolution. All in all, we get
\begin{equation}
 H_0^{\delta}(\Lambda\mathfrak{g}\otimes F(M)) \simeq F(C).
\end{equation}
An important remark is that the operator expression (which will be written later) for $\delta$ clearly shows that it commutes with $\mathfrak{g}$:
\begin{equation}
 \delta(\xi k) = \xi\delta(k) \quad \forall \xi\in\mathfrak{g},k\in K.
\end{equation}
Therefore $\delta$ is said to be a $\mathfrak{g}$-morphism.

\subsection{Chevalley--Eilenberg complex}

Let the representation $\rho$ of $\mathfrak{g}$ be over a vector space $K$. In the following, we will write $\xi k$ for $\rho_{\xi}(k)$. Then, we can define a derivation 
$\tilde d:K\longrightarrow \text{Hom}(\mathfrak{g},K)\simeq\mathfrak{g}^*\otimes K$ by
\begin{equation}
 \tilde dk(\xi) = \xi k \quad \forall\xi\in\mathfrak{g}.
\end{equation}
Now, we define an other differential $d:\mathfrak{g}^*\longrightarrow\Lambda^2\mathfrak{g}^*\otimes K$ by
\begin{equation}
 d\alpha(\xi_1,\xi_2) = \xi_1(\alpha(\xi_1))-\xi_2(\alpha(\xi_2))-\alpha([\xi_1,\xi_2]) = \rho_{\xi_1}(\alpha(\xi_1))-\rho_{\xi_2}(\alpha(\xi_2))-\alpha([\xi_1,\xi_2]).
\end{equation}
Let us notice that since $\rho$ is a representation, we have $d\circ\tilde d=0$. Now, we can extend $d$ to $\Lambda\mathfrak{g}^*$ by defining it as a superderivation for the wedge product:
\begin{equation*}
 d(\alpha\wedge\beta) = (d\alpha)\wedge\beta + (-1)^{|\alpha|}\alpha\wedge d\beta
\end{equation*}
where, if $d\alpha=\tilde\alpha\otimes k$ and $d\beta=\tilde\beta\otimes k'$, the first term in the RHS has to be read as $\tilde\alpha\wedge\beta\otimes k$ and the second 
$\alpha\wedge\tilde\beta\otimes k'$. Finally, $d$ can be extended one more time to $\Lambda\mathfrak{g}^*\otimes K$ by
\begin{equation}
 d(\alpha\otimes k) = (d\alpha)\otimes k + (-1)^{|\alpha|}\alpha\tilde d k.
\end{equation}
Hence we can define the cohomology of $d$ by writing $d=\oplus d_p$, with $d_p:\Lambda^{p-1}\mathfrak{g}^*\otimes K\longrightarrow\Lambda^{p}\mathfrak{g}^*$ and 
$\Lambda^{-1}\mathfrak{g}^*\otimes K = 0$. Then
\begin{equation}
 H^p_d(\Lambda\mathfrak{g}^*\otimes K) = \text{Ker}(d_{p+1})/\text{Im}(d_p).
\end{equation}
Therefore, since Im$(d_0)=0$ and Ker$(d_1)=\{k\in K|\xi k=0~\forall\xi\in\mathfrak{g}\}$, we find that $H^0_d(\Lambda\mathfrak{g}^*\otimes K)$ is the set of $\mathfrak{g}$-invariants of $K$.

\subsection{BRST bi-complex}

From the results of the two last subsections, we understand that we want to take $K=F(C)\simeq H_0^{\delta}(\Lambda\mathfrak{g}\otimes F(M))$. Hence, we want to compute 
\begin{equation*}
 H^0_d\left(\Lambda\mathfrak{g}^*\otimes H_0^{\delta}(\Lambda\mathfrak{g}\otimes F(M))\right).
\end{equation*}
First, let us notice that $d$ and $\delta$ commute with each other, since $\delta$ is a $\mathfrak{g}$-morphism, that can therefore be extended to a map Id$\otimes\delta$ (that will still be 
written $\delta$) of $\Lambda\mathfrak{g}^*\otimes K$. This map can be shown to commute with $d$: let $\omega$ be an element of $\Lambda\mathfrak{g}^*$ and $k$ an element of 
$\Lambda\mathfrak{g}\otimes F(M):=K$. Then
\begin{eqnarray*}
 d(\delta(\omega\otimes k)) & = & (d\omega)\otimes(\delta k)+(-1)^{|\omega|}\omega\otimes d(\delta k) \\
 \delta(d(\omega\otimes k)) & = & (d\omega)\otimes(\delta k)+(-1)^{|\omega|}\omega\otimes \delta(dk)).
\end{eqnarray*}
Hence, it is enough to show that $d$ and $\delta$ commute over $K$. For $\xi\in\mathfrak{g}$ we get
\begin{equation*}
 \delta(dk(\xi)) = \delta(\xi k) = \xi\delta(k)=d(\delta k)(\xi).
\end{equation*}
This is a standard textbook exercise (see Problem 1.13 of \cite{Fi06}) and can be greatly generalized. Now, writing $\delta=\oplus\delta_p$ with 
an obvious notation, we defined the BRST differential
\begin{equation}
 s = d+\oplus(-1)^p\delta_p.
\end{equation}
Hence we have $s^2=0$ and the miracle turns out to be the following theorem:
\begin{thm}
 Under the assumptions of the above constructions, we have
\begin{equation} \label{BRST_th}
 H_s^n\left(\Lambda\mathfrak{g}^*\otimes\Lambda\mathfrak{g}\otimes F(M)\right) \simeq H^n_d\left(\Lambda\mathfrak{g}^*\otimes H_0^{\delta}(\Lambda\mathfrak{g}\otimes F(M))\right).
\end{equation}
\end{thm}
\begin{proof}
Notice that in the LHS, the module on which $d$ is acting is $\Lambda\mathfrak{g}\otimes F(M)$ while in the RHS it is $ H_0^{\delta}(\Lambda\mathfrak{g}\otimes F(M))$. The module with which we 
will be working at each point of the proof will be clear from the context.

This proof is from Figueroa O'Farrill \cite{Fi06} (theorem 1.1). The idea is to exhibit two maps. One will map the $s$-cocycles to $d$-cocycles and the $s$-coboundaries to $d$-coboundaries, and 
the other will do the inverse. It is for the second map that we will need that $\delta$ defined a resolution. Let us emphasize this point, which is absolutely crucial and often not clearly stated: 
if the Koszul complex is not a resolution, this procedure would not work.

First, let us define $C^{p,q}:=\Lambda^p\mathfrak{g}^*\otimes\Lambda^q\mathfrak{g}\otimes F(M)$ and $\mathcal{C}^n:=\oplus_{p-q=n}C^{p,q}$. $n$ is called the ghost number. Take 
$\omega\in\mathcal{C}^n$, then we can write
\begin{equation}
 \omega = \omega_0+\omega_1+\dots
\end{equation}
with $\omega_i\in C^{n+i,i}$. Then if $\omega$ is a $s$-cocycle, the cocycle condition decomposes to $\delta\omega_0=0$ and $d\omega_i+\delta\omega_{i+1}=0~\forall i\geq0$ due to the bigrading. 
The equation that interest us is
\begin{equation} \label{cocycle_cond}
 d\omega_0=-\delta\omega_1
\end{equation}
for some $\omega_1\in C^{n+1,1}$. Similarly, if $\omega=s\Phi$ is a $s$-coboundary, then we can write $\Phi=\Phi_0+\Phi_1+\dots$ with $\Phi_i\in C^{n-1+i,i}$ and the coboundary condition implies
\begin{equation} \label{coboundary_cond}
 \omega_0 = d\Phi_0+\delta\Phi_1.
\end{equation}
Now, defining $K^n=\Lambda^n\mathfrak{g}\otimes F(M)$, we have 
\begin{equation*}
 H_n^{\delta}(\Lambda\mathfrak{g}\otimes F(M)) = K^n/\delta K^{n+1}.
\end{equation*}
Hence, in the RHS of \eqref{BRST_th} we have
\begin{equation*}
 \Lambda^n\mathfrak{g}^*\otimes K = \Lambda^n\mathfrak{g}^*\otimes (K^0/\delta K^{1}) \simeq \frac{\Lambda^n\mathfrak{g}^*\otimes K^0}{\Lambda^n\mathfrak{g}^*\otimes \delta K^1}=\frac{C^{n,0}}{\delta C^{n,1}}.
\end{equation*}
Therefore, still in the RHS of \eqref{BRST_th} we have the $d$-cocycles being
\begin{equation*}
 Z^n_d(\Lambda^n\mathfrak{g}^*\otimes K) \simeq \{\omega\in\mathcal{C}^n|d\omega\in\delta C^{n+1,1}\}.
\end{equation*}
Hence, equation \eqref{cocycle_cond} states that $\omega_0$ is a $d$-cocycle. Similarly, the commutation of $d$ and $\delta$ shows that the $d$-coboundaries are
\begin{equation*}
 B^n_d(\Lambda^n\mathfrak{g}^*\otimes K) \simeq dC^{n-1,0}+\delta C{n,1}
\end{equation*}
and so \eqref{coboundary_cond} states that $\omega_0$ is a $d$-coboundary. Thus the projection $\pi_n:\mathcal{C}^n\longrightarrow C^{n,0}$ of an element to its $C^{n,0}$ component can be lifted 
to a cohomology map $\pi_n^*:H^n_s\longrightarrow H^n_d$.

To go the other way around, let $\omega_0$ by a $d$-cocycle. According to the study above, $\omega_0\in C^{n,0}$ (so $\delta\omega_0=0$) and it exists $\omega_1\in C^{n+1,1}$ such that
\begin{equation*}
 d\omega_0+\delta\omega_1=0.
\end{equation*}
But $\delta d\omega_1=d\delta\omega_1=-d^2\omega_0=0$, hence $d\omega_1$ is a $\delta$-cocycle. Since $\delta$ gives a resolution, $d\omega_1$ has to be a $\delta$-coboundary. Hence, it 
exists $\omega_2\in C^{n+2,2}$ such that
\begin{equation*}
 d\omega_1+\delta\omega_2=0.
\end{equation*}
We can iterate this procedure. We end up with a sequence $(\omega_0,\omega_1,\omega_2\dots)$ such that $d\omega_i+\delta\omega_{i+1}=0~\forall i$. This is just the decomposition of the 
$s$-cocycle condition, therefore $\omega=\omega_0+\omega_1+\omega_2+\dots$ is a $s$-cocycle.

Finally, if $\omega_0$ is a $d$-coboundary, it is simple to see that it is also $s$-coboundary. Indeed, as seen before, we can write
\begin{eqnarray*}
 \omega_0 & = & d\phi_0+\delta\phi_1 \\
          & = & s\phi_0+s\phi_1-d\phi_1 \quad \text{since }\delta\phi_0=0.
\end{eqnarray*}
Once again, we can compute $\delta d\phi_1=0\Rightarrow d\phi_1=-\delta\phi_2$ for some $\phi_2\in C^{n+1,2}$ since $\delta$ gives a resolution. We can iterate this procedure up to 
$N=$dim$(\mathfrak{g})$. At this point we have
\begin{equation*}
 \omega_0 = s\phi_0+s\phi_1+\dots+s\phi_{N+1}-d\phi_{N+1}.
\end{equation*}
And since $d\phi_{N+1}=0$ since $\phi_{N+1}\in C^{n+N,N+1}$, we have $\omega_0=s\phi$. Hence, we have a map from $d$-cocycles to $s$-cocycles that maps $d$-coboundaries to $s$-coboundaries (it 
is indeed the same map with $\omega_i=0~\forall i\geq1$). Therefore, we have a map $\tilde \pi_n^*:H^n_d\longrightarrow H^n_s$ and this proves the theorem.
\end{proof}

\section{Poisson superalgebra}

In this section, we want to show that the BRST differential defines a so-called Poisson superalgebra on $\Lambda\mathfrak{g}^*\otimes\Lambda\mathfrak{g}\otimes F(M)$. This structure is crucial 
since it is its deformation that will provide the quantization. Let us start with the crucial definitions.

\subsection{Poisson superalgebra}

\begin{defi}
 $(P,.)$ is an associative supercommutative superalgebra if $P$ is a $\mathbb{Z}_2$ graded vector space:
\begin{equation}
 P = P_0\oplus P_1
\end{equation}
with $.$ an internal bilinear operation preserving the grading such that:
\begin{subequations}
 \begin{align}
  & a.(b.c) = (a.b).c \\
  & a.b = (-1)^{|a||b|}b.a \\
 \end{align}
\end{subequations}
with $|a|=i$ for $a\in P_i$.
\end{defi}

In this section, we will not write $.$ for the multiplication. Moreover, all our superalgebras will be associative and supercommutative. Therefore, we will specify it only when needed.

$P_0$ is called the even part of $P$, and $P_1$ the odd one. Any associative and commutative algebra can be seen as an associative supercommutative superalgebra without odd part. The typical 
non-trivial example is the following: let $M$ be some smooth variety and $\Lambda M$ its exterior algebra. Then $(\Lambda M,\wedge)$ has a structure superalgebra with:
\begin{eqnarray*}
     \Lambda M & = & (\Lambda M)^0\oplus (\Lambda M)^1 \\
 (\Lambda M)^0 & = & \oplus_n\Lambda^{2n}M \\
 (\Lambda M)^1 & = & \oplus_n\Lambda^{2n+1}M.
\end{eqnarray*}
This result is a consequence of $\omega\wedge\rho = (-1)^{pq}\rho\wedge\omega$ for $\omega$ a $p$-form and $\rho$ a $q$-form.

\begin{defi}
 $(P,[,])$ is a Lie superalgebra if $[,]$ is a supercommutator
\begin{equation}
 [a,b] = - (-1)^{|a||b|}[b,a]
\end{equation}
that obeys the super-Jacobi relation:
\begin{equation}
 [a,[b,c]] = [[a,b],c]+ (-1)^{|a||b|}[b,[a,c]].
\end{equation}
\end{defi}

\begin{defi}
 $(P,.,\{,\})$ is a Poisson superalgebra if
\begin{itemize}
 \item $(P,.)$ is an associative supercommutative superalgebra.
 \item $(P,\{,\})$ is a Lie superalgebra.
 \item $\{,\}$ is a linear graded derivation over the multiplication:
\begin{equation}
 \{a,bc\} = \{a,b\}c + (-1)^{|a||b|}b\{a,c\}
\end{equation}
that preserves the grading.
\end{itemize}
\end{defi}
Now, the tensor product of two Poisson superalgebra is given by:
\begin{subequations}
 \begin{align}
  & |a\otimes u|\cong(|a|+|u|)[2] \label{tensor_superalgebres_a}\\
  & (a\otimes u)(b\otimes v) = (-1)^{|u||b|}ab\otimes uv \\
  & \{a\otimes u;b\otimes v\} = (-1)^{|u||b|}(\{a,b\}\otimes uv + ab\otimes\{u,v\}). \label{tensor_superalgebres_b}
 \end{align}
\end{subequations}
%
and the Poisson superalgebra structure survives when we take tensor product
\begin{lemm} \label{prod_superalgebres}
Let $P$ and $Q$ be two Poisson superalgebras. Then with the above definitions, $P\otimes Q$ is a Poisson superalgebra.
\end{lemm}
\begin{proof}
Associativity follow from associativity of $P$ and $Q$. Let $p_1, p_2\in P$ and $q_1, q_2\in Q$.
\begin{eqnarray*}
 (p_1\otimes q_1)(p_2\otimes q_2) & = & (-1)^{|q_1||p_2|}(p_1p_2)\otimes(q_1q_2) \\
                                  & = & (-1)^{|q_1||p_2|+|p_1||p_2|+|q_1||q_2|}(p_2p_1)\otimes(q_2q_1) \\
                                  & = & (-1)^{|q_1||p_2|+|p_1||p_2|+|q_1||q_2|+|p_1||q_2|}(p_2\otimes q_2)(p_1\otimes q_1) \\
                                  & = & (-1)^{|p_1\otimes q_1||p_2\otimes q_2|}(p_2\otimes q_2)(p_1\otimes q_1).
\end{eqnarray*}
Super-Jacobi and super-Leibniz can be proven in a similar way, but it is quite cumbersome and we will not do it here.
\end{proof}

%

\begin{defi}
 $P$ is a graded Poisson superalgebra if it is $\mathbb{Z}$-graded such that the two internal operations respect this graduation and such that the $\mathbb{Z}$--graduation is a 
reduction of the $\mathbb{Z}$--graduation:
\begin{subequations}
 \begin{align}
  & P = \oplus_{n\in\mathbb{Z}}P^n \\
  & P^nP^m \subseteq P^{n+m} \qquad \{P^n,P^m\} \subseteq P^{m+n} \\
  & P_0 = \oplus_{n\in\mathbb{Z}}P^{2n} \qquad P_1 = \oplus_{n\in\mathbb{Z}}P^{2n+1}.
 \end{align}
\end{subequations}
\end{defi}

\begin{defi}
A Poisson derivative of degree $k$ is a morphism $D:P^n\longrightarrow P^{n+k}$ such that
\begin{subequations}
 \begin{align}
  & D(ab) = (Da)b + (-1)^{k|a|}a(Db) \\
  & D(\{a,b\}) = \{Da,b\} + (-1)^{k|a|}\{a,Db\}.
 \end{align}
\end{subequations}
This derivative is said to be inner if it exists $Q\in P^k$ such that $Da=\{Q,a\}\forall a\in P$.
\end{defi}

\begin{lemm} \label{inner_Poisson}
 Any inner derivative is a Poisson derivative.
\end{lemm}
\begin{proof}
Let $D$ be an inner derivative of degree $k$. Then:
\begin{eqnarray*}
 D(\{a,b\}) & = & \{Q,\{a,b\}\} \\
            & = & \{\{Q,a\},b\} + (-1)^{k|a|}\{a,\{Q,b\}\} \\
            & = & \{Da,b\} + (-1)^{k|a|}\{a,Db\}
\end{eqnarray*}
\end{proof}

\subsection{BRST Poisson superalgebra}

We will follow here the presentation of \cite{FiKi91} of this structure. It starts by noticing
\begin{equation*}
 \Lambda\mathfrak{g}^*\otimes\Lambda\mathfrak{g} \simeq \Lambda(\mathfrak{g}\otimes\mathfrak{g}^*).
\end{equation*}
Then one can easily define a Poisson structure on $\Lambda(\mathfrak{g}\otimes\mathfrak{g}^*)$: take $x,y\in\mathfrak{g}$ and 
$\alpha,\beta\in\mathfrak{g}^*$ and define:
\begin{equation}
 \{\alpha,x\}=\{x,\alpha\}=\alpha(x) \quad \{x,y\}=\{\alpha,\beta\}=0
\end{equation}
then expand it to $\Lambda(\mathfrak{g}\otimes\mathfrak{g}^*)$ as an odd derivation. Since $F(M)$ has a Poisson structure, it can be seen as a Poisson superalgebra without odd component. Then 
$\Lambda\mathfrak{g}^*\otimes\Lambda\mathfrak{g}\otimes F(M)$ is given its Poisson superagebra structure as the tensorial product of the two others. 

The above description is not very handy. It is easier to deal with the action of $s$ over the fields to deal with practical computations. Therefore, let $(c^i)$ be a basis of 
$\Lambda^1\mathfrak{g}^*$ (the ghosts) and $(\bar{c}_i)$ a basis of $\Lambda^1\mathfrak{g}$ (the antighost). In the following, we will write $c_ic_j$ for $c_i\wedge c_j$. Since $s$ is a 
derivation, it is enough to know how it acts on the generators of the complex. We have already discussed how $\delta$ is extended to acts on this complex, but $d$ was defined to act on 
$\Lambda\mathfrak{g}^*\otimes K$. We have to be more specific here and detail how $d$ acts on the various components of $K$ since $K=\Lambda\mathfrak{g}\otimes F(M)$.

To do so, we have to use the operator expression of the Chevalley--Eilenberg derivative $d$. By explicit computation, we can see that $d$ acts on the elements of 
$\Lambda\mathfrak{g}\otimes F(M)$ as a precise differential operator. Thus we can identify the two and write
\begin{equation}
 d = c^i\rho_i-\frac{1}{2}f_{ij}^{~k}c^ic^j\frac{\partial}{\partial c^k}
\end{equation}
with $\rho_i:=\rho(\bar c_i)$. As already advertised, the Koszul--Tate differential $\delta$ has the following operator expression
\begin{equation}
 \delta = f_i\frac{\partial}{\partial\bar c_i}
\end{equation}
this formula being shown in a similar fashion than the formula for $d$.
Both have a natural extension to $\Lambda\mathfrak{g}^*\otimes\Lambda\mathfrak{g}\otimes F(M)$. Then since $s$ is a derivation, it is enough to define it to look at its action over the basis 
elements:
\begin{align*}
 & sc^k = -\frac{1}{2}f_{ij}^{~k}c^ic^j = \{-\frac{1}{2}f_{ij}^{~n}c^ic^j\bar{c}_n,c^k\}, \\
 & s\bar{c}_k = f_k +f_{ik}^jc^i\bar{c}_j = \{\frac{1}{2}f_{ij}^{~n}c^ic^j\bar{c}_n + f_ic^i,\bar{c}_k\}, \\
 & sf = 0.
\end{align*}
Then, it is clear that $s$ is an inner derivative of degree $1$ (i.e. $s=\{Q,\}$ with $Q\in\mathcal{C}^1$) with
\begin{equation}
 Q = -\frac{1}{2}f_{ij}^{~k}c^ic^j\bar{c}_k + f_ic^i.
\end{equation}
Therefore, $s$ is an inner derivative and thus a Poisson derivative, for the Poisson structure of $\Lambda\mathfrak{g}^*\otimes\Lambda\mathfrak{g}\otimes F(M)$ which is form the Poisson structure 
of each of the components of the tensorial product.

\section{QED}

\subsection{Physical states of a theory}

The usual way to get rid of the orbits of the gauge algebra on the configuration space is to add a gauge-dependent term to the action. This process is called ``fixing the gauge''. However, working 
in a fixed gauge can be very cumbersome. In particular, the gauge symmetry might help to prove renormalizability of the theory. The BRST approach will help: fixing the gauge by a BRST-exact 
term, we will still have the invariance under the BRST symmetry. This might be seen as fixing the gauge on-shell and letting it free off-shell. This 
symmetry will make many proofs simpler.

Let $\Psi\in C^1$ our gauge-fixing fermion. The total action is then
\begin{equation}
 S = \int\d^4x\mathscr{L} + s\Psi.
\end{equation}
The observables and physical fields are then leaving in the zeroth cohomology group of $Q$. This can be seen by using very physical argument, as explained in \cite{Weinberg96b}, the origin of 
the argument being unclear. We will here follow Weinberg's presentation. But before, we need to specify how the BRST operator acts on the fields of any given theory. Let us take the vacuum 
vector $|0>$ and define
\begin{equation}
 Q|0> = 0.
\end{equation}
Then any physical states can be built with some creation operators $a^*(p)$: $|\psi> = \left(\prod a^*_i(p_i)\right)|0>$. We will see below how $Q$ naturally acts on those creation operators. 
Then we simply define the action of $Q$ over $|\psi>$ as:
\begin{equation}
 Q|\psi> = \{Q,\prod_i a^*_i(p_i)\}|0>.
\end{equation}
The argument is that any matrix element $<a|b>$ between two physical states $|a>$ and $|b>$ has to be invariant under gauge transformation, that is, invariant under a change of $\Psi$:
\begin{equation*}
 \Psi\longrightarrow\Psi+\delta\Psi
\end{equation*}
this $\delta$ having nothing to do with the Kozsul-Tate derivative. Hence we get
\begin{equation*}
 0 = \delta<a|b> = i<a|\delta S|b> = i<a|s\delta\Psi|b> = i<a|\{Q,\delta\Psi\}|b>.
\end{equation*}
By definition the ket in the RHS of the last equality is $Q\left(\delta\Psi|b>\right)=Q|b>$ since $|b>$ was assumed to be a physical states and therefore gauge-invariant. Hence we obtained, 
for any physical state
\begin{equation}
 Q|b> = 0.
\end{equation}
With the same argument, we see that if two states differ only by a $Q$-exact term, they will give the same contributions to the S-matrix. Hence the physical states lie in the cohomology of the 
BRST operator.

\subsection{Action of BRST differential on fields}

The gauge transformation of the fermionic field $\psi$ and the bosonic one $A^k_{\mu}$ is
\begin{subequations}
 \begin{align}
  & \delta\psi = i\varepsilon^{i}(x)t_{i}\psi \\
  & \delta A_{\mu}^k = \partial_{\mu}\varepsilon^k - f_{ij}^k\varepsilon^iA^j_{\mu}.
 \end{align}
\end{subequations}
Here $\varepsilon^i\in\mathbb{R}$ and is allowed to depend of the space-time position and $t_i$ is a matrix representation of the Lie algebra generator $X_i$. The action of the BRST 
differential over the matter fields is then defined as the action of the gauge transformation with $\varepsilon^i$ being replaced by $c^i$ and the matrix $t_i$ being chosen to be the 
representation $\rho(X_i) = \lambda_i$. The first point is needed to have a differential operator of degree $+1$.
\begin{subequations}
 \begin{align}
  & s\psi = ic^{i}\lambda_{i}(\psi) \\
  & sA_{\mu}^k = \partial_{\mu}c^k - f_{ij}^kc^iA^j_{\mu}.
 \end{align}
\end{subequations}
\begin{propo}
 $s$ is still a differential: $s^2\psi=0$, $s^2A_{\mu}^k=0$.
\end{propo}
\begin{proof}
By direct computation:
\begin{eqnarray*}
 \frac{1}{i}s^2\psi & = & (sc^i)\lambda_i(\psi) - c^is\left(\lambda_i(\psi)\right)\\
                    & = & -\frac{1}{2}f_{ij}^kc^ic^j\lambda_k(\psi) - c^kc^i\lambda_i\left(\lambda_k(\psi)\right).
\end{eqnarray*}
Using the fact that $\lambda_k$ is a representation of the algebra and therefore
\begin{equation*}
 \rho(f_{ik}^mX_m)(\psi) = f_{ik}^m\lambda^m(\psi) = \lambda_i\left(\lambda_k(\psi)\right) - \lambda_k\left(\lambda_i(\psi)\right)
\end{equation*}
we got, using the antisymmetry of the ghosts:
\begin{equation*}
 c^kc^i\lambda_i\left(\lambda_k(\psi)\right) = -\frac{1}{2}c^ic^jf_{ij}^m\lambda_m(\psi).
\end{equation*}
Hence $s^2\psi=0$. For the $A_{\mu}^k$ the computation is quite similar but we have to use Jacobi identity. First, we get:
\begin{equation*}
 s^2A_{\mu}^k = -\frac{1}{2}f_{ij}^k\partial_{\mu}(c^ic^j)+\frac{1}{2}f_{ij}^kf_{lm}^ic^lc^mA_{\mu}^j + f_{ij}^kc^i\left[\partial_{\mu}c^j-f_{lm}^jc^lA^m_{\mu}\right].
\end{equation*}
$\partial_{\mu}$ being a regular derivative we get
\begin{equation*}
 -\frac{1}{2}f_{ij}^k\partial_{\mu}(c^ic^j) = -f_{ij}^kc^i\partial_{\mu}c^j.
\end{equation*}
Moreover $f_{lm}^j$ is antisymmetric under the exchange of $l$ and $m$, while $c^lA^m_{\mu}$ is symmetric, since $A_{\mu}^m$ is a bosonic field. Hence $f_{lm}^jc^lA^m_{\mu} = 0$. We are left 
with 
\begin{equation*}
 s^2A_{\mu}^k = \frac{1}{2}f_{ij}^kf_{lm}^ic^lc^mA_{\mu}^j.
\end{equation*}
Let us notice than $c^lc^mA_{\mu}^j$ is symmetric under the exchange of $l$ and $j$ and under the exchange of $m$ and $j$. Then the Jacobi identity provides
\begin{equation*}
 f_{ij}^kf_{lm} = -f_{il}^kf_{mj}-f_{im}^kf_{jl}.
\end{equation*}
The two last terms are antisymmetric under the exchange of $m$ and $j$ and $j$ and $l$ respectively. Thus they both cancel when contracted with $c^lc^mA_{\mu}^j$ and $s^2A_{\mu}^k=0$.
\end{proof}
Finally, we define the action of the BRST operator $Q$ on $A_{\mu}^k$ and $\psi$ as:
\begin{subequations}
 \begin{align}
  & \{Q,A_{\mu}^k\} = sA_{\mu}^k \\
  & \{Q,\psi\} = s\psi.
 \end{align}
\end{subequations}
This is how we will define the action of $Q$ on the creation operators.

\subsection{Field content of QED}

In QED, the gauge group is $U(1)$, and its Lie algebra is abelian and has only one generator. Hence all the structure constants shall be set to $0$ in this part and the generator index can 
be dropped: we will write $c$ the ghost and $\bar{c}$ the antighost. A dramatic simplification occuring from the vanishing structure constants is $sc=0$.

We follow once again the presentation of \cite{Weinberg96b}. The idea is that we can expand the fields in normal modes (we work in $d=4$).
\begin{subequations}
 \begin{align}
  & A_{\mu}(x) = \left(\frac{1}{\sqrt{2\pi}}\right)^3\int\frac{\d^3p}{\sqrt{2p_0}}\left[a_{\mu}(p)e^{ip.x} + a_{\mu}^*(p)e^{-ip.x}\right] \\
  & c(x) = \left(\frac{1}{\sqrt{2\pi}}\right)^3\int\frac{\d^3p}{\sqrt{2p_0}}\left[c(p)e^{ip.x} + c^*(p)e^{-ip.x}\right] \\
  & \bar{c}(x) = \left(\frac{1}{\sqrt{2\pi}}\right)^3\int\frac{\d^3p}{\sqrt{2p_0}}\left[\bar{c}(p)e^{ip.x} + \bar{c}^*(p)e^{-ip.x}\right].
 \end{align}
\end{subequations}
$a^*_{\mu}$, $c^*$ and $\bar{c}^*$ are the creation operators and $a_{\mu}$, $c$ and $\bar{c}$ are the annihilation ones. We use the same letter for the operators and the fields for the 
ghost and the antighost, but the distinction will be clear from the context. One can similarly expand the $\lambda$ and the fermion $\psi$. But we want to check here that the BRST formalism 
allows us to refind the usual rules of QED and will thus focus our efforts on the gauge field. First we have to find how $Q$ acts on the various creation and annihilation operators. A direct 
computation gives
\begin{align}
 & \{Q,A_{\mu}\} = \partial_{\mu}c \nonumber \\
\Leftrightarrow & \int\frac{\d^3p}{\sqrt{2p_0}}\left[\{Q,a_{\mu}(p)\}e^{ip.x} + \{Q,a_{\mu}^*(p)\}e^{-ip.x}\right] = \int\frac{\d^3p}{\sqrt{2p_0}}\left[ip_{\mu}c(p)e^{ip.x} - ip_{\mu}c^*(p)e^{-ip.x}\right] \nonumber \\
\Leftrightarrow & \begin{cases}
                 \{Q,a_{\mu}(p)\} = ip_{\mu}c(p) \\
                 \{Q,a_{\mu}^*(p)\} = -ip_{\mu}c^*(p).
                \end{cases}
\end{align}
Similarly we find $\{Q,c(p)\}=0$, $\{Q,c^*p\}=0$ (this being coherent with the previous result and $s^2=0$) and
\begin{align}
 \begin{cases}
  & \{Q,\bar{c}^*(p)\} = l^*(p) \\
  & \{Q,\bar{c}(p)\} = l(p)
 \end{cases}
\end{align}
with $l^*(p)$ and $l(p)$ respectively the creation and annihilation operators of the $\lambda$ field.
\begin{thm}
 The Fock space of physical states of QED is ghost and antighost free.
\end{thm}
\begin{proof}
We can prove it by induction. The vacuum state $|0>$ is obviously ghost and antighost free. Let $|\psi>$ be any physical states (therefore $Q|\psi>=0$) without ghosts nor antighost. Then, let 
$|\bar{c},\psi>=\bar{c}^*|\psi>$ be the state with one more antighost. We have:
\begin{equation}
 Q\bar{c}^*|\psi> = l^*|\psi> \neq 0.
\end{equation}
Hence any state with one antighost does not lies in the cohomology of $Q$ and thus is not a physical state. For the ghost, the same argument does not apply but we can notice
\begin{equation*}
 Qa_{\mu}^*|\psi> = -ip_{\mu}c^*|\psi>.
\end{equation*}
Hence, let $v^{\mu}$ be a vector such that $p_{\mu}v^{\mu}\neq0$. We get:
\begin{equation*}
 c^*|\psi> = Q\left(\frac{i}{p_{\mu}v^{\mu}}a_{\mu}^*|\psi>\right).
\end{equation*}
Hence $c^*|\psi>$ is $Q$-exact and therefore equivalent to zero.
\end{proof}
Finally, let us notice that a physical photon is characterized by a impulsion and a polarization. Hence its creation operator can be decomposed in
\begin{equation}
 a_{\mu}^*(p) = e_{\mu}a^*(p).
\end{equation}
With $a^*(p)$ the operator that creates a photon of impulsion $p$ and $e_{\mu}$ the polarization vector that carries the polarization of the created photon: $e_{\mu}e^{\mu}=1$. Using the fact 
that $\{Q,.\}$ is a graded derivation we get
\begin{equation}
 \{Q,a_{\mu}^*(p)\} = -ip_{\mu}c^*(p) = e_{\mu}\{Q,a^*(p)\} \Leftrightarrow \{Q,a^*(p)\} = -ip_{\mu}e^{\mu}.
\end{equation}
Now, let once again $|\psi>$ be any physical state. The state with one more photon is now $|e,\psi>:=e_{\mu}a^*(p)|\psi>$. It is a physical state if
\begin{eqnarray*}
 0 & = & Q|e,\psi> \\
   & = & e_{\mu}\{Q,a^*(p)\}|\psi> \\
   & = & -ie_{\mu}p_{\nu}e^{\nu}|\psi>.
\end{eqnarray*}
Hence, the state with one added photon is physical only if $p_{\mu}e^{\mu}=0$. This is the usual polarization condition coming from Maxwell's theory.

\chapter{From BV to BRST}

This appendix is a brief set of notes of a presentation given by Serguei Barannikov in December 2014. He presented how to see the BV and the BRST formalism as dual to each other. Here we will just 
stick to the ideas and not give many proofs.

\section{Non-degenerate polyvectors}

Let $M$ be a (super)manifold and $f:M\longrightarrow\mathbb{R}$ be a function that have a critical point, i.e. a point $x$ where $df(x)=0$. Then this point is said to be non-degenerate if and only if 
the determinant of the hessian matrix $H_{ij}(f)=\frac{\partial^2f(x)}{\partial x^i\partial x^j}$is non-zero. Physically, this means that $x$ is an isolated extremum of $f$. Now, if we see
$f$ as a function of $\Pi T^*M$ which does not depend of the odd coordinates, we trivially have $\{f,f\}=0$ since $\partial_{X^*_i}f=0$, with the brackets being the Schouten--Nijenhuis brackets. 

Now, let us assume that $M$ has a symplectic structure and take $\gamma$ a Poisson bivector of $M$. Then $\{\gamma,\gamma\}=0$, once again with the Schouten--Nijenhuis 
brackets. Then $\gamma$ induces a Poisson structure over $\mathcal{C}^{\infty}(M)$ by $\{f,g\}_P=\gamma(df\wedge dg)$. The important thing is that we can see $\gamma$ as a function of $\Pi T^*M$.
If we expand $\gamma$ around a critical point:
\begin{equation}
 \gamma = \gamma^{ij}X^*_iX^*_j + \gamma^{ij}_kX^kX^*_iX^*_j+\dots
\end{equation}
we have again a notion of non-degeneration: $\gamma$ is said to be non-degenerate if det$|\gamma^{ij}|\neq0$. Our goal is to unify this two cases.

\begin{defi}
 A solution $S = S_2+S_3+\dots$ ($S_2$ the quadratic part of $S$, $S_3$ its cubic part and so on) of the CME will is non-degenerate if
 \begin{equation*}
  \text{Ker}(Q^S_{lin})=\text{Im}(Q^S_{lin})
 \end{equation*}
with $Q^S_{lin}:=\{S_2,-\}$.
\end{defi}
This definition is justified by the following lemma
\begin{lemm}
 Functions and bivectors non-degenerate in the sense of the above definition are non-degenerate in the usual sense.
\end{lemm}
\begin{proof}
 For a function $f=f_{ij}x^ix^j+\dots$ we have
 \begin{eqnarray*}
  Q^f_{lin}(x^i) & = & 0 \\
  Q^f_{lin}(x^*_j) & = & f_{ij}x^i.
 \end{eqnarray*}
Then Ker$(Q^f_{lin})=$Im$(Q^f_{lin})\Leftrightarrow\forall j,f_{ij}x^i\neq0$. This is equivalent to say that the hessian of $f$ does not admit $0$ as an eigenvalue and therefore that $f$ is not 
degenerate in the usual sense.

For a polyvector, the proof is exactly the same with the roles of $x^i$ and $x^*_j$ exchanged since $\gamma=\gamma^{ij}x^*_ix^*_j+\dots$
\end{proof}

Therefore we have extended the usual notion of non-degeneracy to polyvectors that satisfy the classical BV equation. 

To summarize, if we have a solution of classical BV, it is said to be non-degenerate if its quadratic part has non cohomology. In physics, the gauge invariance makes a degeneration of the action, 
which is removed by gauge-fixing it with the Fadeev--Popov formula. Hence we see that if we work with a gauged-fixed action that obey the classical BV action (and it does by construction of the 
action) it shall has non cohomology. We will now see that this cohomological point of view allows to compute solutions of the CME, or even its quantum version.

\section{Darboux--Morse theorem}

The general ideas behind the Morse theory is to study the topology of a manifold by studying the differential functions on that manifold. There is many
powerful applications of this theory, but we will here be interested by the well-known Morse lemma:
\begin{lemm}
 Let $f$ be a function $\mathcal{C}^\infty$ on a manifold of dimension $n$. Let $m$ be a critical non-degenerate point of $f$. Then, there is a local coordinate 
 system $x_1, \dots, x_n$ having $m$ as its center and such that $f$ written in this coordinate system is
 \begin{equation}
  f(x)=f(m)-x_1^2-\cdots -x_k^2 +x_{k+1}^2+\cdots +x_n^2.
 \end{equation}
\end{lemm}
The idea of the proof is to expand $f$ around $m$ and to use its non-degeneracy to exhibit a local coordinate system where the lemma holds. 

We have extended the notion of non-degeneracy to polyvectors obeying the CME, therefore we are expecting this Morse's lemma to have a generalization. Moreover, this generalization will take 
care of the symplectic structure of $\Pi T^*M$.
\begin{thm}[Darboux--Morse \cite{Va96}]
 Let $S$ be a solution of the CME having a non-degenerate critical point at $x=x^*=0$. Then it exists a symplectomorphism that reduces $S$ to its quadratic part.
\end{thm}
\emph{Hint of proof.} Decompose $S$ with respect to the grading: $S=S_2+\dots+S_n$. Using the grading on the CME we then have $\{S_2,S_n\}=0$. Since $S$ is non-degenerate, 
$S_n$ is a coboundary for $S_2$: $S_n=\{S_2,X\}$ for some $X$. Studying the possible $X$ give a symplectomorphism such $S=S'_2+\dots+S'_{n-1}$. Then we can iterate the procedure.
\begin{flushright}
 $\blacksquare$
\end{flushright}
This theorem will turn to be essential in the next steps.

\section{Quantum BV}

For a quantum version of BV, we take $S$ to be a formal series of powers of $\hbar$:
\begin{equation}
 S = S^0+2i\hbar S^1+\dots
\end{equation}
with $S^0$ non-degenerate. Moreover, let us assume that $S$ is a solution of the quantum master equation (QME):
\begin{equation}
 \{S,S\}-i\hbar\Delta S=0,
\end{equation}
with $\Delta$ the BV laplacian which is a graded derivation for the Schouten--Nijenhuis bracket, as already mentioned \eqref{BV_prod}. Then, the grading induced by $\hbar$ imposes that $S_0$ is 
a solution of the classical master equation, and
\begin{equation} \label{equ_S1}
 \Delta S^0 = \{S^0,S^1\}.
\end{equation}
The fact that $S^0$ is non-degenerate allows to show that a solution $S^1$ of the above equation exists. Indeed, since $S^0$ is a solution of the CME 
and since $\Delta$ is a derivation for the Schouten--Nijenhuis bracket, we have $\{S^0,\Delta S^0\}=0$: $\Delta S^0$ is a cocycle of $S^0$. Since $S^0$ is 
non-degenerate, $\Delta S^0$ has to be a coboundary as well: it exists $S^1$ that satisfies \eqref{equ_S1}.

Actually, one can construct $S^1$ by taking a homotopy operator $H$ (that is: an operator such that $Q_{lin}H+HQ_{lin}=$Id). $H$ is 
the equivalent of a propagator. Thanks to the Darboux--Morse theorem, we can bring $S_0$ to its quadratic part so that $Q_{lin}^{S_0}=\{S_0,-\}$. This is needed since we want to refind the 
classical quadratic action as $\hbar\longrightarrow0$. Applying $Q$ to \eqref{equ_S1} give $Q_{lin}^{S_0}\Delta S^0 = 0$. Now, from this and the definition of $H$ we get
\begin{equation*}
 Q_{lin}^{S_0}H\Delta S^0 = \Delta S^0.
\end{equation*}
Therefore $S^1=H\Delta S^0$, up to a $Q_{lin}^{S_0}$-exact term. Since the physical observables are elements of the cohomology of $Q_{lin}^{S_0}$, we have built $S^1$ in term of $S^0$. The 
difficulty of an explicit computation will lie in the construction of $H$, that has to be performed by looking at the action of $Q_{lin}^{S_0}$ over the fields of the theory. With the same 
strategy, $S$ can be computed to an arbitrary order by iteration.

\section{equivariante cohomology}

Let us take $M$ a manifold, and $\mathfrak{g}$ a Lie algebra acting on $M$:
\begin{equation}
 e_{\alpha}\in\mathfrak{g}\longrightarrow v_{\alpha}= v_{\alpha}^i(x)\frac{\partial}{\partial x^i}\left|[v_{\alpha};v_{\beta}] = \Gamma_{\alpha\beta}^{\gamma}v_{\gamma}\right.
\end{equation}
for $[e_{\alpha};e_{\beta}] = \Gamma_{\alpha\beta}^{\gamma}e_{\gamma}$ (yes, we have changed our notations). Geometrically speaking, we are working on a manifold $\tilde{M}=M\times\Pi\mathfrak{g}$ 
with coordinates $(x^i,c^{\alpha})$. Now, let us define an odd vector field on $\tilde{M}$:
\begin{equation}
 Q = c^{\alpha}v_{\alpha}^i(x)\frac{\partial}{\partial x^i} + \frac{1}{2}\Gamma_{\alpha\beta}^{\gamma}c^{\alpha}c^{\beta}\frac{\partial}{\partial c^{\gamma}} = Q_1+Q_2.
\end{equation}
We have $Q^2=0$ since $\mathfrak{g}$ is a Lie algebra and since it admit a representation on $M$: $Q$ is a homological vector field. In the case of the BV formalism, we can take for $\mathfrak{g}$ 
a structure more general than a Lie algebra. From the manifold $\tilde{M}$ and the homological odd vector field $Q$, there are two natural structures that can be built: the space of differential 
forms $\Pi T^*\tilde{M}$ on $\tilde{M}$ and the space of polyvector fields $\Pi T\tilde{M}$ on $\tilde{M}$. The first will be called the BRST model, and the second the BV formalism. Hence, within 
this approach, BRST and BV formalisms are somehow dual.

\subsection{Polyvectors: BV formalism}

First, let $v$ be a vector field on $\tilde{M}$. In a given coordinate system:
\begin{equation}
 v(x,c) = v_1^i(x,c)\frac{\partial}{\partial x^i} + v_2^{\alpha}(x,c)\frac{\partial}{\partial c^{\alpha}}.
\end{equation}
We define its symbol to be
\begin{equation}
 S_v = v_1^i(x,c)X^*_i + v_1^i(x,c)c^*_{\alpha}
\end{equation}
with $(x^*_i,c^*_{\alpha})$ the coordinates on the fiber, with reverted parity. Let $S_Q$ be the symbol of $Q$:
\begin{equation*}
  S_Q = c^{\alpha}v_{\alpha}^i(x)x^*_i + \frac{1}{2}\Gamma_{\alpha\beta}^{\gamma}c^{\alpha}c^{\beta}c^*_{\alpha}.
\end{equation*}
It has already been shown in the chapter devoted to the BV formalism that $S_Q$ is a solution of the CME.
\begin{lemm}
 For any $f=f(x)$ invariant under the action of $\mathfrak{g}$ (i.e. $\{f,Q_1\}=0$) then $f+S_Q$ is still a solution of the CME
\end{lemm}
\begin{proof}
 Since $f$ is a function of the $x$ only, $\{f,f\}=0$ then
 \begin{eqnarray*}
  \{f+S_Q,f+S_Q\} & = & 2\{f,S_Q\} \\
                  & = & 2c^{\alpha}v_{\alpha}^i(x)\frac{\partial f}{\partial x^i} \\
		  & = & \{f,Q_1\} = 0
 \end{eqnarray*}
where we used in the second line that $f$ is a function of $x$ only.
\end{proof}
This lemma generalized the proposition \ref{prop21}. Hence we see that with this approach we built once again the BV procedure.

The key point is that if $f$ is degenerate in the directions generated by the action of $\mathfrak{g}$ (which is the case if $f$ is the classical action of our physical theory), then 
$f(x)+c^{\alpha}v_{\alpha}^i(x)c^*_i$ is not degenerate at all. The $c_{\alpha}^*$ realize Tate's procedure.

\subsection{Differential forms: BRST model}


Now the coordinates are $(x^i,\bar{x}^i,c^{\alpha},\bar{c}^{\alpha})$. $Q$ acts via an operator $L_Q$. On the BRST side, we also have the de Rham differential:
\begin{equation}
 d^{\tilde{M}} = \bar{x}^i\frac{\partial}{\partial x^i} + \bar{c}^{\alpha}\frac{\partial}{\partial c^{\alpha}}
\end{equation}
and we define $I_Q$ the operator of contracting with $L_Q$. Then it can be shown that $L_Q$ has no cohomology and 
\begin{equation}
 \exp(I_Q)d^{\tilde{M}}\exp(-I_Q) = d^{\tilde{M}} + [d^{\tilde{M}};I_Q].
\end{equation}
Now, let us define the equivariant cohomology. We work in $\Omega_M\otimes W(\mathfrak{g})$, with $W(\mathfrak{g})$ the Weyl algebra of $\mathfrak{g}$.
The components of $W(\mathfrak{g})$ are $c^{\alpha}$ and $\bar{c}^{\alpha}$. We define the Weyl differential as
\begin{subequations}
 \begin{eqnarray}
        d_Wc^{\alpha} & = & \bar{c}^{\alpha} - \frac{1}{2}\Gamma_{\beta\gamma}^{\alpha}c^{\beta}c^{\gamma} = d^{\tilde{M}}(c^{\alpha})+L_{Q_2}c^{\alpha} \\
  d_W\bar{c}^{\alpha} & = & \Gamma_{\beta\gamma}^{\alpha}\bar{c}^{\beta}c^{\gamma} = L_{Q_2}\bar{c}^{\alpha}.
 \end{eqnarray}
\end{subequations}
It is then simple to check
\begin{lemm} 
 We can define the cohomology of $d_W$, indeed $d^2_W=0$.
\end{lemm}
\begin{proof}
 Since the $\bar{c}$ commute with each other we have
 \begin{equation*}
  d^2_W\bar{c}^{\alpha} = \Gamma_{\beta\gamma}^{\alpha}\Gamma_{\delta\varepsilon}^{\beta}c^{\varepsilon}c^{\gamma}\bar{c}^{\delta} + \frac{1}{2}\Gamma_{\beta\gamma}^{\alpha}\Gamma_{\delta\varepsilon}^{\gamma}c^{\delta}c^{\varepsilon}\bar{c}^{\beta}.
 \end{equation*}
 Using Jacobi's relation we get $\Gamma_{\beta\gamma}^{\alpha}\Gamma_{\delta\varepsilon}^{\gamma}c^{\delta}c^{\varepsilon}\bar{c}^{\beta}=-2\Gamma_{\beta\gamma}^{\alpha}\Gamma_{\delta\varepsilon}^{\beta}c^{\varepsilon}c^{\gamma}\bar{c}^{\delta}$
 hence $d^2_W\bar{c}^{\alpha}=0$. Moreover
 \begin{equation*}
  d^2_Wc^{\alpha} = \Gamma_{\beta\gamma}^{\alpha}\bar{c}^{\beta}c^{\gamma} - \frac{1}{2}\Gamma_{\beta\gamma}^{\alpha}(\bar{c}^{\beta} - \frac{1}{2}\Gamma_{\delta\varepsilon}^{\beta}c^{\delta}c^{\varepsilon})c^{\delta} - \frac{1}{2}\Gamma_{\beta\gamma}^{\alpha}c^{\beta}(\bar{c}^{\gamma} - \frac{1}{2}\Gamma_{\delta\varepsilon}^{\gamma}c^{\delta}c^{\varepsilon}).
 \end{equation*} 
Then, using anticommutation between the $c$ and the $\bar{c}$ and between the $c$'s themselves we easily get that the three terms in $c\bar{c}$ combine to vanish and that the two terms in $ccc$ 
are just the opposite to each other. Therefore $d_Wc^{\alpha}=0$.
\end{proof}
Then we study $d^M + d_W = d^{\tilde{M}} + L_{Q_2}$ acting on the ``basic form'' i.e. the forms $\omega$ such that $L_{\alpha}\omega=0$ and $I_{\alpha}\omega=0$
$\forall\alpha$.

It might be not so clear why this construction gives the usual physicists' BRST formalism. However, it was shown in \cite{OuStvBa89} that those two approach coincide. In this article, the 
authors defined the BRST differential by its action over the fields. They show that this gives a differential algebra structure having the BRST differential as its exterior derivative. Then, it 
is shown that the action is a basic local functional, in the sense given above.

One could be worry that in this approach, we seems to have more fields than in the ``usual'' BRST formalism. Actually, the supernumerary fields will be constrained. Finally, the $\bar{c}$ are not 
the $\bar{c}$ of the usual BRST. Indeed, the later are anticommuting. Actually, the usual (anticommuting) $\bar{c}$ will be the unfixed $\bar{x}$.

\chapter{Feynman integrals}

\section{Two-points integrals}

\subsection{A new method for two-points integrals}

Let $G$ be a two-points graph of a scalar theory: res$(G)=2$. Then its evaluation $I_G$ depends on $p^2$ (the momentum flowing through the graph) and a set of parameters $\{\beta_i\}_{i=1,\dots,n}$. Now let us 
assume that the dependence of $I_G$ in $p^2$ can be written
\begin{equation}
 I_G = (p^2)^{\beta_0-D/2}\mathcal{N}(\{\beta_i\})
\end{equation}
for some $\beta_0$ defined by the above definition and let $\tilde{G}$ we the graph obtained by joining the two external edges of $G$. We defined the Feynman evaluation of $\tilde{G}$ to be
\begin{equation}
 I_{\tilde{G}}:=\int\frac{\d^Dp}{\pi^{D/2}}\frac{I_G(p^2)}{(p^2)^{\beta_0}}.
\end{equation}
Now, if we write the Feynman integral of a graph $H$ in the Schwinger representation
\begin{equation*}
 I_H = \int\left(\prod_{i=0}^l\frac{\d\alpha_i}{\Gamma(\alpha_i)}\alpha_i^{\beta_i-1}\right)U_H^{-D/2}\exp\left(-\frac{\mathcal{V}_H}{U_H}\right)
\end{equation*}
with $U_H$ and $\mathcal{V}_H$ its first and second Symanzik polynomials and $l$ the number of internal edges of $H$. The reduced Feynman integral is then defined as
\begin{equation}
 I_H^{(r)} := \int\left(\prod_{i=0}^l\frac{\d\alpha_i}{\Gamma(\alpha_i)}\alpha_i^{\beta_i-1}\right)\delta(1-\sum_{i\in\mathcal{J}}\alpha_i)U_H^{-D/2}\exp\left(-\frac{\mathcal{V}_H}{U_H}\right)
\end{equation}
with $\mathcal{J}$ any no-empty subset of the set labels of the internal edges of $H$, here $\{1,\dots,l\}$. This object is \`a priori not well-defined but the famous (some say folkloric) 
theorem of Cheng and Wu (\cite{Cheng}, page 259) states that the choice of $\mathcal{J}$ does not change the evaluation of the integral if the integral is scale invariant (that is, if the 
dimension of the numerator is the dimension of $p^{-D}$). The scale invariance will also makes the integral divergent. Indeed, the integral of $H$ can be seen as an integral over a 
hypersurface of the real projective space $\mathbb{RP}^l$ (with $l$ the number of internal edges of $H$). Then the integral over $\tilde{H}$ is an integral over $\mathbb{RP}^l$, but the scale 
invariance makes the integrand constant in one direction, thus making $I_{\tilde{H}}$ (if res$(H)=2$) divergent. Hence, those integrals typically appear only in a reduced form.

Now, for $\tilde{G}$, the integral is indeed scale invariant (it is why we have defined $\beta_0$ the way we did) and therefore $I_{\tilde{G}}^{(r)}$ is well-defined. We can now state a useful  
theorem.
\begin{thm} \label{th_lien_int_exp}
 With the above definitions we have
 \begin{equation} 
  \mathcal{N}(\{\beta_i\}) = \Gamma(D/2)I_{\tilde G}^{(r)}.
 \end{equation}
\end{thm}
\begin{proof}
 We simply use the Schwinger trick on $I_{\tilde{G}}$:
 \begin{align*}
  I_{\tilde{G}}^{(r)} & = \int\frac{\d^Dp}{\pi^{D/2}}I_G(p^2)\delta(1-\sum_{i\in\mathcal{J}}\alpha_i)\frac{1}{\Gamma(\beta_0)}\int_0^{+\infty}\d\alpha_0\alpha_0^{\beta_0-1}e^{-\alpha_0p^2} \\
                      & = \frac{1}{\Gamma(\beta_0)}\int_0^{+\infty}\d\alpha_0\alpha_0^{\beta_0-1}\delta(1-\sum_{i\in\mathcal{J}}\alpha_i)\int\frac{\d^Dp}{\pi^{D/2}}I_G(p^2)e^{-\alpha_0p^2}.
 \end{align*}
 Choosing in the definition of $I_{\tilde{G}}^{(r)}$ $\mathcal{J}=\{0\}$ (since here the set of labels of the internal edges of $\tilde{G}$ is $\{0,1,\dots,l\}$) and invoking the theorem of 
 Cheng and Wu we get
 \begin{equation} \label{lien_int_exp}
  \frac{1}{\Gamma(\beta_0)}\int\frac{\d^Dp}{\pi^{D/2}}I_G(p^2)e^{-p^2} = I_{\tilde{G}}^{(r)}.
 \end{equation}
 Now, let us use the simple dependence of $I_G$ in $p^2$ to write
 \begin{equation*}
  \int\frac{\d^Dp}{\pi^{D/2}}I_G(p^2)e^{-p^2} = \frac{2}{\Gamma(D/2)}\mathcal{N}(\{\beta_i\})\int_0^{+\infty}\d pp^{D-1}(p^2)^{\beta_0-D/2}e^{-p^2}
 \end{equation*}
 where we have performed the angular integration since the integrand had no angular dependence. Making the change of variables $u=p^2\Rightarrow\d u/u=2\d p/p$ we get
 \begin{equation*}
  \int\frac{\d^Dp}{\pi^{D/2}}I_G(p^2)e^{-p^2} = \frac{1}{\Gamma(D/2)}\mathcal{N}(\{\beta_i\})\int_0^{+\infty}\d uu^{D/2-1}u^{\beta_0-D/2}e^{-u} = \frac{\Gamma(\beta_0)}{\Gamma(D/2)}\mathcal{N}(\{\beta_i\})
 \end{equation*}
 which, combined with \eqref{lien_int_exp} gives the formula of the theorem
\end{proof}
This formula is useful since the second Symanzik polynomial is vanishing (in a massless theory) for a vacuum graph, due to the absence of external momenta. Therefore, if we can prove that a 
given integral has a simple dependence in the external momenta without computing it we will rather compute the evaluation of the completed graph rather than the one of the initial graph.

Now, before we apply this technique to an integral with a numerator, it will be useful to check that it gives the right result on a known integral, namely the one loop Mellin transform 
computed in chapter (\ref{chap3}).

\subsection{Sanity check}

We will compute once again the integral
\begin{equation}
 I(p^2,\beta_1,\beta_2) = \int\frac{\d^Dp}{\pi^{D/2}}\frac{1}{(q^2)^{\beta_1}((p-q)^2)^{\beta_2}} = (p^2)^{D/2-\beta_1-\beta_2}\mathcal{N}(\beta_1,\beta_2).
\end{equation}
Therefore we define $\beta_0=D-\beta_1-\beta_2$. The above integral is the Feynman integral of the graph
\begin{equation*}
 G = 
 \begin{tikzpicture}[baseline={([yshift=-.5ex]current bounding box.center)}] 
  \draw (-0.5,0) -- (0,0);
  \draw (0.5,0) circle (0.5);
  \draw (1,0) -- (1.5,0);
  \draw (0,0) node {$\bullet$};
  \draw (1,0) node {$\bullet$};
 \end{tikzpicture}
\end{equation*}
hence we will compute the reduced Feynman integral of the graph
\begin{equation*}
 \tilde{G} = 
 \begin{tikzpicture}[baseline={([yshift=-.5ex]current bounding box.center)}] 
  \draw (0.5,0) circle (0.5);
  \draw (0,0) -- (1,0);
  \draw (0,0) node {$\bullet$};
  \draw (1,0) node {$\bullet$};
 \end{tikzpicture}.
\end{equation*}
$\tilde{G}$ has three spanning trees: each are composed of the two vertices (obviously) and of one edge. Hence its first Symanzik polynomial is
\begin{equation}
 U_{\tilde{G}} = \alpha_0\alpha_1 + \alpha_0\alpha_2 + \alpha_1\alpha_2.
\end{equation}
Choosing in the definition of $I_{\tilde{G}}^{(r)}$ $\mathcal{J}=\{0\}$ we get
\begin{equation*}
 I_{\tilde{G}}^{(r)} = \frac{1}{\Gamma(\beta_0)\Gamma(\beta_1)\Gamma(\beta_2)}\int_{(\mathbb{R}_+)^2}\d\alpha_1\d\alpha_2\frac{\alpha_1^{\beta_1-1}\alpha_2^{\beta_2-1}}{(\alpha_1+\alpha_2+\alpha_1\alpha_2)^{D/2}}.
\end{equation*}
Using here the Schwinger trick to write the denominator as an exponential end inverting the order of the integrations we have
\begin{equation*}
 I_{\tilde{G}}^{(r)} = \frac{1}{\Gamma(D/2)\Gamma(\beta_0)\Gamma(\beta_1)}\int_0^{+\infty}\d tt^{D/2-1}\int\d\alpha_1\alpha_1^{\beta_1-1}e^{-t\alpha_1}\int\d\alpha_2\alpha_2^{\beta_2-1}e^{-\alpha_2t(1+\alpha_1)}.
\end{equation*}
With a simple change of coordinate $u=\alpha_2t(1+\alpha_1)$ we recognize the integral of $\alpha_2$ to be simply an integral representation of the $\Gamma$ function. Therefore
\begin{equation*}
 I_{\tilde{G}}^{(r)} = \frac{1}{\Gamma(D/2)\Gamma(\beta_0)\Gamma(\beta_1)}\int_0^{+\infty}\d tt^{D/2-1}\int\d\alpha_1\frac{\alpha_1^{\beta_1-1}e^{-t\alpha_1}}{[t(1+\alpha_1)]^{\beta_2}}.
\end{equation*}
We can recognize the integral representation of the confluent hypergeometric function of the second kind in the integral over $\alpha_1$. This representation is
\begin{equation}
 U(a,b,z) = \frac{1}{\Gamma(a)}\int_0^{+\infty}\d te^{-zt}t^{a-1}(1+t)^{b-a-1}.
\end{equation}
This representation can be found in \cite{Gradshteyn}, page 1023. $U(a,b,z)$ can be defined with this formula. After the change of variables $v=\alpha_1t$ we arrive to
\begin{equation*}
 I_{\tilde{G}}^{(r)} = \frac{1}{\Gamma(D/2)\Gamma(\beta_0)}\int_0^{+\infty}\d tt^{D/2-\beta_2-1}U(\beta_1,1+\beta_1-\beta_2,t).
\end{equation*}
Finally, the integral over $t$ can be computed using another formula from \cite{Gradshteyn}, page 821:
\begin{equation}
 \int_0^{+\infty}\d t t^{b-1}U(a,c,t) = \frac{\Gamma(b)\Gamma(a-b)\Gamma(b-c+1)}{\Gamma(a)\Gamma(a-c+1)}.
\end{equation}
Hence we get the result
\begin{equation}
 I_{\tilde{G}}^{(r)} = \frac{\Gamma(\beta_1+\beta_2-D/2)\Gamma(D/2-\beta_1)\Gamma(D/2-\beta_2)}{\Gamma(D/2)\Gamma(\beta_0)\Gamma(\beta_1)\Gamma(\beta_2)}.
\end{equation}
The manifest symmetry between $\beta_1, \beta_2$ and $\beta_0$ can be restored if we notice $\beta_1+\beta_2-D/2=D/2-\beta_0$. Then we have our result
\begin{equation} \label{eval_graph_compl}
 I_{\tilde{G}}^{(r)} = \frac{\Gamma(D/2-\beta_0)\Gamma(D/2-\beta_1)\Gamma(D/2-\beta_2)}{\Gamma(D/2)\Gamma(\beta_0)\Gamma(\beta_1)\Gamma(\beta_2)}.
\end{equation}
Now, using $\beta_0=D-\beta_1-\beta_2$ and the formula of theorem \ref{th_lien_int_exp} we get
\begin{equation}
 I(p^2,\beta_1,\beta_2) = (p^2)^{D/2-\beta_1-\beta_2}\frac{\Gamma(\beta_1+\beta_2-D/2)\Gamma(D/2-\beta_1)\Gamma(/2-\beta_2)}{\Gamma(D-\beta_1-\beta_2)\Gamma(\beta_1)\Gamma(\beta_2)},
\end{equation}
in agreement with the result found in chapter 3. Now we fill confident enough to use this technique to compute an unknown integral with a numerator.

\subsection{Integral with numerator}

Let us now compute
\begin{equation}
 I_3(\beta_1,\beta_2,\beta_3,p^2) = \int\frac{\d^Dq}{\pi^{D/2}}\frac{(q^2+(p-q)^2+p^2)^{\beta_3}}{(q^2)^{\beta_1}((p-q)^2)^{\beta_2}}.
\end{equation}
A simple dimensional analysis shows that this integral shall be proportional to $(p^2)^{D/2+\beta_3-\beta_2-\beta_1}$. A more formal argument consists into splitting the numerator to write 
$I_3$ as a series of $I$, the integral re-computed in the previous section. Let us start by using the Schwinger trick on the numerator. After a exchange in the order of the integrations we 
obtain
\begin{equation*}
 I_3(\beta_1,\beta_2,\beta_3,p^2) = \frac{1}{\Gamma(-\beta_3)}\int_0^{+\infty}t^{-\beta_3-1}e^{-tp^2}\int\frac{\d^Dq}{\pi^{D/2}}\frac{e^{-t(q^2+(p-q)^2)}}{(q^2)^{\beta_1}((p-q)^2)^{\beta_2}}.
\end{equation*}
Now expanding the exponential in the integral over $q$ and switching the order of the series and the integral we get
\begin{equation*}
 I_3(\beta_1,\beta_2,\beta_3,p^2) = \frac{1}{\Gamma(-\beta_3)}\int_0^{+\infty}t^{-\beta_3-1}e^{-tp^2}\sum_{n=0}^{+\infty}\frac{(-t)^n}{n!}\int\frac{\d^Dq}{\pi^{D/2}}\frac{[q^2+(p-q)^2]^n}{(q^2)^{\beta_1}((p-q)^2)^{\beta_2}}.
\end{equation*}
Using the binomial expansion on every term of the series we recognize the integral computed in the previous subsection:
\begin{equation*}
 I_3(\beta_1,\beta_2,\beta_3,p^2) = \frac{1}{\Gamma(-\beta_3)}\int_0^{+\infty}t^{-\beta_3-1}e^{-tp^2}\sum_{n=0}^{+\infty}\frac{(-t)^n}{n!}\sum_{k=0}^n\binom{n}{k}I(\beta_1-k,\beta_2-(n-k),p^2).
\end{equation*}
We know that $I(\beta_1-k,\beta_2-(n-k),p^2)$ is proportional to $(p^2)^{D/2-(\beta_1-k)-(\beta_2-n+k)}=(p^2)^{D/2-\beta_1-\beta_2+n}$. Hence, without writing explicitly all the term of the 
series we have
\begin{equation*}
 I_3(\beta_1,\beta_2,\beta_3,p^2) = (p^2)^{D/2-\beta_1-\beta_2}\int_0^{+\infty}t^{-\beta_3-1}e^{-tp^2}\sum_{n=0}^{+\infty}\frac{(-t)^n(p^2)^n}{n!}S_n(\beta_1,\beta_2).
\end{equation*}
Finally, performing the change of integration $u=tp^2$ we get an integrand without any $p^2$ and, as expected
\begin{equation}
 I_3(\beta_1,\beta_2,\beta_3,p^2) = (p^2)^{D/2+\beta_3-\beta_1-\beta_2}\mathcal{N}_3(\beta_1,\beta_2,\beta_3).
\end{equation}
Hence, we can use the theorem \ref{th_lien_int_exp}. The remarkable feature here is that we do not have to compute the an other Feynman integral: the graph $\tilde{G}$ is the same than the one 
computed in the previous subsection. Taking $\beta_0=D+\beta_3-\beta_2-\beta_1$ in \eqref{eval_graph_compl} the theorem directly gives
\begin{equation}
 I_3(\beta_1,\beta_2,\beta_3,p^2) = (p^2)^{D/2+\beta_3-\beta_1-\beta_2}\frac{\Gamma(\beta_1+\beta_2-\beta_3-D/2)\Gamma(D/2-\beta_1)\Gamma(D/2-\beta_2)}{\Gamma(D+\beta_3-\beta_1-\beta_2)\Gamma(\beta_1)\Gamma(\beta_2)}.
\end{equation}
We see here that this method is very powerful. An interesting generalization of the theorem \ref{th_lien_int_exp} would be a similar formula for $3$-valent and $n$-valent graphs. We did not 
find the time to look for such a generalization and therefore we will attack the upcoming $3$-valent graph integral in a more ``brute force'' way.

 

%
%
%
%
%
%
%
%

\section{Three-point integral}

Our goal is to compute a three point scalar integral which is given in \cite{BoDa86}. We will roughtly follow their computations.

\subsection{Formula and method}

Thorought this task, repeated use will be made of some equalities. Let us list them and give some references
\begin{equation}
 \int_0^1\d\eta\frac{\eta^{\alpha-1}(1-\eta)^{\beta-1}}{[a\eta+b(1-\eta)]^{\gamma}} = \frac{\Gamma(\alpha)\Gamma(\beta)}{b^{\gamma}\Gamma(\alpha+\beta)}\hyper{\alpha}{\gamma}{\alpha+\beta}{1-\frac{a}{b}} \label{int2F1}
\end{equation}
with $~_2F_1$ the Gauss' hypergeometric function. This is a well-known formula: it is typically known by mathematica. It can be found in \cite{Gradshteyn}, page 317. We also use the classical 
reflexion formula for the hypergemetric function (which defines its analytical continuation)
\begin{align}
 \hyper{a}{b}{c}{1-z} & = \frac{\Gamma(c)\Gamma(c-a-b)}{\Gamma(c-a)\Gamma(c-b)}\hyper{a}{b}{a+b-c+1}{z} \label{prolong} \\
\llcorner & + \frac{\Gamma(c)\Gamma(a+b-c)}{\Gamma(a)\Gamma(b)}z^{c-a-b}\hyper{c-a}{c-b}{c-a-b+1}{z} \nonumber
\end{align}
which is also in \cite{Gradshteyn}, page 1008. Another classical identity which will prove itself very convenient is a Mellin--Barnes representation of the hypergeometric function, 
\cite{Gradshteyn}, page 1005:
\begin{equation}
 \hyper{a}{b}{c}{z} = \frac{\Gamma(c)}{\Gamma(a)\Gamma(b)}\frac{1}{2i\pi}\int_{-i\infty}^{+i\infty}\d s(-z)^s\frac{\Gamma(-s)\Gamma(a+s)\Gamma(b+s)}{\Gamma(c+s)}. \label{MB}
\end{equation}
Let us notice that the contour of integration is such that it separates the poles in $s=j$ from the poles in $s+a=-j$ and $s+b=-j$. In the following, we will close this contour to the right. Now, 
in \cite{BoDa86} a series is given without proof and seems to have been proven by the authors. We give a proof here and write this formula as a lemma
\begin{lemm} (Davydychev and Boos, \cite{BoDa86})
 For $F_4$ the fourth Appell's hypergeometric function defined by
\begin{equation*}
 \appell{a}{b}{c}{d}{x}{y} = \sum_{m=0}^{+\infty}\sum_{j=0}^{+\infty}\frac{(a)_{m+j}(b)_{m+j}}{(c)_m(d)_j}\frac{x^m}{m!}\frac{y^j}{j!}.
\end{equation*}
we have
\begin{equation} \label{seriesAppell}
 \sum_{j=0}^{+\infty}\frac{x^j}{j!}\frac{(\alpha)_j(\beta)_j}{(\gamma)_j}\hyper{-j}{1-\gamma-j}{\delta}{y} = \appell{\alpha}{\beta}{\gamma}{\delta}{x}{xy}
\end{equation}
\end{lemm}
\begin{proof}
 First, we write
 \begin{equation*}
  \appell{\alpha}{\beta}{\gamma}{\delta}{x}{xy} = \sum_{m=0}^{+\infty}\sum_{j=0}^{+\infty}\frac{(\alpha)_{m+j}(\beta)_{m+j}}{(\gamma)_m(\delta)_j}\frac{x^{m+j}}{m!}\frac{y^j}{j!}.
 \end{equation*}
 Then, using $m! = (m+j)!/(m+1)_j$. Then we can relabel the sum over $j$ to get
 \begin{equation*}
  \appell{\alpha}{\beta}{\gamma}{\delta}{x}{xy} = \sum_{m=0}^{+\infty}\sum_{k=m}^{+\infty}\frac{(\alpha)_{k}(\beta)_{k}}{(\gamma)_{m}(\delta)_{k-m}}(m+1)_{k-m}\frac{x^k}{k!}\frac{y^{k-m}}{(k-m)!}.
 \end{equation*}
 We can inverse the order of the sums, relabel once again the second sum and separate the factors implying just a sum over $k$ to get
 \begin{equation*}
  \appell{\alpha}{\beta}{\gamma}{\delta}{x}{xy} = \sum_{k=0}^{+\infty}\frac{x^k}{k!}\frac{(\alpha)_{k}(\beta)_{k}}{(\gamma)_{k}}\sum_{j=0}^{k}\frac{(\gamma)_k}{(\gamma)_{k-j}(\delta)_j}(k-j+1)_j\frac{y^j}{j!}.
 \end{equation*}
 In this we can use $(\gamma)_k=(\gamma)_{k-j}(\gamma+k-j)_j$ (from the definition of the Pochhammer symbol in term of the Gamma function). Then we recognize that the last sum is the 
hypergeometric function is the formula \eqref{seriesAppell} since $(-k)_n\neq0\Leftrightarrow k\geq n$ and $(-x)_n=(-1)^n(x-n+1)_n$ and $(1-\gamma+k)_n = (-1)^n(\gamma+k-n)_n$.
\end{proof}
Now, in a minkowskian space-time with signature $(-+\dots+)$ in $D$ dimension we want to compute the integral
\begin{equation}
 J(\mu,\nu,\rho) = \int \frac{\d^Dr}{(r^2)^{\mu}((p-r)^2)^{\nu}((q-r)^2)^{\rho}}. \label{int3points}
\end{equation}
After Wick rotating this integral we arrive to a $D$-dimensional integral in euclidean space
\begin{equation}
 J(\mu,\nu,\rho) = i\int \frac{\d^D_Er}{(r^2)^{\mu}((p-r)^2)^{\nu}((q-r)^2)^{\rho}}.
\end{equation}
We will start by writing this as a definite one-dimensional integral. The integrand will be split using \eqref{prolong}. Each of the two terms will be written with the Mellin--Barnes representation 
\eqref{MB}. The order of the integrations will be changed and we will compute the scalar integrals using \eqref{int2F1}. Then we will use once again \eqref{prolong}. Hence we will have four 
Mellin--Barnes integral, which we will compute with the residues theorem. Finally, the six arising series will be written in a closed form using \eqref{seriesAppell}.

\subsection{First computations}

Using the $\alpha$-representation
\begin{equation*}
 \frac{1}{A^k} = \frac{1}{\Gamma(k)}\int_0^{+\infty}\d\alpha\alpha^{k-1}e^{-A\alpha}
\end{equation*}
three times and exchanging the order of the integrations, we can bring the integral over $r$ to a Gaussian integral. Then we perform the change of coordinates
\begin{align*}
 \begin{cases}
  & \alpha_1 = \lambda(1-\xi)(1-\eta) \\
  & \alpha_2 = \lambda\xi(1-\eta) \\
  & \alpha_3 = \lambda\eta.
 \end{cases}
\end{align*}
A useful remark is $\sum_i\alpha_i=\lambda$. Then the integral over $\lambda$ is esay: it is just a $\Gamma$ function. The integral over $\eta$ can be done thanks to \eqref{int2F1}. Finally,
as written in we end up with
\begin{align}
 J(\mu,\nu,\rho) & = i\pi^{D/2}\frac{\Gamma(\mu+\nu+\rho-D/2)\Gamma(D/2-\rho)}{\Gamma(\mu)\Gamma(\nu)\Gamma(D/2)}(p^2)^{D/2-\mu-\nu-\rho} \nonumber \\
\times & \int_0^1\d\xi\xi^{D/2-\mu-\rho-1}(1-\xi)^{D/2-\nu-\rho-1}\hyper{\rho}{\mu+\nu+\rho-D/2}{D/2}{1-z}
\end{align}
with 
\begin{equation}
 z = \frac{(1-\xi)q^2 + \xi k^2}{\xi(1-\xi)p^2}.
\end{equation}
Up to here, everything has been checked.

\subsection{Mellin--Barnes}

Now, we want to compute the last integral
\begin{equation}
 I(z) = \int_0^1\d\xi\xi^{D/2-\mu-\rho-1}(1-\xi)^{D/2-\nu-\rho-1}\hyper{\rho}{\mu+\nu+\rho-D/2}{D/2}{1-z}.
\end{equation}
Using \eqref{prolong} we split this equation into two integrals: 
\begin{align*}
  I(z) & = \int_0^1\d\xi\xi^{D/2-\mu-\rho-1}(1-\xi)^{D/2-\nu-\rho-1}\left[\frac{\Gamma(D/2)\Gamma(D-\mu-\nu-2\rho)}{\Gamma(D/2-\rho)\Gamma(D-\mu-\nu-\rho)}\hyper{\rho}{\mu+\nu+\rho-D/2}{2\rho+\mu+\nu-D+1}{z} \right. \\
\llcorner & \left.+ \frac{\Gamma(D/2)\Gamma(\mu+\nu+2\rho-D)}{\Gamma(\rho)\Gamma(\mu+\nu+\rho-D/2)}\hyper{D/2-\rho}{D-\mu-\nu-\rho}{D-\mu-\nu-2\rho+1}{z}\right] \\
       & = I_1(z) + I_2(z).
\end{align*}
Using the Mellin--Barnes representation \eqref{MB} and exchanging the order of the integrations we arrive to
\begin{align*}
 I_1(z) & = \frac{\Gamma(D/2)\Gamma(D-\mu-\nu-2\rho)\Gamma(2\rho+\mu+\nu-D+1)}{\Gamma(D/2-\rho)\Gamma(D-\mu-\nu-\rho)\Gamma(\rho)\Gamma(\mu+\nu+\rho-D/2)} \\
& \times \frac{1}{2i\pi}\int_{-i\infty}^{+i\infty}\d s\frac{\Gamma(-s)\Gamma(\rho+s)\Gamma(\mu+\nu+\rho-D/2+s)}{\Gamma(2\rho+\mu+\nu-D+1+s)}(-1)^s \\
\llcorner & \times \int_0^1\d\xi\xi^{D/2-\mu-\rho-1}(1-\xi)^{D/2-\nu-\rho-1}\left[\frac{(1-\xi)q^2 + \xi k^2}{\xi(1-\xi)p^2}\right]^s.
\end{align*}
The integration over $\xi$ can be computed using \eqref{int2F1}. We arrive to
\begin{align*}
 I_1(z) & = \frac{\Gamma(D/2)\Gamma(D-\mu-\nu-2\rho)\Gamma(2\rho+\mu+\nu-D+1)}{\Gamma(D/2-\rho)\Gamma(D-\mu-\nu-\rho)\Gamma(\rho)\Gamma(\mu+\nu+\rho-D/2)} \\ 
 \times \frac{1}{2i\pi}\int_{-i\infty}^{+i\infty}\d s & \left(-\frac{q^2}{p^2}\right)^s\frac{\Gamma(-s)\Gamma(\rho+s)\Gamma(\mu+\nu+\rho-D/2+s)\Gamma(D/2-\mu-\rho-s)\Gamma(D/2-\nu-\rho-s)}{\Gamma(2\rho+\mu+\nu-D+1+s)\Gamma(D-2\rho-\mu-\nu-2s)} \\
\llcorner & \times\hyper{D/2-\mu-\rho-s}{-s}{D-2\rho-\mu-\nu-2s}{1-\frac{k^2}{q^2}}.
\end{align*}
Finally, we split the remaining hypergeometric function with the formula \eqref{prolong}. Some simplifications occurs, but we are left with quite a long expression for $I_1$ which we will not 
write down.

For $I_2$, performing the same procedure we have
\begin{align*}
 I_2(z) & = \frac{\Gamma(D/2)\Gamma(D-\mu-\nu-2\rho+1)\Gamma(2\rho+\mu+\nu-D)}{\Gamma(D/2-\rho)\Gamma(D-\mu-\nu-\rho)\Gamma(\rho)\Gamma(\mu+\nu+\rho-D/2)}\left(\frac{q^2}{p^2}\right)^{D-2\rho-\mu-\nu} \\ 
 \times \frac{1}{2i\pi}\int_{-i\infty}^{+i\infty}\d s & \left(-\frac{q^2}{p^2}\right)^s\frac{\Gamma(-s)\Gamma(D/2-\rho+s)\Gamma(D-\mu-\nu-\rho+s)\Gamma(\mu+\rho-D/2-s)\Gamma(\nu+\rho-D/2-s)}{\Gamma(2\rho+\mu+\nu-D-2s)\Gamma(D-2\rho-\mu-\nu+1+s)} \\
\llcorner & \times\hyper{\nu+\rho-D/2-s}{2\rho+\mu+\nu-D-s}{2\rho+\mu+\nu-D-2s}{1-\frac{k^2}{q^2}}
\end{align*}
which, as $I_1$, we split using \eqref{prolong}.

All in all we arrive to
\begin{align}
 & J(\mu,\nu,\rho) = \frac{i\pi^{D/2}}{2i\pi}\frac{\Gamma(2\rho+\mu+\nu-D)\Gamma(D-2\rho-\mu-\nu)}{\Gamma(\mu)\Gamma(\nu)\Gamma(\rho)\Gamma(D-\rho-\mu-\nu)}(p^2)^{D/2-\mu-\nu-\rho}\Bigg((2\rho + \mu+\nu-D) \nonumber \\ 
\times & \left[\Gamma(D/2-\rho-\nu)\int_{-i\infty}^{+i\infty}\d s\left(-\frac{q^2}{p^2}\right)^s\frac{\Gamma(-s)\Gamma(\rho+s)\Gamma(\mu+\nu+\rho-D/2+s)\Gamma(D/2-\mu-\rho-s)}{\Gamma(2\rho+\mu+\nu-D+1+s)\Gamma(D-2\rho-\mu-\nu-s)}\right. \nonumber \\ 
\times & \hyper{D/2-\mu-\rho-s}{-s}{\rho+\nu-D/2+1}{\frac{k^2}{q^2}}  + \Gamma(\rho+\nu-D/2)\left(\frac{k^2}{q^2}\right)^{D/2-\rho-\nu} \nonumber \\
\times & \int_{-i\infty}^{+i\infty}\d s\left(-\frac{q^2}{p^2}\right)^s\frac{\Gamma(\rho+s)\Gamma(\mu+\nu+\rho-D/2+s)\Gamma(D/2-\nu-\rho-s)}{\Gamma(2\rho+\mu+\nu-D+1+s)} \nonumber \\ 
\times & \left.\hyper{D/2-\nu-\rho-s}{D-2\rho-\mu-\nu-s}{D/2-\rho-\nu+1}{\frac{k^2}{q^2}}\right] + (D-2\rho-\mu)\left(\frac{q^2}{p^2}\right)^{D-2\rho-\mu-\nu} \nonumber \\
\times & \left[\Gamma(D/2-\rho-\nu)\int_{-i\infty}^{+i\infty}\d s\left(-\frac{q^2}{p^2}\right)^s\frac{\Gamma(D/2-\rho+s)\Gamma(D-\mu-\nu-\rho+s)\Gamma(\nu+\rho-D/2-s)}{\Gamma(2\rho+\mu+\nu-D-2s)\Gamma(D-2\rho-\mu-\nu+1+s)} \right.\nonumber \\
\times & \hyper{\rho+\nu-D/2-s}{2\rho+\mu+\nu-D-s}{\nu+\rho-D/2+1}{\frac{k^2}{q^2}} + \Gamma(\rho+\nu-D/2)\left(\frac{k^2}{q^2}\right)^{D/2-\rho-\nu} \nonumber \\
\times & \int_{-i\infty}^{+i\infty}\d s\left(-\frac{q^2}{p^2}\right)^s\frac{\Gamma(-s)\Gamma(D/2-\rho+s)\Gamma(D-\mu-\nu-\rho+s)\Gamma(\mu+\rho-D/2-s)}{\Gamma(2\rho+\mu+\nu-D-2s)\Gamma(D-2\rho-\mu-\nu+1+s)} \nonumber \\
\times & \left.\hyper{-s}{\rho+\mu-D/2-s}{D/2-\rho-\nu+1}{\frac{k^2}{q^2}}\right]\Bigg).
\end{align}
In order to reduce the overall complexity of the computation we define $\mathcal{F}_1$, $\mathcal{F}_2, \mathcal{F}_3$ and $\mathcal{F}_4$ such that 
\begin{align}
 J(\mu,\nu,\rho) & = i\pi^{D/2}\frac{\Gamma(2\rho+\mu+\nu-D)\Gamma(D-2\rho-\mu-\nu)}{\Gamma(\mu)\Gamma(\nu)\Gamma(\rho)\Gamma(D-\rho-\mu-\nu)}(p^2)^{D/2-\mu-\nu-\rho}\Bigg((2\rho + \mu+\nu-D) \\ 
 \times & \left[\Gamma(D/2-\rho-\nu)\mathcal{F}_1 + \Gamma(\rho+\nu-D/2)\left(\frac{k^2}{q^2}\right)^{D/2-\rho-\nu}\mathcal{F}_2\right] + (D-2\rho-\mu)\left(\frac{q^2}{p^2}\right)^{D-2\rho-\mu-\nu} \\
 \times & \left[\Gamma(D/2-\rho-\nu)\mathcal{F}_3 + \Gamma(\rho+\nu-D/2)\left(\frac{k^2}{q^2}\right)^{D/2-\rho-\nu}\mathcal{F}_4\right]\Bigg). \label{def_F} \\
\end{align}

\subsection{Contour integrals and series}

We can compute the remaining integrals with the theorem of residues. Since they are Mellin--Barnes integrals, they separate the poles $s\in-\mathbb{N}$ from the poles $(a+s)\in-\mathbb{N}$ and 
$(b+s)\in-\mathbb{N}$. Hence, closing the contour integrals on the positive side, only the terms $\Gamma(-s)$ and $\Gamma(\dots-s)$ will contribute. Hence, we will have two contributions from 
$\mathcal{F}_1$ and $\mathcal{F}_4$ and one from $\mathcal{F}_2$ and $\mathcal{F}_3$. Each of these six contributions is a series and they can be computed with the relation \eqref{seriesAppell}.

From the theorem of the residues we get
\begin{align*}
 \mathcal{F}_1 & = \Bigg(\sum_{j=0}^{+\infty}\frac{1}{j!}\left(\frac{q^2}{p^2}\right)^j\frac{\Gamma(\rho+j)\Gamma(\rho+\mu+\nu-D/2+j)\Gamma(D/2-\mu-\rho-j)}{\Gamma(\mu+\nu+2\rho-D+j+1)\Gamma(D-\mu-\nu-2\rho-j)} \\
 \times & \hyper{D/2-\mu-\rho-j}{-j}{\nu+\rho-D/2+1}{\frac{k^2}{q^2}} \\
  + & \left(-\frac{q^2}{p^2}\right)^{D/2-\mu-\rho} \sum_{j=0}^{+\infty}\frac{1}{j!}\left(\frac{q^2}{p^2}\right)^j\frac{\Gamma(\mu+\rho-D/2-j)\Gamma(D/2-\mu+j)\Gamma(\nu+j)}{\Gamma(\nu+\rho-D/2+1+j)\Gamma(D/2-\nu-\rho-j)} \\
 \times & \hyper{-j}{\mu+\rho-D/2-j}{\nu+\rho-D/2+1}{\frac{k^2}{q^2}}\Bigg).
\end{align*}
Now, we will use four times the reflexion formula $\Gamma(z)\Gamma(1-z)=\pi/\sin(\pi z)$. Twice to get rid of the denominators that fordid us to us \eqref{seriesAppell} (since there is two Gamma functions) and twice 
to bring on the denominator a Gamma function that is to the numerators. Hence we arrive to
\begin{align*}
 \mathcal{F}_1 & = \Bigg(\sum_{j=0}^{+\infty}\frac{1}{j!}\left(\frac{q^2}{p^2}\right)^j\frac{\Gamma(\rho+j)\Gamma(\rho+\mu+\nu-D/2+j)}{\Gamma(1-D/2+\mu+\rho+j)}\frac{\sin(\pi(D-2\rho-\mu-\nu-j))}{\sin(\pi(D/2-\mu-\rho-j))} \\
 \times & \hyper{D/2-\mu-\rho-j}{-j}{\nu+\rho-D/2+1}{\frac{k^2}{q^2}} \\ 
 + & \left(-\frac{q^2}{p^2}\right)^{D/2-\mu-\rho} \sum_{j=0}^{+\infty}\frac{1}{j!}\left(\frac{q^2}{p^2}\right)^j\frac{\Gamma(D/2-\mu+j)\Gamma(\nu+j)}{\Gamma(1+D/2-\mu-\rho+j)}\frac{\sin(\pi(D/2-\rho-\nu-j))}{\sin(\pi(\mu+\rho-D/2-j))} \\
 \times & \hyper{-j}{\mu+\rho-D/2-j}{\nu+\rho-D/2+1}{\frac{k^2}{q^2}}\Bigg).
\end{align*}
Now, we can use $\sin(A+j\pi) = (-1)^j\sin(A)$ to take the sinuses out of the series. Finally, we easily find complete the Pochhammer symbols to be allow to use the formula \eqref{seriesAppell} 
and we get
\begin{align}
 \mathcal{F}_1 & = \frac{\sin(\pi(D-2\rho-\mu-\nu))}{\sin(\pi(D/2-\mu-\rho))}\frac{\Gamma(\rho)\Gamma(\mu+\nu+\rho-D/2)}{\Gamma(1-D/2+\mu+\rho)}\appell{\rho}{\mu+\nu+\rho-D/2}{1-D/2+\mu+\rho}{\rho+\nu-D/2+1}{\frac{q^2}{p^2}}{\frac{k^2}{p^2}} \nonumber \\
 + & \left(-\frac{q^2}{p^2}\right)^{D/2-\mu-\rho}\frac{\sin(\pi(D/2-\rho-\nu))}{\sin(\pi(\mu+\rho-D/2))}\frac{\Gamma(D/2-\mu)\Gamma(\nu)}{\Gamma(1+D/2-\mu-\rho)}\appell{D/2-\mu}{\nu}{1+D/2-\mu-\rho}{\rho+\nu-D/2+1}{\frac{q^2}{p^2}}{\frac{k^2}{p^2}}.
\end{align}
For $\mathcal{F}_2$ and $\mathcal{F}_3$ the situation is much simpler: we do not have to use the reflexion formula at all. We find 
\begin{equation}
 \mathcal{F}_2 = \left(-\frac{q^2}{p^2}\right)^{D/2-\nu-\rho}\frac{\Gamma(D/2-\nu)\Gamma(\mu)}{\Gamma(\rho+\mu-D/2+1}\appell{D/2-\nu}{\mu}{\rho+\mu-D/2+1}{D/2-\nu-\rho+1}{\frac{q^2}{p^2}}{\frac{k^2}{p^2}}
\end{equation}
and 
\begin{equation}
 \mathcal{F}_3 = \left(-\frac{q^2}{p^2}\right)^{\rho+\mu-D/2}\frac{\Gamma(\nu)\Gamma(D/2-\mu)}{\Gamma(D/2-\rho-\mu+1)}\appell{\nu}{D/2-\mu}{D/2-\rho-\mu+1}{\nu+\rho-D/2+1}{\frac{q^2}{p^2}}{\frac{k^2}{p^2}}.
\end{equation}
The computation of $\mathcal{F}_4$ is very similar to the one of $\mathcal{F}_1$, so we will not detail it. We end up with
\begin{align}
 \mathcal{F}_4 & = \frac{\sin(\pi(2\rho+\mu+\nu-D))}{\sin(\pi(\rho+\mu-D/2))}\frac{\Gamma(D/2-\rho)\Gamma(D-\rho-\mu-\nu)}{\Gamma(1-\rho-\mu+D/2)} \nonumber \\
 \times & \appell{D/2-\rho}{D-\mu-\nu-\rho}{1+D/2-\rho-\mu}{D/2-\rho-\nu+1}{\frac{q^2}{p^2}}{\frac{k^2}{p^2}} \nonumber \\
 + & \left(-\frac{q^2}{p^2}\right)^{\rho+\mu-D/2}\frac{\sin(\pi(\rho+\nu-D/2))}{\sin(\pi(D/2-\rho-\mu))}\frac{\Gamma(\mu)\Gamma(D/2-\nu)}{\Gamma(1-D/2+\rho+\mu)} \nonumber \\
 \times & \appell{\mu}{D/2-\nu}{1+\rho+\mu-D/2}{D/2-\rho-\nu+1}{\frac{q^2}{p^2}}{\frac{k^2}{p^2}}.
\end{align}

\subsection{Results}

At this stage, we have computed all the terms appearing in the integral $I(z)$, therefore we have a closed formula for $J(\mu,\nu,\rho)$. However, let us performe some simplifications. First, let 
us notice that the second term of $\mathcal{F}_1$ and $\mathcal{F}_3$ have the same arguments in their Appell functions and the same powers of the impulsions. The second term of $\mathcal{F}_4$ and $\mathcal{F}_2$ have the same 
similarities.

This suggests to look for simplifactions. So, we will carefully develop $J(\mu,\nu,\rho)$ starting form \eqref{def_F} and the above expressions for the $\mathcal{F}$s. We get rid of the sinuses 
by using the reflexion formula of the Gamma function. One term that will cause some trouble is from the second term of $\mathcal{F}_1$
\begin{equation*}
 (-1)^{D/2-\mu-\rho}\frac{\Gamma(2\rho+\mu+\nu-D+1)\Gamma(D-2\rho-\mu-\nu)\Gamma(\mu+\rho-D/2)\Gamma(D/2-\mu)\Gamma(\nu)}{\Gamma(1+\rho+\nu-D/2)},
\end{equation*}
which is similar to a term from $\mathcal{F}_3$
\begin{equation*}
 (-1)^{\rho+\nu-D/2}\frac{\Gamma(2\rho+\mu+\nu-D)\Gamma(1+D-2\rho-\mu-\nu)\Gamma(D/2-\mu-\rho)\Gamma(D/2-\mu)\Gamma(\nu)}{\Gamma(D/2-\rho-\mu+1)}.
\end{equation*}
Using again the reflexion formula for the Gamma function (and the fact that the sinus function is odd), we can rewrite the sum of these terms as
\begin{align}
 & \left[(-1)^{D/2-\mu-\rho}\frac{\sin(\pi(D/2-\rho-\nu))}{\sin(\pi(D-2\rho-\mu-\nu))} + (-1)^{D/2-\rho-\nu}\frac{\sin(\pi(D/2-\rho-\mu))}{\sin(\pi(D-2\rho-\mu-\nu))}\right] \nonumber \\
 & \times\Gamma(\rho+\mu-D/2)\Gamma(D/2-\rho-\nu)\Gamma(\nu)\Gamma(D/2-\mu). \label{problem1}
\end{align}
Similarly, we have a term from $\mathcal{F}_2$
\begin{equation*}
 (-1)^{D/2-\nu-\rho}\frac{\Gamma(2\rho+\mu+\nu-D+1)\Gamma(D-2\rho-\mu-\nu)\Gamma(D/2-\nu)\Gamma(\mu)\Gamma(\rho+\nu-D/2)}{\Gamma(\rho+\mu-D/2+1}
\end{equation*}
similar to another one from the second term of $\mathcal{F}_4$
\begin{equation*}
 (-1)^{\mu+\rho-D/2}\frac{\Gamma(2\rho+\mu+\nu-D)\Gamma(D-2\rho-\mu-\nu+1)\Gamma(D/2-\nu)\Gamma(\mu)\Gamma(D/2-\rho-\nu)}{\Gamma(1-\rho-\mu+D/2}.
\end{equation*}
With the same computation than the one performed above we write their sum as
\begin{align}
 & \left[(-1)^{D/2-\nu-\rho}\frac{\sin(\pi(D/2-\rho-\mu))}{\sin(\pi(D-2\rho-\mu-\nu))} + (-1)^{D/2-\rho+\mu}\frac{\sin(\pi(D/2-\rho-\nu))}{\sin(\pi(2\rho+\mu+\nu-D))}\right] \nonumber \\
 & \times\Gamma(\rho+\nu-D/2)\Gamma(D/2-\rho-\mu)\Gamma(\mu)\Gamma(D/2-\nu). \label{problem2}
\end{align}
Hence, all in all, we arrive to 
\begin{align}
 & J(\mu,\nu,\rho) = \frac{i\pi^{D/2}(p^2)^{D/2-\mu-\nu-rho}}{\Gamma(\mu)\Gamma(\nu)\Gamma(\rho)\Gamma(D-\rho-\mu-\nu)}\Bigg( \\
 & \Gamma(D/2-\rho-\nu)\Gamma(D/2-\mu-\rho)\Gamma(\rho)\Gamma(\mu+\nu-D/2)\appell{\rho}{\mu+\nu+\rho-D/2}{1-D/2+\mu+\rho}{\rho+\nu-D/2+1}{\frac{q^2}{p^2}}{\frac{k^2}{p^2}} \nonumber \\
 + & \left(-\frac{q^2}{p^2}\right)^{D/2-\mu-\rho}\left[(-1)^{D/2-\mu-\rho}\frac{\sin(\pi(D/2-\rho-\nu))}{\sin(\pi(D-2\rho-\mu-\nu))} + (-1)^{D/2-\rho-\nu}\frac{\sin(\pi(D/2-\rho-\mu))}{\sin(\pi(D-2\rho-\mu-\nu))}\right] \nonumber \\
 \times & \Gamma(\rho+\mu-D/2)\Gamma(D/2-\rho-\nu)\Gamma(\nu)\Gamma(D/2-\mu)\appell{D/2-\mu}{\nu}{1+D/2-\mu-\rho}{\rho+\nu-D/2+1}{\frac{q^2}{p^2}}{\frac{k^2}{p^2}} \nonumber \\
 + & \left(-\frac{k^2}{p^2}\right)^{D/2-\nu-\rho}\left[(-1)^{D/2-\nu-\rho}\frac{\sin(\pi(D/2-\rho-\mu))}{\sin(\pi(D-2\rho-\mu-\nu))} + (-1)^{D/2-\rho+\mu}\frac{\sin(\pi(D/2-\rho-\nu))}{\sin(\pi(2\rho+\mu+\nu-D))}\right] \nonumber \\
 \times & \Gamma(\rho+\nu-D/2)\Gamma(D/2-\mu-\rho)\Gamma(D/2-\nu)\Gamma(\mu)\appell{\mu}{D/2-\nu}{1+\rho+\mu-D/2}{D/2-\rho-\nu+1}{\frac{q^2}{p^2}}{\frac{k^2}{p^2}} \nonumber \\
 + & \left(\frac{q^2}{p^2}\right)^{D/2-\rho-\mu}\left(\frac{k^2}{p^2}\right)^{D/2-\rho-\nu}\appell{D/2-\rho}{D-\mu-\nu-\rho}{1+D/2-\rho-\mu}{D/2-\rho-\nu+1}{\frac{q^2}{p^2}}{\frac{k^2}{p^2}} \nonumber \\
 \times & \Gamma(\rho+\mu-D/2)\Gamma(D/2-\rho)\Gamma(D-\rho-\mu-\nu)\Gamma(\rho+\nu-D/2)\Bigg). \nonumber
\end{align}
This expression coincides with the expression of \cite{BoDa86} if the terms between brackets in \eqref{problem1} and \eqref{problem2} are just one. Let us notice that this is true if 
$(-1)^{D/2-\mu-\rho}=\sin(\pi(D/2-\mu-\rho))$ and $(-1)^{D/2-\rho-\nu}=\sin(\pi(D/2-\rho-\nu))$. If this is true, the denominators might become ill-defined, but are cancelled by the 
numerators.


\backmatter

\bibliographystyle{unsrturl}
\bibliography{thesis}

\addcontentsline{toc}{chapter}{Bibliography}

\end{document}